%% file: main.tex
\documentclass[french,12pt,a4paper]{book}
\usepackage{tipa}
\usepackage{gensymb} 
\usepackage[utf8]{inputenc}
\usepackage[LGR, T1]{fontenc}
\usepackage[main=french,british]{babel}
\usepackage[default]{opensans} 
\usepackage{scalerel}
\usepackage{amsmath}
\usepackage{amsthm}
\usepackage{amsfonts}
\usepackage{fancyhdr}
\usepackage{caption}
\usepackage{amssymb}
\usepackage{xcolor} 
\definecolor{Prune}{RGB}{99,0,60} 
\definecolor{B1}{RGB}{49,62,72} 
\definecolor{C1}{RGB}{124,135,143}
\definecolor{D1}{RGB}{213,218,223}
\definecolor{A2}{RGB}{198,11,70}
\definecolor{B2}{RGB}{237,20,91}
\definecolor{C2}{RGB}{238,52,35}
\definecolor{D2}{RGB}{243,115,32}
\definecolor{A3}{RGB}{124,42,144}
\definecolor{B3}{RGB}{125,106,175}
\definecolor{C3}{RGB}{198,103,29}
\definecolor{D3}{RGB}{254,188,24}
\definecolor{A4}{RGB}{0,78,125}
\definecolor{B4}{RGB}{14,135,201}
\definecolor{C4}{RGB}{0,148,181}
\definecolor{D4}{RGB}{70,195,210}
\definecolor{A5}{RGB}{0,128,122}
\definecolor{B5}{RGB}{64,183,105}
\definecolor{C5}{RGB}{140,198,62}
\definecolor{D5}{RGB}{213,223,61}
\usepackage{mdframed}
\usepackage{multirow} 
\usepackage{multicol} 
\usepackage{scrextend} 
\usepackage{tikz}
\usepackage{graphicx}
\usepackage[absolute]{textpos} 
\usepackage{colortbl}
\usepackage{array}
\usepackage{geometry}
\usepackage{titlesec}
\usepackage{hyperref}
\hypersetup{ 
	colorlinks=true,
	linkcolor=black,
	urlcolor=purple}

\usepackage{CJKutf8} 

\newcommand{\textgreek}[1]{\begingroup\fontencoding{LGR}\selectfont#1\endgroup}
\newcommand{\textjap}[1]{\begingroup\begin{CJK}{UTF8}{goth}#1\end{CJK}\endgroup}

\usepackage{cleveref}
\usepackage{mathtools}

\renewcommand{\paragraph}[1]{\medskip\noindent\textbf{#1}}

\definecolor{wine}{rgb}{0.8,0,0.26}
\definecolor{forest}{rgb}{0,0.5,0.16}
\definecolor{orange}{rgb}{1.0,0.66,0}
\definecolor{darkorange}{rgb}{0.5,0.33,0}
\definecolor{darkforest}{rgb}{0.1,0.39,0.1}
\definecolor{darkmagenta}{rgb}{0.6, 0, 0.6}
\definecolor{lightergray}{gray}{0.985}
\definecolor{lightgray}{gray}{0.95}

\input{tikzit.tikzdefs}
\usepackage{tikzit}
\input{tikzit.tikzstyles}
\usepackage{proof}
\usepackage{mathrsfs}
\usepackage{stmaryrd}

\usepackage{quiver}

\input{commands}

\newcommand\summaryname{Abstract}
\newenvironment{abstract}%
    {\small\begin{center}%
    \bfseries{\summaryname} \end{center}
		\begin{changemargin}{2cm}{2cm}
		}{\end{changemargin}}

\newenvironment{citing}[1]{
	\begin{changemargin}{#1}{\rightmargin}
	}{
	\end{changemargin}}

\newmdtheoremenv[linewidth=2,linecolor=A2, topline=false,rightline=false,bottomline=false,backgroundcolor=lightergray,innertopmargin = 0pt]{theorem}{Theorem}[chapter] 
\newmdtheoremenv[linewidth=1.5,linecolor=A5, topline=false,rightline=false,bottomline=false,backgroundcolor=lightergray,innertopmargin = 0pt]{lemma}[theorem]{Lemma} 
\newmdtheoremenv[linewidth=1.2,linecolor=D2, topline=false,rightline=false,bottomline=false,backgroundcolor=lightergray,innertopmargin = 0pt]{corollary}[theorem]{Corollary} 
\newmdtheoremenv[linewidth=1.2,linecolor=B3, topline=false,rightline=false,bottomline=false,backgroundcolor=lightergray,innertopmargin = 0pt]{proposition}[theorem]{Proposition} 

\theoremstyle{definition}
\newtheorem{definition}[theorem]{Definition}
\newtheorem{example}[theorem]{Example}

\theoremstyle{remark}
\newtheorem{remark}[theorem]{Remark}

\surroundwithmdframed[linewidth=1.2,linecolor=B3, leftline=false,topline=false,rightline=false,bottomline=false,backgroundcolor=lightergray,innertopmargin = 0pt]{question}

\usepackage{lipsum} 
\begin{document}
	
	\begin{titlepage}

		\newgeometry{left=5cm,bottom=2cm, top=1cm, right=1cm}
		
		\tikz[remember picture,overlay] \node[opacity=1,inner sep=0pt] at
		(-10mm,-270mm){\includegraphics{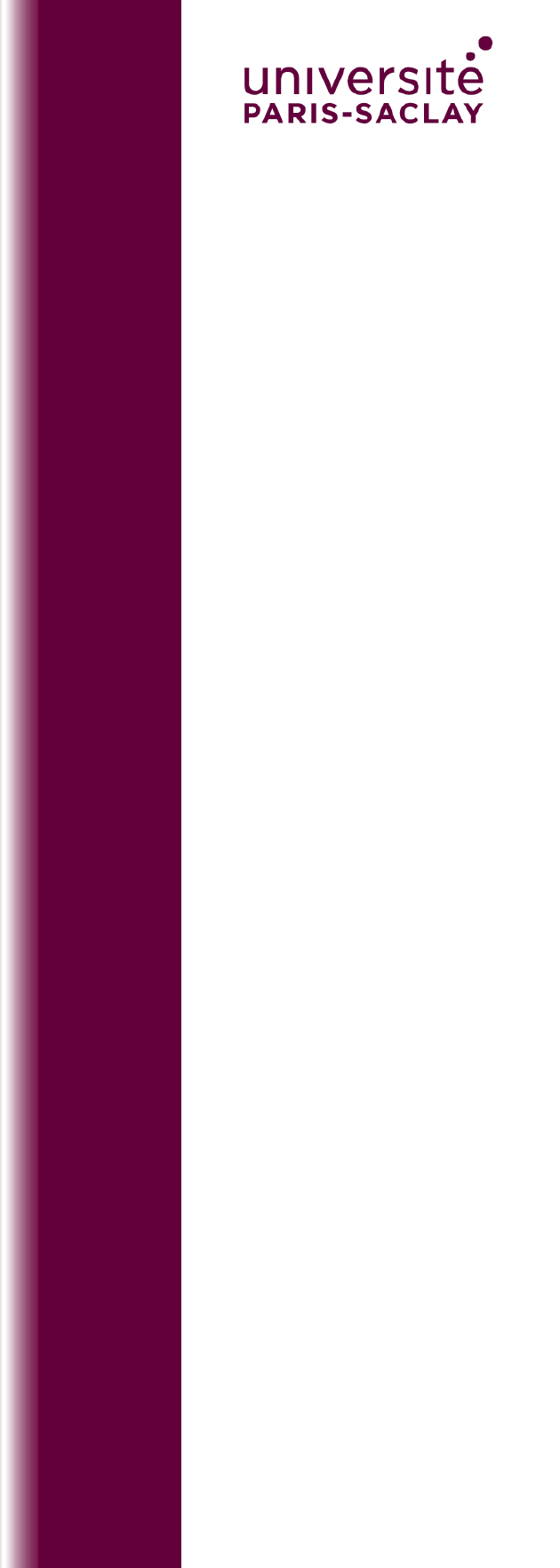}};
		
		
		\color{white}
		
		\begin{picture}(0,0)
			\fontfamily{cmss}\fontseries{m}\fontsize{22}{26}\selectfont
			\put(-137,-743){\rotatebox{90}{\Large\bf T\large{HÈSE DE} \Large{D}\large{OCTORAT}}} \\
			\put(-105,-743){\rotatebox{90}{\large\bf NNT : 2024UPASG030}}
		\end{picture}
		
		
		
		\flushright
		\vspace{10mm} 
		\color{Prune}
		\fontfamily{cmss}\fontseries{m}\fontsize{22}{26}\selectfont
		\selectlanguage{british}

		\Huge The Semantics of Effects: Centrality, Quantum Control
		and Reversible Recursion
		
		\normalsize
		\color{black}
		\selectlanguage{french}
		\Large{\textit{De la sémantique des effets: centralité, contrôle quantique et
		récursivité réversible}} \\
		
		
		\vspace{1.5cm}
		
		\normalsize
		\textbf{Thèse de doctorat de l'université Paris-Saclay} \\
		
		\vspace{6mm}
		
		\small École doctorale n$^{\circ}$ 580, Sciences et Technologies \\
		\small de l’Information et de la Communication (STIC) \\
		\small Spécialité de doctorat: informatique \\
		\small Graduate School : Informatique et Sciences du Numérique \\
		\small Référent : Faculté des sciences d'Orsay \\
		\vspace{6mm}
		
		\footnotesize Thèse préparée au \textbf{Laboratoire Méthodes Formelles}  \\
		\footnotesize (Université Paris-Saclay, CNRS, ENS Paris-Saclay, Inria), \\
		\footnotesize sous la direction de \textbf{Pablo} \textbf{A\scriptsize{RRIGHI}}, 
		professeur Université Paris-Saclay, \\
		\footnotesize et le co-encadrement de \textbf{Benoît V\scriptsize{ALIRON}}, 
		maître de  conférence CentraleSupélec, \\
		\footnotesize et de \textbf{Vladimir Z\scriptsize{AMDZHIEV}}, 
		chercheur Inria. \\
		\vspace{15mm}
		
		\textbf{Thèse soutenue à Gif-sur-Yvette, le 19 juin 2024, par}\\
		\bigskip
		\Large {\color{Prune} \textbf{Louis L\large{EMONNIER}}} 

		\vspace{\fill} 
		
		\bigskip
		
		\flushleft
		\small {\color{Prune} \textbf{Composition du jury}} \\
		{\color{Prune} \scriptsize {Membres du jury avec voix délibérative}} \\
		\vspace{2mm}
		\scriptsize
		\begin{tabular}{|p{9cm}l}
			\arrayrulecolor{Prune}
			\textbf{Jean G\tiny{OUBAULT}--L\tiny{ARRECQ}} &   Président \\ 
			Professeur, ENS Paris-Saclay & \\
			\textbf{Thomas E\tiny{HRHARD}} &  Rapporteur \& Examinateur \\ 
			Directeur de Recherche, CNRS \& Université Paris-Cité   &   \\ 
			\textbf{Laurent R\tiny{EGNIER}} &  Rapporteur \& Examinateur \\ 
			Professeur, Université de Aix-Marseille  &   \\ 
			\textbf{Claudia F\tiny{AGGIAN}} &  Examinatrice \\ 
			Chargée de Recherche, CNRS \& Université Paris-Cité   &   \\ 
			\textbf{Marie K\tiny{ERJEAN}} &  Examinatrice \\ 
			Chargée de Recherche, CNRS \& Université Sorbonne Paris-Nord   &   \\ 
			
		\end{tabular} 
		
		\flushleft
		\small {\color{Prune} \textbf{Membres invités}} \\
		\vspace{2mm}
		\scriptsize
		\begin{tabular}{|p{9cm}l}
			\arrayrulecolor{Prune}
			\textbf{Pablo A\tiny{RRIGHI}} &  Directeur de thèse \\ 
			Professeur, Université Paris-Saclay   &   \\ 
			\textbf{Benoît V\tiny{ALIRON}} &   Co-encadrant \\ 
			Maître de conférence, CentraleSupélec & \\
			\textbf{Vladimir Z\tiny{AMDZHIEV}} &  Co-encadrant \\ 
			Chercheur, Inria   &   
		\end{tabular} 
	\end{titlepage}

\ifthispageodd{\newpage\thispagestyle{empty}\null\newpage}{}
\thispagestyle{empty}
\newgeometry{top=1.5cm, bottom=3.25cm, left=2cm, right=2cm}
\fontfamily{rm}\selectfont

\lhead{}
\rhead{}
\rfoot{}
\cfoot{}
\lfoot{}

\noindent 
\includegraphics[height=2.45cm]{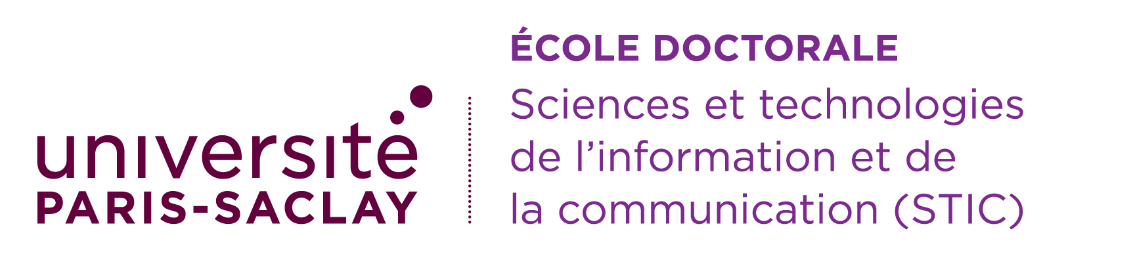}
\vspace{1cm}
\fontfamily{cmss}\fontseries{m}\selectfont

\small

\begin{mdframed}[linecolor=Prune,linewidth=1]

	\textbf{Titre :} De la sémantique des effets: centralité, contrôle quantique
	et récursivité réversible 

	\noindent \textbf{Mots clés :} 
	Informatique quantique
	-- Sémantique
	-- Languages de programmation
	-- Théorie des catégories

	\vspace{-.5cm}
	\begin{multicols}{2}
		\noindent \textbf{Résumé :} 
		Le sujet de cette thèse est axé sur la théorie des langages de
		programmation. Dans un langage de programmation suffisamment bien défini,
		le comportement des programmes peut être étudié à l'aide d'outils empruntés
		à la logique et aux mathématiques, énonçant des résultats sans exécuter le
		code. Ce domaine de l'informatique est appelé \emph{sémantique}. La
		sémantique d'un langage peut se présenter sous plusieurs formes : dans
		notre cas, des sémantiques opérationnelles, des théories équationnelles et
		des sémantiques dénotationnelles. Les premières donnent un sens
		opérationnel aux programmes, au sein de la syntaxe du langage. Elles
		simulent les opérations qu'un ordinateur est censé effectuer s'il exécute
		le programme. Une théorie équationnelle fonctionne également de manière
		syntaxique : elle indique si deux programmes effectuent la même opération
		sans informer sur la procédure.  Enfin, la sémantique dénotationnelle est
		l'étude mathématique des programmes, généralement à l'aide de la théorie
		des catégories. Elle permet par exemple de prouver qu'un programme se
		termine ou non.

		Cette thèse se concentre sur la sémantique des effets dans les langages de
		programmation -- une fonctionnalité ajoutée à un langage, gérant des
		données secondaires ou des résultats probabilistes. Eugenio Moggi, en 1991,
		a publié un travail fondateur sur l'étude de la sémantique des effets,
		soulignant la relation avec les monades en théorie des catégories. La
		première contribution de cette thèse suit directement le travail de Moggi,
		en étudiant la commutativité des effets dans un langage de programmation à
		travers le prisme des monades. Les monades sont la généralisation de
		structures algébriques telles que les monoïdes, qui ont une notion de
		centre : le centre d'un monoïde est une collection d'éléments qui commutent
		avec tous les autres dans le monoïde. Nous fournissons les conditions
		nécessaires et suffisantes pour qu'une monade ait un centre. Nous
		détaillons également la sémantique d'un langage de programmation avec des
		effets qui portent des informations sur les effets qui sont centraux. De
		plus, nous fournissons un lien fort -- un résultat de langage interne --
		entre ses théories équationnelles et sa sémantique dénotationnelle.

		Le deuxième axe de la thèse est l'informatique quantique, perçue comme un
		effet réversible. Le quantique est un domaine émergent de l'informatique
		qui utilise la puissance de la mécanique quantique pour calculer. Au niveau
		des langages de programmation, de nouveaux paradigmes doivent être
		développés pour être fidèles aux opérations quantiques. Les opérations
		quantiques physiquement permises sont toutes réversibles, à l'exception de
		la mesure ; cependant, la mesure peut être reportée à la fin du calcul, ce
		qui nous permet de nous concentrer d'abord sur la partie réversible et
		d'appliquer ensuite la mesure pour obtenir des résultats. Dans le chapitre
		correspondant, nous définissons un langage de programmation réversible,
		avec types simples, qui effectue des opérations quantiques
		\emph{unitaires}. Une sémantique dénotationnelle et une théorie
		équationnelle adaptées au langage sont présentées, et nous prouvons que
		cette dernière est complète. Ce travail vise à fournir des bases solides
		pour l'étude du contrôle quantique d'ordre supérieur.

		En outre, nous étudions la récursion réversible, en fournissant une
		sémantique opérationnelle et dénotationnelle adéquate à un langage de
		programmation fonctionnel, réversible et Turing-complet. La sémantique
		dénotationnelle utilise l'enrichissement dcpo des catégories inverses. Ce
		modèle mathématique sur l'informatique réversible ne se généralise pas
		directement à sa version quantique. Dans la conclusion, nous détaillons les
		limites et l'avenir possible du contrôle quantique d'ordre supérieur.
	\end{multicols}

\end{mdframed}

\vspace{8mm}

\selectlanguage{british}

\begin{mdframed}[linecolor=Prune,linewidth=1]

	\textbf{Title:} The Semantics of Effects: Centrality, Quantum Control and
	Reversible Recursion


	\begin{multicols}{2}
		\noindent \textbf{Abstract:} 
		The topic of this thesis revolves around the theory of programming
		languages. In a sufficiently well-defined programming language, the
		behaviour of programs can be studied with tools borrowed from logic and
		mathematics, allowing us to state results without executing the code. This
		area of computer science is called ``semantics''. The semantics of a
		programming language can take several forms: in this thesis, we work with
		operational semantics, equational theories, and denotational semantics. The
		former gives an operational meaning to programs but within the language's
		syntax. It simulates the operations a computer is supposed to perform if it
		were running the program. An equational theory also works syntactically: it
		indicates whether two programs perform the same operation without giving
		any information on the procedure. Lastly, denotational semantics is the
		mathematical study of programs, usually done with the help of category
		theory. For example, it allows us to prove whether a program terminates.

		This thesis focuses on the semantics of effects in programming languages --
		namely, a feature added to a language, \emph{e.g.} handling side data or
		probabilistic outputs. Eugenio Moggi, in 1991, published foundational work
		on the study of the semantics of effects, highlighting the relationship
		with monads in category theory. The first contribution of this thesis
		directly follows Moggi's work, studying the commutativity of effects in a
		programming language through the prism of monads. Monads are the
		generalisation of algebraic structures such as monoids, which have a notion
		of centre: the centre of a monoid is a collection of elements which commute
		with all others in the monoid. We provide the necessary and sufficient
		conditions for a monad to have a centre. We also detail the semantics of a
		programming language with effects that carry information on which effects
		are central. Moreover, we provide a strong link -- an internal language
		result -- between its equational theories and its denotational semantics.

		The second focus of the thesis is quantum computing, which is seen as a
		reversible effect. Quantum computing is an emergent field in computer
		science that uses the power of quantum mechanics to compute. At the level
		of programming languages, new paradigms need to be developed to be faithful
		to quantum operations. Physically permissible quantum operations are all
		reversible, except measurement; however, measurement can be deferred at the
		end of the computation, allowing us to focus on the reversible part first
		and then apply measurement to obtain results. In the corresponding chapter,
		we define a simply-typed reversible programming language performing quantum
		operations called ``unitaries''. A denotational semantics and an equational
		theory adapted to the language are presented, and we prove that the latter
		is complete. The aim of this work is to provide a solid foundation for the
		study of higher-order quantum control. 

		Furthermore, we study recursion in reversible programming, providing
		adequate operational and denotational semantics to a Turing-complete,
		reversible, functional programming language. The denotational semantics
		uses the dcpo enrichment of rig join inverse categories. This mathematical
		account of higher-order reasoning on reversible computing does not directly
		generalise to its quantum counterpart. In the conclusion, we detail the
		limitations and possible future for higher-order quantum control.
	\end{multicols}
\end{mdframed}

\selectlanguage{french}

\titleformat{\chapter}[display]{\bfseries\huge\color{Prune}}{Chapter\ \thechapter}{.1ex}
{\vspace{0.1ex}
}
[\vspace{3ex}]

%

\newgeometry{top=3cm, bottom=4cm, left=3cm, right=3cm}


\input{ack}

\newgeometry{top=3cm, bottom=4cm, left=2cm, right=2cm}

\tableofcontents

\newgeometry{top=3cm, bottom=4cm, left=3cm, right=3cm}

\chapter*{Résumé en français}
\label{ch:resume}
\addcontentsline{toc}{chapter}{\nameref{ch:resume}}
\begin{citing}{4cm}
	``Lae mathématicien·ne est engagé·e dans la poursuite d'un rêve sans fin,
	comprendre la structure de toute chose.'' --- d'après Charles Ehresmann,
	mathématicien et membre fondateur du groupe Bourbaki.
\end{citing}

\selectlanguage{french}

\input{french-summary}

\selectlanguage{british}

%

\chapter*{Introduction}
\label{ch:intro}
\addcontentsline{toc}{chapter}{\nameref{ch:intro}}
\begin{citing}{5cm}
	``Mathematical Science shows what is. It is the language of the unseen
	relations between things. But to use \& apply that language we must be able
	fully to appreciate, to feel, to seize, the unseen, the unconscious.'' ---
	according to Ada Lovelace.
\end{citing}

\input{introduction}

\chapter{Mathematical Background}
\label{ch:background}
\begin{citing}{5.5cm}
	``You cannot outsmart the model.'' --- Vladimir Zamdzhiev.
\end{citing}

\input{background}

\chapter{Monads and Commutativity}
\label{ch:monads}
\begin{citing}{8cm}
	``Everyone likes monads.'' --- Nima Motamed.
\end{citing}

\input{monads}

\chapter{Simply-typed Quantum Control}
\label{ch:qu-control}
\begin{citing}{5cm}
	``Okay, this is a quantum computer, right? We barely know how it works, it's
	basically magic.'' --- Black Mirror, S06E01.
\end{citing}

\input{quantum-simple}

\chapter{Reversibility and Fixed Points}
\label{ch:reversible}
\begin{citing}{3cm}
	``Should mathematical semantics still be conceived as following in the track
	of pre-existing languages, trying to explain their novel features, and to
	provide ﬁrm foundations for them? Or should it be seen as operating in a more
	autonomous fashion, developing new semantic paradigms, which may then give
	rise to new languages?'' --- Samson Abramsky, in \cite{abramsky2020whither}.
\end{citing}

\input{rev-recursion}

\chapter{Notes on Quantum Recursion}
\label{ch:qu-recursion}
\begin{citing}{7cm}
	``The category of Hilbert spaces is self-dual, has two monoidal structures,
	and its homsets are algebraic domains, but its enrichment and limit behaviour
	is wanting.'' --- Chris Heunen, in \cite{heunen2013l2}.
\end{citing}

\input{quantum-recursion}

\input{guarded-recursion}

\chapter*{Conclusion}
\label{ch:conclusion}
\addcontentsline{toc}{chapter}{\nameref{ch:conclusion}}
\begin{citing}{7cm}
	``What you say does not matter, it is what people remember that matters.''
	--- Benoît Valiron.
\end{citing}

\input{conclusion}

\newpage
\bibliographystyle{alpha}
\bibliography{ref}

\end{document}

%% file: tikzit.tikzstyles
\tikzstyle{empty}=[shape=circle, tikzit fill={rgb,255: red,191; green,191; blue,191}]
\tikzstyle{dot}=[fill=black, draw=black, shape=circle]

\tikzstyle{to}=[draw=black, ->]
\tikzstyle{equal-arrow}=[{-}, double equal sign distance]
\tikzstyle{hook}=[right hook->, draw=black, tikzit draw=magenta]
\tikzstyle{dashed-arrow}=[dashed, ->]
\tikzstyle{mono}=[draw=black, to reversed->]
\tikzstyle{white dot}=[dot,fill=white, inner sep=0.4mm,minimum width=0pt,minimum height=0pt, font=\footnotesize]

\tikzset{tensor/.style={inner sep=2.5pt, draw, circle, path picture={
  \draw[black]
(path picture bounding box.south east) -- (path picture bounding box.north west) (path picture bounding box.south west) -- (path picture bounding box.north east);
}}}
\tikzset{plus/.style={inner sep=2.5pt, draw, circle, path picture={
  \draw[black]
(path picture bounding box.east) -- (path picture bounding box.west) (path picture bounding box.south) -- (path picture bounding box.north);
}}}

%% file: commands.tex
\makeatletter
\newcommand\Label[1]{\refstepcounter{equation}(\theequation)\ltx@label{#1}}
\makeatother

\usepackage{color}
\def\bR{\begin{color}{red}}
\def\bB{\begin{color}{blue}}
\def\bM{\begin{color}{magenta}}
\def\bC{\begin{color}{cyan}}
\def\bW{\begin{color}{white}}
\def\bBl{\begin{color}{black}}
\def\bG{\begin{color}{green}}
\def\bY{\begin{color}{yellow}}
\def\e{\end{color}\xspace}
\newcommand{\bit}{\begin{itemize}}
\newcommand{\eit}{\end{itemize}\par\noindent}
\newcommand{\ben}{\begin{enumerate}}
\newcommand{\een}{\end{enumerate}\par\noindent}
\newcommand{\beq}{\begin{equation}}
\newcommand{\eeq}{\end{equation}\par\noindent}
\newcommand{\beqa}{\begin{eqnarray*}}
\newcommand{\eeqa}{\end{eqnarray*}\par\noindent}
\newcommand{\beqn}{\begin{eqnarray}}
\newcommand{\eeqn}{\end{eqnarray}\par\noindent}

\newcommand{\opn}[1]{\ensuremath{\operatorname{#1}}}

\newcommand{\typ}{\mathtt{type}}

\newcommand{\alt}{~\mid~}

\newcommand{\inl}[1]{\ensuremath{\mathtt{inj}_l{\;#1}}}
\newcommand{\inr}[1]{\ensuremath{\mathtt{inj}_r{\;#1}}}
\newcommand{\ini}[1]{\mathtt{inj}_i{\;#1}}
\newcommand{\pv}[2]{\ensuremath{\langle #1,#2 \rangle}}
\newcommand{\pair}[2]{\ensuremath{#1 \otimes #2}}

\newcommand{\tact}[1][c]{\mathtt{act}_\TT(#1)}
\newcommand{\zact}[1][c]{\mathtt{act}_\SC(#1)}

\newcommand{\xact}[1][c]{\mathtt{act}_\XX(#1)}

\newcommand{\xret}[1]{\mathtt{ret}_\XX~#1}
\newcommand{\sret}[1]{\mathtt{ret}_\SC~#1}
\newcommand{\zret}[1]{\mathtt{ret}_\SC~#1}
\newcommand{\tret}[1]{\mathtt{ret}_\TT~#1}
\newcommand{\rret}[1]{\mathtt{ret}~#1}
\newcommand{\xdo}[3]{\mathtt{do}_\XX~#1\leftarrow #2;~#3}

\newcommand{\zdo}[3]{\mathtt{do}_\SC~#1\leftarrow #2;~#3}
\newcommand{\tdo}[3]{\mathtt{do}_\TT~#1\leftarrow #2;~#3}
\newcommand{\ddo}[3]{\mathtt{do}~#1\leftarrow #2;~#3}

\newcommand{\one}{\mathtt I}

\newcommand{\fold}[1]{\ensuremath{\mathtt{fold}{\;#1}}}

\newcommand{\lett}{\ensuremath{\mathtt{let}}}
\newcommand{\letv}[3]{{\mathtt{let}}\,{#1}={#2}~{\mathtt{in}}~{#3}}

\newcommand{\clauses}[1]{\left\{ ~#1~ \right\}}
\newcommand{\clause}[2]{\mid ~ #1 \isot #2}

\newcommand{\tc}{\mathtt{t}\!\mathtt{t}}
\newcommand{\fc}{\mathtt{f}\!\mathtt{f}}
\newcommand{\OD}[2]{{\rm OD}_{#1}{#2}}
\newcommand{\ODe}[2]{{\rm OD}^{\it ext}_{#1}{#2}}

\newcommand{\ttt}{\ensuremath{\mathbf{tt}}}
\newcommand{\fff}{\ensuremath{\mathbf{ff}}}
\newcommand{\dup}{\ensuremath{\opn{Dup}}}
\newcommand{\garRem}[2]{\ensuremath{\opn{GarbRem}(#1, #2)}}
\newcommand{\nat}{\ensuremath{\mathtt{Nat}}}
\newcommand{\zero}{\ensuremath{\mathtt{zero}}}
\newcommand{\suc}[1]{\ensuremath{\mathtt{S}{\;#1}}}
\newcommand{\ctrl}[1]{\ensuremath{\mathtt{ctrl}{\;#1}}}

\newcommand{\match}[3]{\mathrm{match}(#1,#2,#3)}
\newcommand{\natS}{\ensuremath{\mathbb{N}}}
\newcommand{\natT}{\ensuremath{\mathtt{nat}}}

\newcommand{\boolT}{\ensuremath{\mathbb{B}}}
\newcommand{\1}{I}
\newcommand{\ffix}{\ensuremath{\mathtt{fix}~}}
\newcommand{\subfin}{\triangleleft}
\DeclarePairedDelimiter\floor{\lfloor}{\rfloor}

\newcommand{\finto}{\underset{\mathrm{fin}}{\to}}
\newcommand{\isoto}{\underset{\mathrm{iso}}{\to}}
\newcommand{\termto}{\underset{\mathrm{term}}{\to}}
\newcommand{\nfix}[1]{\ensuremath{\mathtt{fix}^{#1}~}}

\newcommand{\entailiso}{\vdash_{\isoterm}}

\newcommand{\entail}{\vdash}

\newcommand{\entaile}{\vdash_e}

\newcommand{\isot}{\ensuremath{\leftrightarrow}}
\newcommand{\iso}{\ensuremath{\leftrightarrow}}
\newcommand{\isovar}{\ensuremath{\phi}}
\newcommand{\isolambdavar}{\ensuremath{\psi}}
\newcommand{\isobasique}{\clauses{\clause{v_1}{e_1}\mid\dots\clause{v_n}{e_n}}}
\newcommand{\isobreduit}{\clauses{\clause{v_i}{e_i}}_{i \in I}}
\newcommand{\isobreduitj}{\clauses{\clause{v_j}{e_j}}_{j \in J}}

\newcommand{\isoterm}{\ensuremath{\omega}}

\newcommand{\unibasique}{\clauses{\clause{b_1}{e_1}\mid\dots\clause{b_n}{e_n}}}
\newcommand{\unibreduit}{\clauses{\clause{b_i}{e_i}}_{i \in I}}

\newcommand{\nnext}[1]{\mathtt{next}{\;#1}}
\newcommand{\later}{\ensuremath{\blacktriangleright}}

\newcommand\embed\hookrightarrow

\newcommand{\naturalto}{\ensuremath{\Rightarrow}}
\newcommand\natto\Rightarrow
\newcommand{\Id}{\ensuremath{\mathrm{id}}}
\newcommand{\dcpob}{\mathbf{DCPO}_\bot}

\newcommand{\dcpo}{\DCPO}
\newcommand{\Hilbc}{\mathbf{Hilb}^{\aleph_o}}
\newcommand{\Cat}{\mathbf{Cat}}
\newcommand{\abs}[1]{\left\vert #1 \right\vert}
\newcommand{\norm}[1]{\left\Vert #1 \right\Vert}
\newcommand{\dg}{^{\dagger}}
\newcommand{\inv}{^{-1}}
\newcommand{\op}{^{\mathrm{op}}}

\newcommand{\fix}{\mathrm{fix}}
\newcommand{\Ker}{\mathrm{Ker}}
\newcommand{\iim}{\mathrm{Im}}
\newcommand{\iid}{\mathrm{id}}
\newcommand{\comp}{\mathrm{comp}}
\newcommand{\ip}[2]{\langle #1\mid #2 \rangle}
\newcommand{\nnoma}{\!^\noma}

\newcommand{\set}[1]{\ensuremath{\{#1\}}}
\newcommand{\sset}[1]{\ensuremath{\left\{#1\right\}}}
\newcommand{\mynl}{\\[1ex]}

\newcommand{\intf}[1]{\left\llbracket #1 \right\rrbracket}
\newcommand{\interp}[1]{\ensuremath{\intf{#1}}}
\newcommand{\den}[1]{\left\llbracket #1 \right\rrbracket}
\newcommand\sem\den

\newcommand{\rc}[1]{\res{#1^{\circ}}}

\newcommand{\rmctrl}{\ensuremath{\mathrm{ctrl}}}
\newcommand{\rmsucc}{\ensuremath{\mathrm{succ}}}
\newcommand{\rmeval}{\ensuremath{\mathrm{eval}}}
\newcommand{\rmcurry}{\ensuremath{\mathrm{curry}}}

\newcommand{\pinv}{^{\circ}}
\newcommand{\res}[1]{\overline{#1}}

\newcommand{\ov}[1]{\overrightarrow{#1}}

\newcommand{\bra}[1]{\ensuremath{\left\langle #1 \right|}}
\newcommand{\ket}[1]{\ensuremath{\left|  #1 \right\rangle}}
\newcommand{\braket}[2]{\ensuremath{\langle#1|#2\rangle}}
\newcommand{\ketbra}[2]{\ensuremath{\ket{#1}\!\bra{#2}}}

\newcommand{\yo}{\text{\usefont{U}{min}{m}{n}\symbol{'210}}}

\newcommand{\noma}{\text{\usefont{U}{min}{m}{n}\symbol{'005}}\!}

\DeclareFontFamily{U}{min}{}
\DeclareFontShape{U}{min}{m}{n}{<-> udmj30}{}

\newcommand{\secref}[1]{\S \ref{#1}}

%% file: ack.tex
\chapter*{Remerciements/Acknowledgements}
\begin{citing}{4cm}
  ``\'Ecrire une thèse, c'est un peu comme courir un marathon, sauf que tu
  prends du poids au lieu d'en perdre.'' --- LL.
\end{citing}

\vspace{1cm}

\selectlanguage{french}

\paragraph{La langue.}
Quand j'ai commencé à réfléchir à ces remerciements, j'ai rapidement décidé de
me laisser la liberté de passer d'une langue à une autre. C'est bien le seul
endroit où je peux me le permettre, et c'est important que j'y intègre du
français, tout comme de l'anglais, car des futur·es lecteurices de ces
paragraphes ne parlent pas forcément les deux langues.

\paragraph{Avant la thèse.}
Comme nous allons parler de sciences, partons du début. Non pas le début de ma
scolarité, mais disons le moment où j'ai commencé à envisager de faire de la
recherche mon avenir. Je ne suis pas certain que j'aurais suivi cette voie si
mes professeurs de mathématiques de Terminale, Vincent Hugerot et Vincent
Boillon, ne m'avaient pas transmis leur passion. Puis, en sup et en spé, c'est
surtout l'informatique qui attiré mon attention. Cela ne serait jamais arrivé
sans Serge Aubert, le meilleur enseignant que j'ai pu connaître. Il ne cachait
pas sa passion profonde pour l'informatique, qu'elle soit théorique ou pratique
; il transmettait son savoir à la fois avec douceur et fermeté, ce caractère
faisant de lui mon modèle -- et il l'est peut-être toujours.  Je ne peux pas
finir ce paragraphe sans mentionner Pascal Guelfi, qui a profondément bousculé
mon savoir et ma culture mathématiques, et sans qui je n'aurais jamais pu
intégrer une ENS. J'ajouterai que je n'aurais pas pu atteindre la rédaction de
cette thèse sans la formation du département informatique de l'ENS Cachan
(maintenant Paris-Saclay) et ses excellents professeurs.
Observons, à regret, que toutes les personnes mentionnées ci-dessus sont des
hommes. Je n'en ai pas souffert, mais si j'avais été une femme, j'aurais pu
être impactée par ce manque de manque de représentation.

\selectlanguage{british}

\paragraph{Starting into research.}
This thesis has been undirectly helped by how I was introduced in the world of
research and academia, which was first done by Frédéric Dupuis, thanks to Simon
Perdrix. Frédéric m'a permis de découvrir la recherche dans un environnement
sain et encadré, je l'en remercie tout particulièrement. My other research
experience before my PhD journey has happened with Aleks Kissinger and John van
de Wetering. I really enjoyed those months playing with ZH-diagrams and !-boxes
during my stay in Nijmegen \textipa{["nEi""me:G@]}. 

\newpage
\paragraph{My environment.}
During my PhD, I have met and spent time with a crazy amout of nice people.
Merci aux (ex-)membres de l'ADEPS\footnote{Association des Doctorant·es de
l'ENS Paris-Saclay.}, en particulier Céline, Émilie, Giann Karlo, Sofiane.
Merci pour ces soirées, ces conversations qui m'ont permis à la fois de
décrocher du travail et de grandir en tant qu'humain.  Petit mot pour mes
étudiant·es à l'ENS avec qui j'ai passé des moments sympathiques même si ça ne
l'était pas forcément pour elleux.  Je tiens à remercier les personnes qui
dirigent et font tourner le LMF\footnote{Laboratoire d'informatique dans lequel
j'ai réalisé ma thèse.}, ça a été un grand plaisir d'intéragir avec vous.
Merci à Nicolas, pour nos conversations, notre soutien mutuel, nos tentatives
de collaboration, mais qu'on a eu du mal à entretenir par manque de temps.
Merci à Kostia, pour m'avoir accompagné sur toute la première moitié de ma
thèse, pour ces petits moments de non travail au labo ou en conf, tes
régulières victoires à Towerfall et non régulières à Mario Kart, et enfin pour
m'avoir fait découvrir le GT Scalp et plus généralement la communauté
$\lambda$-calcul et logique linéaire. Peut-être un jour sauras-tu que les
normalien·nes peuvent être des gens sympathiques.  Merci à Titouan, qui a été
mon premier collaborateur en dehors de mes encadrants, c'est toujours un
plaisir de travailler avec toi.  Merci à James pour son soutien et pour m'avoir
préparé psychologiquement à déménager au Royaume-Uni. I hope I will keep
hearing from you about \emph{pro}strong \emph{pro}functors.  Merci aux
camarades de conférences pour leur soutien, en particulier Lison, Aloÿs, Tito
et Axel.  I also thank Agustín (ciclismo y asados), Dongho
($\scalerel*{\includegraphics[height=11pt]{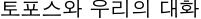}}{B_n}$),
Evi 
(\scalerel*{\includegraphics[height=11pt]{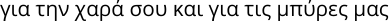}}{By}),
Jérôme (nos conversations du jeudi matin), Kinnari
($\scalerel*{\includegraphics[height=11pt]{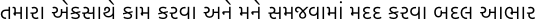}}{{\overline A}}$),
and the members of the QuaCS team
and of the LMF. My stay there has been overall nice and enjoyable.

\paragraph{My supervisors.}
Of course, I would not be writing any of this if Benoît and Vladimir were not
there to help me figure how to write a thesis. 

\selectlanguage{french}

Les mots ne sont pas suffisants pour remercier Benoît de m'avoir pris sous son
aile, de m'avoir fait vivre la recherche comme il l'entendait, avec sa joie,
son humour et sa bienveillance, et surtout de m'avoir laissé une liberté
presque totale, afin d'explorer et de collaborer avec un nombre de personnes
grandissant, tout au long de ma thèse. Benoît est également un expert pour
corriger et améliorer les présentations orales. Tout ce qu'il y a de bon dans
mes interventions est probablement grâce à lui.

\selectlanguage{british}

Before starting at QuaCS, I was supposed to do my PhD thesis with Vladimir in
Nancy. After many \emph{péripéties}\footnote{Roughly, ``twists'' in French.}, I
had funding to work in Saclay and not in Nancy. I was surprised when, a year or
so later, Vladimir announced that he was joining us, and that he could
supervise me. His supervision style is more serious and stricter, which was not
a bad thing, all things considered. \scalerel*{\includegraphics[height=15pt]{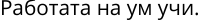}}{By}
 I will probably miss the usual jokes and stories at the coffee break or
at the bar (although I think I know most of them now).

To young people reading this: being successful with a PhD thesis is not only
about having strong scientific skills; it also requires to work with the right
people \emph{for you}. Knowing the people you are going to spend 3+ years with
is very important, and I would even say that this is more important than the
PhD subject.

\paragraph{Collaborating.}
Benoît, Kostia and I have been working with Robin Kaarsgaard, quite
unsuccessfully for now, because the subject I wanted to tackle with them
happened to be harder than what I thought. Thanks to Robin for his patience.
Kommer tid, kommer råd.

\paragraph{Le jury.}
Merci à Laurent Regnier et Thomas Ehrhard qui ont accepté de relire ce
manuscrit en détail, et merci à Marie Kerjean, Claudia Faggian et Jean
Goubault-Larrecq de participer au jury.

\paragraph{Relecture.}
Ce manuscrit est passé sous le peigne fin de plusieurs personnes, à qui je dois
beaucoup. Un gigantesque merci à James, Kostia et Manon, ainsi que mes
encadrants Benoît et Vladimir. À chaque relecture, j'avais des choses à
améliorer ; chaque personne remarquait des typos ou des erreurs différentes,
c'était un processus intéressant, et ça aurait été presque impossible pour moi
de le faire seul.

\paragraph{Mes proches.}
Je tiens à remercier toustes mes ami·es  -- avec une attention particulière
pour Elric, Marie, Adrien, Leela, Lucas, Lucas, Marion, Marion, Jean-Xavier,
Victor, Alexandre, Alexandre, Rosemonde, Cécile, Gaël, Gaspard, Clara, Yoan,
Yan, Valentin, Baptiste, Simon, Loïc, Quentin, Charlie, Arthur, Maxime, Maxime,
Manon, le mot \og ami·e \fg{} n'étant pas toujours suffisant -- qui m'ont
soutenu et accompagné ces dernières années. Il y a presque deux ans, j'ai
commencé à jouer de la basse d'abord grâce à Gaël qui m'a prêté un instrument,
puis en groupe avec Adrien, Hugo, Lucas et Marie. En plus d'avoir eu le plaisir
de découvrir une nouvelle passion, ça m'a permis de me changer les idées quand
j'en avais besoin.

Ma famille a toujours été un pilier inébranlable de ma vie, je leur dois une
grande partie de ce que je suis. Cette thèse est dédiée à ma grand-mère,
Marie-Thérèse, et mon grand-père, Jean-François, qui ne sont plus là pour lire
ces remerciements mais qui m'accompagnent dans tout ce que je fais.

%% file: french-summary.tex
\vspace{2cm}

Le principal objectif de cette thèse réside dans l'exploration des structures
fondamentales de la programmation, avec un fort aspect théorique. En
particulier, sont utilisées des structures provenant des mathématiques et de la
logique pour prouver des propriétés sur les programmes. Ce domaine de
l'informatique théorique est appelé \emph{méthodes formelles}.  Dans cette
thèse, nous nous concentrons sur les langages de programmation \emph{formels}.
Ces derniers sont développés et étudiés pour affirmer et extraire des
propriétés formelles dans des domaines spécifiques de la conception des
langages de programmation.  Le langage de programmation formel le plus standard
est le $\lambda$-calcul (prononcez \og \emph{lambda} calcul \fg{}), introduit
par Alonzo Church \cite{church1932logic} dans le cadre de son programme de
recherche sur les fondements des mathématiques. La présentation du
$\lambda$-calcul est en apparence simple : dans sa syntaxe, on peut former des
fonctions et appliquer ces fonctions. Cependant, ce langage est \emph{Turing
complet}, ce qui signifie que tout programme \emph{calculable} peut être
représenté dans le $\lambda$-calcul. Il est particulièrement pratique de
travailler en $\lambda$-calcul et il est simple d'y ajouter presque n'importe
quel type de fonctionnalité, en ajoutant par exemple des combinateurs. Il est
également flexible, en présentant plusieurs stratégies de calcul : en effet, on
peut par exemple choisir de calculer d'abord le contenu des fonctions ou
l'argument pris par ces fonctions.

Dans cette thèse, l'objectif est d'étudier deux aspects de la conception des
langages de programmation, à savoir le flot de contrôle et les effets. Le
\emph{flot de contrôle} délimite les prises de décision sous-jacentes à
l'exécution des tâches dans un paradigme de programmation. Il existe
différentes manières de contrôler le flot d'un langage de programmation.  Les
langages impératifs, tels que \texttt{Python}, sont contrôlés par des
instructions comme  \og if \fg{} et \og while \fg{}, tandis que les langages
fonctionnels comme \texttt{Caml} sont plus subtils, avec un contrôle réalisé
par des appels de fonctions ; ce dernier s'inspire du $\lambda$-calcul et de
son interprétation mathématique. L'autre aspect qui nous intéresse est celui
des \emph{effets} : un calcul \emph{à effet} se distingue fondamentalement aux
calculs \emph{purs} ; en d'autres termes, traiter des effets en programmation
signifie que l'on fait la distinction entre les opérations de base d'un langage
et ses interactions avec le monde extérieur. Ces effets se présentent sous
différentes formes. Il y a des effets généraux, parfois appelés \emph{effets de
bord}, qui interagissent directement avec un agent extérieur au programme --
par exemple, l'écriture sur une bande séparée, la gestion des entrées et
sorties de données. Il est également possible de rencontrer des \emph{effets
algébriques}, tels que l'introduction de non-déterminisme ou de comportements
probabilistes. Le calcul quantique peut également être vu comme un effet
algébrique. Ces effets sont \emph{algébriques} dans le sens où ils présentent
des caractéristiques issus de l'algèbre en mathématique. Cela signifie, en
particulier, que les questions qui s'appliquent aux structures algébriques
peuvent également s'appliquer aux effets -- par exemple, si deux effets
commutent entre eux ou non.

Dans cette thèse, nous proposons une étude formelle de ces aspects à travers le
prisme de la sémantique, un paradigme en informatique qui attribue des
interprétations logiques et mathématiques aux programmes. Une étude sémantique
d'un langage de programmation nous permet de déduire des propriétés sur les
programmes -- par exemple, s'ils terminent ou non. Les composantes
mathématiques d'une étude sémantique sont, dans notre cas, réalisées avec
l'aide de la théorie des catégories.

\section*{Sémantique}
\label{sec:fr-intro-sem}
\begin{citing}{7cm}
	Du grec ancien \textgreek{σημᾰντῐκός}, \emph{qui donne du sens}.
\end{citing}

La notion de sémantique en théorie de l'informatique a débuté avec Robert Floyd
\cite{floyd1967meanings}, dans une tentative de formaliser ce qui est attendu
d'un langage de programmation d'un point de vue logique. La sémantique d'un
langage de programmation peut revêtir différentes formes. Dans cette thèse,
nous nous intéressons en particulier à trois d'entre elles.
\begin{itemize}
	\item La première est appelée \emph{sémantique opérationnelle}. Une
		sémantique opérationnelle décrit généralement, à travers des règles de
		réécriture ou des règles d'inférence, les opérations qu'un langage de
		programmation est censé effectuer. Par exemple, étant donné un programme
		\emph{formel} $\mathtt{t}$, la sémantique opérationnelle pourrait détailler
		par exemple l'état du programme après une étape de calcul. Cela s'écrit
		souvent $\mathtt t \to \mathtt t'$.
	\item Une autre forme de sémantique consiste à fournir une \emph{théorie
		équationnelle}. Une théorie équationnelle formalise si deux programmes sont
		censés produire le même résultat, et cela s'écrit généralement $\mathtt t =
		\mathtt t'$. Elle est souvent plus générale qu'une sémantique
		opérationnelle, dans le sens où si $\mathtt t \to \mathtt t'$, alors
		$\mathtt t = \mathtt t'$.
	\item La dernière forme de sémantique utilisée dans cette thèse est la
		sémantique dénotationnelle, où cette fois, ce sont les mathématiques qui
		sont utilisées pour donner un sens à un programme. Cette sémantique
		représente les actions d'un programme sous forme d'une fonction qui prend
		en entrée l'état de départ du programme. Elle est considérée comme une
		manière de s'abstraire de la syntaxe du langage. Elle est pratique à
		plusieurs égards : on peut prouver des propriétés sur un langage sans
		dépendre de la syntaxe ; elle nous permet également d'utiliser des travaux
		mathématiques antérieurs réalisés de manière indépendante, et elle fournit
		parfois de nouvelles intuitions sur le paradigme de programmation utilisé.
\end{itemize}

L'isomorphisme de \emph{Curry-Howard} \cite{curry1934logic, howard1980types}
établit un lien fort entre les programmes et les preuves en logique formelle :
c'est en réalité une correspondance bijective entre la sémantique
opérationnelle des programmes et la réécriture en théorie de la preuve. L'ajout
de la sémantique dénotationnelle dans l'étude des langages de programmation, en
particulier avec l'aide de la théorie des catégories, a conduit à une
correspondance cette fois entre les programmes, les preuves et les catégories,
appelée la \emph{correspondance Curry-Howard-Lambek}.

La théorie des catégories est la science des fonctions, décrivant les
structures mathématiques à travers des morphismes qui peuvent être composés ;
contrairement à la théorie des ensembles où l'accent est mis sans surprise sur
les ensembles, et non sur les fonctions. La composition signifie que, étant
donné un morphisme $X \to Y$ et un morphisme $Y \to Z$, il existe un morphisme
$X \to Z$ qui est le résultat de la composition des deux précédents. Une
catégorie est une collection d'objets et de morphismes, et la théorie des
catégories définit un cadre formel pour parler de ces morphismes. Son
vocabulaire permet de formaliser des énoncés généraux sur diverses structures
mathématiques. Ce vocabulaire est la clé de voûte du contenu mathématique de
cette thèse et il est donné et expliqué tout au long du
Chapitre~\ref{ch:background}.

La théorie des catégories a un impact significatif dans l'étude mathématique
des langages de programmation, car les programmes sont eux-mêmes des morphismes
: ils transforment des ensemblese de données en d'autres ensembles de données.
De plus, deux programmes peuvent être composés ; dans de nombreux langages de
programmation, la composition de deux programmes est obtenue par concaténation.
Il est donc naturel d'utiliser la théorie des catégories pour étudier
mathématiquement les langages de programmation.

\section*{Les effets vus de manière externe}
\label{sec:fr-intro-external}

Dans le domaine des langages de programmation, le concept d'effets englobe les
interactions observables entre un programme et son environnement, encapsulant
des actions qui dépassent le domaine qu'on appelle du calcul \emph{pur}.
Contrairement aux calculs purement fonctionnels, qui présentent un comportement
déterministe, les calculs avec effet permettent aux programmes d'interagir avec
des entités externes, de manipuler des états ou d'effectuer des opérations
d'entrée/sortie. Ces effets jouent un rôle crucial dans la définition du
comportement et de la fonctionnalité des langages de programmation.

Les effets peuvent être traités de manière \emph{externe}. Par exemple, dans un
système typé -- c'est-à-dire un langage où des labels spécifiques, appelées
\emph{types}, sont attribuées aux termes -- on peut séparer les types de
calculs purs des types de calculs à effet. Ces derniers sont souvent attribués
à une \emph{modalité}. C'est ce qui est fait dans le métalangage de Moggi
\cite{moggi}, basé sur un $\lambda$-calcul simplement typé avec des types
supplémentaires pour les effets.

De plus, la sémantique dénotationnelle des effets en théorie des catégories a
été largement étudiée par Moggi \cite{moggi, moggi-lics} (voir les détails dans
la section \secref{sub:sem-lambda-effects}). Il montre que les effets sont
correctement interprétés par des \emph{monades}. Ces dernières sont la
généralisation catégorique des monoïdes de la théorie des ensembles, où un
calcul sans effet correspond à l'élément neutre du monoïde, et une composition
d'effets est similaire à la multiplication. Il est naturel de se demander si
les propriétés sur les monoïdes s'appliquent également aux monades. En
particulier, nous nous concentrons sur la question de la commutativité. Deux
éléments $x$ et $y$ dans un monoïde commutent si le produit de $x$ et $y$ est
le même que le produit de $y$ et $x$. Un élément $x$ est \emph{central} s'il
commute avec tous les autres éléments du monoïde. Par conséquent, on peut se
demander ce qu'est un \emph{effet central} dans une monade, qui est la
généralisation directe de la notion de monoïde. Dans le
Chapitre~\ref{ch:monads}, nous fournissons les réponses à la question de la
centralité des effets.

\section*{Les effets vus de manière interne}
\label{sec:fr-intro-internal}

Contrairement aux effets travaillant avec un périphérique externe -- par
exemple, l'entrée/sortie --, les \emph{effets algébriques} sont souvent
considérés internes au langage et l'informatique quantique n'y fait pas
exception.

\paragraph{Informatique quantique et réversibilité.}
Les données quantiques sont caractérisées par la \emph{superposition}. Alors
qu'un bit classique prend ses valeurs dans l'ensemble $\{0,1\}$, un bit
quantique -- souvent écrit \emph{qubit} -- est donné par une superposition :
\[
	\ket{\varphi} = \alpha \ket{0} + \beta \ket{1}
\]
où $\ket{0}$ et $\ket{1}$ sont des vecteurs dans un espace de Hilbert et
$\alpha$ et $\beta$ sont des nombres complexes. Avec cette présentation, on
peut voir dès lors que l'informatique quantique a sa place parmi les effets
algébriques. Cependant, il existe d'autres conditions sur un qubit pour qu'il
soit physiquement admissible. Les vecteurs $\ket{0}$ et $\ket{1}$ doivent être
orthogonaux et les nombres complexes $\alpha$ et $\beta$ doivent vérifier la
condition suivante : $\abs{\alpha}^2 + \abs{\beta}^2 = 1$. Ce sont les
conditions de \emph{normalisation}, et un état $\ket{\varphi}$ qui vérifie ces
conditions est dit \emph{normalisé}.

Pour préserver cette normalisation, les opérations quantiques admissibles --
appelées \emph{unitaires} -- doivent être réversibles. Il existe d'autres
opérations en informatique quantique qui ne sont pas réversibles : une pour
créer des états quantiques, par exemple initialisant à $\ket{0}$, et une pour
détruire des données quantiques, appelée \emph{mesure}. Cette dernière envoie
l'état $\ket{\varphi}$ sur $0$ avec une probabilité $\abs{\alpha}^2$ et sur $1$
avec une probabilité $\abs{\beta}^2$.

Une première approche de la programmation intégrant du quantique est le
$\lambda$-calcul quantique \cite{selinger2009quantum}. Néanmoins, ce langage ne
traite pas la programmation quantique comme un effet algébrique, car il
nécessite de mesurer les qubits pour en tirer des données classiques pour
contrôler le flux d'exécution.

Pour rester dans un effet algébrique quantique, une solution serait de
considérer uniquement les opérations quantiques réversibles, comme le montre
\cite{sabry2018symmetric}. Dans cet article, la réversibilité des fonctions du
langage est assurée grâce à l'aide du pattern-matching réversible.

\paragraph{Pattern-matching réversible.}
Considérons la fonction booléenne constante valant toujours $1$. Cette
fonction n'est pas \emph{réversible} car elle n'est pas injective. Pour être
réversible, une fonction $f \colon X \to Y$ doit être \emph{déterministe en
avant} et \emph{déterministe en arrière}. Le premier est généralement supposé :
cela signifie que tout $x \in X$ a une seule image par la fonction $f$. Le
second signifie que pour tout $y \in Y$, il existe au plus un $x \in X$ tel que
$f(x) = y$. C'est un synonyme d'\emph{injectivité}.

En informatique, le type des bits est donné par $1 \oplus 1$ où la somme
directe $\oplus$ peut par exemple dénoter l'union disjointe d'ensembles. Nous
introduisons également deux combinateurs : l'injection à gauche $\inl\!$ et
l'injection à droite $\inr\!$, tels que les termes $\inl *$ et $\inr *$
représentent respectivement le bit $0$ et le bit $1$. La fonction constante
peut alors être donnée par :
\[
	\left\{ \begin{array}{lcll}
		1 \oplus 1 &\to & 1 \oplus 1 \\
		\inl * &\mapsto & \inr * \\
		\inr * &\mapsto & \inr *
	\end{array} \right.
\]
Cette présentation fournit une intuition sur l'injectivité grâce à la syntaxe,
car les deux sorties possibles sont injectées du même côté. Pour assurer la
réversibilité de la mise en correspondance de motifs, nous forçons les motifs
d'un même côté à être \emph{orthogonaux}, une notion qui fait écho en algèbre
linéaire, où les vecteurs sont orthogonaux s'ils appartiennent à des parties
distinctes d'une somme directe. Étant donné n'importe quel terme $t$, $\inl t$
et $\inr t$ sont orthogonaux. Nous voyons dans le chapitre~\ref{ch:qu-control}
comment cette notion d'orthogonalité peut être formalisée syntaxiquement.

Étant donnés deux programmes $f \colon \B \to \B$ et $g \colon \B \to \B$ qui
sont réversibles, le programme suivant :
\begin{equation}
	\label{eq:switch}
	\phi =
	\left\{ \begin{array}{lcll}
		\B \otimes \B &\to & \B \otimes \B \\
		(0,x) &\mapsto & (0,fg(x)) \\
		(1,x) &\mapsto & (1,gf(x))
	\end{array} \right.
\end{equation}
est par exemple réversible.

\paragraph{Sémantique de la programmation quantique.}
Il existe, dans la littérature, de nombreux modèles sémantiques différents pour
les langages de programmation quantiques \cite{benoit-phd,
malherbe2013categorical, pagani2014quantitative, cho2016von,
jia2022variational, tsukada2024enriched}. Cependant, toutes ces approches sont
des modèles appropriés pour l'informatique quantique avec un \emph{contrôle
classique} uniquement : les tests, tels que les instructions $\mathtt{if}$ ou
les appels récursifs, sont contrôlés par des données classiques. Dans le
$\lambda$-calcul quantique, un qubit doit être mesuré avant d'influencer le
contrôle d'un programme. Ces modèles ne nous concernent donc pas, car
l'utilisation de la mesure détruit la superposition quantique, et elle ne
préserve donc pas l'effet quantique susmentionné. Dans cette thèse, nous
souhaitons préserver cet effet, d'où un accent sur un \emph{flux de contrôle
quantique}.

En informatique quantique, l'exemple de fonction réversible donné ci-dessus
(\ref{eq:switch}), où les bits sont généralisé à des qubits, est appelé le
\emph{quantum switch} \cite{chiribella2013quantum}. La fonction $\lambda f .
\lambda g . \phi$ ne peut pas être exprimée dans le $\lambda$-calcul quantique
ni dans aucun langage avec un contrôle classique, car le flux du programme
$\phi$ doit être contrôlé par des données quantiques. L'un des objectifs de
cette thèse est donc d'établir des bases solides pour la sémantique d'un
langage de programmation avec un flux de contrôle quantique.

\section*{Contribution de la Thèse}

Cette thèse aborde les effets d'un point de vue algébrique. Tout d'abord, dans
le chapitre~\ref{ch:monads}, nous étudions la question de la commutativité des
effets à travers leur sémantique dénotationnelle -- à savoir, les monades
fortes. Nous commençons par poser les bases catégoriques pour définir ce qu'est
le centre d'une monade (voir le théorème~\ref{th:centralisability}) et ce
qu'est une \emph{sous-monade centrale} (voir le théorème~\ref{th:centrality}).
Nous proposons ensuite une syntaxe proche du métalangage de Moggi pour capturer
les effets centraux, et nous introduisons à la fois des théories et une
sémantique dénotationnelle pour ce métalangage, que nous avons appelé
\emph{Central Submonad Calculus}. Nous montrons un résultat de langage interne
(voir le théorème~\ref{th:internal-language}), prouvant que les théories
équationnelles du Central Submonad Calculus sont essentiellement équivalentes
aux modèles de ce calcul.

En seconde partie, dans le chapitre~\ref{ch:qu-control}, nous nous concentrons
sur un sujet plus spécifique, qui est l'informatique quantique vue comme un
effet algébrique réversible. Nous présentons un langage de programmation
réversible qui capture cet effet de manière interne. En particulier, le langage
est conçu pour manipuler des états quantiques \emph{normalisés} et pour
préserver cette normalisation. Cela se fait par l'introduction d'une notion
syntaxique d'\emph{orthogonalité}, et de l'équivalent syntaxique des bases
orthonormales appelé une \emph{décomposition orthogonale}. Nous proposons une
théorie équationnelle et une sémantique dénotationnelle pour le langage, et
nous prouvons la \emph{complétude} (voir le théorème~\ref{th:qu-completeness})
: étant donnés deux termes bien typés, ils sont égaux dans la théorie
équationnelle si et seulement s'ils sont égaux dans la sémantique
dénotationnelle.

Ensuite, nous abordons la question des types de données infinis et de la
récursivité dans la programmation réversible, dans le but de l'adapter à
l'effet réversible quantique. Dans le Chapitre~\ref{ch:reversible}, nous
introduisons un langage réversible similaire au précédent, cette fois-ci sans
effets quantiques, mais où des types de données inductifs et la récursivité
sont ajoutées. Nous donnons une sémantique opérationnelle au langage, où les
opérations \emph{d'ordre supérieur}, telles que les appels récursifs, sont
considérées séparément des opérations réversibles. Nous fournissons également
une sémantique dénotationnelle dans les \emph{catégories join inverse rig} qui
ont les propriétés exactes nécessaires pour modéliser le langage. Nous montrons
que ce modèle est \emph{adéquat} par rapport à la sémantique opérationnelle
(voir le théorème~\ref{th:rev-adeq}), et nous fournissons ensuite un résultat
proche de la \emph{complétude totale} (voir le théorème~\ref{th:computable}),
montrant que toute fonction calculable dans le modèle concret des injections
partielles est représentable par une fonction dans le langage.

Enfin, le chapitre~\ref{ch:qu-recursion} contient des commentaires sur la
récursion dans le contexte de l'effet réversible quantique. En effet, il
s'avère que cet effet ne peut pas être étudié comme une monade. De plus, les
techniques utilisées dans le Chapitre~\ref{ch:reversible} ne se généralisent
pas au cas quantique : les catégories en jeu ne sont pas enrichies dans $\dcpo$
et ne semblent pas être tracées de manière convenable. Sur une note plus
positive, nous proposons une solution potentielle à la récursion quantique avec
l'aide de la récursion gardée, un cadre dans lequel les appels récursifs sont
\emph{gardés} par des modalités de retard. Pour ce faire, nous établissons un
modèle catégorique pour la récursion quantique gardée, et prouvons que ce
modèle est adapté pour interpréter la récursion (voir le
théorème~\ref{th:contractive-S}) et les types inductifs (voir le
théorème~\ref{th:param}).

%% file: introduction.tex
\vspace{2cm}

The primary focus of this thesis resides in the exploration of foundational
structures within programming, with a strong theoretical aspect. We use
structures drawn from mathematics and logic to prove properties on programs.
This area of theoretical computer science is called \emph{formal methods}. In
particular in this thesis, we set our gaze on \emph{formal} programming
languages. They are developed and studied to assert and extract formal
properties in specific areas of programming language design. The most standard
formal programming language is the $\lambda$-calculus, introduced by Alonzo
Church \cite{church1932logic} as a part of his research programme in the
foundations of mathematics. The presentation of the $\lambda$-calculus is in
appearance simple: in its syntax, one can either form functions, or apply those
functions.  However, this language is \emph{Turing complete}, meaning that any
\emph{computable} program can be represented in the $\lambda$-calculus. It is
especially convenient to work with it because it is simple to add almost any
kind of feature to it, by adding combinators for example. It is also flexible,
with different possible computation strategies.

In this thesis, we set ourselves to study two aspects of language design,
namely, control flow and effects. The \emph{control flow} delineates the
decision-making processes underlying task execution within a programming
paradigm. There are various ways to control the flow of a programming language.
Imperative languages, such as \texttt{Python}, are controlled by statements
like ``if'' and ``while'', whereas functional languages such as \texttt{Caml}
are more subtile with control through functions calls; and the latter takes
inspiration in the $\lambda$-calculus and its mathematical interpretation.
The other aspect we care about is \emph{effects}: an
\emph{effectful} computation is a fundamental distinction from \emph{pure}
computation; in other words, dealing with effects in programming means
distinguishing between the core operations of a language and its interactions
with the outside world. These effects come in different forms. There are some
general effects, sometimes called \emph{side} effects, that interact directly
with an agent outside of the program -- for example, writing on a separate
tape, managing inputs and outputs of data. We also encounter \emph{algebraic
effects}, such as the introduction of non-determinism or probabilistic
behaviour. Quantum computation can also be seen as an algebraic effect. These
effects are \emph{algebraic} in the sense that they exhibit algebraic
characteristics. This means, in particular, that the questions that apply to
algebraic structures can also apply to effects -- for example, whether two
effects commute with each other.

In this thesis, we propose a formal study of these aspects through the lens of
semantics, a paradigm in computer science that assigns logical and mathematical
interpretations to programs. A semantic study of a programming language allows
us to derive statements about programs -- for example, whether they terminate.
The mathematical components of a semantic study is, in our case, done with the
help of category theory.

\section*{Semantics}
\label{sec:intro-sem}
\begin{citing}{6cm}
	From Ancient Greek \textgreek{σημᾰντῐκός}, \emph{which gives meaning}.
\end{citing}

Semantics started with Robert Floyd \cite{floyd1967meanings}, as an attempt at
formalising what is expect from a programming language with a logic point of
view. The semantics of a programming language can come in many different
shapes. In this thesis, we are concerned with three in particular. 
\begin{itemize}
	\item The first one is called \emph{operational semantics}. An
		operational semantics usually outlines, through rewriting rules
		or inference rules, the operations a programming language is
		expected to perform. For example, given a \emph{formal} program
		$\mathtt{t}$, the operational semantics could detail for
		example what is the state of the program after one
		computational step.  This is often written $\mathtt t \to
		\mathtt t'$.
	\item Another form of semantics consists in providing an \emph{equational
		theory}. An equational theory formalises whether two programs are expected
		to eventually perform the same operation, and we write that $\mathtt t =
		\mathtt t'$. It is usually more general than an operational semantics, in
		the sense that if $\mathtt t \to \mathtt t'$, then we have $\mathtt t =
		\mathtt t'$.
	\item The final form of semantics used in this thesis is denotational
		semantics. As mentioned above, the denotational semantics of a
		programming language involves mathematics. It represents the
		actions of a program as a function on the inputs. It is thought
		as a way of abstracting away from the syntax of the language.
		It is practical in several ways: one can prove properties about
		a language without depending on the syntax, it also allows us
		to use previous mathematical work realised independently, and
		it sometimes provides new intuitions on the matters at hand.
\end{itemize}

The \emph{Curry-Howard} isomorphism \cite{curry1934logic, howard1980types}
establishes a strong link between programs and proofs in formal logic: it states
a one-to-one correspondence between the operational semantics of programs and
proof-theoretic rewriting. The addition of denotational semantics in the study
of programming language, especially with the help of category theory, led to a
one-to-one correspondence between programs, proofs and categories, called the
\emph{Curry-Howard-Lambek correspondence}.

Category theory describes mathematical structures through morphisms -- for
example, functions or relations -- that can be composed; as opposed to set
theory in which the emphasis is on sets, and not functions. Composition means
that given a morphism $X \to Y$ and a morphism $Y \to Z$, there exists a
morphism $X \to Z$ that is the result of the two former morphisms
\emph{composed}. A category is a collection of objects and morphisms, and
category theory defines a framework for talking about morphisms. This
vocabulary helps formalise general statements on various mathematical
structures. This vocabulary is the cornerstone of the mathematical content of
this thesis. It is given and explained along Chapter~\ref{ch:background}.

Category theory is especially meaningful in the mathematical study of
programming language, because programs are themselves morphisms: they transform
bits of data into bits of data. Moreover, two programs can be composed; in
imperative programming languages for instance, the composition of two programs
is obtained by concatenation. Thus it is only natural to use category theory to
study programming languages.

\section*{Effects as External Behaviour}
\label{sec:intro-external}

In the domain of programming languages, the concept of ``effects'' encompasses
the interactions between a program and its environment, encapsulating actions
that extend beyond the realm of pure computation. Unlike purely functional
computations, which adhere strictly to mathematical principles and exhibit
deterministic behaviour, effectful computations enable programs to interact
with external entities, manipulate states, or perform input/output operations.
These effects play a pivotal role in shaping the behaviour and functionality of
programming languages. 

Because of this distinction between internal states and the environment, a
natural approach is to consider effects \emph{externally} to the program. For
example, in a typed system -- \emph{i.e.}~a language where secific labels,
called \emph{types}, are assigned to terms -- one can separate the types of
pure computations from the types of effectful computations. The latter are
often assigned a \emph{modality}. This is what is done in Moggi's metalanguage
\cite{moggi}, based on a typed $\lambda$-calculus with an additional type for
effects.

Moreover, the denotational semantics of effects in category theory have been
extensively studied by Moggi \cite{moggi, moggi-lics} (see details in
\secref{sub:sem-lambda-effects}). He shows that effects are suitably
interpreted by \emph{monads}. The latter are the category theoretical
generalisation of monoids in set theory, where a computation without effect
corresponds to the neutral element of the monoid, and a composition of effects
is akin to the multiplication. It is only natural to wonder whether properties
on monoids also apply to monads. In particular, we focus on the question of
commutativity. Two elements $x$ and $y$ in a monoid are commutative if the
product of $x$ and $y$ is the same as the product of $y$ and $x$. An element
$x$ is central if it commutes with all other elements in the monoid.
Consequently, one can wonder about what is a \emph{central effect} in a monad.
In Chapter~\ref{ch:monads}, we provide the answers to the question of
centrality of effects.

\section*{Effects as an Internal Behaviour}
\label{sec:intro-internal}

Contrary to effects working with an actual external device -- for example,
input/output --, \emph{algebraic} effects -- for instance, probabilities or non
determinism -- are often considered internally to the language. Quantum
computing is not an exception in that regard.

\paragraph{Quantum computing and reversibility.}
Quantum data is characterised by \emph{superpositions}. While a classical bit
takes its values in the set $\{ 0,1 \}$, a quantum bit -- often written
\emph{qubit} -- is given by the superposition:
\[
	\ket \varphi = \alpha \ket 0 + \beta \ket 1
\]
where $\ket 0$ and $\ket 1$ are vectors in a Hilbert space and $\alpha$ and
$\beta$ are complex numbers. With this presentation, we can already see that
quantum computing has its place among algebraic effects. However, there are
more conditions on a qubit for it to be physically admissible. The vectors
$\ket 0$ and $\ket 1$ need to be orthogonal and the complex numbers $\alpha$
and $\beta$ need to verify $\abs\alpha^2 + \abs\beta^2 = 1$. These are called
\emph{normalisation conditions}, and a state $\ket \varphi$ that verifies these
conditions is said \emph{normalised}.

To preserve this normalisation, the admissible quantum operations -- called
\emph{unitaries} -- need to be reversible. There are other operations in
quantum computing that are not reversible: one to create quantum data, for
example initialising at $\ket 0$, and one to destroy quantum data, called
\emph{measurement}. The latter maps the state $\ket \varphi$ to $0$ with
probability $\abs\alpha^2$ and to $1$ with probability $\abs\beta^2$.

A first approach to programming with quantum effects is the quantum
$\lambda$-calculus \cite{selinger2009quantum}. Nevertheless, that language
does not handle quantum programming as an algebraic effect, since it requires
measurement into classical data to control the flow of execution.

To stay within a quantum algebraic effect, a solution would be to consider only
the quantum reversible operations, as shown in \cite{sabry2018symmetric}. In
that paper, the reversibility of functions of the language is ensured through
the help of reversible pattern-matching.

\paragraph{Reversible pattern-matching.}
Consider the constant boolean function $1$. This function is not
\emph{reversible} because it is not injective. To be reversible, a function $f
\colon X \to Y$ has to \emph{forward deterministic} and \emph{backward
deterministic}. The former is traditionally assumed: it means that any $x \in
X$ has a single image by the function $f$. The latter means that for all $y \in
Y$, there exists a most one $x \in X$ such that $f(x) = y$. It is a synonym of
\emph{injectivity}.

In computer science, the type of bits is given by $1 \oplus 1$ where the direct
sum $\oplus$ can for example be the disjoint union of sets. We also introduce
two combinators, the left injection $\inl\!$ and the right injection $\inr\!$,
such that the terms $\inl *$ and $\inr *$ respectively respresent the bit $0$
and the bit $1$. The constant boolean function can then be given by:
\[
	\left\{ \begin{array}{lcll}
		1 \oplus 1 &\to & 1 \oplus 1 \\
		\inl * &\mapsto & \inr * \\
		\inr * &\mapsto & \inr * 
	\end{array} \right.
\]
This presentation provides an intuition on injectivity thanks to the syntax, as
both possible outputs are injected to the same side.  To ensure reversibility
of pattern-matching, we force the patterns of a same side to be
\emph{orthogonal}, a notion that echoes in linear algebra, where vectors are
orthogonal if they belong to separate parts of a direct sum. Given any term
$t$, $\inl t$ and $\inr t$ are orthogonal. We see in
Chapter~\ref{ch:qu-control} how this notion of orthogonality can be formalised
syntactically.

Given two programs $f \colon \B \to \B$ and $g \colon \B \to \B$ that
are reversible, the following program:
\begin{equation}
	\label{eq:switch}
	\phi =
	\left\{ \begin{array}{lcll}
		\B \otimes \B &\to & \B \otimes \B \\
		(0,x) &\mapsto & (0,fg(x)) \\
		(1,x) &\mapsto & (1,gf(x))
	\end{array} \right.
\end{equation}
is for instance reversible.

\paragraph{Semantics of quantum computing.}
There are, in the literature, many different semantic models of quantum
programming languages \cite{benoit-phd, malherbe2013categorical,
pagani2014quantitative, cho2016von, jia2022variational, tsukada2024enriched}.
However, all those approaches are proper model for quantum computing with
\emph{classical control}: tests, such as $\mathtt{if}$ statements or recursive
calls, are controlled by classical data. In the quantum $\lambda$-calculus, a
qubit has to be measured before influencing the control of a program. These
models are not of interest to this thesis, because the use of measurement
breaks superposition, therefore it does not preserve the aforementioned quantum
effect. In this thesis, we wish to preserve the effect. Hence a focus on a
\emph{quantum-controlled} flow. 

In quantum computing, the example of reversible function given above
(\ref{eq:switch}), where bits are replaced with qubits, is called the
\emph{quantum switch} \cite{chiribella2013quantum}. The function $\lambda f .
\lambda g . \phi$ cannot be expressed in the quantum $\lambda$-calculus nor any
language with classical control, because the flow of the program $\phi$ needs
to be controlled by quantum data. One of the goals of this thesis is to lay
foundations for the semantics of a language with a quantum control flow.

\section*{Contribution of the Thesis}

This thesis tackles effects with an algebraic point of view. First, in
Chapter~\ref{ch:monads}, we study the question of commutativity of effects
through their denotational semantics -- namely, strong monads. We start by
laying out categorical grounds to define what is the centre of a monad (see
Theorem~\ref{th:centralisability}) and what is a \emph{central} submonad (see
Theorem~\ref{th:centrality}). We then provide a syntax close to Moggi's
metalanguage to capture central effects, and we introduce both theories and
denotational semantics for this metalanguage, that we called the \emph{Central
Submonad Calculus}. We show an internal language result (see
Theorem~\ref{th:internal-language}), proving that equational theories of the
Central Submonad Calculus are basically equivalent to models of this calculus.

Secondly, in Chapter~\ref{ch:qu-control}, we focus on a more specific subject,
which is quantum computing seen as a reversible algebraic effect. We provide a
reversible programming language that captures this effect internally. In
particular, the language is designed to manipulate \emph{normalised} quantum
states and to preserve this normalisation. This is done through the
introduction of a syntactical notion of \emph{orthogonality} and of
\emph{orthogonal decomposition}, which is the syntactical equivalent to an
orthonormal basis. We provide an equational theory and a denotational semantics
for the language, and we prove \emph{completeness} (see
Theorem~\ref{th:qu-completeness}): given two well-typed terms, they are equal
in the equational theory if and only if they are equal in the denotational
semantics.

Then, we tackle the question of infinite data types and recursion in reversible
programming, as an attempt to adapt it to the quantum reversible effect. In
Chapter~\ref{ch:reversible}, we introduce a reversible language akin to the one
before, this time without quantum effects, but where inductive data types and
recursion are added. We give an operational semantics to the language, where
\emph{higher-order} operations, such as recursive calls, are considered
separately to reversible operations. We also provide a denotational semantics
in \emph{join inverse rig categories} which have the exact properties needed to
model the language. We show that this model is \emph{adequate} with regard to
the operational semantics (see Theorem~\ref{th:rev-adeq}), and we later provide
a result close to \emph{full completeness} (see Theorem~\ref{th:computable}),
showing that any computable function in the concrete model of partial
injections is representable by a function in the language.

Finally, Chapter~\ref{ch:qu-recursion} contains comments on recursion in the
context of quantum reversible effects. Indeed, it turns out that this effect
cannot be studied as a monad. Moreover, the techniques used in
Chapter~\ref{ch:reversible} do not generalise to the quantum case: the
categories at play are not enriched in $\dcpo$ and do not seem to be properly
traced. On a more positive note, we provide a potential solution to quantum
recursion with the help of guarded recursion, a framework in which recursive
calls are \emph{guarded} by delay modalities. To do so, we lay out a
categorical model for guarded quantum recursion, and prove that this model is
suitable to interpret recursion (see Theorem~\ref{th:contractive-S}) and
inductive types (see Theorem~\ref{th:param}).

%% file: background.tex
\begin{abstract}
	We introduce the background material in mathematics, and especially in
	category theory, necessary to navigate the thesis seamlessly. In some
	sections, basic notions of type theory and programming languages are
	presented and linked to their categorical semantics.

	\paragraph{References.} This background chapter is only made up of earlier
	work, published by several different authors. It is meant as an introduction
	to the material this thesis requires, and not as a literature review. The few
	proofs inserted here and there are provided by the author, for didactic
	purposes.
\end{abstract}

\section{Category theory: some definitions}
\label{sec:background-category}

In this thesis, category theory is used as vocabulary to express a mathematical
point of view on programs and programming languages; and what we often
refer to as \emph{interpretation}, \emph{denotational semantics} or
\emph{denotation} is a map from a syntax -- or a \emph{language} -- to a
category, usually written $\sem -$. In other words, given a piece of syntax
$\mathtt t$, which we refer to as a \emph{term} of the syntax, its
interpretation $\sem{\mathtt t}$ is given in a fixed category $\CC$. If a few
coherence properties are satisfied, we allow ourselves to call $\CC$ a
\emph{model} of the language. We provide examples along the chapter, in
\secref{sub:ccc}, \secref{sub:dcpo} and \secref{sub:inductive-types}.

In this section, we recall the definitions required to work with category theory
as a denotational model of programming languages. The author recommends the book of
Tom Leinster \cite{leinster2016basic}, which introduces category theory with more
background, details and examples.

\begin{definition}[Category]
	\label{def:cat-category}
	A \emph{category} $\CC$ is a collection of objects -- usually written with
	capital Latin letters $X,Y,Z,\dots$ --  and a collection of morphisms --
	written $f:X\to Y$ to indicate that $f$ is a morphism from $X$ to $Y$ -- such
	that:	
	\begin{itemize}
		\item for every object $X$, there is a morphism $\iid_X \colon X\to X$,
		\item for every pair of morphisms $f \colon X\to Y$, $g \colon Y \to Z$, there is a
			morphism $g\circ f \colon X\to Z$ called the composition of $f$ and $g$,
		\item composition is associative: $(f \circ g) \circ h = f \circ (g \circ h)$,
		\item and for every morphism $f \colon X\to Y$, we have $\iid_Y \circ f= f =f\circ \iid_X$.
	\end{itemize}
	We write $\CC(X,Y)$ for the collection of morphisms from $X$ to $Y$.
\end{definition}

\begin{example}
	A very well-known category is the one with objects that are sets, and
	morphisms that are functions between sets; we write $\Set$ for this category.
	Note that whenever $X$ and $Y$ are sets, $\Set(X,Y)$ is also a set.
\end{example}

\begin{example}
	Vector spaces -- additive groups together with the outer action of a field
	$\mathbb K$ -- with linear maps as morphisms also form a category, written
	$\Vect$.
\end{example}

A category $\CC$ is called \emph{locally small} if all $\CC(X,Y)$ are sets.
They are then called \emph{homsets}, short for ``sets of homomorphisms''. A
locally small category $\CC$ is \emph{small} if its collection of objects is a
set.

\begin{remark}
	Throughout the thesis, given two morphisms $f \colon X \to Y$ and $g \colon
	Y \to Z$, we write $gf \colon X \to Z$ for the composition $g \circ f \colon
	X \to Z$ when it is not ambiguous.
\end{remark}

Category theory is better pictured with diagrams to represent morphisms. In
a category $\CC$, the composition of two morphisms $f \colon X \to Y$ and 
$g \colon Y \to Z$ is seen as the diagram:
\[ \begin{tikzcd}
	X & Y & Z 
	\arrow["f", from=1-1, to=1-2]
	\arrow["g", from=1-2, to=1-3]
\end{tikzcd}\]
Diagrams are sound thanks to associativity. It allows us to write the following
diagram:
\[ \begin{tikzcd}
	X & Y & Z & T
	\arrow["f", from=1-1, to=1-2]
	\arrow["g", from=1-2, to=1-3]
	\arrow["h", from=1-3, to=1-4]
\end{tikzcd}\]
without any need to be precise about in which order the composition is taken.
Moreover, given $h \colon X \to Z$, the condition $h = g \circ f$ is
described as the \emph{commutativity} of the following diagram:
\[ \begin{tikzcd}
	X & Y & Z 
	\arrow["f", from=1-1, to=1-2]
	\arrow["g", from=1-2, to=1-3]
	\arrow[bend right=30, "h"', from=1-1, to=1-3]
\end{tikzcd}\]

\begin{example}[Opposite Category]
	Given a category $\CC$, one can form the opposite category, written $\CC^{\rm
	op}$, with the same objects as $\CC$, such that there is a morphism $Y \to X$
	for every morphism $X \to Y$ in $\CC$. These new morphisms respect the same
	diagrams as in $\CC$, but with reversed arrows.
\end{example}

Once we master the definition of a category, we can introduce some vocabulary
on morphisms. We define what an isomorphism is. This definition echoes to the
one in set theory.

\begin{definition}[Isomorphism]
	Given a category $\CC$ and a morphism $f \colon X \to Y$ in that category, we
	say that $f$ is an \emph{isomorphism} if there exists a (unique) $f\inv
	\colon Y \to X$ such that $f \circ f\inv = \iid_Y$ and $f\inv \circ f =
	\iid_X$.
\end{definition}

Given an isomorphism $f \colon X \to Y$, we say that the objects $X$ and $Y$
are isomorphic. We also introduce notions akin to injective and to surjective
functions.

\begin{definition}[Monomorphism]
	Given a category $\CC$ and a morphism $f \colon X \to Y$ in that category, we
	say that $f$ is a \emph{monomorphism} -- or that $f$ is \emph{monic} -- if for all
	objects $Z$ and all morphisms $g_1, g_2 \colon Z \to X$, if $f \circ g_1 = f
	\circ g_2$, then $g_1 = g_2$.
\end{definition}

\begin{definition}[Epimorphism]
	Given a category $\CC$ and a morphism $f \colon X \to Y$ in that category, we
	say that $f$ is an \emph{epimorphism} -- or that $f$ is \emph{epic} -- if for all
	objects $Z$ and all morphisms $g_1, g_2 \colon Y \to Z$, if $g_1 \circ f =
	g_2 \circ f$, then $g_1 = g_2$.
\end{definition}

\begin{remark}
	In $\Set$, the category of sets and functions, monomorphisms (resp.
	epimorphisms) are exactly injective (resp.  surjective) functions.
\end{remark}

A morphism that is monic and epic is not necessarily an isomorphism.

\begin{lemma}[\cite{maclane}]
	Given a morphism $f$ in a category $\CC$, $f$ is a monomorphism iff it is an
	epimorphism in $\CC^{\rm op}$.
\end{lemma}

\begin{definition}[Functor]
	\label{def:cat-functor}
	Given two categories $\CC$ and $\DD$, a \emph{functor} $F \colon \CC \to \DD$
	is a function on objects and on morphisms, such that: for all objects $X$ in $\CC$,
	there is an object $F(X)$ in $\DD$, and for all morphsisms $f\colon X \to Y$,
	there is a morphism $F(f)\colon F(X) \to F(Y)$ in $\DD$, and
	\begin{itemize}
		\item for all objects $X$ in $\CC$, $F(X)$ is an object of $\DD$;
		\item for all morphisms $X \to Y$ in $\CC$, $F(f) \colon F(X) \to F(Y)$ is
			a morphism in $\DD$;
		\item for all objects $X$ in $\CC$, $F(\iid_X) = \iid_{F(X)}$;
		\item for all pairs of morphisms $f \colon X \to Y, g \colon Y \to Z$,
			$F(gf) = F(g) F(f)$.
	\end{itemize}
	Functors are often written with Latin capital letters $F$ and $G$.
\end{definition}

We abuse notations and sometimes drop the parenthesis when applying a functor.
For example, the object $F(X)$ is often written $FX$ when it is not ambiguous.

\begin{example}[Identity functor]
	\label{ex:identity-functor}
	Given any category $\CC$, one can define the identity functor $\iid_\CC 
	\colon \CC \to \CC$ that maps any object to itself and any morphism to itself.
\end{example}

\begin{example}
	We define $U \colon \Vect \to \Set$ that maps a vector space to its
	underlying set and that maps a linear map to itself, now seen as a function
	between sets. $U$ is a functor, and is called the \emph{forgetful functor},
	because it forgets the structure of a vector space.
\end{example}

\begin{example}[Hom functor]
	\label{ex:hom-functor}
	Given any locally small category $\CC$ and an object $X$ of $\CC$, the assignment
	$\CC(-,X)$ forms a functor $\CC^{\rm op} \to \Set$ that maps an object $Y$ to
	the set $\CC(Y,X)$ and a morphism $f \colon Y \to Z$ in $\CC$ to the function
	between sets $(- \circ f) \colon \CC(Z,X) \to \CC(Y,X)$. A similar functor,
	namely $\CC(X,-) \colon \CC \to \Set$, can also be defined.
\end{example}

\begin{example}[Category of small categories]
	One can define the category $\Cat$, with small categories as objects and
	functors between them as morphisms. The identity functor $\CC \to \CC$ is
	described in Example~\ref{ex:identity-functor}. Given two functors $F \colon
	\CC \to \DD$ and $G \colon \DD \to \mathbf{E}$, it is routine to show that $G
	\circ F$ is a functor $\CC \to \mathbf E$.
\end{example}

Category theory is the theory of \emph{arrows}, trying to establish morphisms whenever
it is possible. Given two categories $\CC$ and $\DD$, and two functors $F$ and $G$
between them, we obtain the following diagram:
\[
	\begin{tikzcd}
		\CC & \DD
		\arrow["F", bend left=30, from=1-1, to=1-2]
		\arrow["G"', bend right=30, from=1-1, to=1-2]
	\end{tikzcd}
\]
This diagram does not necessarily commute. However, it can be filled with a new
kind of arrow, as pictured below.
\[
	\begin{tikzcd}
		\CC & \DD
		\arrow["F"{name=0}, bend left=30, from=1-1, to=1-2]
		\arrow["G"'{name=1}, bend right=30, from=1-1, to=1-2]
		\arrow[shorten <=2pt, shorten >=2pt, Rightarrow, from=0, to=1]
	\end{tikzcd}
\]
That new arrow is called a \emph{natural transformation} and its definition is
as follows.

\begin{definition}[Natural transformation]
	\label{def:cat-nat}
	Given two categories $\CC$ and $\DD$, given two functors $F,G \colon \CC \to
	\DD$, a \emph{natural transformation} $\alpha \colon F \Rightarrow G$ is a
	collection of morphisms indexed by the objects of $\CC$ such that, for all
	morphisms $f \colon X \to Y$, the following diagram:
	\[\begin{tikzcd}
		F(X) & G(X) \\
		F(Y) & G(Y)
		\arrow["\alpha_X", from=1-1, to=1-2]
		\arrow["\alpha_Y"', from=2-1, to=2-2]
		\arrow["F(f)"', from=1-1, to=2-1]
		\arrow["G(f)", from=1-2, to=2-2]
	\end{tikzcd} \]
	commutes.
\end{definition}

\begin{example}
	Given a set $M$, we define a functor $T \defeq M \times - \colon \Set \to
	\Set$.  Moreover, given	a function $(- \cdot -) \colon M \times M \to M$
	and an element $e \in M$, there is a natural transformation $\eta \colon
	\iid_\Set \natto T$ and a natural transformation $\mu \colon T \circ T \natto
	T$, such that for all sets $X$:
	\[
		\eta_X = \left\{ \begin{array}{lcl}
			X &\to & M \times X \\
			x &\mapsto & (e,x)
		\end{array} \right.
		\qquad
		\mu_X = \left\{ \begin{array}{lcl}
			M \times (M \times X) &\to & M \times X \\
			(m_1, (m_2,x)) &\mapsto & (m_1 \cdot m_2,x)
		\end{array} \right. 
	\]
	In fact, if $(M,\cdot,e)$ is a monoid, then $T$ is a monad (see
	Definition~\ref{def:monad}).
\end{example}

\begin{example}[Functor Category]
	Given two categories $\CC$ and $\DD$, we write $[\CC \to \DD]$ or $\DD^\CC$
	for the category of functors $\CC \to \DD$ and natural transformations
	between them. Given a functor $F \colon \CC \to \DD$, the identity natural
	transformation $\iid_F \colon F \natto F$ is a natural transformation whose components are all
	the identity; and given two natural transformations $\alpha \colon F \natto
	G$ and $\beta \colon G \natto H$, for all $f \colon X \to Y$ in $\CC$, the
	diagram:
	\[\begin{tikzcd}
		F(X) & G(X) & H(X) \\
		F(Y) & G(Y) & H(Y)
		\arrow["\alpha_X", from=1-1, to=1-2]
		\arrow["\beta_X", from=1-2, to=1-3]
		\arrow["\alpha_Y"', from=2-1, to=2-2]
		\arrow["\beta_Y"', from=2-2, to=2-3]
		\arrow["F(f)"', from=1-1, to=2-1]
		\arrow["G(f)", from=1-2, to=2-2]
		\arrow["H(f)", from=1-3, to=2-3]
	\end{tikzcd} \]
	commutes, and thus $\beta \circ \alpha$ defined as the pointwise composition
	is a natural transformation.
	This composition of natural transformations is also called the \emph{vertical composition}
	because of the following diagram:
	\[
		\begin{tikzcd}
			\CC && \DD
			\arrow[""{name=0, anchor=center, inner sep=0}, "F", bend left=50, from=1-1, to=1-3]
			\arrow[""{name=1, anchor=center, inner sep=0}, "G"{description}, from=1-1, to=1-3]
			\arrow[""{name=2, anchor=center, inner sep=0}, "H"', bend right=50, from=1-1, to=1-3]
			\arrow["\alpha", shorten <=2pt, shorten >=4pt, Rightarrow, from=0, to=1]
			\arrow["\beta", shorten <=4pt, shorten >=2pt, Rightarrow, from=1, to=2]
		\end{tikzcd}
	\]
	Some functor categories are often used in the literature, and therefore have
	a name of their own. For example, given a small category $\CC$,
	$\Set^{\CC^{\rm op}}$ is called the \emph{category of presheaves} over $\CC$. 
\end{example}

As a cultural note, a category of presheaves is a topos: a cartesian closed
category whose objects and morphisms carry a logical meaning, we say that a
topos has an \emph{internal logic}. The details on cartesian closed categories
are found later in this chapter.

\begin{example}[Yoneda embedding]
	\label{ex:yoneda}
	Given a category $\CC$, the \emph{Yoneda embedding} is a functor $\yo \colon
	\CC \to \Set^{\CC^{\rm op}}$ 
	($\yo$ is the Japanese hiragana ``yo''
	after \textjap{米田 信夫（よねだ のぶお）}, Yoneda Nobuo) 
	such that for all object $X$ in $\CC$, $\yo(X) = \CC(-,X)$ (see
	Example~\ref{ex:hom-functor}), and for all morphism $f \colon X \to Y$ in
	$\CC$, $\yo(f)$ is a natural transformation $\CC(-,X) \natto \CC(-,Y)$ whose
	components are morphisms $\yo(f)_Z \colon \CC(Z,X) \to \CC(Z,Y) :: g \mapsto
	f \circ g$ in $\Set$.
\end{example}

\begin{definition}[Adjunction]
	\label{def:cat-adjunction}
	Given two categories $\CC$ and $\DD$, we say that two functors $F \colon \CC
	\to \DD$ and $G \colon \DD \to \CC$ are respectfully \emph{left adjoint} and
	\emph{right adjoint} if for all objects $X$ in $\CC$ and $Y$ in $\DD$, there
	is a bijection $\DD(FX,Y) \cong \CC(X,GY)$ that is natural in $X$ and $Y$. An
	adjunction can be written with a diagram, as follows:
	\[ \begin{tikzcd}
		\CC & \bot & \DD 
		\arrow[bend left=30, "F", from=1-1, to=1-3]
		\arrow[bend left=30, "G", from=1-3, to=1-1]
	\end{tikzcd} \]
	and is also written $F \dashv G$.
\end{definition}

An adjunction also gives rise to two natural transformations:
\begin{itemize}
	\item $\varepsilon \colon FG \natto \iid_\DD$, called the \emph{counit},
	\item $\eta \colon \iid_\CC \natto GF$, called the \emph{unit},
\end{itemize}
such that for every object $X$ in $\CC$ and every object $Y$ in $\DD$:
\[
	\iid_{FX} = \varepsilon_{FX} \circ F(\eta_X),
	\qquad
	\iid_{GY} = G(\varepsilon_Y) \circ \eta_{GY}.
\]

\begin{definition}[Initial and terminal object]
	\label{def:cat-init-termi}
	Given a category $\CC$, an object $X$ of $\CC$ is said to be \emph{initial} if for
	every object $Y$ in $\CC$, there is a unique morphism $X \to Y$. Conversely,
	an object $X$ of $\CC$ is said to be \emph{terminal} if for every object $Y$ in
	$\CC$, there is a unique morphism $Y \to X$.
\end{definition}

An initial object is often written $0$, and a terminal object is often written
$1$. Moreover, given a terminal object $1$ and any object $X$, we write $!_X$
for the unique morphism $X \to 1$.

\begin{definition}[Product]
	\label{def:cat-product}
	Given a category $\CC$ and two objects $X_1$ and $X_2$ of $\CC$, a
	\emph{product} of $X_1$ and $X_2$ is an object of $\CC$, usually written $X_1
	\times X_2$, equipped with two morphisms $\pi_1 \colon X_1 \times X_2 \to
	X_1$ and $\pi_2 \colon X_1 \times X_2 \to X_2$, such that for every object
	$Y$ and morphisms $f_1 \colon Y \to X_1$ and $f_2 \colon Y \to X_2$, there is
	a unique morphism $f \colon Y \to X_1 \times X_2$ such that the following
	diagram:
	\[ \begin{tikzcd}
		\ & Y & \ \\
		X_1 & X_1 \times X_2 & X_2
		\arrow["f_1"', from=1-2, to=2-1]
		\arrow["f_2", from=1-2, to=2-3]
		\arrow["\pi_1", from=2-2, to=2-1]
		\arrow["\pi_2"', from=2-2, to=2-3]
		\arrow["f", dashed, from=1-2, to=2-2]
	\end{tikzcd} \]
	commutes. The unique morphism obtained is often written $\pv{f_1}{f_2}$.
\end{definition}

\begin{example}
	Given two categories $\CC$ and $\DD$, their product $\CC \times \DD$ is also
	a category.
\end{example}

\begin{remark}
	Given a category $\CC$ with products for any pair of objects, observe that for
	all objects $X$ of $\CC$, $-\times X \colon \CC \to \CC$ is a functor, as well
	as $X \times - \colon \CC \to \CC$. Actually, $- \times - \colon \CC \times
	\CC \to \CC$ is a functor, and is commonly called a \emph{bifunctor}.
\end{remark}

\subsection{Cartesian closed categories and $\lambda$-calculus}
\label{sub:ccc}

Cartesian closed categories are the main tool to study the semantics of
functional classical programming languages -- \emph{classical} as opposed to
\emph{quantum}, which is one focus of this thesis. The adjective
\emph{cartesian} refers to the products, as introduced in
Definition~\ref{def:cat-product}. The notion of closure is more subtle: one of
the main properties is that function spaces are themselves objects of the
category.  Formally, a cartesian category $\CC$ is closed if for all objects
$Y$, the functor $- \times Y \colon \CC \to \CC$ admits a right adjoint $[Y \to
- ] \colon \CC \to \CC$, sometimes written $- ^ Y$. This adjunction embodies
the notion of currying, meaning that a program $(A \times B) \to C$ is
equivalent to a program $A \to (B \to C)$.

The author recommends reading the lecture notes of Awodey and Bauer
\cite{awodey-bauer-lecture} to have complete details about the topic of this
section, and more about the link between logic, categories and programming
languages.

\begin{definition}[Cartesian closed category]
	\label{def:cat-ccc}
	A \emph{cartesian closed} category $\CC$ is a category that has the following
	properties:
	\begin{itemize}
		\item $\CC$ has a terminal object, usually written $1$;
		\item for all pairs of objects $X$ and $Y$, there is a product $X \times Y$
			in $\CC$;
	\end{itemize}
	such that for all objects $Y$, the assignment $(- \times Y) \colon \CC \to
	\CC$ is a left adjoint functor.
\end{definition}

For all objects $Y$, we write $[Y \to -] \colon \CC \to \CC$ for the right
adjoint of $(- \times Y) \colon \CC \to \CC$. Given a pair of objects $X$ and
$Y$, the object $[X \to Y]$ is called the \emph{exponential}. 

\begin{example}
	The category $\Set$ of sets and functions between them is cartesian closed.
	Any singleton set is a terminal object.  The product of two sets $X$ and $Y$
	is the usual cartesian product $X \times Y$, which is the set:
	\[
		\{ (x,y) \alt x \in X, y \in Y \}
	\]
	and the exponential of $X$ and $Y$ is the set of functions from $X$ to $Y$, namely:
	\[
		\{ f \alt f \colon X \to Y \}.
	\]
\end{example}

There are many more examples of cartesian closed categories, such as the
category of dcpos and Scott continuous functions, introduced later in
\secref{sub:dcpo}.

Cartesian closed categories are remarkable because of their link with
$\lambda$-calculi. The latter is a paradigm for computation, at the same level
as Turing machines and recursive functions, and its raw presentation -- the
untyped $\lambda$-calculus -- is known to represent all computable functions.
In this thesis, we rather focus on typed $\lambda$-calculi, and typed
programming languages in general, because of their link to logic and category
theory. Next, we introduce briefly a simply-typed $\lambda$-calculus.

\paragraph{Simply-typed $\lambda$-calculus.}
First, we give a definition of the types, that are generated by the following grammar:
\begin{equation*}
	A ::= ~1 \alt A \times A \alt A \to A 
\end{equation*}
Note that, throughout the thesis, we might be less formal, writing for example:
\begin{equation}
	\label{eq:simple-types}
	A,B,\dots ~ ::= ~ 1 \alt A \times B \alt A \to B 
\end{equation}
for readability; and the two definitions are to be regarded as identical. To
define a language, one must also provide a set of terms, usually also
introduced by a grammar. In our case, the terms of the simply-typed
$\lambda$-calculus are given by: 
\begin{equation}
	\label{eq:simple-terms}
	M,N,\dots ~ ::= ~x \alt * \alt \lambda x^A. M 
	\alt MN \alt \pv M N \alt \pi_i M
\end{equation}
where $x$ can range among a set of variables $\set{x,y,z, \dots}$, $A$ is a
type as introduced above and $i \in \set{1,2}$. Note that variables can be free
in a term, or bound by a $\lambda$-abstraction. A term without any free
variables is called a \emph{closed} term. In order to avoid conflicts between
variables we will always work up to $\alpha$-conversion and use Barendregt's
convention~\cite[p.26]{henk1984lambda} which consists in keeping all bound and
free variables names distinct.

The types allows us to formalise what a well-typed term is, through typing
rules. A typing judgement is written $\Gamma \vdash M \colon A$, where $M$ is a
term of (\ref{eq:simple-terms}), $A$ is a type of (\ref{eq:simple-types}) and
$\Gamma$ is a context that contains variables each associated with a type $x_1
\colon A_1,~x_2 \colon A_2, \dots,~x_n \colon A_n$. The rule to form
\emph{correct} typing judgements are presented in a way that is usual in logic,
\emph{i.e.}~with inference rules. Those typing rules in the case of the
simply-typed $\lambda$-calculus are introduced in
Figure~\ref{fig:simple-formation-rules}.

\begin{figure}[!h]
	\[ \begin{array}{c}
	\infer{
    \Gamma,x \colon A\vdash x \colon A}{}
  \qquad
  \infer{
    \Gamma\vdash MN \colon B
  }{
    \Gamma\vdash M \colon A\to B
    &
    \Gamma\vdash N \colon A}
  \\[1.5ex]
  \infer{
    \Gamma\vdash * \colon 1}{}
  \qquad
  \infer{\Gamma\vdash\lambda x^A.M \colon A\to B}{\Gamma,x \colon A\vdash M \colon B}
  \qquad
  \infer{
    \Gamma\vdash\pi_i M  \colon  A_i}{\Gamma\vdash M \colon A_1\times A_2}
  \\[1.5ex]
  \infer{
    \Gamma\vdash \pv{M}{N} \colon  A\times B
  }{
    \Gamma\vdash M \colon A
    &
    \Gamma\vdash N \colon B
  }
	\end{array} \]
	\caption{Typing rules of the simply-typed $\lambda$-calculus.}
	\label{fig:simple-formation-rules}
\end{figure}

One important concept in programming language design is substitution, which
allows us to replace each occurrence of a free variable with a term of the syntax; and
this term can also contain variables.

\begin{definition}[Substitution]
	\label{def:lambda-substitution}
	Given two well-typed terms $M,N$ of (\ref{eq:simple-terms}), we write
	$M[N/x]$ for the term where the free occurrences of $x$ in $M$ are replaced
	by $N$.
\end{definition}

Whenever a term is introduced, it is important to show that it can be well-typed
with the typing rules.

\begin{lemma}[\cite{henk1984lambda}]
	\label{lem:type-subst}
	Given two well-typed terms $\Gamma,~x\colon A \vdash M \colon B$ and $\Gamma
	\vdash N \colon A$, the judgement $\Gamma \vdash M[N/x] \colon B$ is valid.
\end{lemma}

The computational behaviour of a programming language is formalised through an
operational semantics. In a small step operational semantics, $M \to N$
informally means that the term $M$ evaluates to $N$ after one computational
step.  The most prominent rule of the $\lambda$-calculus is called
$\beta$-reduction:
\begin{equation}
	\label{eq:beta-reduction}
	(\lambda x^A . M) N \to M[N/x].
\end{equation}

\paragraph{Equational Theory.}
Instead of working operationally, one can consider equations between terms.
This new point of view loses information on the computational aspect of the
language, but gains in convenience. The equational theory of the simply-typed
$\lambda$-calculus is given in Figure~\ref{fig:moggi-rules}.

\begin{figure}[!h]
	\begin{center}
	\resizebox{.9\hsize}{!}{
		$
		\begin{array}{c}
    \infer[(refl)]{
      \Gamma\vdash M=M \colon A
    }{
      \Gamma\vdash M \colon A
    }
    \qquad
    \infer[(symm)]{
      \Gamma\vdash M=N \colon A
    }{
      \Gamma\vdash N=M \colon A
    }
    \\[1.5ex]
    \infer[(trans)]{
      \Gamma\vdash M=P \colon A
    }{
      \Gamma\vdash M=N \colon A
      &
      \Gamma\vdash N=P \colon A
    }
    \\[1.5ex]
    \infer[(1.\eta)]{
      \Gamma,x \colon 1 \vdash * = x \colon A
    }{}
    \qquad
    \infer[(subst)]{
      \Gamma\vdash N[M/x] = P[M/x] \colon B
    }{
      \Gamma\vdash M \colon A
      &
      \Gamma, x \colon A \vdash N = P \colon B
    }
    \\[1.5ex]
    \infer[(\pv{}{} .eq)]{
      \Gamma\vdash \pv M N = \pv{M'}{N'} \colon A\times B
    }{
      \Gamma\vdash M=M' \colon A
      &
      \Gamma\vdash N=N' \colon B
    }
		\qquad
    \infer[(\times.\beta)]{
      \Gamma\vdash\pi_i\pv{M_1}{M_2} = M_i \colon A_i
    }{
      \Gamma\vdash M_1 \colon A_1
      &
      \Gamma\vdash M_2 \colon A_2
    }
    \\[1.5ex]
    \infer[(\times.\eta)]{
      \Gamma\vdash \pv{\pi_1 M}{\pi_2 M} = M \colon A\times B
    }{
      \Gamma\vdash M \colon A\times B
    }
    \\[1.5ex]
    \infer[(app.eq)]{
      \Gamma\vdash MN=M'N' \colon B
    }{
      \Gamma\vdash M=M' \colon A\to B
      &
      \Gamma\vdash N=N' \colon A
    }
    \\[1.5ex]
    \infer[(\lambda.eq)]{
      \Gamma\vdash \lambda x^A.M = \lambda x^A.N \colon A\to B
    }{
      \Gamma, x \colon A \vdash M = N \colon B
    }
		\qquad
    \infer[(\lambda.\beta)]{
      \Gamma\vdash (\lambda x^A. M) N = M[N/x] \colon B
    }{
      \Gamma,x \colon A \vdash M \colon B
      &
      \Gamma\vdash N \colon A
    }
    \\[1.5ex]
    \infer[(\lambda.\eta)]{
      \Gamma\vdash \lambda x^A. Mx = M \colon A\to B
    }{
      \Gamma\vdash M \colon A\to B
    }
    \qquad
    \infer[(weak)]{
      \Gamma,x \colon A \vdash M=N \colon B
    }{
      \Gamma\vdash M=N \colon B
    }
	\end{array}
	$
	}
\end{center}
  \caption{Equational rules of the simply-typed $\lambda$-calculus.}
  \label{fig:moggi-rules}
\end{figure}

\paragraph{Denotational semantics.}
The point of denotational semantics is to provide a mathematical interpretation
to well-typed terms in a programming language to extract properties, design or
a better understanding of the language. This requires first an interpretation
for the types. Let us fix a cartesian closed category $\CC$, and give a
semantics to types as objects in $\CC$. The semantics of types is defined by
induction on their grammar:
\[
	\sem 1 = 1 
	\qquad \sem{A \times B} = \sem A \times \sem B 
	\qquad \sem{A \to B} = [\sem A \to \sem B]
\]
A context $\Gamma = x_1 \colon A_1, \dots, x_n \colon A_n$ is interpreted as an
object in $\CC$ given by the product of the interpretations of all the types
involved: $\sem{A_1} \times \dots \times \sem{A_n}$. 

%
%

Given a well-typed term $\Gamma \vdash M \colon A$ in the simply-typed
$\lambda$-calculus, we write its denotational interpretation $\sem{\Gamma
\vdash M \colon A}$. If we fix a cartesian closed category $\CC$ (see
Definition~\ref{def:cat-ccc}), the interpretation $\sem{\Gamma \vdash M \colon
A}$ is given as a morphism in $\CC$ from the interpretation of the context
$\Gamma$ to the interpretation of the type $A$.  The interpretation of term
judgements can then be defined by induction on the typing rules. The details
can be found in Figure~\ref{fig:lambda-den-semantics}, where $\rmcurry_{X,Y,Z}$
is the natural isomorphism $\CC(X \times Y,Z) \cong \CC(X, [Y \to Z])$ given by
the adjunction $(- \times Y) \dashv [Y \to -]$, and $\rmeval_{X,Y} \colon [X
\to Y] \times X \to Y$ is the counit of the adjunction.

\begin{figure}[!h]
	\begin{align*}
		\sem{\Gamma \vdash M \colon A} &\in \CC(\sem\Gamma, \sem A) \\
		\sem{\Gamma \vdash * \colon 1} &= ~!_{\sem\Gamma} \\
		\sem{\Gamma, x \colon A \vdash x \colon A} &= \pi_{\sem A} \\
		\sem{\Gamma \vdash MN \colon B} &= \rmeval_{\sem A,\sem B} \circ
		\pv{\sem{\Gamma \vdash M \colon A \to B}}{\sem{\Gamma \vdash N \colon A}}
		\\
		\sem{\Gamma \vdash \lambda x^A . M \colon A \to B} &=
		\rmcurry_{\sem\Gamma,\sem A, \sem B} (\sem{\Gamma, x \colon A \vdash M
		\colon B}) \\
		\sem{\Gamma \vdash \pi_i M \colon A} &= \pi_i \circ \sem{\Gamma \vdash M
		\colon A_1 \times A_2} \\
		\sem{\Gamma \vdash \pv M N \colon A \times B} &=
		\pv{\sem{\Gamma \vdash M \colon A}}{\sem{\Gamma \vdash N \colon B}}
	\end{align*}
	\caption{Denotational semantics of terms in the simply-typed $\lambda$-calculus.}
	\label{fig:lambda-den-semantics}
\end{figure}

\paragraph{Relationship between the semantics.}
So far, we have introduced some operational semantics, an equational theory
and a denotational semantics to a simply-typed $\lambda$-calculus. However,
we have yet to show what links them. 

For example, we can state a \emph{soundness} result between an operational semantics
and the denotational semantics:
\begin{center}
	Given $\Gamma \vdash M \colon A$, if $M \to N$, then $\sem{\Gamma \vdash M
	\colon A} = \sem{\Gamma \vdash N \colon A}$.
\end{center}
which is often simple to prove, by induction on the rules of the operational
semantics. The converse is expected to be trickier, and is not necessarily
true. The converse of a soundness statement is called \emph{adequacy}.

In addition, a relationship between the equational theory and the denotational
semantics can be given, called \emph{completeness}:
\begin{center}
	$\Gamma \vdash M = N \colon A$ if and only if 
	$\sem{\Gamma \vdash M \colon A}
	= \sem{\Gamma \vdash N \colon A}$.
\end{center}

Soundness is often the minimal requirement to consider a category as a model of
a specific language equipped with an operational semantics or an equational
theory. Usually, adequacy or completeness is also expected; with this stronger
property, the model can be used to study and to improve a programming language.
For example, as shown is \secref{sub:dcpo}, if the language can be interpreted
in a model that allows for fixed points, then fixed points can be safely added
to the syntax.

\subsection{Symmetric monoidal categories}

We introduce \emph{symmetric monoidal} categories, which are more general
than cartesian categories. They have applications as models of linear logic
and of quantum computing.

\begin{definition}
	A \emph{monoidal} category $\CC$ is a category equipped with the following
	structure:
	\begin{itemize}
		\item a bifunctor $\otimes \colon \CC \times \CC \to \CC$, called the
			\emph{tensor product};
		\item an object $I$ called the \emph{unit};
		\item a natural isomorphism $\alpha_{X,Y,Z} \colon (X \otimes Y) \otimes Z
			\to X \otimes (Y \otimes Z)$, called the \emph{associator};
		\item a natural isomorphism $\lambda_X \colon I \otimes X \to X$, called
			the \emph{left unitor};
		\item a natural isomorphism $\rho_X \colon X \otimes I \to X$, called the
			\emph{right unitor};
	\end{itemize}
	such that, for all objects $X$, $Y$, $Z$ and $T$, the two diagrams: 
	\[
		\begin{tikzcd}
			((X \otimes Y) \otimes Z) \otimes T 
			&
			(X \otimes Y) \otimes (Z \otimes T)
			&
			X \otimes (Y \otimes (Z \otimes T))
			\\
			(X \otimes (Y \otimes Z)) \otimes T
			&&
			X \otimes ((Y \otimes Z) \otimes T)
			\arrow["\alpha_{X,Y,Z \otimes T}", from=1-2, to=1-3]
			\arrow["\alpha_{X \otimes Y,Z,T}", from=1-1, to=1-2]
			\arrow["\iid_X \otimes \alpha_{Y,Z,T}", from=2-3, to=1-3]
			\arrow["\alpha_{X,Y,Z} \otimes \iid_T", from=1-1, to=2-1]
			\arrow["\alpha_{X,Y \otimes Z,T}", from=2-1, to=2-3]
		\end{tikzcd}
	\]
	\[
		\begin{tikzcd}
			(X \otimes I) \otimes Y 
			&&
			X \otimes (I \otimes Y) 
			\\
			& X \otimes Y &
			\arrow["\alpha_{X,I,Y}", from=1-1, to=1-3]
			\arrow["\iid_X \otimes \lambda_Y", from=1-3, to=2-2]
			\arrow["\rho_X \otimes \iid_Y"', from=1-1, to=2-2]
		\end{tikzcd}
	\]
	commute.
\end{definition}

Diagrams such as the ones above are sometimes called \emph{coherence
conditions}, because they picture out loud conditions that one would expect to
have.

\begin{definition}
	A \emph{symmetric monoidal} category $\CC$ is a monoidal category equipped
	with a natural isomorphism $\sigma_{X,Y} \colon X \otimes Y \to Y \otimes X$,
	called the \emph{symmetry} such that, for all objects $X$, $Y$ and $Z$, the
	diagrams: 
	\[
		\begin{tikzcd}
			X \otimes I
			&&
			I \otimes X \\
			& X &
			\arrow["\sigma_{X,I}", from=1-1, to=1-3]
			\arrow["\rho_X"', from=1-1, to=2-2]
			\arrow["\lambda_X", from=1-3, to=2-2]
		\end{tikzcd}
	\]
	\[
		\begin{tikzcd}
			(X \otimes Y) \otimes Z 
			&&
			(Y \otimes X) \otimes Z \\
			X \otimes (Y \otimes Z)
			&&
			Y \otimes (X \otimes Z) \\
			(Y \otimes Z) \otimes X
			&&
			Y \otimes (Z \otimes X)
			\arrow["\sigma_{X,Y} \otimes \iid_Z", from=1-1, to=1-3]
			\arrow["\alpha_{X,Y,Z}", from=1-1, to=2-1]
			\arrow["\alpha_{Y,X,Z}", from=1-3, to=2-3]
			\arrow["\sigma_{X, Y \otimes Z}", from=2-1, to=3-1]
			\arrow["\iid_Y \otimes \sigma_{X,Z}", from=2-3, to=3-3]
			\arrow["\alpha_{Y,Z,X}", from=3-1, to=3-3]
		\end{tikzcd}
	\]
	commute, and for all objects $X$ and $Y$, $\sigma_{X,Y} \circ \sigma_{Y,X} =
	\iid$.
\end{definition}

\begin{example}
	Any category with finite products (see Definition~\ref{def:cat-product}) is,
	in particular, a symmetric monoidal category. Note that the converse is not
	true.
\end{example}

The example above covers many instances of categories, such as $\Set$ or
$\Vect$.

More examples of symmetric monoidal categories are given in the thesis. One
noticeable difference between a cartesian product and a monoidal product, is
that the monoidal one does not allow for copying in general. Indeed, with
products, the morphism $\pv \iid \iid \colon X \to X \times X$ necessarily
exists, whereas there is in general no canonical morphism $X \to X \otimes X$
in a monoidal category. This hints at the fact that symmetric monoidal
categories are the right tool to reason about a linear $\lambda$-calculus,
where each variable is used exactly once.

\subsection{Enriched categories}

Categories in computer science are usually \emph{locally small}, meaning that
given two objects $A$ and $B$, there is a \emph{set} of morphisms $A\to B$.
Enrichment is the study of the structure of those sets of morphisms, which
could be vector spaces or topological spaces for example, more details can be
found in~\cite{KELLY196515, kelly1982basic, maranda_1965}. 

\begin{definition}[Enriched Category]
	\label{def:enriched-cat}
	Given a monoidal category $(\VV, \otimes, I,\alpha, \lambda, \rho)$, a
	category $\CC$ \emph{enriched} in $\VV$ (sometimes called a $\VV$-category)
	is given by:
	\begin{itemize}
		\item a collection of objects of $\CC$;
		\item an object $\CC(X,Y)$ in $\VV$ for all objects $X$ and $Y$ in $\CC$;
		\item a morphism $\iid_X \colon I \to \CC(X,X)$ in $\VV$, for all objects
			$X$ in $\CC$, and that is called the \emph{identity};
		\item a morphism $\comp_{X,Y,Z} \colon \CC(Y,Z) \otimes \CC(X,Y) \to
			\CC(X,Z)$ in $\VV$ for all objects $X$, $Y$ and $Z$ in $\CC$, called the
			\emph{composition};
	\end{itemize}
	such that for all objects $X$, $Y$, $Z$ and $T$, the following diagrams:
	\[
		\begin{tikzcd}
			(\CC(Z,T) \otimes \CC(Y,Z)) \otimes \CC(X,Y)
			&&
			\CC(Y,T) \otimes \CC(X,Y) \\
			&& \CC(X,T) \\
			\CC(Z,T) \otimes (\CC(Y,Z) \otimes \CC(X,Y))
			&&
			\CC(Z,T) \otimes \CC(X,Z)
			\arrow["\alpha", from=1-1, to=3-1]
			\arrow["\comp_{Y,Z,T} \otimes \iid", from=1-1, to=1-3]
			\arrow["\iid \otimes \comp_{X,Y,Z}", from=3-1, to=3-3]
			\arrow["\comp_{X,Y,T}", from=1-3, to=2-3]
			\arrow["\comp_{X,Z,T}"', from=3-3, to=2-3]
		\end{tikzcd}
	\]
	\[
		\begin{tikzcd}
			I \otimes \CC(X,Y)
			&&
			\CC(Y,Y) \otimes \CC(X,Y) \\
			& \CC(X,Y) &
			\arrow["\iid_Y \otimes \iid_{\CC(X,Y)}", from=1-1, to=1-3]
			\arrow["\rho_{\CC(X,Y)}"', from=1-1, to=2-2]
			\arrow["\comp_{X,Y,Y}", from=1-3, to=2-2]
		\end{tikzcd}
	\]
	\[
		\begin{tikzcd}
			\CC(X,Y) \otimes I
			&&
			\CC(X,Y) \otimes \CC(X,X) \\
			& \CC(X,Y) &
			\arrow["\iid_{\CC(X,Y)} \otimes \iid_X", from=1-1, to=1-3]
			\arrow["\lambda_{\CC(X,Y)}"', from=1-1, to=2-2]
			\arrow["\comp_{X,X,Y}", from=1-3, to=2-2]
		\end{tikzcd}
	\]
	commute.
\end{definition}

\begin{example}
	A locally small category is $\Set$-enriched. This is obtained directly with
	the definition of category and the facts that homsets are sets, thus they are
	objects of the category $\Set$. The coherence conditions are satisfied thanks
	to the associativity of composition in a category and to the axioms of the
	identity morphism.
\end{example}

\begin{example}
	A cartesian closed category is enriched over itself. Indeed, composition is
	obtained by currying the following morphism:
	\[
		\begin{tikzcd}
			{[Y \to Z]} \times {[X \to Y]} \times X 
			&&&
			{[Y \to Z]} \times Y
			&&& 
			Z
			\arrow["\iid_{[Y \to Z]} \times \rmeval_{X,Y}", from=1-1, to=1-4]
			\arrow["\rmeval_{Y,Z}", from=1-4, to=1-7]
		\end{tikzcd}
	\]
	and it is routine to show that it verifies the coherence conditions.
\end{example}

\begin{definition}[Enriched Functor]
	Given two $\VV$-enriched categories $\CC$ and $\DD$, a $\VV$-enriched functor
	$F$ maps every object $X$ of $\CC$ to an object of $\DD$, written $FX$, and
	provides, for all objects $X$ and $Y$ in $\CC$, a morphism $F_{X,Y} \colon
	\CC(X,Y) \to \DD(FX,FY)$ in $\VV$ such that:
	\[
		F_{X,X} \circ \iid_X = \iid_{FX} 
		\qquad F_{X,Z} \circ \comp_{X,Y,Z} = \comp_{FX,FY,FZ} \circ (F_{Y,Z}
		\otimes F_{X,Y})
	\]
	for all objects $X$, $Y$ and $Z$ in $\CC$.
\end{definition}

\section{Fixed Points}
\label{sec:fixed-points}

This section introduces fixed point theorems that are relevant to this thesis;
namely, fixed points in partially ordered sets, which allow for the
interpretation of recursion and while loops, and initial algebras, which are a
canonical tool to provide a semantics for inductive and recursive data types. 

\subsection{Dcpos}
\label{sub:dcpo}

We work with the notion of partially ordered sets, usually called posets.

\begin{definition}[Directed subset]
	\label{def:directed}
	A non-empty subset $D$ of a poset $X$ is \emph{directed} if every pair of
	elements in the subset $D$ has an upper bound also in $D$. 
\end{definition}

\begin{definition}[Dcpo]
	\label{def:dcpo}
	A dcpo -- short for \emph{directed complete partial order} -- is a poset $X$
	such that every directed subset $D \subseteq X$ has a supremum in $X$. A
	\emph{pointed} dcpo $(X,\bot)$ is a dcpo $X$ that has a least element, that
	we usually write $\bot$. If $D$ is directed, we write $\sup D$ for its upper
	bound.
\end{definition}

\begin{example}
	The set of booleans $\B = \{ 0,1 \}$ with equality as an order is a dcpo.
	We can add a bottom element $\bot$, and we write $\B_\bot$ for the set
	$\{ \bot, 0, 1 \}$ with the following order:
	\[
		\bot \leq 0 \qquad \bot \leq 1.
	\]
	The partially ordered set $(\B_\bot, \leq)$ is a pointed dcpo. We say that
	the order is \emph{flat}. A partial order is sometimes drawn for a better
	view of its behaviour. The dcpo $\B_\bot$ is then pictured as:
	\[
		\begin{tikzcd}
			0 && 1 \\
			& \bot &
			\arrow[no head, from=1-1, to=2-2]
			\arrow[no head, from=1-3, to=2-2]
		\end{tikzcd}
	\]
	where $x \leq y$ iff there is a line between $x$ and $y$ and $x$ is
	\emph{below} $y$.
\end{example}

A common example of pointed dcpo used in the semantics of programming
languages, \emph{e.g.} PFC \cite{plotkin1977lcf}, is the flat dcpo of natural
numbers, also called by Plotkin \emph{the standard collection of domains for
arithmetic}. This is given in the next example.

\begin{example}
	The flat dcpo of natural numbers $\N_\bot$ is given by:
	\[
		\begin{tikzcd}
			0 & 1 & \dots & n & \dots \\
			&& \bot &&
			\arrow[no head, from=1-1, to=2-3]
			\arrow[no head, from=1-2, to=2-3]
			\arrow[no head, from=1-3, to=2-3]
			\arrow[no head, from=1-4, to=2-3]
			\arrow[no head, from=1-5, to=2-3]
		\end{tikzcd}
	\]
	Its directed subsets are simple: they are either sigletons or $\{ \bot, n\}$
	with $n \in \N$.
\end{example}

Next, we define functions between dcpos called \emph{Scott continuous} after
Dana Scott. He is famous for his contribution to \emph{domain theory}, which
encompasses all the content of this subsection about dpcos. A detailed account
of domain theory can be found in \cite{abramsky95domain}. A continuous function
needs first to be monotone -- \emph{i.e.}, given $x \leq y$, then $f x \leq f
y$; what French speakers would prefer to call an \emph{increasing}
function.

\begin{definition}[Scott continuous]
	Given two dcpos $X$ and $Y$, a monotone function $f \colon X \to Y$ is
	\emph{Scott continuous} if for every directed subset $D \subseteq X$, $f(\sup
	D) = \sup f(D)$.
\end{definition}

\begin{example}
	The function $f \colon \N_\bot \to \N_\bot$ defined as:
	\[
		f = \left\{ \begin{array}{lcl}
			n &\mapsto & n+1 \\
			\bot &\mapsto & \bot
		\end{array}
		\right.
	\]
	is Scott continuous.
\end{example}

If $X$ and $Y$ are dcpos, then the set of Scott continuous functions
$X \to Y$ is also a dcpo.

Note that dcpos and Scott continuous maps form a category $\DCPO$. The category
of pointed dcpos and Scott continuous maps is written $\dcpob$. The category
$\dcpob$ is cartesian closed, which means that it can be used for the
denotational semantics of a $\lambda$-calculus.

\begin{theorem}[Kleene's Fixed Point \cite{stoltenberg1994domains}]
	\label{th:kleene}
	If $(X,\bot)$ is a pointed dcpo and $f \colon X \to X$ is a Scott continuous
	function, then the function $f$ has a least fixed point, obtained as $\fix f
	= \sup \set{f^n(\bot) \alt n \in \N}$.
\end{theorem}

\paragraph{Recursion.}
We can add a new term to our $\lambda$-calculus to capture \emph{recursion}.
With the same types given in (\ref{eq:simple-types}), we add to the grammar in
(\ref{eq:simple-terms}) the following:
\begin{equation}
	\label{eq:fix-term}
	M,N,\dots ::= \cdots \alt \ffix M
\end{equation}
with the typing rule given below.
\[
	\infer{\Gamma \vdash \ffix M \colon A}{\Gamma \vdash M \colon A \to A}
\]
The operational behaviour of this new term is as expected:
\[
	\ffix M \to M (\ffix M)
\]
and its denotational semantics in $\dcpob$ is given by Kleene's fixed point
(see Theorem~\ref{th:kleene}), taken pointwise. Details can be found in the
original paper by Plotkin \cite{plotkin1977lcf}. We can easily see that this
semantics is sound, namely
\[
	\sem{\ffix M} = \sem{M (\ffix M)}.
\]

\subsection{Initial Algebras}
\label{sub:inductive-types}

Inductive data types are written in the syntax as some least fixed point of a
type judgement, \emph{e.g.} $\mu X . A$. As an example, the type of natural
numbers is given by $\mu X . 1 + X$ and the type of trees, whose nodes have
type $A$, is $\mu X . 1 + (X \times A \times X)$. As said earlier, types are
represented as objects in the category; but handling inductive types means that
we need to handle type variables too, and thus a type judgement $\Theta \vdash
A$, where $\Theta$ is a set of type variables, is an \emph{object mapping}, or
rather a functor. We show below how to consider fixed points of functors in our
categorical setting. 

\begin{definition}[Algebra]
	\label{def:initial-algebra}
	Given a functor $F \colon \CC \to \CC$, a pair $(X,f)$ composed of an object
	$X$ and a morphism $f \colon FX \to X$ is a called an \emph{$F$-algebra}.
	Given two $F$-algebras $(X,f)$ and $(Y,g)$, a morphism $h \colon X \to Y$
	is an \emph{$F$-algebra homomorphism} if the following diagram commutes:
	\[
		\begin{tikzcd}
			FX & X \\
			FY & Y
			\arrow["f", from=1-1, to=1-2]
			\arrow["g", from=2-1, to=2-2]
			\arrow["Fh"', from=1-1, to=2-1]
			\arrow["h", from=1-2, to=2-2]
		\end{tikzcd}
	\]
	When the category of $F$-algebras and $F$-algebra homomorphisms has an
	initial object, the latter is called the \emph{initial $F$-algebra}.
\end{definition}

\begin{lemma}[Lambek's lemma \cite{adamek2018fixed}]
	\label{lem:lambek}
	Given an endofunctor $F:\CC\to\CC$ and an initial $F$-algebra
	$(X,\alpha\colon FX\to X)$, then $\alpha$ is an isomorphism.
\end{lemma}

With Lambek's lemma, we know that an initial algebra provides an object $X$
such that $X\cong FX$. Therefore we can see that the object $X$ is a fixed
point of the endofunctor $F$, as requested. However, we need a stonger notion
of algebra for the denotational semantics of inductive types. Hence the next
definition.

\begin{definition}[Parameterised Initial Algebra]
	\label{def:para-initial-algebra}
	Given two categories $\CC$ and $\DD$ and a functor $F \colon \CC \times \DD
	\to \DD$, a \emph{parameterised initial algebra} for $F$ is a pair $(F\nnoma,
	\phi^F)$, such that:
	\begin{itemize}
		\item $F\nnoma \colon \CC \to \DD$ is a functor;
		\item $\phi^F \colon F \circ \pv{\iid}{F\nnoma} \natto F\nnoma \colon \CC \to
			\DD$ is a natural isomorphism;
		\item for every object $X$ in $\CC$, the pair $(F\nnoma,\phi^F_X)$ is an
			initial $F(X,-)$-algebra.
	\end{itemize}
\end{definition}

\begin{remark}
	Observe that the previous definition with $\CC = 1$, the category with one
	object and the identity, we recover Definition~\ref{def:initial-algebra}.
	The notion of parameterised initial algebra is then more general.
\end{remark}

The existence of parameterised initial algebras is given by the theorem found
in \cite[Corollary~7.2.4]{fiore04axiomatic} and recalled in
Theorem~\ref{th:param-alg}.

\begin{definition}
	\label{def:alg-compact}
	A category $\CC$ is \emph{parameterised $\DCPO$-algebraically complete}
	if all functors as described in Definition~\ref{def:para-initial-algebra}
	admit a parameterised initial algebra.
\end{definition}

In the following, we present sufficient conditions, outlined by Fiore in
\cite{fiore04axiomatic}, for a category to be $\dcpo$-algebraically complete.

\begin{definition}[Ep-pair \cite{fiore04axiomatic}]
	\label{def:ep-pair}
	Given a $\dcpo$-category $\CC$, a morphism $e \colon X \to Y$ in $\CC$ is
	called an \emph{embedding} if there exists a morphism $p \colon Y \to X$ such
	that $p \circ e = \iid_X$ and $e \circ p \leq \iid_Y$. The morphisms $e$ and
	$p$ form an \emph{embedding-projection pair} $(e,p)$, also called
	\emph{ep-pair}.
\end{definition}

\begin{remark}
	Similarly to embedding, a morphism $p$ is called a \emph{projector} if it is
	part of an ep-pair $(e,p)$.
\end{remark}

We recall that an \emph{ep-zero} \cite[Definition~7.1.1]{fiore04axiomatic} is
an object $0$ such that:
\begin{itemize}
	\item $0$ is an initial object;
	\item given any morphism $f \colon 0 \to Y$, $f$ is an embedding;
	\item $0$ is a terminal object;
	\item given any morphism $g \colon X \to 0$, $g$ is a projector.
\end{itemize}

\begin{theorem}[\cite{fiore04axiomatic}]
	\label{th:param-alg}
	A $\DCPO$-category $\CC$ with an ep-zero and colimits of $\omega$-chains
	of embeddings is parameterised $\DCPO$-algebraically complete.
\end{theorem}

Actually, a category that verifies the conditions above has stronger
properties: it is parameterised $\DCPO$-algebraically $\omega$-compact, namely
it has parameterised initial algebras and parameterised final coalgebras for
all $\dcpo$-functors. However, we do not need such a strong result in this
thesis.

The latest theorem above means that any $F$ as introduced in
Definition~\ref{def:para-initial-algebra} admits a parameterised initial
algebra given that $\CC$ is a $\DCPO$-category with an ep-zero and colimits of
$\omega$-chains of embeddings.

\paragraph{Inductive types.}
We take an example inspired from the metalanguage FPC
\cite{gunter1992semantics}, with details in \cite[Chapter~8]{fiore04axiomatic}.
We are given the following types:
\begin{equation}
	\label{eq:fpc-types}
	A,B ~::=~ X \alt A + B \alt A \otimes B \alt \alt \mu X . A
\end{equation}
with the typing rules:
\[
	\begin{array}{c}
		\infer{\Theta, X \vdash X}{}
		\qquad
		\infer[\star \in \set{+,\otimes}]{\Theta \vdash A \star B}{
			\Theta \vdash A
			&
			\Theta \vdash B
		}
		\qquad
		\infer{\Theta \vdash \mu X . A}{\Theta, X \vdash A}
	\end{array}
\]

Their semantics is given in a symmetric monoidal $\dcpo$-enriched category
$\CC$ with coproducts that verifies the hypothesis of
Theorem~\ref{th:param-alg}. The semantics of a type context $\Theta$ is
$\sem\Theta = \CC^{\abs\Theta}$ and thus the interpretation of a type judgement
$\Theta \vdash A$ is given by a functor $\sem{\Theta \vdash A} \colon
\CC^{\abs\Theta} \to \CC$. This interpretation is defined by induction on the
typing rules, and the only non-trivial case is the fixed point constructor,
whose semantics is given by:
\[
	\sem{\Theta \vdash \mu X . A} = \sem{\Theta, X \vdash A}\nnoma
\]
where $(-)\nnoma$ is defined in Definition~\ref{def:para-initial-algebra}.

\section{Restriction and Inverse Categories}
\label{sec:restriction}

A significant part of this thesis is the study of reversiblity in programming
languages and in their semantics. This section is concerned with presenting the
categorical tools in the literature that formalise invertibility. Note that
\emph{reversible} does not mean \emph{bijective}: a partial injection is
reversible, in the sense that all of its possible outputs have a unique and
deterministic corresponding input. Let us say more about partial injections.
Given two sets $X$ and $Y$, the image of a partial function $f$ is given as
follows: $f(x)$ if $x$ is in the domain of $f$, and $\bot$ or \emph{undefined}
if $x$ is not in the domain of $f$. A simple example would be the partial
injection $f \colon \set{0,1} \to \set{0,1}$ such that $f(0) = 1$ and $f$ is
undefined on $1$. Sets and partial injections form a category $\PInj$, which is
the canonical example among restriction and inverse categories, the focus of
this section.

For further reading, the author recommends the original work of Guo
\cite{guo2012products}, Giles \cite{giles2014investigation} and Kaarsgaard
\cite{kaarsgaard2017thesis}.

\subsection{Basic structure}
\label{sub:inv-basic}

The axiomatisation of inverse categories gives the conditions for the morphisms
of a category to be \emph{partial injections}. First, the notion of
restriction allows us to capture the \emph{actual} domain of a morphism through a
partial identity function. Historically, \emph{inverse} categories
\cite{kastl1979inverse} were introduced before \emph{restriction} categories,
but the latter are more convenient to introduce the subject.

\begin{definition}[Restriction~\cite{cockett2002restriction-I}]
	\label{def:restr} 
	A \emph{restriction} structure is an operator that maps each morphism $f \colon
	X\rightarrow Y$ to a morphism $\res f \colon X\rightarrow X$ such that for
	all $g \colon X \to Z$ and $h \colon Y \to T$, we have:
	\begin{align*}
		&\ f \circ \res f = f,  &\quad&
		\res f \circ \res g = \res g \circ \res f, \\
		&\ \res{f \circ \res g} = \res f \circ \res g, &\quad&
		\res h \circ f = f \circ \res{h \circ f}. 
	\end{align*}
	A morphism $f$ is said to be \emph{total} if $\res f = \iid_X$. A category
	with a restriction structure is called a \emph{restriction category}.
\end{definition}

\begin{remark}
	Note that the definition implies that for all $f \colon X \to Y$, there is a unique
	$\res f \colon X\rightarrow X$.
\end{remark}

\begin{example}
	Given two sets $X$ and $Y$, any partial injection $f \colon X \to Y$ defined
	on a subset $X'\subseteq X$ and undefined on $X\setminus X'$ is given as
	follows:
	\[
		\left\{ \begin{array}{ll}
			f(x) & \text{when } x \in X' \\
			\text{undefined} & \text{when } x \notin X'
		\end{array} \right.
	\]
	Then, the restriction of $f$, the morphism,
	$\res f \colon X \to X$, is given by:
	\[
		\left\{ \begin{array}{ll}
			x & \text{when } x \in X' \\
			\text{undefined} & \text{when } x \notin X'
		\end{array} \right.
	\]
	which is the identity on $X'\subseteq X$ and undefined on $X \setminus X'$.
	This example shows that $\PInj$ is a restriction category.
\end{example}

\begin{definition}[Restriction Functor~\cite{cockett2002restriction-I}]
	\label{def:restr-cat}
	Given two restriction categories $\CC$ and $\DD$, a functor $F \colon \CC
	\rightarrow \DD$ is a \emph{restriction functor} if $\res{F(f)} =
	F \res f$ for all morphism $f$ of $\CC$. The definition is
	canonically extended to bifunctors. 
\end{definition}

To interpret reversibility, we need to introduce a notion of
reversed process, a process that exactly reverses another process.
This is given by a generalised notion of inverse.

\begin{definition}[Inverse category]
  \label{def:invcat}
	An \emph{inverse category} is a restriction category where all morphisms are
	partial isomorphisms; meaning that for $f \colon X \to Y$, there exists
	$f^{\circ} \colon Y \to X$ such that $f^{\circ} \circ f = \res f$ and
	$f \circ f\pinv = \rc f$, and is called the \emph{partial inverse}. 
\end{definition}

\begin{lemma}
	In an inverse category, the partial inverse $f^{\circ} \colon Y \to X$
	of a morphism $f \colon X \to Y$ is unique.
\end{lemma}
\begin{proof}
	Assume there exists two morphisms $g,h \colon \colon Y \to X$ such that
	\[
		gf = \res f \qquad fg = \res g \qquad hf = \res f \qquad fh \res h.
	\]
	Therefore, we have $gf = \res f = hf$. We have then:
	\begin{align*}
		g \res h &= g f h &\text{(hyp. above)} \\
		&= h f h &\text{(observation above)} \\
		&= \res f h &\text{(hyp. above)} \\
		&= h \res{fh} & \text{(Def. \ref{def:restr}, eq. 4)} \\
		&= h \res{\res h} &\text{(hyp. above)} \\
		&= h \res h & \\
		&= h & \text{(Def. \ref{def:restr}, eq. 1)} 
	\end{align*}
	Thus, $g \res h = h$. Symmetrically, we can prove that $h \res g = g$. We then obtain:
	\begin{align*}
		\res h &= \res{g \res h} &\text{(observation above)} \\
		&= \res g \res h & \text{(Def. \ref{def:restr}, eq. 3)} \\
		&= \res h \res g & \text{(Def. \ref{def:restr}, eq. 2)} \\
		&= \res{h \res g} & \text{(Def. \ref{def:restr}, eq. 3)} \\
		&= \res{g} & \text{(Def. \ref{def:restr}, eq. 1)} 
	\end{align*}
	Finally, we have $h = g \res h = g \res g = g$.
\end{proof}

\begin{remark}
	Given an inverse category $\CC$, $(-)\pinv$ is actually a contravariant
	functor $\CC\op \to \CC$. We can also observe that if $\CC$ is an inverse
	category, then $\CC\op$ is also.
\end{remark}

\begin{example}
	\label{ex:simple-pinj}
	In $\PInj$, let us consider the partial function $f\colon\{0,1\}\to\{0,1\}$
	as $f(0)=1$ and undefined on $1$. Its restriction $\res f$
	is undefined on $1$ also but $\res f(0)=0$. Its \emph{inverse}
	$f^\circ$ is undefined on $0$ and such that $f^\circ(1)=0$.
\end{example}

The example above generalises and $\PInj$ is an actual inverse category. Even
more, it is \emph{the} inverse category: \cite{kastl1979inverse} proves that
every locally small inverse category is isomorphic to a subcategory of $\PInj$.

\begin{definition}[Restriction compatible~\cite{guo2012products}]
  \label{def:restr-compati}
	Two morphisms $f,g \colon X\to Y$ in a restriction category $\CC$ are
	\emph{restriction compatible} if $f\res g = g\res f$. The relation is written
	$f \smile g$. 
\end{definition}

\begin{example}
	\label{ex:pinj-compati}
	We build upon Example~\ref{ex:simple-pinj}, consider an additional partial
	injection $g \colon \set{0,1} \to \set{0,1}$ such that $g(0) = 1$ and $g(1) =
	0$. We have $f \smile g$. The morphism $g$ is \emph{more defined} than $f$,
	and this intuition is made precise by the next definition.
\end{example}

\begin{definition}[Partial order~\cite{cockett2002restriction-I}]
	\label{def:restr-order}
	Let $f,g \colon X\to Y$ be two morphisms in a restriction category. We then
	define $f \leq g$ as $g\res f = f$.
\end{definition}

The next lemma -- which is rather an observation -- links the latest introduced
notions of compatibility and order between maps in a restriction category.

\begin{lemma}
	\label{lem:order-compati}
	Given $\CC$ a restriction category and two morphisms $f,g \colon X \to Y$ in
	$\CC$, if $f \leq g$ then $f \smile g$.
\end{lemma}
\begin{proof}
	Remember that $f \leq g$ means that $g \res f = f$. We can precompose by
	$\res g$ and get:
	\begin{align*}
		f \res g &= g \res f \res g &\ \\
		&= g \res g \res f & \text{(Def. \ref{def:restr}, eq. 2)} \\
		&= g \res f & \text{(Def. \ref{def:restr}, eq. 1)} 
	\end{align*}
	and $f \res g = g \res f$ is the definition of compatibility.
\end{proof}

\begin{definition}[Inverse compatible~\cite{kaarsgaard2017join}]
	Given $\CC$ an inverse category, $f,g \colon X \to Y$ in $\CC$ are inverse
	compatible if $f\smile g$ and $f^{\circ} \smile g^{\circ}$, noted $f\asymp
	g$.
\end{definition}

\begin{definition}
  A set $S$ of morphisms of the same type $A\to B$ is restriction compatible
  (\textit{resp.}  inverse compatible) if all elements of $S$ are
  pairwise restriction compatible (\textit{resp.} inverse compatible).
\end{definition}

This thesis makes use only of inverse categories, but note that most of the
definitions below have a counterpart for restriction categories.

\begin{definition}[Joins~\cite{guo2012products}]
	\label{def:join}
	An inverse category $\CC$ is equipped with \emph{joins} if for all inverse
	compatible sets $S$ of morphisms $X \to Y$, there exists a morphism in $\CC$
	written $\bigvee_{s\in S} s \colon X \to Y$ such that, for all $t \colon X
	\to Y$ and for all $s\in S$, $s\leq t$, the following holds:
  \begin{align*}
    s &\leq \bigvee\limits_{s\in S} s,
    & \bigvee\limits_{s\in S} s&\leq t,
    &\res{\bigvee\limits_{s\in S} s}
    &=\bigvee\limits_{s\in S} \res s, \\
    f\circ\left(\bigvee\limits_{s\in S} s\right)
    &= \bigvee\limits_{s\in
      S} fs,
    & \left(\bigvee\limits_{s\in
      S} s\right)\circ g
    &=\bigvee\limits_{s\in S} sg.
  \end{align*}
	Such a category is called a \emph{join inverse category}.
\end{definition}

\begin{example}
	The category $\PInj$ is a join inverse category. Following
	Example~\ref{ex:pinj-compati}, $f$ and $g$ are inverse compatible and $f \vee
	g = g$.
\end{example}

Building up from Definition~\ref{def:restr}, a \emph{join restriction
functor} is a restriction functor that preserves all thus constructed joins.

\begin{remark}[Zero]
	\label{rem:zero}
	Given a join inverse category $\CC$, and since $\emptyset \subseteq \CC(X,Y)$
	with all of its elements that are inverse compatible, there exists a morphism
	$0_{X,Y} \doteq\bigvee_{s\in\emptyset} s \colon X \to Y$, called \emph{zero
	map}.  It satisfies the following equations, for all $f \colon Y \to Z$ and
	$g \colon Z \to X$: 
	\[
		f \circ 0_{X,Y}=0_{X,Z} 
		\qquad 0_{X,Y} \circ g=0_{Z,Y} 
		\qquad 0_{X,Y}^{\circ} = 0_{Y,X}
		\qquad \res{0_{X,Y}} = 0_{X,X}.
	\]
	Moreover, $0_{X,Y}$ is the least element
	in $\CC(X,Y)$ for the order introduced above, in
	Definition~\ref{def:restr-order}.
\end{remark}

\begin{example}
	Given two sets $X,Y$, the morphism $0_{X,Y} \colon X \to Y$ mentioned above
	in $\PInj$ is the partial injection $X \to Y$ that is defined nowhere.
\end{example}

\begin{lemma}
	\label{lem:res-zero-zero}
	Given an inverse category $\CC$ and a morphism $f \colon X \to Y$ such that
	$\res f = 0_{X,X}$, then $f = 0_{X,Y}$.
\end{lemma}
\begin{proof}
	By Definition~\ref{def:restr}, we know that $f = f \res f$, and thus $f = f \res f
	= f 0_{X,X} = 0_{X,Y}$.
\end{proof}

\begin{lemma}
	\label{lem:inv-ortho-compati}
	Given an inverse category $\CC$ and two morphisms $f, g \colon X \to Y$ in
	$\CC$ such that $f\pinv g = 0_{X,X}$ and $fg\pinv = 0_{Y,Y}$, then $f$ and
	$g$ are inverse compatible.
\end{lemma}
\begin{proof}
	Our goal is to prove that $f \res g = g \res f$, but we are going to prove
	that both are equal to zero. Definition~\ref{def:restr} ensures that $\res{f \res
	g} = \res f \res g$, and then Definition~\ref{def:invcat} gives that $\res f \res g
	= f\pinv f g\pinv g$, which is then equal to $f\pinv 0_{Y,Y} g$ by
	hypothesis, and thus $\res{f \res g} = \res f \res g = 0_{X,X}$;
	Lemma~\ref{lem:res-zero-zero} ensures then that $f \res g = 0_{X,Y}$. Note
	that $\res{g \res f} = \res f \res g = \res g \res f$ by
	Definition~\ref{def:restr}, thus $g \res f = 0_{X,Y}$ too. We have proven that $f
	\smile g$. The proof that $f\pinv \smile g\pinv$ is similar.
\end{proof}

\subsection{Additional Structure}
\label{sub:inv-additional}

\begin{definition}[Restriction Zero]
	An inverse category $\CC$ has a \emph{restriction zero} object $0$
	iff for all objects $X$ and $Y$, there exists a unique morphism $0_{X,Y}
	\colon X \to Y$ that factors through $0$ and satisfies $\res{0_{X,Y}}
	=0_{X,X}$.
\end{definition}

\begin{remark}
	In a join inverse category with a restriction zero, the zero morphisms
	given by the join structure (see Remark~\ref{rem:zero}) coincide with the
	ones of the previous definition.
\end{remark}

\begin{lemma}
	\label{lem:restr-ep-zero}
	Given an inverse category $\CC$ with a restriction zero $0$ and that is
	$\DCPO$-enriched, $0$ is an ep-zero.
\end{lemma}
\begin{proof}
	The restriction zero $0$ is an initial and terminal object by definition.
	The pairs of embedding and projector are given by any morphism and its
	inverse $-\pinv$.
\end{proof}

\begin{definition}[Disjointness tensor~\cite{giles2014investigation}]
  \label{def:dis-tensor}
	A join inverse category $\CC$ is said to have a \textit{disjointness tensor} if
	it is equipped with a symmetric monoidal restriction bifunctor $(- \oplus -)
	\colon \CC\times\CC\rightarrow\CC$, with as unit a restriction zero $0$ and
	for all objects $X$ and $Y$, morphisms $\iota_l \colon X\rightarrow X \oplus
	Y$ and $\iota_r \colon Y\rightarrow X \oplus Y$ that are total, jointly epic,
	and such that their inverses are jointly monic and $\rc{\iota_l}~\rc{\iota_r}
	= 0_{X \oplus Y}$. The morphisms $\iota$ are called \emph{injections}.
\end{definition}

\begin{remark}
	A precise way of writing the injections would be $\iota_l^{X,Y} \colon X \to
	X \oplus Y$ and $\iota_r^{X,Y} \colon Y \to X \oplus Y$. However, we choose
	to loosen the notations when it is not ambiguous. This choice is motivated by
	the recurrent use of the contravariant functor $(-)\pinv$, and $\iota_l\pinv$
	appears to be more readable than $(\iota_l^{X,Y})\pinv$.
\end{remark}

The last requirement in the previous definition can be described as the
injections having \emph{orthogonal} outputs; with a set theoretic vocabulary,
we could say that their images have an empty intersection. We show in the next
lemma that this intuition is verified, with a simpler presentation than in
Definition~\ref{def:dis-tensor}.

\begin{lemma}
	\label{lem:orthogonal-injections}
	In a join inverse category $\CC$ with disjointness tensor, for all
	objects $X$ and $Y$, we have $\iota_l\pinv \circ \iota_r = 0_{Y,X}$.
\end{lemma} 
\begin{proof}
	We know from Definition~\ref{def:dis-tensor} that $\rc{\iota_l}~\rc{\iota_r}
	= 0_{X \oplus Y}$, thus by postcomposition
	$\iota_l\pinv~\rc{\iota_l}~\rc{\iota_r} = 0_{X \oplus Y, X}$. Since $\CC$ is
	a restriction category (see Definition~\ref{def:restr}), $\iota_l\pinv~\rc{\iota_l}
	= \iota_l\pinv$, and thus $\iota_l\pinv~\rc{\iota_r} = 0_{X \oplus Y, X}$.
	With Definition~\ref{def:invcat}, we have $\rc{\iota_r} =
	\iota_r\iota_r\pinv$, thus $\iota_l\pinv \iota_r \iota_r\pinv = 0_{X \oplus
	Y, X}$. Finally, we precompose by $\iota_r$, to obtain $\iota_l\pinv \iota_r
	\iota_r\pinv \iota_r = 0_{Y,X}$. Definition~\ref{def:invcat} again tells us
	that $\iota_r\pinv ~\iota_r = \res{\iota_r}$, and also $\iota_r~\res{\iota_r}
	= \iota_r$, hence the equality: $\iota_l\pinv \iota_r = 0_{Y,X}$.
\end{proof}

This orthogonality between morphisms introduced by the disjointness tensor
shows how pattern-matching could be handled. The author of this thesis has
proven in \cite{nous2021invcat} that the disjointness tensor is not the only
way of performing pattern-matching within an inverse category, however it
remains the canonical way.

\begin{remark}
	\label{rem:loose-inner-product}
	As we have started to describe the previous lemma as an \emph{orthogonality}
	assertion, we are only a few inches away from defining an \emph{inner
	product}, in a vector space fashion -- although we are definitely not working
	with vector spaces. In the same scenery as
	Lemma~\ref{lem:orthogonal-injections}, given two morphisms $f \colon X \to Z$
	and $g \colon Y \to Z$, we allow ourselves to call the morphism $f\pinv \circ
	g \colon Y \to X$ the inner product of $f$ and $g$, in a very loose way. Note
	that $g\pinv \circ f \colon X \to Y$ can also be said to be their inner
	product, and is not the same morphism in general.
\end{remark}

A programming language usually involves pairs or tuples of terms, often denoted
with a symmetric monoidal category. The next definition proposes a definition
of a model for a (simple) reversible programming language handling
pattern-matching.

\begin{definition}[Rig]
	\label{def:rig}
	Let us consider a join inverse category equipped with a symmetric monoidal
	tensor product $(\otimes,1)$ and a disjointness tensor $(\oplus,0)$ that are
	join preserving, and such that there are isomorphisms $\delta_{A,B,C} : A
	\otimes (B\oplus C) \rightarrow (A\otimes B)\oplus (A\otimes C)$ and $\nu_A
	\colon A \otimes 0 \rightarrow 0$. This is called a \emph{join inverse rig
	category}.
\end{definition}

It is proven in \cite{kaarsgaard2017join} that a join inverse category can be
considered enriched in $\DCPO$ without loss of generality, showing that there
is a way of working with fixed points and general recursion in reversible
settings. The next results prove that the operations involved in a reversible
programming language preserve this $\dcpo$-enriched structure.

\begin{lemma}[\cite{kaarsgaard2017join}]
	\label{lem:restr-fun-dcpo}
	Let $\CC$ and $\DD$ be join inverse rig $\dcpo$-categories, and $F \colon \CC
	\to \DD$ be a join preserving restriction functor. Then $F$ is a
	$\DCPO$-functor.
\end{lemma}

A conclusion of this lemma is that the functors used to interpret simple data
types in a programming language are $\dcpo$-functors.

\begin{corollary}
	\label{cor:dcpofun}
  Let $\CC$ be a join inverse rig category. The functors
		$(-\otimes -): \CC\times\CC\to\CC$ and
		$(-\oplus -): \CC\times\CC\to\CC$
  are $\DCPO$-functors.
\end{corollary}

Separately, it is important that the inverse structure also preserves the
enrichment.

\begin{lemma}[\cite{kaarsgaard2017join}]
	Let $\CC$ be a join inverse rig category. The functor
	$(-)\pinv:\CC^{op}\to\CC$ is a $\dcpo$-functor.
\end{lemma}

Finally, we import a result from the literature ensuring that a join
inverse rig category can safely be generalised to carry the interpretation of
inductive types.

\begin{proposition}[\cite{kaarsgaard2017join}]
	\label{prop:omega-chains}
	Any join inverse rig category can be faithfully embedded in a join
	inverse rig category with colimits of $\omega$-chains of embeddings. 
\end{proposition}

This shows, with Lemma~\ref{lem:restr-ep-zero}, that a join inverse rig
$\dcpo$-category can verify the hypotheses of Theorem~\ref{th:param-alg}
whithout loss of generality; and thus this kind of category is a model for
inductive types. The details of join inverse rig $\dcpo$-categories used
as a denotational model is found in Chapter~\ref{ch:reversible}.

\section{Hilbert spaces}
\label{sec:back-hilbert}

One of the main focus in the study of mathematical quantum mechanics is Hilbert
spaces. We assume basic knowledge of linear algebra, such as: vectors, linear
maps, bases, kernels, \emph{etc.} The author recommends the book written by
Heunen and Vicary \cite{heunen2019categories} for a more detailed introduction
to Hilbert spaces, and to category theory applied to quantum computing.

\subsection{Introductory Definitions}
\label{sub:hilb-intro}

Formally, a Hilbert space is a complex vector space equipped with an inner product,
written $\ip - -$, such that this inner product induces a complete metric space. The
inner product is used to compute probabilities of measurement outcomes in
quantum theory. Note that, given a complex number $\alpha$, we write
$\res\alpha$ for its conjugate.

\begin{definition}[Inner product]
	An inner product on a complex vector space $V$ is a function $\ip - - \colon
	V \times V \to \C$, such that:
	\begin{itemize}
		\item for all $x,y \in V$, $\ip x y = \res{\ip y x}$;
		\item for all $x,y,z \in V$ and $\alpha \in \C$,
			\[
				\ip{x}{\alpha y} = \alpha \ip x y,
				\qquad
				\ip{x}{y + z} = \ip x y + \ip x z;
			\]
		\item for all $x \in V$, $\ip x x \geq 0$, and if $\ip x x = 0$, then
			$x = 0$.
	\end{itemize}
	Given a vector space with an inner product, the canonical norm of a vector
	$x$ is defined as $\Vert x \Vert \defeq  \sqrt{\ip x x}$.
\end{definition}

This is enough to state the definition of a Hilbert space.

\begin{definition}[Hilbert space]
	A Hilbert space is a complex vector space $H$ equipped with an inner product
	such that $H$ is \emph{complete} with regard to its canonical norm. By
	complete, we mean that: if a sequence of vectors $(v_i)_{i \in \N}$ is such
	that $\sum_{i = 0}^\infty \Vert v_i \Vert < \infty$, then there exists a
	vector $v \in H$ such that $\Vert v - \sum_{i = 0}^\infty v_i \Vert$ tends to
	zero as $n$ goes to the infinity. The vector $v$ is called a \emph{limit}.
\end{definition}

It is interesting to observe that all finite-dimentional vector spaces with an
inner product are complete. Moreover, any vector space with an inner product
can be \emph{completed}, by adding the adequate limit vectors.

\begin{remark}
	A basis in a Hilbert space is not exactly defined the same way as a basis in
	a vector space. A basis in a Hilbert space is such that any vector is
	\emph{limit} of linear combinations of the elements of the basis. An
	orthonormal basis in a Hilbert space is such that its elements are pairwise
	orthogonal, have norm $1$, and their linear span is dense in the Hilbert
	space.
\end{remark}

\begin{definition}[Bounded linear map]
	Given two Hilbert spaces $H_1$ and $H_2$, a linear map $f \colon H_1 \to H_2$
	is \emph{bounded} if there exists $\alpha \in \R$ such that $\Vert f x \Vert
	\leq \alpha \Vert x \Vert$ for all $x \in H_1$.
\end{definition}

\subsection{Additional Structure}
\label{sub:hilb-additional}

We make use of differents kinds of structure in vector spaces, such as direct
sums and tensor products.

\begin{definition}[Direct sum]
	Given two complex vector spaces $V$ and $W$, one can form their \emph{direct
	sum} $V \oplus W$, whose elements are $(v,w)$ with $v \in V$ and $w \in W$,
	such that, for all $v,v' \in V$ and $w,w' \in W$ and $\alpha,\beta \in \C$,
	$\alpha (v,w) + \beta (v',w') = (\alpha v + \beta v', \alpha w + \beta w')$.
\end{definition}

\begin{remark}
	$V \oplus \set 0$ is isomorphic to $V$, and given $v \in V$, the vector
	$(v,0)$ can be written $v$ when there is no ambiguity. Given $v \in V$ and $w
	\in W$, the vector $(v,w)$ can sometimes be written $v + w$.
\end{remark}

Hilbert spaces are closed under direct sums, with the following inner product:
$\ip{(x_1,y_1)}{(x_2,y_2)}_{X \oplus Y} = \ip{x_1}{x_2}_{X} +
\ip{y_1}{y_2}_{Y}$. However, this is not true for the tensor product: the
\emph{linear algebraic} tensor product of two Hilbert spaces is not necessarily
a Hilbert space. We explain below how we can get around this issue with
\emph{completion}.

\begin{definition}[Tensor product]
	Given two complex vector spaces $V, W$, there is a vector space $V \otimes
	W$, together with a bilinear map $- \otimes- \colon V \times W \to V \otimes
	W :: (v,w) \mapsto v \otimes w$, such that for every bilinear map $h \colon V
	\times W \to Z$, there is a unique linear map $h' \colon V \otimes W \to Z$,
	such that $h = h' \circ \otimes$.
\end{definition}

The tensor product of two Hilbert spaces $X$ and $Y$ is obtained through the
tensor product of the underlying vector spaces, with the inner product $\ip{x_1
\otimes y_1}{x_2 \otimes y_2}_{X \otimes Y} = \ip{x_1}{x_2}_X \ip{y_1}{y_2}_Y$
and then the completion of this space gives the desired Hilbert spaces. We
abuse notation and write $X \otimes Y$ for the resulting Hilbert space.

The category of Hilbert spaces and bounded linear maps between them, written
$\Hilb$, admits several different monoidal structures: with $(\otimes, \C)$ and
with $(\oplus, \set 0)$ -- the latter is in fact a biproduct. They even give it
a \emph{rig} structure, in the sense of Definition~\ref{def:rig}. Classical
computers operate on bits, while quantum computers apply operations on qubits,
written $\ket 0$ and $\ket 1$. They are usually denoted as vectors in the
Hilbert space $\C \oplus \C$ with $\ket 0 \defeq (1,0)$ and $\ket 1 \defeq
(0,1)$ the elements of its canonical basis.

The next lemma outlines another important structure to Hilbert spaces, called
the \emph{adjoint}. It is due to Frigyes Riesz
[\textipa{'fri{\textbardotlessj}ES 'ri:s}].

\begin{lemma}[\cite{riesz1955functional}]
	Given a bounded linear map $f \colon H_1 \to H_2$ between Hilbert spaces,
	there is a unique bounded linear map $f\dg \colon H_2 \to H_1$ such that for
	all $x \in H_1$ and $y \in H_2$, $\ip{f(x)}{y} = \ip{x}{f\dg(y)}$. The map
	$f\dg$ is called the adjoint of $f$.
\end{lemma}

The category $\Hilb$ is equipped with a structure of symmetric monoidal
\emph{dagger} category, meaning that $(-)\dg$ is an involutive contravariant
endofunctor which is the identity on objects. Moreover, the dagger and the
monoidal tensor respect some coherence conditions. For example, given two
bounded linear maps $f \colon H_1 \to H_2$ and $g \colon H'_1 \to H'_2$,
\[
	(f \otimes g)\dg = f\dg \otimes g\dg \colon H_2 \otimes H'_2 \to H_1 \otimes
	H'_1.
\]

\begin{remark}
	Remember that throughout the thesis, given two maps $f \colon H_1 \to H_2$
	and $g \colon H_2 \to H_3$, we sometimes write $gf \colon H_1 \to H_3$ for
	the composition $g \circ f \colon H_1 \to H_3$. In addition, given a complex
	number $\alpha$ and a map $f \colon H_1 \to H_2$, we write $\alpha f \colon
	H_1 \to H_2$ for the multiplication of the vector space $\alpha . f \colon
	H_1 \to H_2$.
\end{remark}

\begin{definition}
	Given a morphism $f \colon H_1 \to H_2$ in $\Hilb$, we say that $f$ is:
	\begin{itemize}
		\item a unitary, if it is an isomophism and $f^{-1} = f\dg$;
		\item a contraction (or contractive map), if for all $x \in A$, $\norm{fx}
			\leq \norm x$;
		\item an isometry, if $f\dg f = \mathrm{id}$;
		\item a coisometry, if $ff\dg = \mathrm{id}$.
	\end{itemize}
\end{definition}

The category $\Hilb$ enjoys many properties, such as being its own opposite
category and having a zero object -- the zero space $\{ 0 \}$. However, it is
not monoidal closed. Besides, neither $\Hilb$ nor its subcategories obtained
with the above definition are inverse categories, because projectors do not
commute.

\subsection{Quantum Computing}
\label{sub:hilb-computing}

Quantum physics is the science of the infinitesimal. Its laws rule the world of
small particules, in a way often described as hardly understandable to the
macroscopic human intuition. In quantum theory, a particle can be in several
states at the same time -- this is called a \emph{superposition} of states.
This superposition holds as long as the particule is not \emph{observed} -- by
the operator by the environment. Once we have a look at the particule, it fixes
itself in one particular state, with a certain probability.

This description can be simplified to the level of bits. Imagine that the two
possible states are $0$ or $1$, usually written respectively $\ket 0$ and $\ket
1$. The general state of a \emph{quantum} bit -- or \emph{qubit} -- is given by
the expression $\alpha \ket 0 + \beta \ket 1$, where $\alpha$ and $\beta$ are
complex numbers.  Once this qubit is observed, it becomes either $0$ with
probability $\abs\alpha^2$, or $1$ with probability $\abs\beta^2$. Because $0$
and $1$ are the only two possibilities in this simple case, we need to have a
proper probability distribution, and thus $\abs\alpha^2 + \abs\beta^2 = 1$.
When this condition is verified, we say that the state $\alpha \ket 0 + \beta
\ket 1$ is \emph{normalised}. There are several states that represent a qubit
with an equal probability to be measured to $0$ and $1$; the most usual ones
are $\ket + = \frac{1}{\sqrt 2} \ket 0 + \frac{1}{\sqrt 2} \ket 1$ and $\ket -
= \frac{1}{\sqrt 2} \ket 0 - \frac{1}{\sqrt 2} \ket 1$. When manipulating two
qubits in states $\ket x$ and $\ket y$ in parallel, we write $\ket x \otimes
\ket y$ or even $\ket{xy}$ for the resulting state.

Quantum computing is the science of performing operations on those qubits --
and more generally, on a quantum superposition of data -- to compute. The most
traditional way of expliciting a quantum algorithm is with a quantum circuit;
the latter is the quantum generalisation of logical circuits. A quantum circuit
is thus a sequence of quantum \emph{logic gates} applied to a fixed number of
qubits. Among quantum logic gates, one can find the $\mathtt{not}$ gate,
similar to the classical one, the Hadamard gate which maps $\ket 0$ to $\ket +$
and $\ket 1$ to $\ket -$, rotation gates that map $\ket 0$ to $\ket 0$ and
$\ket 1$ to $e^{i \pi \theta} \ket 1$, where $\theta$ is a real number, and the
controlled $\mathtt{not}$ on two qubits, which applies the $\mathtt{not}$ to
the second qubits when the first one is $1$, else is the identity.

These operations are usually represented with complex matrices, and the states
are given by vectors in finite-dimensional Hilbert spaces.
\[
	\ket 0 = \begin{pmatrix} 1 \\ 0 \end{pmatrix}
	\qquad
	\ket 1 = \begin{pmatrix} 0 \\ 1 \end{pmatrix}
	\qquad
	\alpha \ket 0 + \beta \ket 1 = \begin{pmatrix} \alpha \\ \beta \end{pmatrix}
\]
Multiples states in parallel, such as $\ket 1 \otimes \ket 0$, are obtained with
the usual tensor product. Thus the gates described above are the following:
\[
	\mathtt{not} = \begin{pmatrix} 0 & 1 \\ 1 & 0 \end{pmatrix}
	\qquad
	\mathtt{had} = \begin{pmatrix}
		\frac{1}{\sqrt 2} & \frac{1}{\sqrt 2} \\
		\frac{1}{\sqrt 2} & - \frac{1}{\sqrt 2}
	\end{pmatrix}
	\qquad
	R_{\theta} = \begin{pmatrix} 1 & 0 \\ 0 & e^{i \pi \theta} \end{pmatrix}
	\qquad
	\mathtt{cnot} = \begin{pmatrix}
		1 & 0 & 0 & 0 \\
		0 & 1 & 0 & 0 \\
		0 & 0 & 0 & 1 \\
		0 & 0 & 1 & 0
	\end{pmatrix}
\]

Any quantum algorithm can be expressed as a finite sequence of gates
\cite{nielsen02quantum}, and we say that quantum circuits are \emph{universal}.

Another important notation in quantum computing besides $\ket \cdot$ is $\bra \cdot$,
which is obtained by taking the dagger:
\[
	\bra 0 = \ket 0 \dg = \begin{pmatrix} 1 & 0 \end{pmatrix}
	\qquad
	\bra 1 = \ket 1 \dg = \begin{pmatrix} 0 & 1 \end{pmatrix}
	\qquad
	\alpha \bra 0 + \beta \bra 1  = \alpha \ket 0 \dg + \beta \ket 1 \dg  =
	\begin{pmatrix} \alpha & \beta \end{pmatrix}
\]

Note that Hilbert spaces and unitaries (resp. contractions) form a dagger
category. They are wide subcategories of $\Hilb$. Unitary maps are of central
importance because they are the proper quantum operations, as solutions of the
Schrödinger equation. One of the most significant unitary maps in quantum
computing is the basis change in $\C \oplus \C$, also known as the
\emph{Hadamard gate}: $\ketbra 0 + + \ketbra 1 -$. We observe here that
$\ketbra 0 +$ and $\ketbra 1 -$ are contractions, and that unitary maps can be
formulated as linear combination of compatible contractive maps. Note also that
the \emph{states}, like $\ket 0$, are isometries. 

We write $\Contr$, the category of \emph{countably-dimensional} Hilbert spaces
and contractive maps, $\Iso$, the category of \emph{countably-dimensional}
Hilbert and isometries between them, and $\Coiso$ the category of
\emph{countably-dimensional} Hilbert spaces and coisometries between them. The
category of \emph{countably-dimensional} Hilbert spaces and bounded linear maps
is written $\Hilbc$ in this thesis.

\begin{definition}[Zero map]
	Given any pair of Hilbert spaces $H_1$ and $H_2$, we write $0_{H_1,H_2}
	\colon H_1 \to H_2$ for the linear map whose image is $\set 0$ (we also write
	that $\Ker (0_{H_1,H_2}) = H_1$). When it is not ambiguous, we write $0 \colon
	H_1 \to H_2$. It is a contractive map.
\end{definition}

\begin{remark}
	Given a Hilbert space $H$, the morphism $0 \colon H \to \set 0$ is unique for
	every $H$ and makes $\set 0$ a terminal object in $\Contr$ and in $\Coiso$.
\end{remark}

Contractions are widely used in the literature for the denotational semantics
of quantum programming languages \cite{heunen2022information,
pablo2022universal, carette2023quantum}. Some recent developments expose axioms
for the categories involved in this thesis \cite{heunen2022axioms,
heunen2022contractions, dimeglio2024dagger}. This better mathematical understanding
of the category theory behind Hilbert spaces can only be beneficial for the
theory of programming languages.

\paragraph{The $\ell^2$ functor.}
As said in \cite{heunen2013l2}, the $\ell^2$ construction is the closest thing
there is to a free Hilbert space. Given a set $X$, the following:
\begin{equation}
	\label{eq:l2}
	\ell^2(X) \defeq \set{\phi \colon X \to \C \alt \sum_{x \in X}
	\abs{\phi(x)}^2 < \infty}
\end{equation}
is actually a Hilbert space. Even more, $\ell^2(-)$ is a functor $\PInj \to
\Hilb$; given a morphism $f \colon X \to Y$ in $\PInj$, we have:
\[
	\ell^2(f)(\phi) = \phi \circ f\pinv.
\]
This functor comes with many properties (see \cite{heunen2013l2}) except the one
programming language theorists would want: it has no adjoints. Because of this,
a \emph{quantum effect} based on $\ell^2$ cannot be studied as a usual
computational effect (see \secref{sub:l2effect} where this point is discussed,
and see the next section \secref{sec:background-monads} for the \emph{usual}
semantics of effects).

Given an element $x \in X$, its counterpart in $\ell^2(X)$ -- in other words,
the $\phi$ such that $\phi(x) = 1$ and for all $y \neq x$, $\phi(y) = 0$ -- is
written $\ket x$. The family $(\ket x)_{x \in X}$ is called the \emph{canonical
basis} of $\ell^2(X)$.

\section{Monads}
\label{sec:background-monads}

We introduce some background on strong and commutative monads and their
premonoidal structure. Monads appear to be the most usual tool to interpret
effects in a programming language. This is thanks to Moggi's work \cite{moggi,
moggi-lics}.

Monads are the generalisation of monoids in category theory, where the
operation is the composition. One is likely to hear at least one the sentence
``a monad is just a monoid in the category of endofunctors''. While this
sentence is correct, we introduce a bit more background to understand the later
use of monads in Chapter~\ref{ch:monads}.

\subsection{Strong and Commutative Monads}

We begin by recalling the definition of a monad.

\begin{definition}[Monad]\label{def:monad}
	A \emph{monad} over a category $\CC$ is an endofunctor $\TT \colon \CC \to
	\CC$ equipped with two natural transformations $\eta \colon \iid \Rightarrow
	\TT$ and $\mu \colon \TT^2 \Rightarrow \TT$ such that the following diagrams
    \[\input{monad-def.tikz}\qquad\qquad\input{monad-def2.tikz}\]
  commute. We call $\eta$ the \emph{unit} of $\TT$ and we say that $\mu$ is the
  \emph{multiplication} of $\TT$.
\end{definition}

Next, we recall the definition of a \emph{strong} monad, which is the main
object of study in Chapter~\ref{ch:monads}. As we already explained in the introduction,
these monads are more computationally relevant (compared to non-strong ones) for
most use cases. The additional structure, called the \emph{monadic strength},
ensures the monad interacts appropriately with the monoidal structure of the
base category.

\begin{definition}[Strong Monad]\label{def:strong-monad}
  A \emph{strong monad} over a monoidal category $(\CC,\otimes,I,\alpha,\lambda, \rho)$
  is a monad $(\TT,\eta,\mu)$ equipped with a natural transformation
  $\tau_{X,Y}:X\otimes\TT Y\to\TT(X\otimes Y),$ called \emph{left strength},
  such that the following diagrams commute:
    \[\stikz[0.83]{../figures/strength-def.tikz}\]
\end{definition}

We now recall the definition of a \emph{commutative} monad which is of central
importance here and in Chapter~\ref{ch:monads}. Compared to a strong monad, a commutative monad
enjoys even stronger coherence properties with respect to the monoidal
structure of the base category (see also \secref{sub:premonoidal}).

\begin{definition}[Commutative Monad]
  \label{def:commutative-monad}
	Let $(\TT, \eta, \mu, \tau)$ be a strong monad on a \emph{symmetric} monoidal
	category $(\CC, \otimes, I, \gamma)$.  The \emph{right strength} $\tau'_{X,Y}
	\colon \TT X \otimes Y \to \TT(X \otimes Y)$ of $\TT$ is given by the
	assignment $\tau'_{X,Y} \eqdef
	\TT(\gamma_{Y,X})\circ\tau_{Y,X}\circ\gamma_{\TT X,Y}$. Then, $\TT$ is said
	to be \emph{commutative} if the following diagram commutes:
  \begin{equation}
    \label{eq:commutative-monad}
    \stikz[0.83]{./figures/commutative-monad-def.tikz}
  \end{equation}
\end{definition}

\begin{remark}
	In the literature, the left and right strengths are sometimes called
	``strength" and ``costrength" respectively.
\end{remark}

\begin{definition}[Morphism of Strong Monads \cite{jacobs-coalgebra}]\label{def:morph-monad}
	Given two strong monads $(\TT, \eta^\TT, \mu^\TT, \tau^\TT)$ and
	$(\PP,\eta^\PP, \mu^\PP, \tau^\PP)$ over a category $\CC$, a \emph{morphism
	of strong monads} is a natural transformation $\iota : \TT \Rightarrow \PP$
	that makes the following diagrams commute:
  \[\stikz[0.83]{figures/map-of-monads-def.tikz}\]
\end{definition}

Strong monads over a (symmetric) monoidal category $\CC$
and strong monad morphisms between them form a category which we denote by
writing $\mathbf{StrMnd}(\CC).$
In the situation of Definition \ref{def:morph-monad}, if $\iota$ is a
monomorphism in $\mathbf{StrMnd}(\CC)$, then $\TT$ is said
to be a \emph{strong submonad} of $\PP$ and $\iota$ is said to be a
\emph{strong submonad morphism}.

\begin{definition}[Kleisli category]\label{def:kleisli}
	Given a monad $(\TT,\eta,\mu)$ over a category $\CC$, the \emph{Kleisli
	category} $\CC_\TT$ of $\TT$ is the category whose objects are the same as
	those of $\CC$, but whose morphisms are given by $\CC_\TT(X,Y)=\CC(X,\TT Y)$.
	Composition in $\CC_\TT$ is given by $g\odot f \eqdef \mu_Z\circ\TT g\circ f$
	where $f:X\to\TT Y$ and $g:Y\to\TT Z$. The identity at $X$ is given by the
	monadic unit $\eta_X \colon X \to \TT X.$
\end{definition}

\begin{proposition}[\cite{jacobs-coalgebra}]
  \label{prop:embed}
  If $\iota : \TT \naturalto \PP$ is a submonad morphism, then
  the functor $\II:\CC_\TT\to \CC_\PP,$ defined by $\II(X) = X$ on objects and
  $\II(f:X\to \TT Y) = \iota_Y\circ f \colon X \to \PP Y$ on morphisms,
  is an embedding of categories.
\end{proposition}

The functor $\II$ above is the canonical embedding of $\CC_\TT$ into $\CC_\PP$
induced by the submonad morphism $\iota \colon \TT \naturalto \PP .$

\subsection{Semantics of the $\lambda$-calculus with effects}
\label{sub:sem-lambda-effects}

We present here a brief summary of the work by Moggi \cite{moggi, moggi-lics}
on computational effects. This work has been very influencial, resulting in the
development of the programming language Haskell, among others.

The grammar and typing rules for Moggi's metalanguage is found in
Figure~\ref{fig:moggi-grammars}. Compared to the simply-typed
$\lambda$-calculus, a type construction $\TT(-)$ is added. Given a type $A$,
the type $\TT A$ is called a \emph{monadic} type, and represents the
computational effects of type $A$ allowed in the language. The $\mathtt{re}t$
constructor is to be seen as an introduction rule for monadic types, and
embodies the fact that a \emph{pure} or \emph{non-effectful} computation can be
seen as a monadic computation, but with no effect. The $\mathtt{do}$ operation
performs the sequencing of monadic computations.

The equational theory for monadic types, added on top of the equational theory
for the simply-typed $\lambda$-calculus presented in
Figure~\ref{fig:moggi-rules}, is found in Figure~\ref{fig:moggi-monad-rules}.

\begin{figure}[!h]
  \noindent $\text{(Types)}\quad  A, B ~~::= 1 \alt A \to  B\alt A\times B
  \alt \TT A$ \\
  $\quad$\\
  $\text{(Terms)} \quad M,N ~~ ::=  x \alt * \alt \lambda x^A.M \alt MN \alt \pv M N$ \\
  $\text{ }\qquad\alt \pi_i M \alt \rret M \alt \ddo x M N $
  $\quad$\\
  \[\begin{array}{c}
    \infer{
      \Gamma,x \colon A\vdash x \colon A}{}
    \qquad
    \infer{
      \Gamma\vdash MN \colon B
    }{
      \Gamma\vdash M \colon A\to B
      &
      \Gamma\vdash N \colon A}
    \\[1.5ex]
    \infer{
      \Gamma\vdash * \colon 1}{}
    \qquad
    \infer{\Gamma\vdash\lambda x^A.M \colon A\to B}{\Gamma,x \colon A\vdash M \colon B}
    \qquad
    \infer{
      \Gamma\vdash\pi_i M \colon A_i}{\Gamma\vdash M \colon A_1\times A_2}
    \\[1.5ex]
    \infer{
      \Gamma\vdash \pv{M}{N} \colon  A\times B
    }{
      \Gamma\vdash M \colon A
      &
      \Gamma\vdash N \colon B
    }
    \qquad
    \infer{
      \Gamma\vdash \rret M  \colon \TT A}{\Gamma\vdash M \colon A}
		\qquad
    \infer{
      \Gamma\vdash\ddo x M N \colon \TT B
    }{
      \Gamma\vdash M \colon \TT A
      &
      \Gamma, x \colon A\vdash N \colon \TT B
    }
  \end{array}\]
  \caption{Grammars and typing rules.}
  \label{fig:moggi-grammars}
\end{figure}

\begin{figure}[!h]
  \resizebox{\hsize}{!}{
    $
    \begin{array}{c}
      \infer[(ret.eq)]{
        \Gamma\vdash \rret M = \rret N \colon \TT A
      }{
        \Gamma\vdash M = N \colon A
      }
      \qquad
      \infer[(do.eq)]{
        \Gamma\vdash \ddo x M N = \ddo{x}{M'}{N'} \colon \TT B
      }{
        \Gamma\vdash M=M' \colon \TT A
        &
        \Gamma, x \colon A\vdash N=N' \colon \TT B
      }
      \\[1.5ex]
      \infer[(\TT.\beta)]{
        \Gamma\vdash \ddo{x}{\rret M}{N} = N[M/x] \colon \TT B
      }{
        \Gamma\vdash M \colon A
        &
        \Gamma,x \colon A\vdash N \colon \TT B
      }
      \qquad
      \infer[(\TT.\eta)]{
        \Gamma\vdash \ddo{x}{M}{\rret x} = M \colon \TT A
      }{
        \Gamma\vdash M \colon \TT A
      }
		\end{array}
		$
	}
  \caption{Moggi's equational rules for terms of monadic types.}
  \label{fig:moggi-monad-rules}
\end{figure}

\paragraph{Denotational semantics.}
The denotational semantics of Moggi's metalanguage is obtained in a cartesian
closed category $\CC$ equipped with a strong monad $\TT$. Pure computations
shall still be interpreted in as morphisms in the category $\CC$; while monadic
computations, \emph{e.g.} $\Gamma \vdash M \colon \TT A$, are interpreted as
morphisms $\sem\Gamma \to \TT \sem A$, thus living in the Kleisli category of
$\TT$. The interpretation of $\mathtt{ret}$ is given by the unit of the monad
$\TT$, and the interpretation of of $\mathtt{do}$ is obtained with the
composition in the Kleisli category.

\subsection{Premonoidal Structure of Strong Monads}
\label{sub:premonoidal}

Let $\TT$ be a strong monad on a symmetric monoidal category $(\CC, I,
\otimes)$. Then, its Kleisli category $\CC_\TT$ does \emph{not} necessarily
have a canonical monoidal structure. However, it does have a canonical
\emph{premonoidal structure} as shown by Power and Robinson
\cite{power-premonoidal}. In fact, they show that this premonoidal structure is
monoidal iff the monad $\TT$ is commutative. Next, we briefly recall the
premonoidal structure of $\CC_\TT$ as outlined by them.

For every two objects $X$ and $Y$ of $\CC_\TT$, their tensor product $X \otimes
Y$ is also an object of $\CC_\TT$, but the monoidal product $\otimes$ of $\CC$
does not necessarily induce a bifunctor on $\CC_\TT$. However, by using
the left and right strengths of $\TT$, we can define two families of functors
as follows:
\begin{itemize}
  \item for any object $X$, a functor $(-\otimes_l X) \colon \CC_\TT \to
    \CC_\TT$ whose action on objects sends $Y$ to $Y\otimes X$, and sends
    $f:Y\to \TT Z$ to $\tau'_{Z,X}\circ(f\otimes X):Y\otimes X\to\TT(Z\otimes X)$;
  \item for any object $X$, a functor $(X\otimes_r -) \colon \CC_\TT \to
    \CC_\TT$ whose action on objects sends $Y$ to $X\otimes Y$, and sends
    $f:Y\to \TT Z$ to $\tau_{X,Z}\circ(X\otimes f):X\otimes Y\to\TT(X\otimes Z)$.
\end{itemize}
This categorical data satisfies the axioms and coherence properties of
\emph{premonoidal categories} as explained in \cite{power-premonoidal}, but
which we omit here because it is not essential for the development of our
results.  What is important is to note that in a premonoidal category, $f \otimes_l X'$ and
$X\otimes_r g$ do not always commute. This leads us to the following definition,
which plays a crucial role in the theory of premonoidal categories and
has important links to our development.

\begin{definition}[Premonoidal Centre \cite{power-premonoidal}]\label{def:central-morph}
	Given a strong monad $(\TT, \eta, \mu, \tau)$ on a symmetric monoidal
	category $(\CC, I, \otimes)$, we say that a morphism $f:X\to Y$ in $\CC_\TT$
	is \emph{central} if for any morphism $f':X'\to Y'$ in $\CC_\TT$, the diagram
  \[ \input{def-central-morphism.tikz} \]
  commutes in $\CC_\TT$.
  The \emph{premonoidal centre} of $\CC_\TT$ is the
  subcategory $Z(\CC_\TT)$ which has the same objects as those of $\CC_\TT$ and
  whose morphisms are the central morphisms of $\CC_\TT$.
\end{definition}

In \cite{power-premonoidal}, the authors prove that $Z(\CC_\TT)$, is a
symmetric \emph{monoidal} subcategory of $\CC_\TT$. In particular, this means
that Kleisli composition and the tensor functors $(- \otimes_l X)$ and $(X
\otimes_r -)$ preserve central morphisms. However, it does not necessarily hold
that the subcategory $Z(\CC_\TT)$ is the Kleisli category for a monad over
$\CC$. Nevertheless, in this situation, the left adjoint of the Kleisli
adjunction $\JJ \colon \CC \to \CC_\TT$ always corestricts to $Z(\CC_\TT)$.  We
write $\hat \JJ \colon \CC \to Z(\CC_\TT)$ to indicate this corestriction
(which need not be a left adjoint).

\begin{remark}
  In \cite{power-premonoidal}, the subcategory $Z(\CC_\TT)$ is called the
  centre of $\CC_\TT$.  However, we refer to it as the \emph{premonoidal
  centre} of a premonoidal category to avoid confusion with the new
  notion of the centre of a monad that we introduce next. In the sequel, we show
  that the two notions are very strongly related to each other (Theorem
  \ref{th:centralisability}).
\end{remark}

%% file: figures/monad-def.tikz
\begin{tikzpicture}
	\begin{pgfonlayer}{nodelayer}
		\node [style=empty] (0) at (-2.25, 2) {$\TT^3 X$};
		\node [style=empty] (1) at (2.25, 2) {$\TT^2 X$};
		\node [style=empty] (2) at (-2.25, -2) {$\TT^2 X$};
		\node [style=empty] (3) at (2.25, -2) {$\TT X$};
		\node [style=empty] (4) at (-3.25, 0) {$\TT\mu_X$};
		\node [style=empty] (5) at (0, 2.45) {$\mu_{\TT X}$};
		\node [style=empty] (6) at (3, 0) {$\mu_X$};
		\node [style=empty] (7) at (0, -2.5) {$\mu_X$};
		\node [style=none] (8) at (-2.25, 1.25) {};
		\node [style=none] (9) at (-2.25, -1.25) {};
		\node [style=none] (10) at (2.25, -1.25) {};
		\node [style=none] (11) at (2.25, 1.25) {};
	\end{pgfonlayer}
	\begin{pgfonlayer}{edgelayer}
		\draw [style=to] (2) to (3);
		\draw [style=to] (0) to (1);
		\draw [style=to] (8.center) to (9.center);
		\draw [style=to] (11.center) to (10.center);
	\end{pgfonlayer}
\end{tikzpicture}

%% file: figures/monad-def2.tikz
\begin{tikzpicture}
	\begin{pgfonlayer}{nodelayer}
		\node [style=empty] (0) at (-2.25, 2) {$\TT X$};
		\node [style=empty] (1) at (2.25, 2) {$\TT^2 X$};
		\node [style=empty] (2) at (-2.25, -2) {$\TT^2 X$};
		\node [style=empty] (3) at (2.25, -2) {$\TT X$};
		\node [style=empty] (4) at (-3.125, 0) {$\eta_{\TT X}$};
		\node [style=empty] (5) at (0, 2.5) {$\TT\eta_X$};
		\node [style=empty] (6) at (3, 0) {$\mu_X$};
		\node [style=empty] (7) at (0, -2.5) {$\mu_X$};
		\node [style=none] (8) at (-2.25, 1.25) {};
		\node [style=none] (9) at (-2.25, -1.25) {};
		\node [style=none] (10) at (2.25, 1.25) {};
		\node [style=none] (11) at (2.25, -1.25) {};
	\end{pgfonlayer}
	\begin{pgfonlayer}{edgelayer}
		\draw [style=to] (2) to (3);
		\draw [style=to] (0) to (1);
		\draw [style=equal-arrow] (0) to (3);
		\draw [style=to] (8.center) to (9.center);
		\draw [style=to] (10.center) to (11.center);
	\end{pgfonlayer}
\end{tikzpicture}

%% file: figures/def-central-morphism.tikz
\begin{tikzpicture}
	\begin{pgfonlayer}{nodelayer}
		\node [style=empty] (0) at (-2.5, 2) {$X\otimes X'$};
		\node [style=empty] (1) at (-2.5, -2.25) {$X\otimes Y'$};
		\node [style=empty] (2) at (4.5, 2) {$Y\otimes X'$};
		\node [style=empty] (3) at (4.5, -2.25) {$Y\otimes Y'$};
		\node [style=empty] (4) at (-2.5, 1.5) {};
		\node [style=empty] (5) at (-2.5, -1.75) {};
		\node [style=empty] (6) at (4.5, 1.5) {};
		\node [style=empty] (7) at (4.5, -1.75) {};
		\node [style=empty] (16) at (1, 2.5) {$f\otimes_l X'$};
		\node [style=empty] (17) at (-3.75, 0) {$X\otimes_r f'$};
		\node [style=empty] (18) at (5.75, 0) {$Y\otimes_r f'$};
		\node [style=empty] (19) at (1, -2.75) {$f\otimes_l Y'$};
	\end{pgfonlayer}
	\begin{pgfonlayer}{edgelayer}
		\draw [style=to] (0) to (2);
		\draw [style=to] (6) to (7);
		\draw [style=to] (4) to (5);
		\draw [style=to] (1) to (3);
	\end{pgfonlayer}
\end{tikzpicture}

%% file: monads.tex
\begin{abstract}
	Monads in category theory are algebraic structures that can be used to model
	computational effects in programming languages.  We show how the notion of
	``\emph{centre}'', and more generally ``\emph{centrality}'', \textit{i.e.}, the
	property for an effect to commute with all other effects, may be formulated
	for strong monads acting on symmetric monoidal categories. We identify three
	equivalent conditions which characterise the existence of the centre of a
	strong monad (some of which relate it to the premonoidal centre of Power and
	Robinson) and we show that every strong monad on many well-known naturally
	occurring categories does admit a centre, thereby showing that this new
	notion is ubiquitous. More generally, we study \emph{central submonads},
	which are necessarily commutative, just like the centre of a strong monad. We
	provide a computational interpretation by formulating
	equational theories of lambda calculi equipped with central submonads, we
	describe categorical models for these theories and prove soundness,
	completeness and internal language results for our semantics.

	\paragraph{References.} This chapter, apart from some additions made by the
	author, is a paper \cite{nous23monads} presented at LICS'2023, and coauthered
	with Titouan Carette and Vladimir Zamdzhiev.
\end{abstract}

\section{Introduction}

The importance of monads in programming semantics has been demonstrated in
seminal work by Moggi \cite{moggi-lics,moggi}. The main idea is that monads
allow us to introduce computational effects (e.g., state, input/output,
recursion, probability, continuations) into pure type systems in a controlled
way. The mathematical development surrounding monads has been very successful
and it directly influenced modern programming language design through the
introduction of monads as a programming abstraction into languages such as
Haskell, Scala and others (see \cite{benton-monads}). Inspired by this, we
follow in the same spirit: we start with a mathematical question about monads,
we provide the answer to it and we present a computational interpretation. The
mathematical question that we ask is simple and it is inspired by the theory of
monoids and groups:
\begin{center}
	\label{eq:question}
	\emph{Is there a suitable notion of ``centre'' that may be formulated for
	monads and what is a ``central'' submonad?}
\end{center}

We show that, just as every monoid $M$ (on $\Set$) has a centre, which is a
\emph{commutative} submonoid of $M$, so does every (canonically strong) monad
$\TT$ on $\Set$ and the centre of $\TT$ is a \emph{commutative} submonad of
$\TT$ (\secref{sub:sets}). A central\footnote{Given a group $G$, a
\emph{central subgroup} is a subgroup of the centre of $G$, equivalently, a
subgroup whose elements commute with every element of $G$.} submonad of $\TT$
is simply a submonad of the centre of $\TT$ (Definition \ref{def:central-sub})
and the analogy to the case of monoids and groups is completely preserved. Note
that our construction has nothing to do with the folklore characterisation of
monads as monoid objects in a functor category, wherein the notion of
commutativity is unclear. The relevant analogy with monoids in $\Set$ is fully
explained in Example \ref{ex:free-monad}.  Generalising away from the category
$\Set$, the answer is a little bit more complicated: not every monoid object
$M$ on a symmetric monoidal category $\CC$ has a centre, and neither does every
strong monad on $\CC$ (\secref{sub:bad-monad}).  However, we show that under
some reasonable assumptions, the centre does exist (Theorem
\ref{th:centralisability}) and we have not found any naturally occurring monads
in the literature that are not centralisable (\emph{i.e.}, monads other than the
artificially constructed one we used as a counter-example).  Furthermore, we
show that for many categories of interest, all strong monads on them are
centralisable (\secref{sub:examples}) and we demonstrate that the notion of
centre is ubiquitous. The centre of a strong monad satisfies interesting
universal properties (Theorem \ref{th:centralisability}) which may be
equivalently formulated in terms of our novel notion of \emph{central cone} or
via the \emph{premonoidal centre} of Power and Robinson
\cite{power-premonoidal}. The notion of a central submonad is more general and
it may be defined without using the centre.  When the centre exists, a central
submonad may be equivalently defined as a strong submonad of the centre
(Theorem \ref{th:centrality}).

The computational significance of these ideas is easy to understand: given an
effect, modelled by a strong monad, such that perhaps not \emph{every} pair of
effectful operations commute (\emph{i.e.}, the order of monadic sequencing matters),
identify only those effectful operations which do commute with any other
possible effectful operation. The effectful operations that satisfy this
property are called \emph{central}. When the monad is centralisable, the
collection of \emph{all} central operations determine the centre of the monad
(which is a commutative submonad). \emph{Any} collection of central operations
that may be organised into a strong submonad determines a central submonad
(which also is commutative). We argue that central submonads have greater
computational significance compared to the centre of a strong monad
(\secref{sub:theories}) for two main reasons: (1) central submonads are
strictly more general; (2) central submonads have a simpler and considerably
more practical axiomatisation via an equational theory, whereas the centre of a
monad requires an axiomatisation using a more complicated logical theory.  We
cement our categorical semantics by proving soundness, completeness and
internal language results (See \cite{maietti2005relating} for a
convincing argument why internal language results are important and why
soundness and completeness \emph{alone} might not be sufficient).

\subsection{Related Work}

A notion of commutants for enriched algebraic theories has been defined in
\cite{commutants} from which the author derives a notion of centre of an
enriched algebraic theory. In the case of enriched monads, in other words,
strong monads arising from enriched algebraic theories, their notion of
commutant extends to monad morphisms. While not explicitly stated in the paper,
applying the commutant construction on the identity monad morphism from a monad
to itself provides a notion of centre of a monad that appears to coincide with
ours. However, enriched algebraic theories correspond to $\mathcal{J}$-ary
$\VVV$-enriched monads (See \cite{commutants} for a definition of
$\mathcal{J}$-ary monads w.r.t.  a system of arities $\mathcal{J}$) on a
symmetric monoidal \emph{closed} category $\VVV$ (equivalently $\JJ$-ary strong
monads on $\VVV$). In this chapter, we show that monoidal closure of $\VVV$ is not
necessary to define the centre and neither is the $\mathcal J$-ary assumption
on the monad.  Other related work \cite{garner2016commutativity} considers a
very general notion of commutativity in terms of certain kinds of duoidal
categories.  As a special case of their treatment, the authors are able to
recover the commutativity of bistrong monads and with some additional effort
(not outlined in the paper), it is possible to construct the centre of a
bistrong monad acting on a monoidal \emph{biclosed} category.  Our construction
of the centre appears to coincide with theirs in the special case of strong
monads defined on symmetric monoidal \emph{closed} categories, but as discussed
above, our method does not require any kind of closure of the category.
Therefore, compared to both works \cite{garner2016commutativity,commutants}, as
far as symmetric monoidal (not necessarily closed) categories are concerned,
our methods can be used to construct the centre for a larger class of strong
monads and we establish our main results, together with our universal
characterisation of the centre, under these assumptions. Furthermore, we also
place a heavy emphasis on \emph{central} submonads in this chapter and these kinds
of monads are not discussed in either of these works and neither is there a
computational interpretation (which is our main result in
\secref{sec:computational}).

Another related work is \cite{power-premonoidal}, which introduces
premonoidal categories. We have established important links between our
development and the premonoidal centre (Theorem \ref{th:centralisability}).
While premonoidal categories have been influential in our understanding of
effectful computation, it was less clear (to us) how to formulate an appropriate
computational interpretation of the premonoidal centre for higher-order
languages. We show that under some mild assumptions (which are easily
satisfied see \secref{sec:examples}), the premonoidal centre of the Kleisli
category of a strong monad induces an adjunction into the base category
(Theorem \ref{th:centralisability}) and this allows us to formulate a suitable
computational interpretation by using monads, which are already well-understood
\cite{moggi,moggi-lics} and well-integrated into many programming languages
\cite{benton-monads}.

Staton and Levy introduce the novel notion of \emph{premulticategories}
\cite{premulticategories} in order to axiomatise impure/effectful computation
in programming languages. The notion of centrality plays an important role in
the development of the theory there as well.  However, they do not focus, as we
do, on providing suitable programming abstractions that identify both central
and non-central computations (e.g., by separating them into different types
like us) and from what we can tell from our reading, there are no universal
properties stated for the collection of central morphisms.  Also, our results
provide a computational interpretation in terms of monads, which are standard
and well-understood, so it is easier to incorporate them into existing
languages.

Central morphisms in the context of computational effects have been studied in
\cite{fuhrmann1999lambda}, among other sorts of \emph{varieties} of morphisms:
thunkable, copyable, and discardable.  The author links their notion of central
morphisms with the ones from the premonoidal centre in Power and Robinson
\cite{power-premonoidal}, and also proves under some conditions that those
varieties form a subcategory with similar properties to the original category.
However, they do not mention that a central submonad or a centre can be
constructed out of those central morphisms. More generally, the fact that monads
could be derived from those varieties is not studied at all in that paper.

The work in \cite{fuhrmann1999lambda} has an impact in \cite{dylan2022galois},
where a Galois connection is established between call-by-value
and call-by-name. In that paper, the order in which operations
are done matters, and central computations are mentioned. Again,
the central computations are not linked to submonads in there.

\subsection{Work of the Author}

The author has contributed to the following points.
\begin{itemize}
	\item Equivalent characterisations for a monad to be centralisable (see
		Theorem~\ref{th:centralisability}).
	\item Subsection~\ref{sub:lawvere}, which details the link with Lawvere
		theories.
	\item A language for central submonads, inspired from Moggi's metalanguage.
	\item A notion of equational theories for such a language.
	\item A completeness and internal language result, linking the categorical
		model to a syntactic notion of central submonad.
\end{itemize}

Compared to the paper \cite{nous23monads}, the proofs produced by the author
are added, as well as a section (see \secref{sub:lawvere}) in which we show the
link with Lawvere theories.

\section{The Centre of a Strong Monad}
\label{sec:centralisable}

We begin by showing that any (necessarily strong) monad on $\Set$ has a
centre (\secref{sub:sets}) and we later show how to define the centre of a
strong monad on an arbitrary symmetric monoidal category
(\secref{sub:centre-general}). Unlike the former, the latter submonad does not
always exist, but it does exist under mild assumptions and we show that the
notion is ubiquitous.

\subsection{The Centre of a Monad on $\Set$}
\label{sub:sets}

The results we present next are a special case of our more general development
from \secref{sub:centre-general}, but we choose to devote special attention to
monads on $\Set$ for illustrative purposes.

\begin{definition}[Centre]\label{def:centre-set}
	Given a strong monad $(\TT, \eta, \mu, \tau)$ on $\Set$
	with right strength $\tau'$,
	we say that the \emph{centre} of $\TT$ at $X$, written $\ZZ X$, is the set
	\begin{align*}
		\ZZ X \eqdef \{ t \in \TT X \ & |\ \forall Y \in \Ob(\Set). \forall s \in \TT Y. 
		\mu(\TT\tau'(\tau(t,s))) = \mu(\TT\tau(\tau'(t,s)))  \} .
	\end{align*}
	We write $\iota_X \colon \ZZ X \subseteq \TT X$ for the indicated subset inclusion.
\end{definition}

In other words, the centre of $\TT$ at $X$ is the subset of $\TT X$ which
contains all monadic elements for which \eqref{eq:commutative-monad} holds when
the set $X$ is fixed and the set $Y$ ranges over all sets.


Notice that $\ZZ X \supseteq \eta_X(X),$ \emph{i.e.}, the centre of $\TT$ at $X$
always contains all monadic elements which are in the image of the monadic
unit. This follows easily from the axioms of strong monads.  In fact, the
assignment $\ZZ(-)$ extends to a \emph{commutative submonad} of $\TT$. In
particular, the assignment $\ZZ(-)$ extends to a functor $\ZZ \colon \Set \to \Set$
when we define \( \ZZ f \eqdef \TT f|_{\ZZ X} \colon \ZZ X \to \ZZ Y, \) for any
function $f: X \to Y,$ where $\TT f|_{\ZZ X}$ indicates the restriction of $\TT
f \colon \TT X \to \TT Y$ to the subset $\ZZ X.$ Moreover, for any two sets $X$ and
$Y$, the monadic unit $\eta_X \colon X \to \TT X$, the monadic multiplication $\mu_X
: \TT^2 X \to \TT X$, and the monadic strength $\tau_{X,Y} \colon X \times \TT Y \to
\TT(X \times Y)$ (co)restrict respectively to functions $\eta_X^\ZZ \colon X \to \ZZ
X$, $\mu_X^\ZZ \colon \ZZ^2 X \to \ZZ X$ and $\tau_{X,Y}^\ZZ \colon X \times \ZZ Y \to
\ZZ(X \times Y)$.  That the above four classes of functions (co)restrict as
indicated follows from our more general treatment presented in the next
section. It then follows, as a special case of Theorem
\ref{th:centralisability}, that the data we just described constitutes a
commutative submonad of $\TT$.

\begin{theorem}
	\label{thm:centre-set}
	The assignment $\ZZ(-)$ can be extended to a \emph{commutative submonad} $(\ZZ, \eta^\ZZ, \mu^\ZZ, \tau^\ZZ)$ of $\TT$
	with the inclusions $\iota_X \colon \ZZ X \subseteq \TT X$ being the submonad morphism.
	Furthermore, there is a canonical isomorphism of categories $\Set_\ZZ \cong
	Z(\Set_\TT)\footnote{Theorem \ref{th:centralisability} states precisely in what sense this isomorphism is canonical.}.$
\end{theorem}

The final statement of Theorem \ref{thm:centre-set}
shows that the Kleisli category of $\ZZ$ is canonically
isomorphic to the premonoidal centre of the Kleisli category of $\TT$. Because
of this, we are justified in saying that $\ZZ$ is not just \emph{a} commutative
submonad of $\TT$, but rather it is \emph{the centre} of $\TT,$ which
is necessarily commutative (just like the centre of a monoid is a commutative
submonoid). In \secref{sub:specific-examples} we provide concrete examples of
monads on $\Set$ and their centres and we see that the construction
of the centre aligns nicely with our intuition.

\subsection{The General Construction of the Centre}
\label{sub:centre-general}

Throughout the remainder of
the section, we assume we are given a symmetric monoidal category
$(\CC,\otimes,I,\alpha,\lambda, \rho, \gamma)$ and a strong monad $(\TT, \eta,
\mu, \tau)$ on it with right strength $\tau'$.

In $\Set$, the centre is defined pointwise through subsets of $\TT X$ which
only contain elements that satisfy the coherence condition for a commutative
monad.  However, $\CC$ is an arbitrary symmetric monoidal category, so we
cannot easily form subojects in the required way.  This leads us to the
definition of a \emph{central cone} which allows us to overcome this problem.

\begin{definition}[Central Cone]\label{def:central-cone}
	Let $X$ be an object of $\CC$. A \emph{central cone} of $\TT$ at $X$ is given
	by a pair $(Z, \iota)$ of an object $Z$ and a morphism $\iota \colon Z \to
	\TT X,$ such that for any object $Y,$ the diagram
	\[\input{def-central-cone.tikz}\]
	commutes.
	If $(Z, \iota)$ and $(Z', \iota')$ are two central cones of $\TT$ at $X$,
	then a \emph{morphism of central cones} $\varphi \colon (Z', \iota') \to (Z,
	\iota)$ is a morphism $\varphi \colon Z' \to Z,$ such that $\iota \circ \varphi =
	\iota'.$ Thus central cones of $\TT$ at $X$ form a category.  A
	\emph{terminal central cone} of $\TT$ at $X$ is a central cone $(Z, \iota)$
	for $\TT$ at $X$, such that for any central cone $(Z', \iota')$ of $\TT$ at
	$X$, there exists a unique morphism of central cones $\varphi \colon (Z', \iota')
	\to (Z, \iota).$ In other words, it is the terminal object in the category of
	central cones of $\TT$ at $X$.
\end{definition}

In particular, Definition \ref{def:centre-set} gives a terminal central cone
for the special case of monads on $\Set.$ The names ``central morphism'' (in
the premonoidal sense, see \secref{sub:premonoidal}) and ``central cone''
(above) also hint that there should be a relationship between them. In fact, the
two notions are equivalent.

\begin{proposition}\label{prop:central}
	Let $f:X\to\TT Y$ be a morphism in $\CC$. The pair $(X,f)$ is a central cone
	of $\TT$ at $Y$ iff $f$ is central in $\CC_\TT$ in the premonoidal sense
	(Definition~\ref{def:central-morph}).
\end{proposition}
\begin{proof}
	Let $(X,f)$ be a central cone and let $f' \colon X'\to\TT Y'$ be a morphism.
	The following diagram: \[\scalebox{0.8}{\input{central-cone-to-morph.tikz}} \]
	commutes because: (1) $\CC$ is monoidal; (2) $\tau'$ is natural; (3) $\tau$
	is natural; and (4) the pair $(X, f)$ is a central cone.  Therefore, the
	morphism $f$ is central in the premonoidal sense.\\ For the other direction,
	if $f$ is central in $\CC_\TT$, the following diagram:
	\[\input{central-morph-to-cone.tikz} \] commutes because: (1) $\tau$ is natural;
	(2) $f$ is a central morphism; all remaining subdiagrams commute trivially.
	This shows the pair $(X,f)$ is a central cone.
\end{proof}

From now on, we rely heavily on the fact that central cones and central
morphisms are equivalent notions, and we use Proposition \ref{prop:central}
implicitly in the sequel.  On the other hand, \emph{terminal} central cones are
crucial for our development, but it is unclear how to introduce a similar
notion of ``terminal central morphism'' that is useful. For this reason, we
prefer to work with (terminal) central cones.

It is easy to see that if a terminal central cone for $\TT$ at $X$ exists, then
it is unique up to a unique isomorphism of central cones. Also, one can easily
prove that if $(Z, \iota)$ is a terminal central cone, then $\iota$ is a
monomorphism.  The main definition of this subsection follows next and gives
the foundation for constructing the centre of a strong monad.

\begin{definition}[Centralisable Monad]
	\label{def:centralisable}
	We say that the monad $\TT$ is \emph{centralisable} if, for any object $X$, a
	terminal central cone of $\TT$ at $X$ exists. In this situation, we write
	$(\ZZ X, \iota_X)$ for the terminal central cone of $\TT$ at $X$.
\end{definition}

In fact, for a centralisable monad $\TT$, its terminal central cones induce a
commutative submonad $\ZZ$ of $\TT$, as the next theorem shows, and its proof
reveals constructively how the monad structure arises from them.

\begin{theorem}\label{th:submonad-from-cone}
	If the monad $\TT$ is centralisable, then the assignment $\ZZ(-)$ extends to a
	commutative monad $(\ZZ, \eta^\ZZ, \mu^\ZZ, \tau^\ZZ)$ on $\CC$. Moreover,
	$\ZZ$ is a commutative submonad of $\TT$ and the morphisms $\iota_X \colon \ZZ X
	\to \TT X$ constitute a monomorphism of strong monads $\iota \colon \ZZ
	\Rightarrow \TT$.
\end{theorem}

This theorem relies on several lemmas that are detailed below.

\begin{lemma}\label{lem:precompose-central}
	If $(X,f \colon X\to \TT Y)$ is a central cone of $\TT$ at $Y,$ then for any
	$g \colon Z\to X$, it follows that $(Z,f\circ g)$ is a central cone of $\TT$
	at $Y$.
\end{lemma}
\begin{proof}
	This is obtained by precomposing the definition of central cone by $g\otimes \id$.
	\[ \input{central-precomposing.tikz} \]
	commutes directly from the definition of central cone for $f$.
\end{proof}

\begin{lemma}\label{lem:postcompose-central}
	If $(X,f \colon X\to \TT Y)$ is a central cone of $\TT$ at $Y$ then for any
	$g \colon Y\to Z$, it follows that $(X,\TT g\circ f)$ is a central cone of
	$\TT$ at $Z$.
\end{lemma}
\begin{proof}
	The naturality of $\tau$ and $\mu$ allow us to push the application of
	$g$ to the last postcomposition, in order to use the central property
	of $f$. In more details, the following diagram:
	\[\scalebox{0.8}{\input{central-postcomposing.tikz}}\] commutes, because:
	(1) $f$ is a central cone, (2) $\tau'$ is natural, (3) $\tau$ is
	natural, (4) $\mu$ is natural (5) $\tau$ is natural, (6) $\tau'$ is
	natural, (7) $\mu$ is natural.
\end{proof}

\begin{lemma}\label{lem:monic}
	If $(Z,\iota)$ is a terminal central cone of $\TT$ at $X$, then $\iota$ is a monomorphism.
\end{lemma}
\begin{proof}
	Let us consider $f,g:Y\to Z$ such that $\iota\circ f=\iota\circ g$; this
	morphism is a central cone at $X$ (Lemma~\ref{lem:precompose-central}), and
	since $(Z,\iota)$ is a terminal central cone, it factors uniquely through
	$\iota$. Thus $f = g$ and therefore $\iota$ is monic.
\end{proof}

\begin{proof}[Proof of Theorem~\ref{th:submonad-from-cone}]
	First let us describe the functorial structure of $\ZZ$. Recall that $\ZZ$
	maps every object $X$ to its terminal central cone at $X$.  Let $f:X\to
	Y$ be a morphism. We know that $\TT f\circ\iota_X:\ZZ X\to \TT Y$ is a
	central cone according to Lemma~\ref{lem:postcompose-central}.
	Therefore, we define $\ZZ f$ as the unique map such that the following
	diagram commutes:
	\[\input{z-functor.tikz}\]

	It follows directly that $\ZZ$ maps the identity to the identity, and that
	$\iota$ is natural.  $\ZZ$ also preserves composition, which follows by
	the commutative diagram below.
	\[\input{z-functor-comp.tikz}\]
	This proves that $\ZZ$ is a functor. Next, we describe its monad structure
	and after that we show that it is commutative. \\ The monadic unit
	$\eta_X$ is central, because it is the identity morphism in
	$Z(\CC_\TT)$, thus it factors through $\iota_X$ to define $\eta^\ZZ_X$.
	\[\input{central-unit.tikz}\]
	Next, observe that, by definition, $\mu_X\circ\TT\iota_X\circ\iota_{\ZZ X} =
	\iota_X\odot\iota_{\ZZ X}$, where $(- \odot -)$ indicates Kleisli
	composition. Since $\iota$ is central and Kleisli composition preserves
	central morphisms (see Definition~\ref{def:central-morph}, central morphisms
	form a subcategory of the Kleisli category), it follows that this morphism
	factors through $\iota_X$ and we use this to define $\mu^\ZZ_X$ as in the
	diagram below.
	\[\input{central-mult.tikz}\]
	Again, by definition, $\tau_{A,B}\circ(A\otimes\iota_B) = A\otimes_r\iota_B$.
	Central morphisms are preserved by the premonoidal products (see
	\ref{sub:premonoidal}) and therefore, this morphism factors through
	$\iota_{A\otimes B}$ which we use to define $\tau^\ZZ_{A,B}$ as in the
	diagram below.
	\[\input{central-strength.tikz}\]
	Note that the last three diagrams are exactly those of a morphism of strong
	monads (see Definition \ref{def:morph-monad}). Using the fact that $\iota$ is
	monic (see Lemma~\ref{lem:monic}), the following commutative diagram shows
	that $\eta^\ZZ$ is natural.
	\[\input{eta-z-natural.tikz}\]
	(1) definition of $\eta^\ZZ$, (2) $\iota$ is natural, (3) $\eta$ is natural
	and (4) definition of $\eta^\ZZ$. Thus, we have proven that for any $f \colon
	X\to Y$, $\iota_Y\circ\ZZ f\circ\eta^\ZZ_X = \iota_Y\circ\eta^\ZZ_Y\circ f$.
	Besides, $\iota$ is monic, thus $\ZZ f\circ\eta^\ZZ_X =\eta^\ZZ_Y\circ f$
	which proves that $\eta^\ZZ$ is natural.  We will prove all the remaining
	diagrams with the same reasoning. 

	The following commutative diagram shows that $\mu^\ZZ$ is natural.
	\[\input{mu-z-natural.tikz}\]
	(1) definition of $\mu^\ZZ$, (2) $\iota$ is natural, (3) $\mu$ is natural,
	(4) $\iota$ is natural and (5) definition of $\mu^\ZZ$. \\ The
	following commutative diagrams shows that $\tau^\ZZ$ is natural.
	\[\input{tau-z-natural.tikz}\]
	(1) definition of $\tau^\ZZ$, (2) $\iota$ is natural, (3) $\tau$ is natural,
	(4) $\iota$ is natural and (5) definition of $\tau^\ZZ$. \\
	\[\input{tau-z-natural-1.tikz}\]
	(1) definition of $\tau^\ZZ$, (2) $\iota$ is natural, (3) $\tau$ is natural,
	(4) $\iota$ is natural and (5) definition of $\tau^\ZZ$. \\ The following
	commutative diagrams prove that $\ZZ$ is a monad.
	\[\scalebox{0.8}{\input{z-monad-1.tikz}}\]
	(1) and (2) involve the definition of $\mu^\ZZ$ and the naturality of $\iota$
	and $\mu^\ZZ$, (3) is by definition of monad, (4) definition of $\mu^\ZZ$ and
	(5) also.  (6) and (7) involve the definition of $\eta^\ZZ$ and the
	naturality of $\iota$ and $\eta^\ZZ$, (8) is by definition of monad, (9)
	definition of $\mu^\ZZ$ and (10) also. \\
	$\ZZ$ is proven strong with very similar diagrams. The commutative diagram:
	\begin{equation}\label{eq:proof-commutative}
		\scalebox{0.7}{\input{z-monad-commutative.tikz}}
	\end{equation}
	proves that $\ZZ$ is a commutative monad, with (1) $\tau'^\ZZ$ is natural,
	(2) definition of $\tau^\ZZ$, (3) $\tau^\ZZ$ is natural, (4) $\CC$ is
	monoidal, (5) definition of $\tau'^\ZZ$, (6) $\iota$ is natural, (7)
	definition of $\mu^\ZZ$, (8) definition of $\tau^\ZZ$, (9) $\iota$ is
	central, (10) definition of $\tau'^\ZZ$, (11) $\iota$ is natural and (12)
	definition of $\mu^\ZZ$.
\end{proof}

Theorem~\ref{th:submonad-from-cone} shows that centralisable monads always
induce a canonical commutative submonad. Next, we justify why this submonad
should be seen as the centre of $\TT.$ Note that since $\ZZ$ is a submonad of
$\TT$, we know that $\CC_\ZZ$ canonically embeds into $\CC_\TT$ (see
Proposition \ref{prop:embed}).  The next theorem shows that this embedding
factors through the premonoidal centre of $\CC_\TT$, and moreover, the two
categories are isomorphic.

\begin{theorem}\label{th:iso-of-categories}
	In the situation of Theorem \ref{th:submonad-from-cone}, the canonical
	embedding functor $\II \colon \CC_\ZZ\to\CC_\TT$ corestricts to an isomorphism of
	categories $\CC_\ZZ\cong Z(\CC_\TT)$.
\end{theorem}
\begin{proof}
	$\II$ corestricts as indicated follows easily: for any morphism $f \colon X \to
	\ZZ Y$, we have that $\II f = \iota_Y \circ f$ which is central by Lemma
	\ref{lem:precompose-central}.  Let us write $\hat \II$ for the corestriction
	of $\II$ to $Z(\CC_\TT)$. Next, to prove that $\hat \II \colon \CC_\ZZ \to
	Z(\CC_\TT)$ is an isomorphism, we define the inverse functor
	$G:Z(\CC_\TT)\to\CC_\ZZ$.

	On objects, we have $G(X) \eqdef X.$ To define its mapping on morphisms,
	observe that if $f \colon X\to\TT Y$ is a central morphism (in the
	premonoidal sense), then $(X,f)$ is a central cone of $\TT$ at $Y$
	(Proposition \ref{prop:central}) and therefore there exists a unique morphism
	$f^\ZZ \colon X\to\ZZ Y$ such that $\iota_Y \circ f^\ZZ = f$; we define $Gf
	\eqdef f^\ZZ$.  The proof that $G$ is a functor is direct considering that
	any $f^\ZZ$ is a morphism of central cones and that all components of $\iota$
	are monomorphisms.

	To show that $\hat \II$ and $G$ are mutual inverses, let $f:X\to \TT Y$ be a
	morphism of $Z(\CC_\TT)$, \emph{i.e.}, a central morphism. Then, $\hat \II
	Gf=\iota_Y\circ f^\ZZ = f$ by definition of morphism of central cones (see
	Definition~\ref{def:central-cone}). For the other direction, let $g:X\to \ZZ
	Y$ be a morphism in $\CC$.  Then, $\iota_Y\circ G\hat \II g=\iota_Y\circ
	(\iota_Y\circ g)^\ZZ = \iota_Y \circ g$ by Definition~\ref{def:central-cone}
	and thus $G\hat \II g=g$ since $\iota_Y$ is a monomorphism (Lemma
	\ref{lem:monic}).
\end{proof}

It should now be clear that Theorem \ref{th:submonad-from-cone} and Theorem
\ref{th:iso-of-categories} show that we are justified in naming the submonad
$\ZZ$ as \emph{the} centre of $\TT$.  The existence of terminal central cones
is not only sufficient to construct the centre (as we just showed), but it also
is necessary and we show this next. Furthermore, we provide another equivalent
characterisation in terms of the premonoidal structure of the monad.

\begin{theorem}[Centre]\label{th:centralisability}
	Let $\CC$ be a symmetric monoidal category and $\TT$ a strong monad on it.
	The following are equivalent:
	\begin{enumerate}
		\item \label{cond:1} For any object $X$ of $\CC$, $\TT$ admits a terminal
			central cone at $X$;
		\item \label{cond:2} There exists a commutative submonad $\ZZ$ of $\TT$
			(which we call \emph{the centre} of $\TT$) such that the canonical
			embedding functor $\II:\CC_\ZZ\to\CC_\TT$ corestricts to an
			isomorphism of categories $\CC_\ZZ\cong Z(\CC_\TT)$;
		\item \label{cond:3} The corestriction of the Kleisli left adjoint
			$\JJ:\CC\to\CC_\TT$ to the premonoidal centre $\hat \JJ:\CC\to
			Z(\CC_\TT)$ also is a left adjoint.
	\end{enumerate}
\end{theorem}
\begin{proof}
	We follow a circular strategy in order to prove that each of the points
	implies the others.

	$(1\Rightarrow 2):$
	By Theorem~\ref{th:submonad-from-cone} and Theorem~\ref{th:iso-of-categories}.

	$(2\Rightarrow 3):$
	Let us consider the Kleisli left adjoint $\JJ^\ZZ$ associated to the
	monad $\ZZ$. All our hypotheses can be summarised by the diagram
	\[\input{main-theorem-2-to-3.tikz}\] where $\hat \II \colon \CC_\ZZ\cong
	Z(\CC_\TT)$ is the corestriction of $\II$. This diagram commutes,
	because $\ZZ$ is a submonad of $\TT$ (recall also that $\hat \JJ$ is
	the indicated corestriction of $\JJ$, see \secref{sub:premonoidal}).
	Since $\hat \II$ is an isomorphism, then $\hat \JJ = \hat \II \circ
	\JJ^\ZZ$ is the composition of two left adjoints and it is therefore
	also a left adjoint.

	$(3\Rightarrow 1):$
	Let $\mathcal R \colon Z(\CC_\TT)\to\CC$ be the right adjoint of $\hat \JJ$
	and let $\varepsilon$ be the counit of the adjunction.  We will show that the
	pair $(\mathcal R X, \varepsilon_X)$ is the terminal central cone of $\TT$ at
	$X$.\\ First, since $\varepsilon_X$ is a morphism in $Z(\CC_\TT)$, it follows
	that it is central. Thus the pair $(\mathcal R X, \varepsilon_X)$ is a
	central cone of $\TT$ at $X$. Next, let $\Phi \colon Z(\CC_\TT)[\hat \JJ Y,
	X] \cong \CC[Y, \mathcal R X]$ be the natural bijection induced by the
	adjunction. If $f \colon Y \to \TT X$ is central, meaning a morphism of
	$Z(\CC_\TT)$, the diagram below left commutes in $Z(\CC_\TT)$, or
	equivalently, the diagram below right commutes in $\CC$:
	\[\input{adjoint-diagram-unique-centre.tikz}\] Note that the pair $(Y, f)$ is
	equivalently a central cone for $\TT$ at $X$ (by Proposition
	\ref{prop:central}).  Thus $f$ uniquely factors through the counit
	$\varepsilon_X: \mathcal RX\to\TT X$ and therefore $(\mathcal
	RX,\varepsilon_X)$ is the terminal central cone of $\TT$ at $X$.
\end{proof}

This theorem shows that Definition \ref{def:centralisable} may be stated by
choosing any one of the above equivalent criteria. We note that the first
condition is the easiest to verify in practice. The second one is the
most useful for providing a computational interpretation, as we do in the
sequel. The third condition provides an important link to premonoidal
categories.

\begin{example}
	\label{ex:free-monad}
	Given a monoid $(M,e,m)$, consider the free monad induced by $M$, also known
	as the \emph{writer monad}, which we write as $\TT = (- \times M) \colon \Set
	\to \Set$.  The centre $\ZZ$ of $\TT$ is given by the commutative monad
	$(-\times Z(M)) \colon \Set \to \Set$, where $Z(M)$ is the centre of the
	monoid $M$ and where the monad data is given by the (co)restrictions of the
	monad data of $\TT$. Note that $\TT$ is a commutative monad iff $M$ is a
	commutative monoid. See also Example \ref{ex:counter}.
\end{example}

\subsection{A Non-centralisable Monad}
\label{sub:bad-monad}

In $\Set$, the terminal central cones used to define the centre are defined by
taking appropriate subsets.  One may wonder what happens if not every subset of
a given set is an object of the category. The following example describes such
a situation, which gives rise to a non-centralisable strong monad.

\begin{example}
	\label{ex:counter}
	Consider the Dihedral group $\mathbb{D}_4,$ which has $8$ elements. Its
	centre $Z(\mathbb{D}_4)$ is non-trivial and has 2 elements. Let $\CC$ be the
	full subcategory of $\Set$ with objects that are finite products of the set
	$\mathbb{D}_4$ with itself. This category has a cartesian structure, and the
	terminal object is the singleton set (which is the empty product). Notice
	that every object in this category has a cardinality that is a power of $8$.
	Therefore the cardinality of every homset of $\CC$ is a power of $8$. Since
	$\CC$ has a cartesian structure and since $\mathbb{D}_4$ is a monoid, we can
	consider the writer monad $\MM \eqdef (- \times \mathbb D_4) \colon \CC \to \CC$
	induced by $\mathbb{D}_4$, which can be defined in the same way as in Example
	\ref{ex:free-monad}. It follows that $\MM$ is a strong monad on $\CC$.
	However, it is easy to show that this monad is not centralisable. Assume (for
	contradiction) that there is a monad $\ZZ \colon \CC \to \CC$ such that
	$\CC_\ZZ\cong Z(\CC_\MM)$ (see Theorem \ref{th:centralisability}). Next,
	observe that the homset $Z(\CC_\MM)[1,1]$ has the same cardinality as the
	centre of the monoid $\mathbb D_4$, \emph{i.e.}, its cardinality is $2$. However,
	$\CC_\ZZ$ cannot have such a homset since $\CC_\ZZ[X,Y] = \CC[X,\ZZ Y]$ which
	must have cardinality a power of $8$.  Therefore there exists no such monad
	$\ZZ$ and $\MM$ is not centralisable.
\end{example}

Besides this example and any further attempts at constructing non-centralisable
monads for this sole purpose, we do not know of any other strong monad in the
literature that is not centralisable.
In the next section, we present many
examples of centralisable monads and classes of centralisable monads which show
that our results are widely applicable.

\section{Examples of Centres of Strong Monads}
\label{sec:examples}

In this section, we show how we can make use of the mathematical results we
already established in order to reason about the centres of monads of interest.

\subsection{Categories whose Strong Monads are Centralisable}
\label{sub:examples}

We saw earlier that every (strong) monad on $\Set$ is centralisable. In fact,
this is also true for many other naturally occurring categories.  For example,
in many categories of interest, the objects of the category have a suitable
notion of subobject (e.g., subsets in $\Set$, subspaces in $\mathbf{Vect}$) and
the centre can be constructed in a similar way to the one in $\Set$.

\begin{example}
	Let $\DCPO$ be the category whose objects are directed-complete partial
	orders and whose morphisms are Scott-continuous maps between them.  Every
	strong monad on $\DCPO$ with respect to its cartesian structure is
	centralisable. The easiest way to see this is to use Theorem
	\ref{th:centralisability} \eqref{cond:1}. Writing $\TT \colon \DCPO \to \DCPO$ for
	an arbitrary strong monad on $\DCPO$, the terminal central cone of $\TT$ at
	$X$ is given by the subdcpo $\ZZ X \subseteq \TT X$ which has the underlying
	set
	\( \ZZ X \eqdef \{ t \in \TT X \ |\ \forall Y \in \Ob(\DCPO). \forall s \in \TT Y.\  \) \(\mu(\TT\tau'(\tau(t,s))) = \mu(\TT\tau(\tau'(t,s)))  \} . \)
	That $\ZZ X$ (with the inherited order) is a subdcpo of $\TT X$ follows
	easily by using the fact that $\mu, \tau, \tau'$ and $\TT$ are
	Scott-continuous. Therefore, the construction is fully analogous to the one
	in $\Set$.
\end{example}

\begin{example}
	Let $\TOP$ be the category whose objects are topological
	spaces, and whose morphisms are continuous maps between them.
	Every strong monad on $\TOP$ with respect to its cartesian structure
	is centralisable.
	Using Theorem \ref{th:centralisability} \eqref{cond:1} and writing $\TT :
	\TOP \to \TOP$ for an arbitrary strong monad on $\TOP$, the terminal central
	cone of $\TT$ at $X$ is given by the space $\ZZ X \subseteq \TT X$ which has
	the underlying set $ \ZZ X \eqdef \{ t \in \TT X \ |\ \forall Y \in
	\Ob(\TOP). \forall s \in \TT Y.$  $\mu(\TT\tau'(\tau(t,s))) =
	\mu(\TT\tau(\tau'(t,s)))  \} $ and whose topology is the subspace topology
	inherited from $\TT X$.
\end{example}

\begin{example}
	\label{ex:meas}
	Every strong monad on the category $\Meas$ (whose objects are measurable
	spaces and the morphisms are measurable maps between them) is centralisable.
	The construction is fully analogous to the previous example, but instead of
	the subspace topology, we equip the underlying set with the subspace
	$\sigma$-algebra inherited from $\TT X$ (which is the smallest
	$\sigma$-algebra that makes the subset inclusion map measurable).
\end{example}

\begin{example}
	Let $\mathbf{Vect}$ be the category whose objects are vector spaces, and
	whose morphisms are linear maps between them. Every strong monad on
	$\mathbf{Vect}$ with respect to the usual symmetric monoidal structure is
	centralisable. One simply defines the subset $\ZZ X$ as in the other examples
	and shows that this is a linear subspace of $\TT X$. That this is the
	terminal central cone is then obvious.
\end{example}

The above categories, together with the category $\Set$, are not meant to
provide an exhaustive list of categories for which all strong monads are
centralisable. Indeed, there are many more categories for which this is true.
The purpose of these examples is to illustrate how we may use Theorem
\ref{th:centralisability} \eqref{cond:1} to construct the centre of a strong
monad.  Changing perspective, the proof of the next proposition
uses Theorem \ref{th:centralisability} \eqref{cond:3}.

\begin{proposition}\label{prop:suffcond}
	Let $\mathbf{C}$ be a symmetric monoidal closed category that is total --
	\emph{i.e.}, a locally small category whose Yoneda embedding has a left
	adjoint. Then all strong monads over $\mathbf{C}$ are centralisable.
\end{proposition}
\begin{proof}
	This proof was provided by Titouan Carette in \cite{nous23monads}.
\end{proof}

\begin{example}
	Any category which is the Eilenberg-Moore category of a commutative monad
	over $\Set$ is total \cite{kelly1986survey}. Furthermore it is symmetric
	monoidal closed \cite{keigher1978symmetric}, thus all strong monads on it are
	centralisable.  This includes: the category $\Set_* $ of pointed sets and
	point preserving functions (algebras of the lift monad); the category
	$\mathbf{CMon}$ of commutative monoids and monoid homomorphisms (algebras of
	the commutative monoid monad); the category $\mathbf{Conv}$ of convex sets
	and linear functions (algebras of the distribution monad); and the category
	$\mathbf{Sup}$ of complete semilattices and sup-preserving functions
	(algebras of the powerset monad).
\end{example}

\begin{example}
	Any presheaf category $\Set^{\mathbf{C}^{\mathrm{op}}}$ over a small category
	$\mathbf{C}$ is total \cite{kelly1986survey} and cartesian closed, thus all strong monads on it (with respect to the cartesian structure) are
	centralisable. This includes: the category $\Set^{A^{\mathrm{op}}}$, where $A$ is the
	category with two objects and two parallel arrows, which can be seen as the
	category of directed multi-graphs and graph homomorphisms; the category
	$\Set^{G^{\mathrm{op}}}$, where $G$ is a group seen as a category, which can be seen
	as the category of $G$-sets (sets with an action of $G$) and equivariant
	maps; and the topos of trees $\Set^{\mathbb{N}^{\mathrm{op}}}$. If
	$\mathbf{C}$ is symmetric monoidal, then the Day convolution product makes
	$\Set^{\mathbf{C}^{\mathrm{op}}}$ symmetric monoidal closed \cite{day1970construction},
	hence all strong monads on it with respect to the Day convolution monoidal
	structure also are centralisable.
\end{example}

\begin{example}
	Any Grothendieck topos is cartesian closed and total, therefore it satisfies
	the conditions of Proposition \ref{prop:suffcond}.
\end{example}

\subsection{Specific Examples of Centralisable Monads}
\label{sub:specific-examples}

In this subsection, we consider specific monads and construct their centres.

\begin{example}
	\label{ex:commutative}
	Every commutative monad is naturally isomorphic to its centre.
\end{example}

\begin{example}
	\label{ex:continuation}
	Let $S$ be a set and consider the well-known continuation monad $\TT =
	[[-,S],S]:\Set\to\Set$.  Note that, if $S$ is the empty set or a singleton
	set, then $\TT$ is commutative, so we are in the situation of Example
	\ref{ex:commutative}.  Otherwise, when $S$ is not trivial, one can prove
	(details omitted here) that $\ZZ X = \eta_X(X) \cong X$. Therefore, the
	centre of $\TT$ is trivial and it is naturally isomorphic to the identity
	monad.
\end{example}

\begin{example}
	\label{ex:list}
	Consider the well-known list monad $T:\Set\to\Set$ that is given by $TX = \bigsqcup_{n\geq 0} X^n.$
	Then, the centre of $\TT$ is naturally isomorphic to the identity monad.
\end{example}

Example \ref{ex:continuation} shows that the centre of a monad may be trivial
in the sense that it is precisely the image of the monadic unit and this is the
least it can be.  At the other extreme, Example \ref{ex:commutative} shows that
the centre of a commutative monad coincides with itself, as one would expect.
Thus, the monads that have interesting centres are those monads which are
strong but not commutative, and which have non-trivial centres, such as the one
in Example \ref{ex:free-monad}. Another interesting example of a strong monad
with a non-trivial centre is provided next.

\begin{example}
	\label{ex:semiring}
	Every semiring $(S,+,0,\cdot,1)$ induces a monad $\TT_S:\Set\to\Set$
	\cite{jakl2022bilinear}.  This monad maps a set $X$ to the set of finite
	formal sums of the form $\sum s_i x_i$, where $s_i$ are elements of
	$S$ and $x_i$ are elements of $X$.  The monad $\TT_S$ is commutative iff
	$S$ is commutative as a semiring.  The centre
	$\ZZ$ of $\TT_S$ is induced by the commutative semiring $Z(S)$, \emph{i.e.},
	by the centre of $S$ in the usual sense. Therefore, $\ZZ = \TT_{Z(S)}.$
\end{example}

\begin{example}
	Any Lawvere theory $\lawT$ \cite{hyland2007lawvere}
	induces a finitary monad on $\Set$. The centre
	of this monad is the monad induced by the centre of $\lawT$
	in the sense of Lawvere theories \cite{wraith-lecture}.
	This is detailed in \secref{sub:lawvere}.
\end{example}

\begin{example}
	\label{ex:valuations-monad}
	The valuations monad $\mathcal V \colon \DCPO \to \DCPO$
	\cite{jones-plotkin,jones90} is similar in spirit to the Giry monad on
	measurable spaces \cite{giry}. It is an important monad in domain theory
	\cite{domain-theory} that is used to combine probability and recursion for
	dcpo's.
	Given a dcpo $X$, the valuations monad $\VVV$ assigns
	the dcpo $\VVV X$ of all Scott-continuous \emph{valuations} on $X$, which are
	Scott-continuous functions $\nu \colon \mathcal \sigma(X) \to [0,1]$ from the
	Scott-open sets of $X$ into the unit interval that satisfy some additional
	properties that make them suitable to model probability (details omitted
	here, see \cite{jones90} for more information).
	The category $\DCPO$ is cartesian closed and the valuations monad $\mathcal V
	\colon \DCPO \to \DCPO$ is strong, but its commutativity on $\DCPO$ has been
	an open problem since 1989 \cite{jones90,jones-plotkin,central-valuations,valuation-monads,valuation-statistical}.  The
	difficulty in (dis)proving the commutativity of $\VVV$ boils down to
	(dis)proving the following Fubini-style equation
	\[ \int_X \int_Y \chi_{U}(x,y)d \nu d \xi = \int_Y \int_X \chi_{U}(x,y)d \xi d \nu  \]
	holds for any dcpo's $X$ and $Y$, any Scott-open subset $U \in \sigma(X \times Y)$
	and any two valuations $\xi \in \VVV X$ and $\nu \in \VVV Y.$
	In the above equation, the notion of integration is given by the
	\emph{valuation integral} (see \cite{jones90} for more information).

	The \emph{central valuations monad} \cite{central-valuations}, is the submonad $\ZZ \colon \DCPO
	\to \DCPO$ that maps a dcpo $X$ to the dcpo $\ZZ X$ which has all
	\emph{central valuations} as elements. Equivalently:
	\begin{align*}
		\ZZ X \eqdef \Biggl\{ \xi \in \VVV(X)\ |\ \forall Y \in \Ob(\DCPO). \forall U \in \sigma(X \times Y). \\
		\forall \nu \in \VVV(Y) .
		\int_X \int_Y \chi_{U}(x,y)d \nu d \xi = \int_Y \int_X \chi_{U}(x,y)d \xi d \nu \Biggr\} .
	\end{align*}
	But this is precisely the centre of $\VVV$, which can be seen using Theorem
	\ref{th:centralisability} \eqref{cond:1} after unpacking the definition of
	the monad data of $\VVV$.  Therefore, we see that the main result of
	\cite{central-valuations} is a special case of our more general categorical
	treatment.  We wish to note, that the centre of $\VVV$ is quite large. It
	contains all three commutative submonads identified in
	\cite{valuation-monads} and all of them may be used to model lambda calculi
	with recursion and discrete probabilistic choice (see
	\cite{valuation-monads,central-valuations}).
\end{example}

\subsection{Link with Lawvere theories}
\label{sub:lawvere}

Commutants for Lawvere theories \cite{hyland2007lawvere} were defined
in Wraith's lecture notes \cite{wraith-lecture} but were only
studied in details by Lucyshyn-Wright \cite{rory2018convex} later.
The centre of a Lawvere theory is a special case of commutant.

In a Lawvere theory $\lawT$, we say that $f:A^n\to A^{n'}$ and $g:A^m\to
A^{m'}$ commute if and only if $f^{m'}\circ g^n$ (also written $f\star g$) and
$g^{n'}\circ f^m$ (also written $g\star f$) are equal, up to isomorphism.  If
$\lawS$ is a full subcategory of $\lawT$, one can define the commutant of
$\lawS$ in $\lawT$, meaning a full subcategory of $\lawT$ whose morphisms
commute with the morphisms of $\lawS$. This commutant is written $\lawS^\bot$,
and is also a Lawvere subtheory of $\lawT$.  Considering this, $\lawT^\bot$ is
seen as the \emph{centre} of the Lawvere theory $\lawT$ \cite{wraith-lecture};
and any subtheory of $\lawT^\bot$ is a \emph{central subtheory} of $\lawT$.

What about monads? Models of a Lawvere theory $\lawT$ are
finite-product-preserving functors $\lawT\to\Set$ and they form a category
$\Mod(\lawT,\Set)$.  This category is adjoint to $\Set$ through a forgetful and
free functors. Those adjunctions give birth to a monad. This monad is on
$\Set$, thus it is strong, centralisable and finitary since it originates from
a Lawvere theory. Thus given a Lawvere theory $\lawT$, we obtain a monad $\TT$
whose centre $\ZZ$ is a commutative submonad of $\TT$ and is finitary, which
means that there exists a corresponding Lawvere theory. This Lawvere theory is
a commutative subtheory of $\lawT$, as proven next.

The connection between Lawvere theories and finitary monads
is extensively detailed in \cite{street1972formal,garner2014lawvere,garner2018enriched}.
To get a Lawvere theory out of a finitary monad $\ZZ$ on $\Set$,
one needs to look at the opposite category of a skeleton
of $\Set_\ZZ$ \cite{hyland2007lawvere}, noted here $\mathfrak s\Set_\ZZ^{op}$.
This Lawvere theory is commutative because
$\Set_\ZZ$ is monoidal. Moreover, $\Set_\ZZ$ is embedded in
$\Set_\TT$, then $\mathfrak s\Set_\ZZ^{op}$ is embedded in
$\mathfrak s\Set_\TT^{op}$ ; the latter being equivalent
to $\lawT$.

\begin{theorem}
	\label{th:linklawvere}
	Given a Lawvere theory $\lawT$, its $\Set$-monad $\TT$ is
	centralisable and its centre $\ZZ$
	has a corresponding Lawvere theory $\mathfrak s\Set_\ZZ^{op}$
	that is equivalent to $\lawT^\bot$.
\end{theorem}
\begin{proof}
	This is a direct application of the point (2) of
	Theorem~\ref{th:centralisability}.
\end{proof}

This connection helps motivate a similar theory for commutants in the general
context of strong monads. However, the litterature on Lawvere theories is not
enough to grasp all those monads on symmetric monoidal category: in this
subsection, we have only given the example for the category $\Set$, and in
general, in the literature, the category is often required to be closed.

\section{Central Submonads}
\label{sec:real-central}

So far, we focused primarily on \emph{the} centre of a strong monad. Now we
focus our attention on \emph{central submonads} of a strong monad which we
define by taking inspiration from the notion of central subgroup in group
theory. Just like central subgroups, central submonads are more general
compared to the centre. The centre of a strong monad, whenever it
exists, can be intuitively understood as the largest central submonad, so the
two notions are strongly related. We will later see that central submonads are
more interesting computationally.

\begin{theorem}[Centrality]\label{th:centrality}
	Let $\CC$ be a symmetric monoidal category and $\TT$ a strong monad on it.
	Let $\SC$ be a strong submonad of $\TT$ with $\iota:\SC\Rightarrow\TT$ the
	strong submonad monomorphism.  The following are equivalent:
	\begin{enumerate}
		\item[1)] \label{ccond:1} For any object $X$ of $\CC$,
			$(\SC X,\iota_X)$ is a central cone for $\TT$ at $X$;
		\item[2)] \label{ccond:2} the canonical embedding functor
			$\II:\CC_\SC\to\CC_\TT$ corestricts to an embedding
			of categories $\hat\II:\CC_\SC\to Z(\CC_\TT)$.
	\end{enumerate}
	Furthermore, these conditions imply that $\SC$ is a commutative submonad
	of $\TT$. Under the additional assumption that $\TT$ is centralisable,
	these conditions also are equivalent to:
	\begin{enumerate}
		\item [3)] $\SC$ is a submonad of the centre of $\TT$, and thus is
			commutative.
	\end{enumerate}
\end{theorem}
\begin{proof}
	\

	$(1\Rightarrow 2):$
	The proof of Th.~\ref{th:iso-of-categories} contains the necessary
	elements for this proof.  In details, we know that all the components
	of $\iota$ are central, and we also know that precomposing a central
	morphism keeps being central (see Lemma~\ref{lem:precompose-central}).

	$(2\Rightarrow 1):$
	The hypothesis ensures that $\hat\II(id_X)=\iota_X$ is central.

	The diagram in (\ref{eq:proof-commutative}) proves that the centre of a
	centralisable monad is commutative.  Assuming (1) -- or (2) -- is true,
	then the same diagram replacing $\ZZ$ by $\SC$ proves that $\SC$ is a
	commutative monad.

	$(1\Rightarrow 3):$
	Moreover, each $\iota^\SC_X:\SC X\Rightarrow\TT X$ factorises through
	the terminal central cone $\iota^\ZZ_X$.  A strong monad morphism
	$\SC\Rightarrow\ZZ$ arises from those factorisations.

	$(3\Rightarrow 1):$
	Let us write $\ZZ$ the centre of $\TT$, $\iota^\SC:\SC\Rightarrow\ZZ$
	and $\iota^\ZZ:\ZZ\Rightarrow\TT$ the submonad morphisms. The
	components of $\iota^\ZZ$ are terminal central cones, and are in
	particular central, so $\iota^\ZZ\circ\iota^\SC$ is also central by
	Lemma~\ref{lem:precompose-central}.  Thus the components of the
	submonad morphism from $\SC$ to $\TT$ are central.
\end{proof}

\begin{definition}[Central Submonad]\label{def:central-sub}
	Given a strong submonad $\SC$ of $\TT$, we say that $\SC$ is a \emph{central
	submonad} of $\TT$ if it satisfies any one of the above equivalent criteria
	from Theorem \ref{th:centrality}.
\end{definition}

Just like the centre of a strong monad, any central submonad also is
commutative and the above theorem (Theorem~\ref{def:central-sub}) shows that
central submonads have a similar structure to the centre of a strong monad. The
final statement shows that we may see the centre (whenever it exists) as the
largest central submonad of $\TT$. The centre of a strong monad often does
exist (as we already argued), so the last criterion also provides a simple way
to determine whether a submonad is central or not.

\begin{example}
	By the above theorem, every centre described in \secref{sec:examples} is a
	central submonad.
\end{example}

\begin{example}
	Let $\TT$ be a strong monad on a symmetric monoidal category $\CC,$ such that
	all unit maps $\eta_X \colon X \to \TT X$ are monomorphisms (this is often
	the case in practice). Then, the identity monad on $\CC$ is a central
	submonad of $\TT.$
\end{example}

\begin{example}
	Given a monoid $M,$ let $\TT = (M \times -)$ be the monad on $\Set$ from
	Example~\ref{ex:free-monad}. Any submonoid $S$ of $Z(M)$ induces a central
	submonad $(S \times -)$ of $\TT$.
\end{example}

\begin{example}
	Given a semiring $R$, consider the monad $\TT_R$ from Example
	\ref{ex:semiring}. Any subsemiring $S$ of $Z(R)$ induces a central submonad
	$\TT_{S}$ of $\TT_R.$
\end{example}

\begin{example}
	A notion of \emph{central Lawvere subtheory} can be introduced in an obvious
	way. It induces a central submonad of the monad induced by the
	original Lawvere theory.
\end{example}

\begin{example}
	The three commutative submonads identified in \cite{valuation-monads} are
	central submonads of the valuations monad $\VV$ from Example
	\ref{ex:valuations-monad}, because each one of them is a commutative submonad
	of the centre of $\VV$ \cite{central-valuations}.
\end{example}

\begin{remark}
	Given an arbitrary monoid $M$ (on $\Set$), there could be a commutative
	submonoid $S$ of $M$ that is not central (\emph{i.e.}, its elements do not commute
	with all elements of $M$). The same holds for strong monads.  For instance,
	let $M = \mathbb{D}_4$ (see Example \ref{ex:counter}) and let $S$ be the
	submonoid of $M$ that contains only the rotations (of which there are
	four). Then, $S$ is a commutative submonoid that is not central.  By taking
	the free monads induced by these monoids (see Example \ref{ex:free-monad})
	on $\Set$, we get an example of a commutative submonad that is not central.
	Moreover, if we take $\DD$ to be the full subcategory of $\Set$ whose
	objects have cardinality that is different from two, then $\DD$ has a
	cartesian structure and the writer monads induced by $S$ and $M$ on $\DD$
	give an example of a non-centralisable strong monad that admits a
	commutative non-central submonad. In this situation, the identity
	monad on $\DD$ gives an example of a central (commutative) submonad
	even though the ambient monad (induced by $M$) is not centralisable.
\end{remark}

\section{Computational Interpretation}
\label{sec:computational}

In this section, we provide a computational interpretation of our ideas. We
consider a simply-typed lambda calculus together with a strong monad $\TT$ and
a \emph{central submonad} $\SC$ of $\TT$.  We call this system the
\emph{Central Submonad Calculus (CSC)}.  We describe its equational theories,
formulate appropriate categorical models for it and we prove soundness,
completeness and internal language results for our semantics.

\subsection{Syntactic Structure of the Central Submonad Calculus}
\label{sub:computational-syntax}

We begin by describing the types we use. The grammar of types (see Figure
\ref{fig:grammars}) are just the usual ones with one addition -- we extend the
grammar by adding the family of types $\SC A$. The type $\TT A$ represents the
type of monadic computations for our monad $\TT$ that produce values of type
$A$ (together with a potential side effect described by $\TT$).  The type $\SC
A$ represents the type of \emph{central} monadic computations for our monad
$\TT$ that produce values of type $A$ (together with a potential \emph{central}
side effect that is in the submonad $\SC$).  Some terms and formation rules can
be expressed in the same way for types of the form $\SC A$ or $\TT A$ and in
this case we simply write $\XX A$ to indicate that $\XX$ may range over $\{\SC,
\TT \}.$

The grammar of terms and their formation rules are described in Figure
\ref{fig:grammars}.  The first six rules in Figure \ref{fig:grammars} are just
the usual formation rules for a simply-typed lambda calculus with pair types.
Contexts are considered up to permutation and without repetition and all
judgements we consider are implicitly closed under weakening (which is
important when adding constants).  The $\mathtt{ret}_\XX\ M$ term is used as an
introduction rule for the monadic types and it allows us to see the pure (\emph{i.e.},
non-effectful) computation described by the term $M$ as a monadic one.  The
term $\iota M$ allows us to view a \emph{central} monadic computation as a
monadic (not necessarily central) one. Semantically, it corresponds to applying
the $\iota$ submonad inclusion we saw in previous sections.  Finally, we have
two terms for monadic sequencing that use the familiar $\mathtt{do}$-notation.
The monadic sequencing of two central computations remains central, which is
represented via the $\mathtt{do}_\SC$ terms; the $\mathtt{do}_\TT$ terms are
used for monadic sequencing of (not necessarily central) computations.

\begin{figure}
	\noindent $\text{(Types)}\quad  A, B ~~::= 1 \alt A \to  B\alt A\times B
	\alt \SC A \alt \TT A$ \\
	$\quad$\\
	$\text{(Terms)} \quad M,N ~~ ::=  x \alt * \alt \lambda x^A.M \alt MN \alt \pv M N$ \\
	$\text{ }\qquad\alt \pi_i M \alt \sret M \alt \iota M \alt \tret M$\\
	$\text{ }\qquad\alt \zdo x M N \alt \tdo x M N $
	$\quad$\\
	\[\begin{array}{c}
		\infer{
			\Gamma,x:A\vdash x:A}{}
		\qquad
		\infer{
			\Gamma\vdash MN:B
		}{
			\Gamma\vdash M:A\to B
			&
			\Gamma\vdash N:A}
		\\[1.5ex]
		\infer{
			\Gamma\vdash *:1}{}
		\qquad
		\infer{\Gamma\vdash\lambda x^A.M:A\to B}{\Gamma,x:A\vdash M:B}
		\qquad
		\infer{
			\Gamma\vdash\pi_i M \colon A_i}{\Gamma\vdash M:A_1\times A_2}
		\\[1.5ex]
		\infer{
			\Gamma\vdash \pv{M}{N}: A\times B
		}{
			\Gamma\vdash M:A
			&
			\Gamma\vdash N:B
		}
		\qquad
		\infer{
			\Gamma\vdash \xret M :\XX A}{\Gamma\vdash M:A}
		\\[1.5ex]
		\infer{
			\Gamma\vdash \iota M :\TT A}{\Gamma\vdash M:\SC A}
		\qquad
		\infer{
			\Gamma\vdash\xdo x M N \colon \XX B
		}{
			\Gamma\vdash M:\XX A
			&
			\Gamma, x:A\vdash N:\XX B
		}
	\end{array}\]
	\caption{Grammars and formation rules.}
	\label{fig:grammars}
\end{figure}

\subsection{Equational Theories of the Central Submonad Calculus}
\label{sub:theories}

Next, we describe equational theories for our calculus. We follow the
vocabulary and the terminology in \cite{maietti2005relating} in order to
formulate an appropriate notion of $\CSC$-theory.

\begin{definition}[$\CSC$-theory]
	A $\CSC$-theory is an extension of the Central Submonad Calculus (see
	\secref{sub:computational-syntax}) with new ground types, new term constants
	(which we assume are well-formed in any context, including the empty one) and
	new equalities between types and between terms.
\end{definition}

In a $\CSC$-theory, we have four types of judgements: the judgement $ \vdash A :
\typ$ indicates that $A$ is a (simple) type; the judgement $\vdash A = B :
\typ$ indicates that types $A$ and $B$ are equal; the judgement $\Gamma \vdash
M \colon A$ indicates that $M$ is a well-formed term of type $A$ in context
$\Gamma$, as usual; finally, the judgement $\Gamma \vdash M = N \colon A$ indicates
that the two well-formed terms $M$ and $N$ are equal.

Type judgements and term judgements are described in Figure \ref{fig:grammars}
and type equality judgements in Figure \ref{fig:type-rules}. Following the
principle of judgemental equality, we add type conversion rules in Figure
\ref{fig:variable-conversion-rule}. The rules in Figure \ref{fig:moggi-rules}
are the usual rules that describe the equational theory of the simply-typed
lambda calculus.  As often done by many authors, we implicitly identify terms
that are $\alpha$-equivalent.  The rules for $\beta$-equivalence and
$\eta$-equivalence are explicitly specified.

In Figure \ref{fig:rules}, we present the equational rules for monadic
computation.  The rules on the first three lines -- \emph{(ret.eq), (do.eq),
$(\XX.\beta)$, $(\XX.\eta)$, ($\XX$.assoc)} -- axiomatise the structure of a
strong monad. Because of this, these rules are stated for both monads $\TT$ and
$\SC.$ The rules \emph{($\iota$.mono), ($\iota \SC$.ret)} and \emph{($\iota
\SC$.comp)} are used to axiomatise the structure of $\SC$ as a submonad of
$\TT.$ Intuitively, these rules can be understood as specifying that central
monadic computations can be seen as (general) monadic computations of the
ambient monad $\TT$.  The remainder of the rules are used to axiomatise the
behaviour of $\SC$ as a \emph{central} submonad of $\TT$.  The rule
$(\SC.central)$ is undoubtedly the most important one, because it ensures that
central computations commute with any other (not necessarily central)
computation when performing monadic sequencing with the $\TT$ monad.

\begin{figure*}
	\[\begin{array}{c}
		\infer{
			\vdash A = A \colon \typ
		}{
			\vdash A \colon \typ
		}
		\qquad
		\infer{
			\vdash B=A:\typ
		}{
			\vdash A=B:\typ
		}
		\qquad
		\infer{
			\vdash A=C:\typ
		}{
			\vdash A=B:\typ
			&
			\vdash B=C:\typ
		}
		\\[1.5ex]
		\infer{
			\vdash A \times B = A'\times B' \colon \typ
		}{
			\vdash A = A' :\typ
			&
			\vdash B = B' :\typ
		}
		\qquad
		\infer{
			\vdash A \to B = A'\to B' \colon \typ
		}{
			\vdash A = A' :\typ
			&
			\vdash B = B' :\typ
		}
		\\[1.5ex]
		\infer{
			\vdash \XX A = \XX B :\typ
		}{
			\vdash A = B:\typ
		}
	\end{array}\]
	\caption{Equational rules for types.}
	\label{fig:type-rules}
\end{figure*}

\begin{figure*}
	\[\begin{array}{c}
		\infer{
			\Gamma, x \colon B \vdash M \colon D
		}{
			\vdash A=B:\typ
			&
			\vdash C=D:\typ
			&
			\Gamma, x \colon A \vdash M \colon C
		}
		\\[1.5ex]
		\infer{
			\Gamma, x \colon B \vdash M = N \colon D
		}{
			\vdash A=B:\typ
			&
			\vdash C=D:\typ
			&
			\Gamma, x \colon A \vdash M = N \colon C
		}
	\end{array}\]
	\caption{Type conversion rules.}
	\label{fig:variable-conversion-rule}
\end{figure*}

\begin{figure*}
	\resizebox{\hsize}{!}{
		$
		\begin{array}{c}
			\infer[(ret.eq)]{
				\Gamma\vdash \xret M = \xret N \colon \XX A
			}{
				\Gamma\vdash M = N \colon A
			}
			\qquad
			\infer[(do.eq)]{
				\Gamma\vdash \xdo x M N = \xdo{x}{M'}{N'} \colon \XX B
			}{
				\Gamma\vdash M=M' \colon \XX A
				&
				\Gamma, x:A\vdash N=N' \colon \XX B
			}
			\\[1.5ex]
			\infer[(\XX.\beta)]{
				\Gamma\vdash \xdo{x}{\xret M}{N} = N[M/x] \colon \XX B
			}{
				\Gamma\vdash M \colon A
				&
				\Gamma,x:A\vdash N \colon \XX B
			}
			\qquad
			\infer[(\XX.\eta)]{
				\Gamma\vdash \xdo{x}{M}{\xret x} = M \colon \XX A
			}{
				\Gamma\vdash M \colon \XX A
			}
			\\[1.5ex]
			\infer[(\XX.assoc)]{
				\Gamma\vdash \xdo{y}{(\xdo x M N)}{P} = \xdo{x}{M}{\xdo y N P} \colon \XX C
			}{
				\Gamma\vdash M \colon \XX A
				&
				\Gamma\vdash N \colon \XX B
				&
				\Gamma,x:A,y:B\vdash P \colon \XX C
			}
			\\[1.5ex]
			\infer[(\SC.central)]{
				\Gamma \vdash \tdo{x}{\iota M}{\tdo{y}{N}{P}} = \tdo{y}{N}{\tdo{x}{\iota M}{P}} \colon \TT C
			}{
				\Gamma\vdash M:\SC A
				&
				\Gamma\vdash N \colon \TT B
				&
				\Gamma,x:A,y:B\vdash P:\TT C
			}
			\\[1.5ex]
			\infer=[(\iota.mono)]{
				\Gamma\vdash\iota M=\iota N:\TT A
			}{
				\Gamma\vdash M=N:\SC A
			}
			\qquad
			\infer[(\iota\SC.ret)]{
				\Gamma\vdash \iota~\sret M = \tret M \colon \TT A
			}{
				\Gamma\vdash M \colon A
			}
			\\[1.5ex]
			\infer[(\iota\SC.comp)]{
				\Gamma\vdash \tdo{x}{\iota M}{\iota N} = \iota~\zdo x M N \colon \TT B
			}{
				\Gamma \vdash M \colon \SC A
				&
				\Gamma, x:A \vdash N \colon \SC B
			}
		\end{array}
		$
		}
	\caption{Equational rules for terms of monadic types of $\mathrm{CSC}$.}
	\label{fig:rules}
\end{figure*}

\begin{example}
	\label{ex:theory-writer}
	Let us consider an example of a $\CSC$-theory.  Given a monoid $(M,e,m)$ we now
	axiomatise the writer monad induced by $M$.  A theory for this monad does not
	add any new types, but it adds constants for each element $c$ of $M$:
	$\Gamma\vdash\tact:\TT 1.$
	In this specific theory, we may think of the side-effect computed by monadic
	sequencing as being simply an element of $M$.  The term $\tact$ can be
	understood as performing the monoid multiplication on the right with argument
	$c$, \emph{i.e.}, it applies the function $m(-,c)$ to whatever is the current state
	of the program.

	Let $S$ be a submonoid of the centre $Z(M)$ of $M$.  This makes $S$ a
	\emph{central} submonoid of $M$ (this can be defined in a similar way to
	central subgroups).  We enrich the theory with the following constant and
	rule for each $s$ in $S$:
	\[
		\infer{\Gamma\vdash\zact[s]:\SC 1}{}
		\qquad
		\infer{\Gamma\vdash\iota~\zact[s] = \tact[s]:\TT 1}{}
	\]
	The application of $\xret$ is equivalent to acting on the monoid
	data with the neutral element:
	\[
		\infer{\Gamma\vdash\xret * = \xact[e] \colon \SC 1}{}
	\]
	Of course, the actions compose:
	\[
		\infer{
			\begin{array}{c}
				\Gamma\vdash\xdo{*}{\xact}{\xdo{*}{\xact[c']}{M}} \\
				= \xdo{*}{\xact[m(c,c')]}{M}:\XX A
			\end{array}
		}{
			\Gamma\vdash M:\XX A
		}
	\]
	where we have used some (hopefully obvious) syntactic sugar.
	We write $\Th_M$ to refer to this theory.
\end{example}

\begin{remark}
	\label{rem:why-not-centre}
	As we have now seen, the equational theories of central submonads admit a
	presentation that is similar in spirit to that of the simply-typed
	$\lambda$-calculus.  However, that is not the case with \emph{the} centre of
	a strong monad. The reason is that the theory $\Th$ can introduce a central
	effect -- one that commutes with all others -- as a constant $c$ that is not
	assigned the type $\SC A$, but the type $\TT A$, for some $A$. However, the
	centre, being the largest central submonad, must contain all such effects, so
	the constant $c$ has to be equal to a term of the form $\iota c'$.  One
	solution to this problem would be to use a more expressive logic and
	introduce a rule as follows (writing inline because of space):
	given $c \colon \TT A$ and $x:A,y:B\vdash P \colon \TT C$, such that
	$\forall N:\TT B.~\vdash \tdo{x}{c}{\tdo{y}{N}{P}} = \tdo{y}{N}{\tdo{x}{c}{P}} \colon \TT C$
	then $\exists c':\SC A.~\vdash c = \iota c' \colon \TT A.$
	However, the addition of such
	a rule seems unnecessary to prove our main point and it increases the
	complexity of the logic.  Because of this, our choice is to focus on central
	submonads.  Another reason to prefer central submonads over the centre is
	that they are more general and it is not required to identify \emph{all}
	central effects (which would be the case for the centre). Overall, our choice
	for central submonads is motivated by the advantages they provide in terms of
	generality, simplicity and practicality of their equational theories compared
	to the centre.
\end{remark}

Now that we have introduced theories, we explain how they can be translated
into one another in an appropriate way.

\begin{definition}[$\mathrm{CSC}$-translation]
	A \emph{translation} $V$ between two $\CSC$-theories $\Th$ and $\Th'$ is
	a function that maps types of $\Th$ to types of $\Th'$ and terms of $\Th$ to
	terms of $\Th'$ that preserves the provability of all type judgements, term
	judgements, type equality judgements and term equality judgements.
	Moreover, such a translation is required to satisfy the following
	structural requirements on types:
	\[\begin{array}{c}
		V(1)=1
		\qquad
		V(\TT A) = \TT V(A)
		\qquad
		V(\SC A) = \SC V(B)
		\\
		V(A\to B) = V(A)\to V(B)
		\qquad
		V(A\times B) = V(A)\times V(B)
	\end{array}\]
	and on terms:
	\[\begin{array}{c}
		V(*) = *
		\\
		V(\lambda x^A.M) = \lambda x^A.V(M)
		\qquad
		V(MN) = V(M)V(N)
		\\
		V(\pv M N) = \pv{V(M)}{V(N)}
		\qquad
		V(\pi_i M) = \pi_i V(M)
		\\
		V(\iota M) = \iota V(M)
		\qquad
		V(\xret M) = \xret V(M)
		\\
		V(\xdo x M N) = \xdo{x}{V(M)}{V(N)}
	\end{array}\]
\end{definition}

\begin{remark}
	The above equations do not imply preservation of the relevant judgements for
	\emph{constants}. Because of this, the first part of the definition also is
	necessary.
\end{remark}

Of course, it is easy to see that $\CSC$-theories and $\CSC$-translations form a
category. However, in order to precisely state our main result, we have to
consider the 2-categorical structure of $\CSC$-theories.  Intuitively, we may view
every $\CSC$-theory as a category itself (with types as objects
and terms as morphisms) and every $\CSC$-translation as a functor that strictly
preserves the relevant structure. Then, intuitively, an appropriate
notion of a 2-morphism would be a natural transformation between such functors.
This is made precise (in non-categorical terms) by our next definition.

\begin{definition}[$\CSC$-translation Transformation]\label{def:trans-trans}
	Given two $\CSC$-theories $\Th$ and $\Th'$,
	and two $\CSC$-translations $V$ and $V'$ between
	them, a \emph{$\CSC$-translation transformation} $\alpha:V\Rightarrow V'$
	is a type-indexed family of term judgements $x:V(A)\vdash \alpha_A \colon V'(A)$
	such that, for any valid judgement $x:A\vdash f:B$ in $\Th$
	\[ x:V(A)\vdash \alpha_B[V(f)/x] = V'(f)[\alpha_A/x] \colon V'(B) \]
	also is derivable in $\Th'.$
\end{definition}

\begin{proposition}
	$\mathrm{CSC}$-theories, $\mathrm{CSC}$-translations
	and $\mathrm{CSC}$-translation transformations
	form a $2$-category $\TR(\mathrm{CSC})$.
\end{proposition}
\begin{proof}
	Direct with Definition~\ref{def:trans-trans}.
\end{proof}

\subsection{Categorical Models of CSC}
\label{sub:models}

Now we describe what are the appropriate categorical models for providing a
semantic interpretation of our calculus.

\begin{definition}[$\CSC$-model]
	A \emph{$\CSC$-model} is a cartesian closed category $\CC$ equipped with both a
	strong monad $\TT$ and a central submonad $\SC^\TT$ of $\TT$ with submonad
	monomorphism written as $\iota^\TT \colon \SC^\TT \naturalto \TT.$ We often
	use a quadruple $(\CC,\TT,\SC^\TT, \iota^\TT)$ to refer to a $\CSC$-model.
\end{definition}

We will soon show that $\CSC$-models correspond to $\CSC$-theories in a precise way.
This correspondence covers $\CSC$-translations too and for
this we introduce our next definition.

\begin{definition}[$\CSC$-model Morphism]
	Given two $\CSC$-models $(\CC,\TT,\SC^\TT, \iota^\TT)$ and $(\DD,\MM,\SC^\MM,
	\iota^\MM)$, a \emph{$\CSC$-model morphism} is a strict cartesian closed functor
	$F:\CC\to\DD$ that satisfies the following additional coherence properties:
	\[\begin{array}{c}
		F(\TT X) = \MM(FX)
		\qquad
		F(\SC^\TT X) = \SC^\MM (FX)
		\\[1.5ex]
		F\iota^\TT_X = \iota^\MM_{FX}
		\qquad
		F\eta^\TT_X = \eta^\MM_{FX}
		\\[1.5ex]
		F\mu^\TT_X = \mu^\MM_{FX}
		\qquad
		F\tau^\TT_{X,Y} = \tau^\MM_{FX,FY}.
	\end{array}\]
\end{definition}

Notice that a $\CSC$-model morphism \emph{strictly} preserves all of the relevant
categorical structure. This is done on purpose so that we can establish an
exact correspondence with $\CSC$-translations, which also strictly preserve the
relevant structure. To match the notion of a $\CSC$-translation transformation, we
just have to consider natural transformations between $\CSC$-model morphisms.

\begin{proposition}
	$\CSC$-models, $\CSC$-model morphisms and natural transformations between
	them form a $2$-category $\Mod(\CSC)$.
\end{proposition}
\begin{proof}
	Direct.
\end{proof}

\subsection{Semantic Interpretation}
\label{sub:denotational}

Now we explain how to introduce a denotational semantics for our theories using
our models.  An interpretation of a $\CSC$-theory $\Th$ in a $\CSC$-model $\CC$
is a function $\sem{-}$ that maps types of $\Th$ to objects of $\CC$ and
well-formed terms of $\Th$ to morphisms of $\CC$. We provide the details below.

For each ground type $G$, we assume there is an appropriate corresponding
object $\sem{G}$ of $\CC$.  The remaining types are interpreted as objects in
$\CC$ as follows:
$
\sem{1} \eqdef 1 ;
\sem{A\to B} \eqdef \sem B^{\sem A} ;
\sem{A\times B} \eqdef \sem A\times\sem B ;
\sem{\SC A} \eqdef \SC\sem A ;
\sem{\TT A} \eqdef \TT\sem A .
$
Variable contexts $\Gamma = x_1 \colon A_1 \dots x_n \colon A_n$ are interpreted as usual as
$\sem\Gamma\defeq\sem{A_1}\times\dots\times \sem{A_n}$. Terms are interpreted
as morphisms $ \sem{\Gamma\vdash M:A} \colon \sem\Gamma \to \sem A $ of $\CC$.  When
the context and the type of a term $M$ are understood, then we simply write
$\sem M$ as a shorthand for $\sem{\Gamma\vdash M:A}$.  The interpretation of
term constants and the terms of the simply-typed $\lambda$-calculus is defined
in the usual way (details omitted). The interpretation of the monadic terms is
given by:
\begin{align*}
	\sem{\Gamma\vdash \xret M :\XX A} &= \eta^\XX_{\sem A}\circ\sem M \\
	\sem{\Gamma\vdash \iota M :\TT A} &= \iota_{\sem A}\circ\sem M \\
	\sem{\Gamma\vdash\xdo x M N \colon \XX B} &= \mu^\XX_{\sem B}\circ\XX\sem
	N\circ\tau^\XX_{\sem\Gamma,\sem A}\circ\pv{\iid}{\sem M}
\end{align*}
where we use $\XX$ to range over $\TT$ or its central submonad $\SC.$

\begin{definition}[Soundness and Completeness]
	An interpretation $\sem{-}$ of a $\CSC$-theory $\Th$ in a $\CSC$-model $\CC$
	is said to be \emph{sound} if for any type equality judgement $\vdash
	A=B:\typ$ in $\Th$, we have that $\sem A=\sem B$ in $\CC$, and for any
	equality judgement $\Gamma\vdash M=N :A$ in $\Th$, we have that
	$\sem{\Gamma\vdash M \colon A}=\sem{\Gamma\vdash N:A}$ in $\CC$. An
	interpretation $\sem{-}$ is said to be \emph{complete} when $\vdash A=B:\typ$
	iff $\sem A = \sem B$ and $\Gamma\vdash M=N \colon A$ iff $\sem{\Gamma\vdash
	M \colon A}=\sem{\Gamma\vdash N \colon A}.$ If, moreover, the interpretation
	is clear from context, then we may simply say that the model $\CC$ itself is
	sound and complete for the $\CSC$-theory $\Th$.
\end{definition}

\begin{remark}
	There are different definitions of what constitutes a ``model'' in the
	literature. For example, a ``model'' in \cite{crole1994} corresponds to a
	sound interpretation in our sense.
\end{remark}

\begin{example}
	A categorical model for the $\CSC$-theory $\Th_M$ of
	Example~\ref{ex:theory-writer} is given by the category $\Set$ together with
	the writer monad $\TT \eqdef (-\times M) \colon \Set \to \Set$ and the central
	submonad $\SC \eqdef (-\times S) \colon \Set \to \Set$.  More specifically, the
	monad data for $\TT$ is given by:
	\begin{align*}
		& \eta_A     \colon A \to A\times M :: a \mapsto (a,e) \\
		& \mu_A      \colon (A\times M)\times M \to A\times M :: ((a,c),c') \mapsto (a,m(c,c')) \\
		& \tau_{A,B} \colon A\times(B\times M) \to (A\times B)\times M ::\\
		& \qquad\quad  (a,(b,c)) \mapsto  ((a,b),c)
	\end{align*}
	and the monad data for $\SC$ is defined in the same way by (co)restricting to
	the submonoid $S$.  The interpretation of the term constants is given by:
	\begin{align*}
		\sem{\Gamma\vdash \tact:\TT 1} &\colon \sem\Gamma \to 1\times M :: \gamma \mapsto (*,c) \\
		\sem{\Gamma\vdash \zact:\SC 1} &\colon \sem\Gamma \to 1\times S :: \gamma \mapsto (*,c)
	\end{align*}
	This interpretation of the theory $\Th_M$ is sound and complete.
\end{example}

\subsection{Equivalence between Theories and Models}
\label{sub:equivalence}

Our final result in this chapter is to show that $\CSC$-theories and
$\CSC$-models are strongly related. To do this, we define the \emph{syntactic
$\CSC$-model} $S(\Th)$ of $\CSC$-theory $\Th$, and the \emph{internal language}
$L(\CC)$ that maps a $\CSC$-model $\CC$ to its internal language viewed as a
$\CSC$-theory.  These two assignments give rise to our desired equivalence
(Theorem \ref{th:internal-language}).

\paragraph{The Syntactic $\CSC$-model.}
\label{sub:syntactic-fun}
Assume throughout the subsection that we are given a $\mathrm{CSC}$-theory
$\Th$.  We show how to construct a sound and complete model $S(\Th)$ of $\Th$
by building its categorical data using the syntax provided by $\Th.$

\begin{definition}[Syntactic Category]
	Let $S(\Th)$ be the category whose objects are the types of $\Th$ modulo type
	equality, \emph{i.e.}, the objects are equivalence classes $[A]$ of types with $A'
	\in [A]$ iff $\vdash A' = A \colon \typ$ in $\Th.$ The morphisms $S(\Th)([A],
	[B])$ are equivalence classes of judgements $[x:A\vdash f:B],$ where $(x:A'
	\vdash f' \colon B') \in [x:A\vdash f \colon B]$ iff $\vdash A' = A \colon
	\typ$ and $\vdash B' = B \colon \typ$ and $x:A \vdash f = f' \colon B.$
	Identities are given by $[x \colon A \vdash x \colon A]$ and composition is
	defined by
	\[ [y \colon B \vdash g \colon C] \circ [x \colon A \vdash f \colon B'] = [x
	\colon A \vdash g[f/y] \colon C] , \]
	with $B' \in [B].$
\end{definition}

\begin{lemma}
	\label{lem:syntactic-category}
	The above definition is independent of the choice of representatives and the
	syntactic category $S(\Th)$ is a well-defined cartesian closed category.
\end{lemma}
\begin{proof}
	Suppose given two morphisms $f \colon  A\to B, g \colon B\to C$, and a choice
	$[x \colon A'\vdash f' \colon B'_f] = f$ and $[y \colon B'_g\vdash g' \colon
	C']= g$. Note that $B=[B'_f]=[B'_g]$, and in particular $y \colon B'_f\vdash
	g' \colon C'$ is derivable with $[y \colon B'_g\vdash g' \colon C']=[y \colon
	B'_f\vdash g' \colon C']$. Thus, $x \colon A'\vdash g'[f'/y]  \colon C'$ is
	derivable.  We then prove that the choice $[x \colon A'\vdash f' \colon B'_f]
	= f$ and $[y \colon B'_f\vdash g' \colon C']= g$ does not matter.  We
	consider now new term judgments for some terms $f''$ and $g''$ such that $[x
	\colon A'\vdash f' \colon B'_f] = [x \colon A''\vdash f'' \colon B''_f]$ and
	$[y \colon B'_f\vdash g' \colon C'] = [y \colon B''_f\vdash g'' \colon C'']$.
	By definition, $[A']=[A'']$, $[B'_f]=[B''_f]$ and $[C']=[C'']$, and we wish
	to prove that $[x \colon A'\vdash g'[f'/y]  \colon C'] = [x \colon A''\vdash
	g''[f''/y]  \colon C'']$.

	\[
		\begin{array}{l}
			\Pi_2 = \left\{
				\begin{array}{l}
					\infer[(\lambda.eq)]{
						x \colon A'\vdash (\lambda y^{B'_f}.g')f' =
						(\lambda y^{B''_f}.g'')f'' \colon C'
					}{
						\infer{
							x \colon A'\vdash  \lambda y^{B'_f}.g'
							= \lambda y^{B''_f}.g''  \colon 
							B'_f\to C'
						}{
							x \colon A,y \colon B'\vdash g'=g'' \colon C'
						}
						&
						x \colon A'\vdash f'=f'' \colon C'
					}
				\end{array} \right.
		\end{array}
	\]
	\[
		\begin{array}{l}
			\Pi_1 = \left\{
				\begin{array}{l}
					\infer[(trans)]{
						x \colon A'\vdash (\lambda y^{B'_f}.g')f' = g''[f''/y] \colon C'
					}{
						\Pi_2
						&
						\infer[(\lambda.\beta)]{
							x \colon A'\vdash (\lambda y^{B''_f}.g'')f'' =
							g''[f''/y] \colon C'
						}{
							x \colon A',y \colon B''_f \vdash g'' \colon C'
							&
							x \colon A' \vdash f'' \colon B''_f
						}
				} \end{array} \right.
		\end{array}
	\]
	\[
		\infer[(trans)]{
			x \colon A'\vdash g'[f'/y] = g''[f''/y] \colon C'
		}{
			\infer[(\lambda.\beta)]{
				x \colon A'\vdash g'[f'/y] = (\lambda y^{B'_f}.g')f' \colon C'
			}{
				x \colon A',y \colon B'_f\vdash g' \colon C'
				&
				x \colon A'\vdash f' \colon C'
			}
			&
			\Pi_1
		}
	\]

	Thus, it is safe to define $g\circ f$ as $[x \colon A'\vdash g'[f'/y]  \colon C']$.

	Given a choice of $A'$ in $[A]$, $[x \colon A'\vdash x \colon A']$ is the
	identity morphism for the type $[A]$. Considering $[x \colon A'\vdash
	f \colon B']$ and $[y \colon C'\vdash g \colon A']$, we have:
	\[
		[x \colon A'\vdash f \colon B']\circ [x \colon A'\vdash x \colon A']
		= [x \colon A' \vdash f[x/x] \colon B']
		= [x \colon A' \vdash f \colon B'],\]
	and
	\[ [x \colon A'\vdash x \colon A'] \circ [y \colon C'\vdash g \colon A']
	= [y \colon C'\vdash x[g/x] \colon A']
	= [y \colon C'\vdash g  \colon A']. \]
	One can notice that, for example, $x \colon A'\vdash f \colon B'$ has
	conveniently be chosen with the right type $A'$. It is
	authorised, because we have proven above that the choice of
	representative does not matter in composition matters.

	The cartesian closure is a usual result for a syntactic
	category from a simply-typed $\lambda$-calculus, and it is
	preserved in our context.
\end{proof}

\begin{remark}
	Note that by using \emph{Scott's trick} \cite{scott-trick} we can take
	quotients without having to go up higher in the class hierarchy, so foundational
	issues can be avoided.
\end{remark}

\begin{lemma}[\cite{awodey-bauer-lecture}]
	\label{lem:strong-syntactic}
	The following assignments:
	\begin{align*}
		\TT([A])   &= [\TT A] \\
		\TT( [x \colon A \vdash f \colon B] ) &= [y \colon \TT A \vdash \tdo x y {\tret f} \colon \TT B] \\
		\eta_{[A]}     &= [ x \colon A\vdash \tret x \colon \TT A ] \\
		\mu_{[A]}      &= [ x \colon \TT\TT A \vdash \tdo y x y \colon \TT A ] \\
		\tau_{[A],[B]} &= [ x \colon A\times \TT B \vdash 
		\tdo{y}{\pi_2 x}{\tret\pv{\pi_1 x}{y}} \colon \TT(A\times B) ]
	\end{align*}
	are independent of the choice of representatives and define a strong monad $(\TT, \eta, \mu, \tau)$ on $S(\Th)$.
\end{lemma}

\begin{lemma}\label{lem:syntactic-central-submonad}
	In a similar way to Lemma \ref{lem:syntactic-category}, we can define a
	strong monad $(\SC, \eta^\SC, \mu^\SC, \tau^\SC)$ on $S(\Th)$ by using the
	corresponding monadic primitives.  Then, the assignment:
	\[
		\iota_{[A]} = [x \colon \SC A\vdash \iota x \colon \TT A]
	\]
	is independent of the choice of representative and gives a strong submonad
	monomorphism $\iota \colon \SC \naturalto \TT$ that makes $\SC$ a central submonad
	of $\TT.$
\end{lemma}
\begin{proof}
	In all the following proofs, we consider convenient members of equivalence
	classes, because the choice of representative does not change the result,
	thanks to Lemma~\ref{lem:syntactic-category}.

	We prove that $\iota$ is a submonad morphism:

	$\begin{array}{cl}
		&\iota_A\circ\eta^\SC_A \\
		\stackrel{def.}{=}
		& [y  \colon \SC A\vdash \iota y \colon \TT A] \circ [x \colon A\vdash \zret x \colon \SC A] \\
		\stackrel{comp.}{=}
		& [x \colon A\vdash \iota~\zret x \colon \TT A] \\
		\stackrel{(\iota\SC.ret)}{=} & [x \colon A\vdash \tret x  \colon \TT A] \\
		\stackrel{def.}{=} & \eta^\TT_A
	\end{array}$
		\\[3.5ex]
	$\begin{array}{cl}
		&\mu^\TT_A\circ\TT\iota_A\circ\iota_{\SC A} \\
		\stackrel{def.}{=}
		& [z \colon \TT\TT A\vdash \tdo{y}{z}{y} \colon \TT A] \\
		&\circ~[y' \colon \TT\SC
		A\vdash \tdo{x}{y'}{\tret\iota x} \colon \TT\TT A] \circ [x' \colon \SC\SC
		A \vdash \iota x' \colon \TT\SC A] \\
		\stackrel{comp.}{=}
		& [x' \colon \SC\SC A\vdash\tdo{y}{\left(\tdo{x}{\iota x'}{\tret \iota
		x}\right)}{y} \colon \TT A] \\
		\stackrel{(\TT.assoc)}{=} & [x' \colon \SC\SC A\vdash \tdo{x}{\iota
		x'}{\tdo{y}{\tret \iota x}{y}} \colon \TT A] \\
		\stackrel{(\TT.\beta)}{=} & [x' \colon \SC\SC A\vdash \tdo{x}{\iota
		x'}{\iota x} \colon \TT A] \\
		\stackrel{(\iota\SC.comp)}{=} & [x' \colon \SC\SC A\vdash
		\iota~\zdo{x}{x'}{x} \colon \TT A] \\
		\stackrel{comp.}{=}
		& [y \colon \SC A\vdash \iota y  \colon \TT A] \circ [x' \colon \SC\SC A\vdash
		\zdo{x}{x'}{x} \colon \SC A] \\
		\stackrel{def.}{=} & \iota_A\circ\mu^\SC_A
	\end{array}$
		\\[3.5ex]
	$\begin{array}{cl}
		&\iota_{A\times B}\circ\tau^\SC_{A,B} \\
		\stackrel{def.}{=}
		& [x \colon \SC(A\times B)\vdash \iota x \colon \TT(A\times B)] \\
		&\circ~
		[z \colon A\times\SC B \vdash \zdo{y}{\pi_2 z}{\zret\pv{\pi_1 z}{y}}  \colon 
		\SC(A\times B) ] \\
		\stackrel{comp.}{=} & [z  \colon A\times\SC B \vdash
		\iota\left(\zdo{y}{\pi_2 z}{\zret\pv{\pi_1
		z}{y}}\right) \colon \TT(A\times B)] \\
		\stackrel{(\iota\SC.comp)}{=} & [z  \colon A\times\SC
		B\vdash\tdo{y}{\iota~\pi_2 z}{\iota~\zret\pv{\pi_1
		z}{y}} \colon \TT(A\times B)] \\
		\stackrel{(\iota\SC.ret)}{=} & [z  \colon A\times\SC
		B\vdash\tdo{y}{\iota~\pi_2 z}{\tret\pv{\pi_1 z}{y}} \colon \TT(A\times
		B)] \\
		\stackrel{(\times.\beta)}{=} & [z  \colon A\times\SC
		B\vdash\tdo{y}{\pi_2\pv{\pi_1 z}{\iota \pi_2 z}}{\tret\pv{\pi_1
		\pv{\pi_1 z}{\iota \pi_2 z}}{y}} \colon \TT(A\times B)] \\
		\stackrel{comp.}{=} & [x \colon A\times\TT B \vdash \tdo{y}{\pi_2
		x}{\tret\pv{\pi_1 x}{y}} \colon \TT(A\times B)] \\
		&\circ~ [z  \colon A\times\SC
		B\vdash \pv{\pi_1 z}{\iota \pi_2 z} \colon A\times \TT B] \\
		\stackrel{def.}{=} & \tau^\TT_{A,B}\circ(A\times\iota_B)
	\end{array}$

	Moreover, $\iota$ is a monomorphism because of the $(\iota.mono)$ rule.

	Finally,
	$\ZZ$ is a central submonad of $\TT$:
	\[\begin{array}{cl}
		& dst_{A,B}\circ(\iota\times\TT B) \\
		\stackrel{def.+comp.}{=} & [z \colon \SC A\times\TT B\vdash \\
		&\hspace{-.5cm} \tdo{x}{\left(\tdo{y}{\iota~\pi_1 z}{\tret\left(\tdo{y'}{\pi_2
		z}{\tret\pv{y}{y'}}\right)}\right)}{x}  \colon \TT(A\times B)] \\
		\stackrel{(\TT.assoc)}{=} & [z \colon \SC A\times\TT
		B\vdash \\
		&\hspace{-.5cm} \tdo{y}{\iota~\pi_1 z}{\tdo{x}{\tret\left(\tdo{y'}{\pi_2
		z}{\tret\pv{y}{y'}}\right)}{x}} \colon \TT(A\times B)] \\
		\stackrel{(\TT.\beta)}{=} & [z \colon \SC A\times\TT
		B\vdash\tdo{y}{\iota~\pi_1 z}{\tdo{y'}{\pi_2
		z}{\tret\pv{y}{y'}}} \colon  \TT(A\times B)] \\
		\stackrel{(\SC.central)}{=} & [z \colon \SC A\times\TT
		B\vdash\tdo{y'}{\pi_2 z}{\tdo{y}{\iota~\pi_1
		z}{\tret\pv{y}{y'}}} \colon  \TT(A\times B)] \\
		\stackrel{(\TT.\beta)}{=} & [z \colon \SC A\times\TT
		B\vdash \\
		&\hspace{-.5cm} \tdo{y'}{\pi_2 z}{\tdo{x}{\tret\left(\tdo{y}{\iota~\pi_1
		z}{\tret\pv{y}{y'}}\right)}{x}} \colon \TT(A\times B)] \\
		\stackrel{(\TT.assoc)}{=} & [z \colon \SC A\times\TT
		B\vdash \\
		&\hspace{-.5cm} \tdo{x}{\left(\tdo{y'}{\pi_2
		z}{\tret\left(\tdo{y}{\iota~\pi_1
		z}{\tret\pv{y}{y'}}\right)}\right)}{x}  \colon \TT(A\times B)] \\
		\stackrel{comp.+def.}{=} & dst'_{A,B}\circ(\iota\times\TT B)
	\end{array}\]
\end{proof}

Now we can prove our completeness result.

\begin{theorem}[Completeness]
	\label{th:completeness}
	The quadruple $(S(\Th), \TT, \SC, \iota)$ is a sound and complete $\CSC$-model
	for the $\CSC$-theory $\Th$.
\end{theorem}
\begin{proof}
	There exists an (obvious) interpretation $\sem{-}$ of $\Th$ into $S(\Th)$
	which follows the structure outlined in \secref{sub:denotational}.  Standard
	arguments then show that $\Gamma\vdash M = N \colon A$ in $\Th$ iff
	$\sem{\Gamma\vdash M \colon A} = \sem{\Gamma\vdash N \colon A}$ in $S(\Th)$.
\end{proof}

\begin{remark}
	Note that the obvious canonical interpretation of $\Th$ in $S(\Th)$ is
	initial as one may expect: any sound interpretation of $\Th$ in a
	$\CSC$-model $\CC$ factorises uniquely through the canonical interpretation
	via a $\CSC$-model morphism.
\end{remark}

\paragraph{Internal Language.}
\label{sub:internal-language}
With completeness proven, we now wish to establish an internal language result.

\begin{definition}[Internal Language]
	\label{def:language}
	Given a $\CSC$-model $\CC$, we define a $\CSC$-theory $L(\CC)$
	as follows:
	\begin{itemize}
		\item For each object $A$ of $\CC$ we add a ground type which we name $A^*$.
		\item Every ground type $A^*$ is interpreted in $\CC$ by setting $\sem{A^*}
			\eqdef A.$ This uniquely determines an interpretation on all types.
		\item If $A$ and $B$ are two (not necessarily ground) types, we add a type
			equality $\vdash A = B  \colon \typ$ iff $\sem A = \sem B$.
		\item For every morphism $f \colon A \to B$ in $\CC$, we add a term constant
			$\vdash c_f \colon A^* \to B^*.$ Its interpretation in $\CC$ is defined to be
			$\sem{c_f} \eqdef \mathrm{curry}(f \circ \cong) \colon 1 \to B^A$, \emph{i.e.}, it is
			defined by currying the morphism $f$ in the obvious way. This uniquely
			determines an interpretation on all well-formed terms.
		\item New term equality axioms $\Gamma\vdash M=N \colon B$ iff $\sem{\Gamma \vdash
			M \colon B} = \sem{\Gamma \vdash N \colon B}.$
	\end{itemize}
\end{definition}

\begin{theorem}
	For any $\CSC$-model $\CC$ the above definition gives a well-defined $\CSC$-theory
	$L(\CC).$ Moreover, the model $\CC$ is sound and complete for $L(\CC).$
\end{theorem}
\begin{proof}
	Well-definedness is straightforward and follows by a simple induction
	argument using the fact that the semantic interpretation $\sem{-}$ defined
	in \secref{sub:denotational} is always sound.  Completeness is then
	immediate by the last condition in Definition~\ref{def:language}.
\end{proof}

\paragraph{Equivalence Theorem.}
Finally, we show that both the construction of the syntactic category and the
assignment of the internal language give rise to appropriate equivalences.

\begin{theorem}
	\label{th:internal-language}
	The relationship between the internal language and the syntactic model enjoys
	the following properties in the 2-categories $\Mod(\CSC)$ and $\TR(\CSC)$,
	respectively:
	\begin{enumerate}
		\item For any $\CSC$-model $\CC$, we have that $\CC \simeq SL(\CC)$,
		\emph{i.e.}, there exist $\CSC$-model morphisms $F \colon \CC \to SL(\CC)$
	and $G \colon  SL(\CC) \to \CC$ such that $F \circ G \cong \Id$ and $\Id
\cong G \circ F.$ \item For any $\CSC$-theory $\Th$, we have that $\Th \simeq
	LS(\Th)$, \emph{i.e.}, there exist $\CSC$-translations $V \colon \Th \to
			LS(\Th)$ and $W \colon  LS(\Th) \to \Th$ such that $V \circ W \cong \Id$
			and $\Id \cong W \circ V.$
	\end{enumerate}
\end{theorem}
\begin{proof}
	Given $\CC$ an object of $\Mod(\CSC)$, we wish to prove that $\CC$ is
	equivalent to $SL(\CC)$. To do so, we introduce two strict cartesian
	closed functors $F \colon \CC\to SL(\CC)$ and $G \colon SL(\CC)\to\CC$, such that
	there are isomorphisms $\Id \Rightarrow GF$ and $FG\Rightarrow \Id$.
	\begin{itemize}
		\item $F$ maps an object $A$ of $\CC$ to $[A^*]$.  It maps a
			morphism $f \colon A \to B$ to $[x:A^*\vdash c_f x \colon B^*]$.
		\item $G$ maps an object $[A]$ to $\sem{A}$, the interpretation of the type
			$A$ in $\CC$, because the choice of representative of $[A]$ does not
			change the interpretation. $G$ maps a morphism $[x \colon A \vdash g
			\colon B]$ to $\sem{x \colon A\vdash g \colon B}$.
	\end{itemize}
	Then it is easy to check that $GF = \Id$ and $FG = \Id.$ Therefore $\CC$ is
	isomorphic to $SL(\CC)$.  Furthermore, given a $\CSC$-theory $\Th$, we wish
	to prove that $\Th$ is equivalent to $LS(\Th)$. To do so, we introduce two
	$\CSC$-translations $V \colon \Th\to LS(\Th)$ and $W \colon LS(\Th)\to\Th$
	such that there are isomorphic $\CSC$-translation transformations
	$VW\Rightarrow \Id$ and $\Id \Rightarrow WV$.
	\begin{itemize}
		\item $V$ maps a type $A$ in $\Th$ to $[A]^*$, and term
			judgements $x \colon A\vdash f \colon B$ to $x \colon [A]^*\vdash
			c_{[x \colon A\vdash f \colon B]} x \colon [B]^*$.
		\item Observe that for each type $A$ in $LS(\Th)$, there is a
			type of the form $[B]^*$ such that $\vdash A=[B]^*
			 \colon \typ$ in $LS(\Th)$.  We define $W(A) \eqdef B$ (the
			choice of $B$ does not matter).  Then, for term
			constants we define $W( \vdash c_{[x \colon A \vdash f \colon B]}
			\colon B^*) \eqdef (\vdash \lambda x. f  \colon A \to B)$ and
			this uniquely determines the action of $W$ on the
			remaining terms (the choice of $f$ does not matter).
	\end{itemize}
	Given a type $A$ in $\Th$, $x \colon W(V(A))\vdash x \colon A$ is derivable
	in $\Th$ because $\vdash W(V(A))=A \colon \typ$, and $\alpha_A \colon x
	\colon W(V(A))\vdash x \colon A$ defines an isomorphic $\CSC$-translation
	transformation: postcomposing (resp.  composing) it with $x \colon A\vdash x
	\colon W(V(A))$ gives $x \colon W(V(A))\vdash x \colon W(V(A))$ (resp. $x
	\colon A\vdash x \colon A$). Given a type $A'$ in $LS(\Th)$, the same is true
	for $\beta_{A'} = x \colon A'\vdash x \colon V(W(A'))$. Thus, for every
	$\CSC$-theory, $\Th$ is equivalent to $LS(\Th)$.
\end{proof}

\begin{remark}
	We introduced type equalities so that we can prove Theorem
	\ref{th:internal-language}. This is also the approach taken in
	\cite{maietti2005relating} and without this, technical difficulties arise.
	Theory translations are defined strictly (up to equality, not up to
	isomorphism) and in order to match this with the corresponding notion of
	model morphism, we use type equalities. Without type equalities, the symmetry
	within Theorem~\ref{th:internal-language} can only be established if we make
	further changes. One potential solution would be to weaken the notion of
	theory translation by requiring that it preserves types up to type
	isomorphism (\emph{i.e.}, make it strong instead of strict), but this is technically
	cumbersome.
\end{remark}

\section{Conclusion and Future Work}
\label{sec:conclude}

We showed that, under some mild assumptions, strong monads indeed admit a
centre, which is a commutative submonad, and we provided three equivalent
characterisations for the existence of this centre (Theorem
\ref{th:centralisability}) which also establish important links to the theory
of premonoidal categories. In particular, every (canonically strong) monad on
$\Set$ is centralisable (\secref{sub:sets}) and we showed that the same is true
for many other categories of interest (\secref{sub:examples}) and we identified
specific monads with interesting centres (\secref{sub:specific-examples}).
More generally, we considered central submonads and we provided a computational
interpretation of our ideas (\secref{sec:computational}) which has the added
benefit of allowing us to easily keep track of which monadic operations are
central, \emph{i.e.}, which effectful operations commute under monadic sequencing with
any other (not necessarily central) effectful operation.  We cemented our
semantics by proving soundness, completeness and internal language results.

One direction for future work is to consider a theory of \emph{commutants} or
\emph{centralisers} for monads (in the spirit of
\cite{commutants,garner2016commutativity}) and to develop a computational
interpretation with the expected properties (soundness, completeness and
internal language).
Another opportunity for future work includes studying the relationship between
the centres of strong monads and distributive laws. In particular, given two
strong monads and a strong/commutative distributive law between them, can we
show that the distributive law also holds for their centres (or for some
central submonads)? If so, this would allow us to use the distributive law to
combine not just the original monads, but their centres/central submonads as
well. Moreover, the interaction of the centre with operations on monadic
theories can be investigated.

Our definition of central submonads makes essential use of the notion of
monomorphism of strong monads. Another possibility for future work is to
investigate an alternative approach where we consider an appropriate class of
factorisation systems instead of monomorphisms to define central submonads. Yet
another possibility for future work is to investigate if central submonads of a
given strong monad have some interesting poset structure.

A natural generalisation of monads is the notion of \emph{arrows} -- or
\emph{strong promonads}. A promonad is a monoid in the category of profunctors,
and profunctors are to functors what relations are to functions. Arrows give
then a more general framework to study computational effects, and are
particularly meaningful for effects in reversible computing
\cite{alimarine2005arrows, heunen2018arrows}. A final direction for future work
is the study of equational theories and internal language for arrows.

%% file: figures/def-central-cone.tikz
\begin{tikzpicture}
	\begin{pgfonlayer}{nodelayer}
		\node [style=empty] (0) at (-6, 4) {$Z\otimes\TT Y$};
		\node [style=empty] (1) at (1, 4) {$\TT X\otimes\TT Y$};
		\node [style=empty] (2) at (8, 4) {$\TT (X\otimes\TT Y)$};
		\node [style=empty] (3) at (8, 0.5) {$\TT^2 (X\otimes Y)$};
		\node [style=empty] (4) at (8, -3) {$\TT (X\otimes Y)$};
		\node [style=empty] (5) at (-6, 0.5) {$\TT X\otimes\TT Y$};
		\node [style=empty] (6) at (-6, -3) {$\TT(\TT X\otimes Y)$};
		\node [style=empty] (7) at (1, -3) {$\TT^2 (X\otimes Y)$};
		\node [style=empty] (8) at (-6, 3.5) {};
		\node [style=empty] (9) at (-6, 1) {};
		\node [style=empty] (10) at (-6, 0) {};
		\node [style=empty] (11) at (-6, -2.5) {};
		\node [style=empty] (12) at (8, 3.5) {};
		\node [style=empty] (13) at (8, 1) {};
		\node [style=empty] (14) at (8, 0) {};
		\node [style=empty] (15) at (8, -2.5) {};
		\node [style=empty] (16) at (-2.5, 4.5) {$\iota\otimes\TT Y$};
		\node [style=empty] (17) at (4.5, 4.5) {$\tau'_{X,\TT Y}$};
		\node [style=empty] (18) at (9.25, 2.25) {$\TT\tau_{X,Y}$};
		\node [style=empty] (19) at (9.25, -1.25) {$\mu_{X\otimes Y}$};
		\node [style=empty] (20) at (-7.25, 2.25) {$\iota\otimes\TT Y$};
		\node [style=empty] (21) at (-7.25, -1.25) {$\tau_{\TT X,Y}$};
		\node [style=empty] (22) at (-2.5, -3.5) {$\TT\tau'_{X,Y}$};
		\node [style=empty] (23) at (4.5, -3.5) {$\mu_{X\otimes Y}$};
	\end{pgfonlayer}
	\begin{pgfonlayer}{edgelayer}
		\draw [style=to] (0) to (1);
		\draw [style=to] (1) to (2);
		\draw [style=to] (6) to (7);
		\draw [style=to] (7) to (4);
		\draw [style=to] (8) to (9);
		\draw [style=to] (10) to (11);
		\draw [style=to] (14) to (15);
		\draw [style=to] (12) to (13);
	\end{pgfonlayer}
\end{tikzpicture}

%% file: figures/central-cone-to-morph.tikz
\begin{tikzpicture}
	\begin{pgfonlayer}{nodelayer}
		\node [style=empty] (0) at (-12, 1.5) {$X\otimes X'$};
		\node [style=empty] (1) at (-12, -9.5) {$\TT(X\otimes Y')$};
		\node [style=empty] (2) at (-4, 1.5) {$\TT Y\otimes X'$};
		\node [style=empty] (3) at (-4, -9.5) {$\TT(\TT Y\otimes Y')$};
		\node [style=empty] (4) at (-12, 1) {};
		\node [style=empty] (5) at (-12, -3.5) {};
		\node [style=empty] (6) at (12, 1) {};
		\node [style=empty] (7) at (12, -3.5) {};
		\node [style=empty] (8) at (-8, 2) {$f\otimes X'$};
		\node [style=empty] (9) at (-13.25, -1.25) {$X\otimes f'$};
		\node [style=empty] (10) at (8, 2) {$\TT(Y\otimes f')$};
		\node [style=empty] (11) at (-8, -10.25) {$\TT(f\otimes Y')$};
		\node [style=empty] (12) at (4, 1.5) {$\TT (Y\otimes X')$};
		\node [style=empty] (13) at (0, 2) {$\tau'_{Y,X'}$};
		\node [style=empty] (14) at (-12, -4) {$X\otimes\TT Y'$};
		\node [style=empty] (15) at (-12, -4.5) {};
		\node [style=empty] (16) at (-12, -9) {};
		\node [style=empty] (17) at (12, -4) {$\TT^2(Y\otimes Y')$};
		\node [style=empty] (18) at (12, -4.5) {};
		\node [style=empty] (19) at (4, -9.5) {$\TT^2(Y\otimes Y')$};
		\node [style=empty] (20) at (12, -9) {};
		\node [style=empty] (21) at (13.5, -1.25) {$\TT\tau_{Y,Y'}$};
		\node [style=empty] (22) at (0, -10.25) {$\TT\tau'_{Y,Y'}$};
		\node [style=empty] (23) at (-13.25, -6.75) {$\tau_{X,Y'}$};
		\node [style=empty] (24) at (12, 1.5) {$\TT (Y\otimes\TT Y')$};
		\node [style=empty] (25) at (13.5, -6.75) {$\mu_{Y\otimes Y'}$};
		\node [style=empty] (26) at (12, -9.5) {$\TT(Y\otimes Y')$};
		\node [style=empty] (27) at (8, -10.25) {$\mu_{Y\otimes Y'}$};
		\node [style=empty] (28) at (-4, -4) {$\TT Y\otimes\TT Y'$};
		\node [style=empty] (29) at (-4, -4.5) {};
		\node [style=empty] (30) at (-4, -9) {};
		\node [style=empty] (31) at (-8, -4.5) {$f\otimes\TT Y'$};
		\node [style=empty] (32) at (-5.25, -6.75) {$\tau_{\TT Y,Y'}$};
		\node [style=empty] (33) at (-2.25, -2) {};
		\node [style=empty] (34) at (11.25, 1) {};
		\node [style=empty] (36) at (-4, 1) {};
		\node [style=empty] (37) at (-4, -1.5) {};
		\node [style=empty] (38) at (-5.5, -0.25) {$\TT Y\otimes f'$};
		\node [style=empty] (39) at (4.25, -1.5) {$\tau'_{Y,\TT Y'}$};
		\node [style=none] (40) at (-9.25, -0.5) {\textcolor{blue}{(1)}};
		\node [style=none] (41) at (0.75, 0.25) {\textcolor{blue}{(2)}};
		\node [style=none] (42) at (-9, -7) {\textcolor{blue}{(3)}};
		\node [style=none] (43) at (3.75, -6) {\textcolor{blue}{(4)}};
		\node [style=empty] (44) at (-4, -2.25) {$\TT Y\otimes\TT Y'$};
		\node [style=empty] (45) at (-8.25, -2.5) {$f\otimes\TT Y'$};
	\end{pgfonlayer}
	\begin{pgfonlayer}{edgelayer}
		\draw [style=to] (0) to (2);
		\draw [style=to] (6) to (7);
		\draw [style=to] (4) to (5);
		\draw [style=to] (1) to (3);
		\draw [style=to] (2) to (12);
		\draw [style=to] (15) to (16);
		\draw [style=to] (18) to (20);
		\draw [style=to] (3) to (19);
		\draw [style=to] (12) to (24);
		\draw [style=to] (19) to (26);
		\draw [style=to] (14) to (28);
		\draw [style=to] (29) to (30);
		\draw [style=to] (33) to (34);
		\draw [style=to] (36) to (37);
		\draw [style=to] (14) to (44);
	\end{pgfonlayer}
\end{tikzpicture}

%% file: figures/central-morph-to-cone.tikz
\begin{tikzpicture}
	\begin{pgfonlayer}{nodelayer}
		\node [style=empty] (0) at (-6, 3) {$Z\otimes\TT Y$};
		\node [style=empty] (1) at (4.75, 3) {$\TT X\otimes\TT Y$};
		\node [style=empty] (2) at (14, 3) {$\TT (X\otimes\TT Y)$};
		\node [style=empty] (3) at (14, -2) {$\TT^2 (X\otimes Y)$};
		\node [style=empty] (4) at (14, -7) {$\TT (X\otimes Y)$};
		\node [style=empty] (5) at (-6, -2) {$\TT X\otimes\TT Y$};
		\node [style=empty] (6) at (-6, -7) {$\TT(\TT X\otimes Y)$};
		\node [style=empty] (7) at (3, -7) {$\TT^2 (X\otimes Y)$};
		\node [style=empty] (8) at (-6, 2.5) {};
		\node [style=empty] (9) at (-6, -1.5) {};
		\node [style=empty] (10) at (-6, -2.5) {};
		\node [style=empty] (11) at (-6, -6.5) {};
		\node [style=empty] (12) at (14, 2.5) {};
		\node [style=empty] (13) at (14, -1.5) {};
		\node [style=empty] (14) at (14, -2.5) {};
		\node [style=empty] (15) at (14, -6.5) {};
		\node [style=empty] (16) at (0.5, 3.5) {$f\otimes\TT Y$};
		\node [style=empty] (17) at (10.25, 3.5) {$\tau'_{X,\TT Y}$};
		\node [style=empty] (18) at (15.25, 0.5) {$\TT\tau_{X,Y}$};
		\node [style=empty] (19) at (15.25, -4.5) {$\mu_{X\otimes Y}$};
		\node [style=empty] (20) at (-7.25, 0.5) {$f\otimes\TT Y$};
		\node [style=empty] (21) at (-7.25, -4.5) {$\tau_{\TT X,Y}$};
		\node [style=empty] (22) at (-2, -7.5) {$\TT\tau'_{X,Y}$};
		\node [style=empty] (23) at (9, -7.5) {$\mu_{X\otimes Y}$};
		\node [style=empty] (25) at (-5, 2.5) {};
		\node [style=empty] (26) at (0, 0.25) {};
		\node [style=empty] (27) at (-1, 1.75) {$f\otimes\TT Y$};
		\node [style=empty] (28) at (2, 0.5) {};
		\node [style=empty] (29) at (4.5, 2.5) {};
		\node [style=empty] (30) at (1.5, 0) {$\TT X\otimes\TT Y$};
		\node [style=empty] (31) at (11, 0) {$\TT (X\otimes\TT Y)$};
		\node [style=empty] (32) at (6.75, 0.5) {$\tau'_{X,\TT Y}$};
		\node [style=empty] (33) at (11.25, 0.5) {};
		\node [style=empty] (34) at (13.5, 2.5) {};
		\node [style=empty] (37) at (0.5, -5) {$\TT(Z\otimes Y)$};
		\node [style=empty] (38) at (0.5, -2.5) {};
		\node [style=empty] (39) at (0.5, -4.5) {};
		\node [style=empty] (40) at (-0.75, -3.75) {$\tau_{Z,Y}$};
		\node [style=empty] (41) at (-4.5, -6.5) {};
		\node [style=empty] (42) at (-0.75, -5.5) {};
		\node [style=none] (43) at (-3.5, -4) {\textcolor{blue}{(1)}};
		\node [style=empty] (44) at (0.5, -2) {$Z\otimes\TT Y$};
		\node [style=empty] (45) at (-5.5, 2.25) {};
		\node [style=empty] (46) at (-0.25, -1.5) {};
		\node [style=empty] (47) at (-2.75, -1.5) {$f\otimes\TT Y$};
		\node [style=empty] (48) at (-3.75, -5.5) {$\TT(f\otimes Y)$};
		\node [style=none] (49) at (7, -3.5) {\textcolor{blue}{(2)}};
	\end{pgfonlayer}
	\begin{pgfonlayer}{edgelayer}
		\draw [style=to] (0) to (1);
		\draw [style=to] (1) to (2);
		\draw [style=to] (6) to (7);
		\draw [style=to] (7) to (4);
		\draw [style=to] (8) to (9);
		\draw [style=to] (10) to (11);
		\draw [style=to] (14) to (15);
		\draw [style=to] (12) to (13);
		\draw [style=to] (25) to (26);
		\draw [style=equal-arrow] (28) to (29);
		\draw [style=to] (30) to (31);
		\draw [style=equal-arrow] (33) to (34);
		\draw [style=to] (38) to (39);
		\draw [style=equal-arrow] (45) to (46);
		\draw [style=to] (44) to (5);
		\draw [style=to] (42) to (41);
	\end{pgfonlayer}
\end{tikzpicture}

%% file: figures/central-precomposing.tikz
\begin{tikzpicture}
	\begin{pgfonlayer}{nodelayer}
		\node [style=empty] (0) at (-6, 3) {$X\otimes\TT X'$};
		\node [style=empty] (1) at (0, 3) {$\TT Y\otimes\TT X'$};
		\node [style=empty] (2) at (6, 3) {$\TT (Y\otimes\TT X')$};
		\node [style=empty] (3) at (6, 0) {$\TT^2 (Y\otimes X')$};
		\node [style=empty] (4) at (6, -3) {$\TT (Y\otimes X')$};
		\node [style=empty] (5) at (-6, 0) {$\TT Y\otimes\TT X'$};
		\node [style=empty] (6) at (-6, -3) {$\TT(\TT Y\otimes X')$};
		\node [style=empty] (7) at (0, -3) {$\TT^2 (Y\otimes X')$};
		\node [style=empty] (8) at (-6, 2.5) {};
		\node [style=empty] (9) at (-6, 0.5) {};
		\node [style=empty] (10) at (-6, -0.5) {};
		\node [style=empty] (11) at (-6, -2.5) {};
		\node [style=empty] (12) at (6, 2.5) {};
		\node [style=empty] (13) at (6, 0.5) {};
		\node [style=empty] (14) at (6, -0.5) {};
		\node [style=empty] (15) at (6, -2.5) {};
		\node [style=empty] (16) at (-3, 3.5) {$f\otimes\TT X'$};
		\node [style=empty] (17) at (3, 3.5) {$\tau'_{Y,\TT X'}$};
		\node [style=empty] (18) at (7.25, 1.5) {$\TT\tau_{Y,X'}$};
		\node [style=empty] (19) at (7.25, -1.5) {$\mu_{Y\otimes X'}$};
		\node [style=empty] (20) at (-7.25, 1.5) {$f\otimes\TT X'$};
		\node [style=empty] (21) at (-7.25, -1.5) {$\tau_{\TT Y,X'}$};
		\node [style=empty] (22) at (-3, -3.5) {$\TT\tau'_{Y,X'}$};
		\node [style=empty] (23) at (3, -3.5) {$\mu_{Y\otimes X'}$};
		\node [style=empty] (24) at (-12.25, 3) {$Z\otimes\TT X'$};
		\node [style=empty] (27) at (-9, 3.5) {$g\otimes\TT X'$};
	\end{pgfonlayer}
	\begin{pgfonlayer}{edgelayer}
		\draw [style=to] (0) to (1);
		\draw [style=to] (1) to (2);
		\draw [style=to] (6) to (7);
		\draw [style=to] (7) to (4);
		\draw [style=to] (8) to (9);
		\draw [style=to] (10) to (11);
		\draw [style=to] (14) to (15);
		\draw [style=to] (12) to (13);
		\draw [style=to] (24) to (0);
	\end{pgfonlayer}
\end{tikzpicture}

%% file: figures/central-postcomposing.tikz
\begin{tikzpicture}
	\begin{pgfonlayer}{nodelayer}
		\node [style=empty] (0) at (-13, 13) {$X\otimes\TT X'$};
		\node [style=empty] (1) at (6.75, 13) {$\TT Y\otimes\TT X'$};
		\node [style=empty] (2) at (13.75, 13) {$\TT Z\otimes\TT X'$};
		\node [style=empty] (3) at (13.75, 3) {$\TT^2 (Z\otimes X')$};
		\node [style=empty] (4) at (13.75, -3.25) {$\TT (Z\otimes X')$};
		\node [style=empty] (5) at (-13, -3.25) {$\TT Z\otimes\TT X'$};
		\node [style=empty] (6) at (-6, -3.25) {$\TT(\TT Z\otimes X')$};
		\node [style=empty] (7) at (1, -3.25) {$\TT^2 (Z\otimes X')$};
		\node [style=empty] (8) at (-13, 12.5) {};
		\node [style=empty] (9) at (-13, 0.75) {};
		\node [style=empty] (10) at (-13, -0.25) {};
		\node [style=empty] (11) at (-13, -2.75) {};
		\node [style=empty] (12) at (13.75, 7.5) {};
		\node [style=empty] (13) at (13.75, 3.5) {};
		\node [style=empty] (14) at (13.75, 2.5) {};
		\node [style=empty] (15) at (13.75, -2.75) {};
		\node [style=empty] (16) at (-3, 13.75) {$f\otimes\TT X'$};
		\node [style=empty] (17) at (15, 10.75) {$\tau'_{Z,\TT X'}$};
		\node [style=empty] (18) at (15, 5.5) {$\TT\tau_{Z,X'}$};
		\node [style=empty] (19) at (15, 0) {$\mu_{Z\otimes X'}$};
		\node [style=empty] (20) at (-14.5, 7) {$f\otimes\TT X'$};
		\node [style=empty] (21) at (-9.5, -4) {$\tau_{\TT Z,X'}$};
		\node [style=empty] (22) at (-2.25, -4) {$\TT\tau'_{Z,X'}$};
		\node [style=empty] (23) at (7, -4) {$\mu_{Z\otimes X'}$};
		\node [style=empty] (24) at (13.75, 12.5) {};
		\node [style=empty] (25) at (13.75, 8.5) {};
		\node [style=empty] (26) at (13.75, 8) {$\TT (Z\otimes\TT X')$};
		\node [style=empty] (27) at (-13, 0.25) {$\TT Y\otimes\TT X'$};
		\node [style=empty] (28) at (-14.75, -1.5) {$\TT g\otimes\TT X'$};
		\node [style=empty] (29) at (10.25, 13.75) {$\TT g\otimes\TT X'$};
		\node [style=empty] (30) at (6.75, 8) {$\TT (Y\otimes\TT X')$};
		\node [style=empty] (31) at (6.75, 3) {$\TT^2 (Y\otimes X')$};
		\node [style=empty] (32) at (6.75, 12.5) {};
		\node [style=empty] (33) at (6.75, 8.5) {};
		\node [style=empty] (34) at (6.75, 7.5) {};
		\node [style=empty] (35) at (6.75, 3.5) {};
		\node [style=empty] (36) at (6.75, 0.25) {$\TT (Y\otimes X')$};
		\node [style=empty] (37) at (7.75, 0) {};
		\node [style=empty] (38) at (12.75, -2.75) {};
		\node [style=empty] (39) at (6.75, 2.5) {};
		\node [style=empty] (40) at (6.75, 0.75) {};
		\node [style=empty] (41) at (5.5, 10.75) {$\tau'_{Y,\TT X'}$};
		\node [style=empty] (42) at (5.5, 5.5) {$\TT\tau_{Y,X'}$};
		\node [style=empty] (43) at (5.5, 1.75) {$\mu_{Y\otimes X'}$};
		\node [style=empty] (44) at (10.25, 8.75) {$\TT (g\otimes\TT X')$};
		\node [style=empty] (45) at (10.25, 3.75) {$\TT^2( g\otimes X')$};
		\node [style=empty] (46) at (11, -0.25) {$\TT( g\otimes X')$};
		\node [style=empty] (47) at (-6, 0.25) {$\TT(\TT Y\otimes X')$};
		\node [style=empty] (48) at (1, 0.25) {$\TT^2 (Y\otimes X')$};
		\node [style=empty] (49) at (1, -0.25) {};
		\node [style=empty] (50) at (1, -2.75) {};
		\node [style=empty] (51) at (-6, -0.25) {};
		\node [style=empty] (52) at (-6, -2.75) {};
		\node [style=empty] (53) at (-9.5, 1) {$\tau_{\TT Y,X'}$};
		\node [style=empty] (54) at (-2.25, 1) {$\TT\tau'_{Y,X'}$};
		\node [style=empty] (55) at (3.75, 0.75) {$\mu_{Y\otimes X'}$};
		\node [style=empty] (56) at (-8, -1.5) {$\TT (\TT g\otimes X')$};
		\node [style=empty] (57) at (3.25, -1.5) {$\TT^2( g\otimes X')$};
		\node [style=none] (58) at (-4, 7) {\textcolor{blue}{(1)}};
		\node [style=none] (59) at (10.25, 11.25) {\textcolor{blue}{(2)}};
		\node [style=none] (60) at (10.25, 5.75) {\textcolor{blue}{(3)}};
		\node [style=none] (61) at (10.25, 1.5) {\textcolor{blue}{(4)}};
		\node [style=none] (62) at (-11.25, -1.5) {\textcolor{blue}{(5)}};
		\node [style=none] (63) at (-2.5, -1.5) {\textcolor{blue}{(6)}};
		\node [style=none] (64) at (6.75, -1.5) {\textcolor{blue}{(7)}};
	\end{pgfonlayer}
	\begin{pgfonlayer}{edgelayer}
		\draw [style=to] (0) to (1);
		\draw [style=to] (1) to (2);
		\draw [style=to] (6) to (7);
		\draw [style=to] (7) to (4);
		\draw [style=to] (8) to (9);
		\draw [style=to] (10) to (11);
		\draw [style=to] (14) to (15);
		\draw [style=to] (12) to (13);
		\draw (5) to (6);
		\draw [style=to] (24) to (25);
		\draw [style=to] (30) to (26);
		\draw [style=to] (32) to (33);
		\draw [style=to] (34) to (35);
		\draw [style=to] (31) to (3);
		\draw [style=to] (37) to (38);
		\draw [style=to] (39) to (40);
		\draw [style=to] (27) to (47);
		\draw [style=to] (47) to (48);
		\draw [style=to] (48) to (36);
		\draw [style=to] (51) to (52);
		\draw [style=to] (49) to (50);
	\end{pgfonlayer}
\end{tikzpicture}

%% file: figures/z-functor.tikz
\begin{tikzpicture}
	\begin{pgfonlayer}{nodelayer}
		\node [style=empty] (0) at (-1.5, 1.5) {$\ZZ X$};
		\node [style=empty] (1) at (1.5, 1.5) {$\ZZ Y$};
		\node [style=empty] (2) at (-1.5, -1.5) {$\TT X$};
		\node [style=empty] (3) at (1.5, -1.5) {$\TT Y$};
		\node [style=empty] (4) at (-1.5, 1) {};
		\node [style=empty] (5) at (1.5, 1) {};
		\node [style=empty] (6) at (1.5, -1) {};
		\node [style=empty] (7) at (-1.5, -1) {};
		\node [style=empty] (8) at (-2.25, 0) {$\iota_X$};
		\node [style=empty] (9) at (2.25, 0) {$\iota_Y$};
		\node [style=empty] (10) at (0, -2.25) {$\TT f$};
		\node [style=empty] (11) at (0, 2.25) {$\ZZ f$};
	\end{pgfonlayer}
	\begin{pgfonlayer}{edgelayer}
		\draw [style=to] (4) to (7);
		\draw [style=to] (5) to (6);
		\draw [style=to] (2) to (3);
		\draw [style=dashed-arrow] (0) to (1);
	\end{pgfonlayer}
\end{tikzpicture}

%% file: figures/z-functor-comp.tikz
\begin{tikzpicture}
	\begin{pgfonlayer}{nodelayer}
		\node [style=empty] (0) at (-1.5, 3) {$\ZZ A$};
		\node [style=empty] (1) at (-1.5, 0) {$\ZZ B$};
		\node [style=empty] (2) at (1.5, 3) {$\TT A$};
		\node [style=empty] (3) at (1.5, 0) {$\TT B$};
		\node [style=empty] (4) at (-1, 3) {};
		\node [style=empty] (5) at (-1, 0) {};
		\node [style=empty] (6) at (1, 0) {};
		\node [style=empty] (7) at (1, 3) {};
		\node [style=empty] (8) at (0, 3.75) {$\iota_A$};
		\node [style=empty] (9) at (0, 0.75) {$\iota_B$};
		\node [style=empty] (10) at (2.25, 1.5) {$\TT g$};
		\node [style=empty] (11) at (-2.25, 1.5) {$\ZZ g$};
		\node [style=empty] (12) at (-1.5, -3) {$\ZZ C$};
		\node [style=empty] (13) at (-2.25, -1.5) {$\ZZ f$};
		\node [style=empty] (14) at (-5.5, 0) {$\ZZ(f\circ g)$};
		\node [style=empty] (15) at (-1, -3) {};
		\node [style=empty] (16) at (1, -3) {};
		\node [style=empty] (17) at (1.5, -3) {$\TT C$};
		\node [style=empty] (18) at (0, -2.25) {$\iota_C$};
		\node [style=empty] (19) at (2.25, -1.5) {$\TT f$};
		\node [style=empty] (20) at (5.5, 0) {$\TT(f\circ g)$};
	\end{pgfonlayer}
	\begin{pgfonlayer}{edgelayer}
		\draw [style=to] (4) to (7);
		\draw [style=to] (5) to (6);
		\draw [style=to] (2) to (3);
		\draw [style=dashed-arrow, bend right=75] (0) to (12);
		\draw [style=to] (3) to (17);
		\draw [style=to] (15) to (16);
		\draw [style=to] (0) to (1);
		\draw [style=to] (1) to (12);
		\draw [style=to, bend left=75] (2) to (17);
	\end{pgfonlayer}
\end{tikzpicture}

%% file: figures/central-unit.tikz
\begin{tikzpicture}
	\begin{pgfonlayer}{nodelayer}
		\node [style=empty] (0) at (-1.5, 1) {$X$};
		\node [style=empty] (1) at (0, -1.5) {$\TT X$};
		\node [style=empty] (2) at (1.75, 1) {$\ZZ X$};
		\node [style=empty] (3) at (-1.25, 1) {};
		\node [style=empty] (4) at (1.25, 1) {};
		\node [style=empty] (5) at (0.25, -1.25) {};
		\node [style=empty] (6) at (1.75, 0.75) {};
		\node [style=empty] (7) at (-1.425, -0.4) {$\eta_X$};
		\node [style=empty] (8) at (1.5, -0.4) {$\iota_X$};
		\node [style=empty] (9) at (-0.025, 1.525) {$\eta^\ZZ_X$};
		\node [style=none] (10) at (-0.5, -1) {};
		\node [style=none] (11) at (-1.375, 0.5) {};
	\end{pgfonlayer}
	\begin{pgfonlayer}{edgelayer}
		\draw [style=dashed-arrow] (3) to (4);
		\draw [style=to] (11.center) to (10.center);
		\draw [style=to] (6) to (5);
	\end{pgfonlayer}
\end{tikzpicture}

%% file: figures/central-mult.tikz
\begin{tikzpicture}
	\begin{pgfonlayer}{nodelayer}
		\node [style=empty] (0) at (-3, 2) {$\ZZ^2 X$};
		\node [style=empty] (1) at (-3, -1) {$\TT\ZZ X$};
		\node [style=empty] (2) at (-3, 1.5) {};
		\node [style=empty] (3) at (-3, -0.5) {};
		\node [style=empty] (4) at (0.5, -1) {$\TT^2 X$};
		\node [style=empty] (5) at (4, -1) {$\TT X$};
		\node [style=empty] (6) at (4, 2) {$\ZZ X$};
		\node [style=empty] (7) at (4, 1.5) {};
		\node [style=empty] (8) at (4, -0.5) {};
		\node [style=empty] (9) at (0.5, 2.75) {$\mu^\ZZ_X$};
		\node [style=empty] (10) at (-1, -1.5) {$\TT\iota_X$};
		\node [style=empty] (11) at (2.5, -1.5) {$\mu_X$};
		\node [style=empty] (12) at (-3.75, 0.5) {$\iota_{\ZZ X}$};
		\node [style=empty] (13) at (4.75, 0.5) {$\iota_X$};
	\end{pgfonlayer}
	\begin{pgfonlayer}{edgelayer}
		\draw [style=to] (7) to (8);
		\draw [style=to] (1) to (4);
		\draw [style=to] (4) to (5);
		\draw [style=to] (2) to (3);
		\draw [style=dashed-arrow] (0) to (6);
	\end{pgfonlayer}
\end{tikzpicture}

%% file: figures/central-strength.tikz
\begin{tikzpicture}
	\begin{pgfonlayer}{nodelayer}
		\node [style=empty] (0) at (-3.5, 2) {$A\otimes\ZZ B$};
		\node [style=empty] (1) at (2.5, 2) {$\ZZ(A\otimes B)$};
		\node [style=empty] (2) at (-3.5, -1) {$A\otimes\TT B$};
		\node [style=empty] (3) at (2.5, -1) {$\TT(A\otimes B)$};
		\node [style=empty] (4) at (-3.5, 1.5) {};
		\node [style=empty] (5) at (-3.5, -0.5) {};
		\node [style=empty] (6) at (2.5, 1.5) {};
		\node [style=empty] (7) at (2.5, -0.5) {};
		\node [style=empty] (8) at (-4.75, 0.5) {$A\otimes\iota_B$};
		\node [style=empty] (9) at (3.5, 0.5) {$\iota_{A\otimes B}$};
		\node [style=empty] (10) at (-0.625, 2.625) {$\tau^\ZZ_{A,B}$};
		\node [style=empty] (11) at (-0.625, -1.5) {$\tau_{A,B}$};
	\end{pgfonlayer}
	\begin{pgfonlayer}{edgelayer}
		\draw [style=to] (4) to (5);
		\draw [style=to] (2) to (3);
		\draw [style=to] (6) to (7);
		\draw [style=dashed-arrow] (0) to (1);
	\end{pgfonlayer}
\end{tikzpicture}

%% file: figures/eta-z-natural.tikz
\begin{tikzpicture}
	\begin{pgfonlayer}{nodelayer}
		\node [style=empty] (0) at (-6, 5) {$X$};
		\node [style=empty] (1) at (6, 5) {$\ZZ X$};
		\node [style=empty] (2) at (-6, -3) {$Y$};
		\node [style=empty] (3) at (6, -3) {$\TT Y$};
		\node [style=empty] (4) at (-6, 4.5) {};
		\node [style=empty] (5) at (-6, -2.5) {};
		\node [style=empty] (6) at (6, 4.5) {};
		\node [style=empty] (7) at (6, 1) {$\ZZ Y$};
		\node [style=empty] (8) at (6, 1.5) {};
		\node [style=empty] (9) at (6, 0.5) {};
		\node [style=empty] (10) at (6, -2.5) {};
		\node [style=empty] (11) at (0, -3) {$\ZZ Y$};
		\node [style=empty] (12) at (0, 3) {$\TT X$};
		\node [style=empty] (13) at (-0.5, 3.25) {};
		\node [style=empty] (14) at (0.5, 3.25) {};
		\node [style=empty] (15) at (5.5, 4.5) {};
		\node [style=empty] (16) at (-5.5, 4.5) {};
		\node [style=empty] (17) at (0.5, 2.5) {};
		\node [style=empty] (18) at (5.5, -2.5) {};
		\node [style=empty] (19) at (0, 5.75) {$\eta^\ZZ_X$};
		\node [style=empty] (20) at (-7, 1.25) {$f$};
		\node [style=empty] (21) at (7, 3.25) {$\ZZ f$};
		\node [style=empty] (22) at (1.5, 0.25) {$\TT f$};
		\node [style=empty] (23) at (0, -1.75) {$\eta_y$};
		\node [style=empty] (24) at (-3, -3.75) {$\eta^\ZZ_Y$};
		\node [style=empty] (25) at (3, -3.75) {$\iota_Y$};
		\node [style=empty] (26) at (6.75, -1) {$\iota_Y$};
		\node [style=empty] (27) at (-3.75, 3.25) {$\eta_X$};
		\node [style=empty] (28) at (3.75, 3.25) {$\iota_X$};
		\node [style=none] (29) at (0, 4.25) {\textcolor{blue}{(1)}};
		\node [style=none] (30) at (4, 1.25) {\textcolor{blue}{(2)}};
		\node [style=none] (31) at (-2.5, 0.5) {\textcolor{blue}{(3)}};
		\node [style=none] (32) at (-1.75, -2.25) {\textcolor{blue}{(4)}};
	\end{pgfonlayer}
	\begin{pgfonlayer}{edgelayer}
		\draw [style=to] (0) to (1);
		\draw [style=to] (6) to (8);
		\draw [style=to] (4) to (5);
		\draw [style=to] (2) to (11);
		\draw [style=mono] (11) to (3);
		\draw [style=mono] (9) to (10);
		\draw [style=to, bend left] (2) to (3);
		\draw [style=to] (16) to (13);
		\draw [style=to] (15) to (14);
		\draw [style=to] (17) to (18);
	\end{pgfonlayer}
\end{tikzpicture}

%% file: figures/mu-z-natural.tikz
\begin{tikzpicture}
	\begin{pgfonlayer}{nodelayer}
		\node [style=empty] (0) at (-8, 5) {$\ZZ^2 X$};
		\node [style=empty] (1) at (8, 5) {$\ZZ X$};
		\node [style=empty] (2) at (-8, -7) {$\ZZ^2 Y$};
		\node [style=empty] (3) at (8, -7) {$\TT Y$};
		\node [style=empty] (4) at (-8, 4.5) {};
		\node [style=empty] (5) at (-8, -6.5) {};
		\node [style=empty] (6) at (0, -7) {$\ZZ Y$};
		\node [style=empty] (7) at (8, -1) {$\ZZ Y$};
		\node [style=empty] (8) at (8, 4.5) {};
		\node [style=empty] (9) at (8, -0.5) {};
		\node [style=empty] (10) at (8, -1.5) {};
		\node [style=empty] (11) at (8, -6.5) {};
		\node [style=empty] (12) at (-3, 2) {$\TT^2 X$};
		\node [style=empty] (13) at (3, 2) {$\TT X$};
		\node [style=empty] (14) at (-7.5, 4.5) {};
		\node [style=empty] (15) at (-3.5, 2.25) {};
		\node [style=empty] (16) at (3.5, 2.25) {};
		\node [style=empty] (17) at (7.5, 4.5) {};
		\node [style=empty] (18) at (0, -3) {$\TT^2 Y$};
		\node [style=empty] (19) at (-7.5, -6.5) {};
		\node [style=empty] (20) at (-0.5, -3.25) {};
		\node [style=empty] (21) at (0.5, -3.25) {};
		\node [style=empty] (22) at (7.25, -6.75) {};
		\node [style=empty] (23) at (-2.5, 1.5) {};
		\node [style=empty] (24) at (-0.25, -2.5) {};
		\node [style=empty] (25) at (3.25, 1.5) {};
		\node [style=empty] (26) at (7.5, -6.25) {};
		\node [style=empty] (27) at (-9, -1) {$\ZZ^2 f$};
		\node [style=empty] (28) at (9, 2) {$\ZZ f$};
		\node [style=empty] (29) at (9, -4) {$\iota_Y$};
		\node [style=empty] (30) at (4, -7.75) {$\iota_Y$};
		\node [style=empty] (31) at (0, 5.75) {$\mu^\ZZ_X$};
		\node [style=empty] (32) at (-6, 3) {$\iota\iota$};
		\node [style=empty] (33) at (-4.25, -4) {$\iota\iota$};
		\node [style=empty] (34) at (5.25, -0.25) {$\TT f$};
		\node [style=empty] (35) at (0, 2.5) {$\mu_X$};
		\node [style=empty] (36) at (-2.5, -0.75) {$\TT^2 f$};
		\node [style=empty] (37) at (2.5, -5) {$\mu_Y$};
		\node [style=empty] (38) at (-4, -7.75) {$\mu^\ZZ_Y$};
		\node [style=none] (39) at (0, 4) {\textcolor{blue}{(1)}};
		\node [style=none] (40) at (-5.5, -1) {\textcolor{blue}{(2)}};
		\node [style=none] (41) at (2, -1) {\textcolor{blue}{(3)}};
		\node [style=none] (42) at (6.5, -1.75) {\textcolor{blue}{(4)}};
		\node [style=none] (43) at (0, -5.5) {\textcolor{blue}{(5)}};
		\node [style=none] (44) at (5.75, 3) {$\iota$};
	\end{pgfonlayer}
	\begin{pgfonlayer}{edgelayer}
		\draw [style=to] (4) to (5);
		\draw [style=mono] (10) to (11);
		\draw [style=to] (8) to (9);
		\draw [style=to] (2) to (6);
		\draw [style=mono] (6) to (3);
		\draw [style=to] (0) to (1);
		\draw [style=to] (12) to (13);
		\draw [style=to] (14) to (15);
		\draw [style=to] (17) to (16);
		\draw [style=to] (19) to (20);
		\draw [style=to] (21) to (22);
		\draw [style=to] (23) to (24);
		\draw [style=to] (25) to (26);
	\end{pgfonlayer}
\end{tikzpicture}

%% file: figures/tau-z-natural.tikz
\begin{tikzpicture}
	\begin{pgfonlayer}{nodelayer}
		\node [style=empty] (0) at (-13, 6) {$A\otimes\ZZ C$};
		\node [style=empty] (1) at (5, 6) {$\ZZ(A\otimes C)$};
		\node [style=empty] (2) at (-13, -2) {$B\otimes\ZZ C$};
		\node [style=empty] (3) at (5, -2) {$\TT(B\otimes C)$};
		\node [style=empty] (4) at (-6, -2) {$\ZZ(B\otimes C)$};
		\node [style=empty] (5) at (5, 2) {$\ZZ(B\otimes C)$};
		\node [style=empty] (6) at (5, 5.5) {};
		\node [style=empty] (7) at (5, 2.5) {};
		\node [style=empty] (8) at (5, 1.5) {};
		\node [style=empty] (9) at (5, -1.5) {};
		\node [style=empty] (10) at (-13, 5.5) {};
		\node [style=empty] (11) at (-13, -1.5) {};
		\node [style=empty] (12) at (-7, 0.5) {$B\otimes\TT C$};
		\node [style=empty] (13) at (-7, 4) {$A\otimes\TT C$};
		\node [style=empty] (14) at (-1, 4) {$\TT(A\otimes C)$};
		\node [style=empty] (15) at (-7, 3.5) {};
		\node [style=empty] (16) at (-7, 1) {};
		\node [style=empty] (17) at (-5.5, 0.25) {};
		\node [style=empty] (18) at (3.5, -1.5) {};
		\node [style=empty] (19) at (4.25, -1.25) {};
		\node [style=empty] (20) at (-0.75, 3.5) {};
		\node [style=empty] (21) at (-12.5, 5.5) {};
		\node [style=empty] (22) at (-8.5, 4.25) {};
		\node [style=empty] (23) at (4.5, 5.5) {};
		\node [style=empty] (24) at (0.5, 4.25) {};
		\node [style=empty] (25) at (-12.5, -1.75) {};
		\node [style=empty] (26) at (-8.25, 0) {};
		\node [style=empty] (27) at (-4, 6.75) {$\tau^\ZZ_{A,C}$};
		\node [style=empty] (28) at (-4, 4.5) {$\tau_{A,C}$};
		\node [style=empty] (29) at (-9.5, -2.75) {$\tau^\ZZ_{B,C}$};
		\node [style=empty] (30) at (-1, -2.75) {$\iota_{B\otimes C}$};
		\node [style=empty] (31) at (6, 0.25) {$\iota_{B\otimes C}$};
		\node [style=empty] (32) at (-3, -0.75) {$\tau_{B,C}$};
		\node [style=empty] (33) at (-11.5, 2) {$f\otimes\ZZ C$};
		\node [style=empty] (34) at (-9.25, 5) {$A\otimes\iota$};
		\node [style=empty] (35) at (-9.25, -1) {$B\otimes\iota$};
		\node [style=empty] (36) at (6.5, 4) {$\ZZ(f\otimes C)$};
		\node [style=empty] (37) at (0, 0.75) {$\TT(f\otimes C)$};
		\node [style=empty] (38) at (-5.5, 2.25) {$f\otimes\TT C$};
		\node [style=empty] (39) at (2.75, 4.5) {$\iota$};
		\node [style=none] (40) at (-4, 5.25) {\textcolor{blue}{(1)}};
		\node [style=none] (41) at (-9, 2) {\textcolor{blue}{(2)}};
		\node [style=none] (42) at (-2.5, 2) {\textcolor{blue}{(3)}};
		\node [style=none] (43) at (2.5, 2.25) {\textcolor{blue}{(4)}};
		\node [style=none] (44) at (-6, -1) {\textcolor{blue}{(5)}};
	\end{pgfonlayer}
	\begin{pgfonlayer}{edgelayer}
		\draw [style=to] (0) to (1);
		\draw [style=to] (10) to (11);
		\draw [style=to] (21) to (22);
		\draw [style=to] (23) to (24);
		\draw [style=to] (13) to (14);
		\draw [style=to] (15) to (16);
		\draw [style=to] (25) to (26);
		\draw [style=to] (2) to (4);
		\draw [style=to] (17) to (18);
		\draw [style=to] (20) to (19);
		\draw [style=to] (6) to (7);
		\draw [style=mono] (4) to (3);
		\draw [style=mono] (8) to (9);
	\end{pgfonlayer}
\end{tikzpicture}

%% file: figures/tau-z-natural-1.tikz
\begin{tikzpicture}
	\begin{pgfonlayer}{nodelayer}
		\node [style=empty] (0) at (-9, 4.25) {$A\otimes\ZZ B$};
		\node [style=empty] (1) at (9, 4.25) {$\ZZ(A\otimes B)$};
		\node [style=empty] (2) at (-9, -3.75) {$A\otimes\ZZ C$};
		\node [style=empty] (3) at (9, -3.75) {$\TT(A\otimes C)$};
		\node [style=empty] (4) at (-2, -3.75) {$\ZZ(A\otimes C)$};
		\node [style=empty] (5) at (9, 0.25) {$\ZZ(A\otimes C)$};
		\node [style=empty] (6) at (9, 3.75) {};
		\node [style=empty] (7) at (9, 0.75) {};
		\node [style=empty] (8) at (9, -0.25) {};
		\node [style=empty] (9) at (9, -3.25) {};
		\node [style=empty] (10) at (-9, 3.75) {};
		\node [style=empty] (11) at (-9, -3.25) {};
		\node [style=empty] (12) at (-3, -1.25) {$A\otimes\TT C$};
		\node [style=empty] (13) at (-3, 2.25) {$A\otimes\TT B$};
		\node [style=empty] (14) at (3, 2.25) {$\TT(A\otimes B)$};
		\node [style=empty] (15) at (-3, 1.75) {};
		\node [style=empty] (16) at (-3, -0.75) {};
		\node [style=empty] (17) at (-1.5, -1.5) {};
		\node [style=empty] (18) at (7.5, -3.25) {};
		\node [style=empty] (19) at (8.25, -3) {};
		\node [style=empty] (20) at (3.25, 1.75) {};
		\node [style=empty] (21) at (-8.5, 3.75) {};
		\node [style=empty] (22) at (-4.5, 2.5) {};
		\node [style=empty] (23) at (8.5, 3.75) {};
		\node [style=empty] (24) at (4.5, 2.5) {};
		\node [style=empty] (25) at (-8.5, -3.5) {};
		\node [style=empty] (26) at (-4.25, -1.75) {};
		\node [style=empty] (27) at (0, 5) {$\tau^\ZZ_{A,B}$};
		\node [style=empty] (28) at (0, 2.75) {$\tau_{A,B}$};
		\node [style=empty] (29) at (-5.5, -4.5) {$\tau^\ZZ_{A,C}$};
		\node [style=empty] (30) at (3, -4.5) {$\iota_{A\otimes C}$};
		\node [style=empty] (31) at (10, -1.5) {$\iota_{A\otimes C}$};
		\node [style=empty] (32) at (1, -2.5) {$\tau_{A,C}$};
		\node [style=empty] (33) at (-7.5, 0.25) {$A\otimes\ZZ f$};
		\node [style=empty] (34) at (-5.25, 3.25) {$A\otimes\iota$};
		\node [style=empty] (35) at (-5.25, -2.75) {$A\otimes\iota$};
		\node [style=empty] (36) at (10.5, 2.25) {$\ZZ(A\otimes f)$};
		\node [style=empty] (37) at (4, -1) {$\TT(A\otimes f)$};
		\node [style=empty] (38) at (-1.5, 0.5) {$A\otimes\TT f$};
		\node [style=empty] (39) at (6.75, 2.75) {$\iota$};
		\node [style=none] (40) at (0, 3.5) {\textcolor{blue}{(1)}};
		\node [style=none] (41) at (-5, 0.25) {\textcolor{blue}{(2)}};
		\node [style=none] (42) at (1.5, 0.25) {\textcolor{blue}{(3)}};
		\node [style=none] (43) at (6.25, 1) {\textcolor{blue}{(4)}};
		\node [style=none] (44) at (-1.75, -2.75) {\textcolor{blue}{(5)}};
	\end{pgfonlayer}
	\begin{pgfonlayer}{edgelayer}
		\draw [style=to] (0) to (1);
		\draw [style=to] (10) to (11);
		\draw [style=to] (21) to (22);
		\draw [style=to] (23) to (24);
		\draw [style=to] (13) to (14);
		\draw [style=to] (15) to (16);
		\draw [style=to] (25) to (26);
		\draw [style=to] (2) to (4);
		\draw [style=to] (17) to (18);
		\draw [style=to] (20) to (19);
		\draw [style=to] (6) to (7);
		\draw [style=mono] (4) to (3);
		\draw [style=mono] (8) to (9);
	\end{pgfonlayer}
\end{tikzpicture}

%% file: figures/z-monad-1.tikz
\begin{tikzpicture}
	\begin{pgfonlayer}{nodelayer}
		\node [style=empty] (0) at (-6, 6) {$\ZZ^3 X$};
		\node [style=empty] (1) at (6, 6) {$\ZZ^2 X$};
		\node [style=empty] (2) at (-6, -2) {$\ZZ^2 X$};
		\node [style=empty] (3) at (6, -2) {$\TT X$};
		\node [style=empty] (4) at (6, 2) {$\ZZ X$};
		\node [style=empty] (5) at (0, -2) {$\ZZ X$};
		\node [style=empty] (6) at (6, 1.5) {};
		\node [style=empty] (7) at (6, -1.5) {};
		\node [style=empty] (8) at (6, 5.5) {};
		\node [style=empty] (9) at (6, 2.5) {};
		\node [style=empty] (10) at (-2, 4) {$\TT^3 X$};
		\node [style=empty] (11) at (-2, 1) {$\TT^2 X$};
		\node [style=empty] (12) at (2, 4) {$\TT^2 X$};
		\node [style=empty] (13) at (2.5, 4.25) {};
		\node [style=empty] (14) at (5.5, 5.5) {};
		\node [style=empty] (15) at (-5.5, 5.5) {};
		\node [style=empty] (16) at (-2.5, 4.25) {};
		\node [style=empty] (17) at (2.25, 3.5) {};
		\node [style=empty] (18) at (5.5, -1.25) {};
		\node [style=empty] (19) at (5.25, -1.75) {};
		\node [style=empty] (20) at (-1.5, 0.75) {};
		\node [style=empty] (21) at (-2, 3.5) {};
		\node [style=empty] (22) at (-2, 1.5) {};
		\node [style=empty] (23) at (-6, 5.5) {};
		\node [style=empty] (24) at (-6, -1.5) {};
		\node [style=empty] (25) at (0, 6.75) {$\mu^\ZZ_{\ZZ X}$};
		\node [style=empty] (26) at (-3, -2.75) {$\mu^\ZZ_X$};
		\node [style=empty] (27) at (6.75, 4) {$\mu^\ZZ_X$};
		\node [style=empty] (28) at (3, -2.75) {$\iota_X$};
		\node [style=empty] (29) at (6.75, 0) {$\iota_X$};
		\node [style=empty] (30) at (-7, 2) {$\ZZ\mu^\ZZ_X$};
		\node [style=empty] (31) at (-3, 2.5) {$\TT\mu_X$};
		\node [style=empty] (32) at (0, 3.5) {$\mu_{\TT X}$};
		\node [style=empty] (33) at (2.75, 1.5) {$\mu_X$};
		\node [style=empty] (34) at (1.25, 0.25) {$\mu_X$};
		\node [style=empty] (35) at (-5.5, -1.5) {};
		\node [style=empty] (36) at (-2.5, 0.5) {};
		\node [style=empty] (37) at (-3.25, 5.25) {$\iota^3$};
		\node [style=empty] (38) at (3.5, 5.25) {$\iota^2$};
		\node [style=empty] (39) at (-4, 0) {$\iota^2$};
		\node [style=empty] (40) at (10.25, 6) {$\ZZ X$};
		\node [style=empty] (41) at (22.25, 6) {$\ZZ^2 X$};
		\node [style=empty] (42) at (10.25, -2) {$\ZZ^2 X$};
		\node [style=empty] (43) at (22.25, -2) {$\TT X$};
		\node [style=empty] (44) at (22.25, 2) {$\ZZ X$};
		\node [style=empty] (45) at (16.25, -2) {$\ZZ X$};
		\node [style=empty] (46) at (22.25, 1.5) {};
		\node [style=empty] (47) at (22.25, -1.5) {};
		\node [style=empty] (48) at (22.25, 5.5) {};
		\node [style=empty] (49) at (22.25, 2.5) {};
		\node [style=empty] (50) at (14.25, 4) {$\TT X$};
		\node [style=empty] (51) at (14.25, 1) {$\TT^2 X$};
		\node [style=empty] (52) at (18.25, 4) {$\TT^2 X$};
		\node [style=empty] (53) at (18.75, 4.25) {};
		\node [style=empty] (54) at (21.75, 5.5) {};
		\node [style=empty] (55) at (10.75, 5.5) {};
		\node [style=empty] (56) at (13.75, 4.25) {};
		\node [style=empty] (57) at (18.5, 3.5) {};
		\node [style=empty] (58) at (21.75, -1.25) {};
		\node [style=empty] (59) at (21.5, -1.75) {};
		\node [style=empty] (60) at (14.75, 0.75) {};
		\node [style=empty] (61) at (14.25, 3.5) {};
		\node [style=empty] (62) at (14.25, 1.5) {};
		\node [style=empty] (63) at (10.25, 5.5) {};
		\node [style=empty] (64) at (10.25, -1.5) {};
		\node [style=empty] (65) at (16.25, 6.75) {$\eta^\ZZ_{\ZZ X}$};
		\node [style=empty] (66) at (13.25, -2.75) {$\mu^\ZZ_X$};
		\node [style=empty] (67) at (23, 4) {$\mu^\ZZ_X$};
		\node [style=empty] (68) at (19.25, -2.75) {$\iota_X$};
		\node [style=empty] (69) at (23, 0) {$\iota_X$};
		\node [style=empty] (70) at (9.25, 2) {$\ZZ\eta^\ZZ_X$};
		\node [style=empty] (71) at (13.25, 2.5) {$\TT\eta_X$};
		\node [style=empty] (72) at (16.25, 3.5) {$\eta_{\TT X}$};
		\node [style=empty] (73) at (19, 1.5) {$\mu_X$};
		\node [style=empty] (74) at (17.5, 0.25) {$\mu_X$};
		\node [style=empty] (75) at (10.75, -1.5) {};
		\node [style=empty] (76) at (13.75, 0.5) {};
		\node [style=empty] (77) at (12.75, 5.25) {$\iota$};
		\node [style=empty] (78) at (20, 5.25) {$\iota^2$};
		\node [style=empty] (79) at (12, 0) {$\iota^2$};
		\node [style=none] (80) at (0, 5) {\textcolor{blue}{(1)}};
		\node [style=none] (81) at (-4.5, 2) {\textcolor{blue}{(2)}};
		\node [style=none] (82) at (0.5, 1.75) {\textcolor{blue}{(3)}};
		\node [style=none] (83) at (4.25, 2.75) {\textcolor{blue}{(4)}};
		\node [style=none] (84) at (-1.25, -0.75) {\textcolor{blue}{(5)}};
		\node [style=none] (85) at (16.25, 5) {\textcolor{blue}{(6)}};
		\node [style=none] (86) at (11.75, 2) {\textcolor{blue}{(7)}};
		\node [style=none] (87) at (16.75, 1.75) {\textcolor{blue}{(8)}};
		\node [style=none] (88) at (20.5, 2.75) {\textcolor{blue}{(9)}};
		\node [style=none] (89) at (15.75, -0.75) {\textcolor{blue}{(10)}};
	\end{pgfonlayer}
	\begin{pgfonlayer}{edgelayer}
		\draw [style=to] (23) to (24);
		\draw [style=to] (0) to (1);
		\draw [style=to] (15) to (16);
		\draw [style=to] (14) to (13);
		\draw [style=to] (10) to (12);
		\draw [style=to] (21) to (22);
		\draw [style=to] (17) to (18);
		\draw [style=to] (20) to (19);
		\draw [style=to] (2) to (5);
		\draw [style=to] (8) to (9);
		\draw [style=mono] (5) to (3);
		\draw [style=mono] (6) to (7);
		\draw [style=to] (35) to (36);
		\draw [style=to] (63) to (64);
		\draw [style=to] (40) to (41);
		\draw [style=to] (55) to (56);
		\draw [style=to] (54) to (53);
		\draw [style=to] (50) to (52);
		\draw [style=to] (61) to (62);
		\draw [style=to] (57) to (58);
		\draw [style=to] (60) to (59);
		\draw [style=to] (42) to (45);
		\draw [style=to] (48) to (49);
		\draw [style=mono] (45) to (43);
		\draw [style=mono] (46) to (47);
		\draw [style=to] (75) to (76);
	\end{pgfonlayer}
\end{tikzpicture}

%% file: figures/z-monad-commutative.tikz
\begin{tikzpicture}
	\begin{pgfonlayer}{nodelayer}
		\node [style=empty] (0) at (-10, 6) {$\ZZ A\otimes\ZZ B$};
		\node [style=empty] (1) at (21, 6) {$\ZZ^2(A\otimes B)$};
		\node [style=empty] (2) at (-10, -18) {$\ZZ^2(A\otimes B)$};
		\node [style=empty] (3) at (21, -18) {$\TT(A\otimes B)$};
		\node [style=empty] (4) at (-10, -6) {$\ZZ(\ZZ A\otimes B)$};
		\node [style=empty] (5) at (21, -6) {$\ZZ(A\otimes B)$};
		\node [style=empty] (6) at (5, -18) {$\ZZ(A\otimes B)$};
		\node [style=empty] (7) at (5, 6) {$\ZZ(A\otimes\ZZ B)$};
		\node [style=empty] (8) at (-10, 5.5) {};
		\node [style=empty] (9) at (-10, -5.5) {};
		\node [style=empty] (10) at (-10, -6.5) {};
		\node [style=empty] (11) at (-10, -17.5) {};
		\node [style=empty] (12) at (21, 5.5) {};
		\node [style=empty] (13) at (21, -5.5) {};
		\node [style=empty] (14) at (21, -6.5) {};
		\node [style=empty] (15) at (21, -17.5) {};
		\node [style=empty] (16) at (-3, 6.5) {$\tau'^\ZZ_{A,\ZZ B}$};
		\node [style=empty] (17) at (12, 6.5) {$\ZZ\tau^\ZZ_{A,B}$};
		\node [style=empty] (18) at (22.25, 0) {$\mu^\ZZ_{A\otimes B}$};
		\node [style=empty] (19) at (-3, -18.75) {$\mu^\ZZ_{A\otimes B}$};
		\node [style=empty] (20) at (-11.25, -12) {$\ZZ\tau'^\ZZ_{A,B}$};
		\node [style=empty] (21) at (-11.25, 0) {$\tau^\ZZ_{\ZZ A,B}$};
		\node [style=empty] (22) at (-2, -15) {$\ZZ\TT(A\otimes B)$};
		\node [style=empty] (23) at (11, -15) {$\TT^2(A\otimes B)$};
		\node [style=empty] (24) at (-8.25, -17.5) {};
		\node [style=empty] (25) at (-3.75, -15.25) {};
		\node [style=empty] (26) at (13, -15.5) {};
		\node [style=empty] (27) at (19.5, -17.5) {};
		\node [style=empty] (28) at (16, -1) {$\ZZ\TT(A\otimes B)$};
		\node [style=empty] (29) at (16, -9) {$\TT^2(A\otimes B)$};
		\node [style=empty] (30) at (16, -1.5) {};
		\node [style=empty] (31) at (16, -8.5) {};
		\node [style=empty] (32) at (20.5, 5.5) {};
		\node [style=empty] (33) at (16.25, -0.5) {};
		\node [style=empty] (34) at (20.25, -17) {};
		\node [style=empty] (35) at (16.25, -9.5) {};
		\node [style=empty] (36) at (-2, -9) {$\ZZ(\TT A\otimes B)$};
		\node [style=empty] (37) at (-8, -6.25) {};
		\node [style=empty] (38) at (-3.75, -8.5) {};
		\node [style=empty] (39) at (-2, -9.5) {};
		\node [style=empty] (40) at (-2, -14.5) {};
		\node [style=empty] (41) at (8, -1) {$\ZZ(A\otimes\TT B)$};
		\node [style=empty] (42) at (5.25, 5.5) {};
		\node [style=empty] (43) at (7.75, -0.5) {};
		\node [style=empty] (44) at (5.5, -11.5) {$\TT(\TT A\otimes B)$};
		\node [style=empty] (45) at (-0.25, -9.5) {};
		\node [style=empty] (46) at (3.5, -11.25) {};
		\node [style=empty] (47) at (6, -12) {};
		\node [style=empty] (48) at (10.25, -14.5) {};
		\node [style=empty] (49) at (11.5, -5.5) {$\TT(A\otimes\TT B)$};
		\node [style=empty] (50) at (8.25, -1.5) {};
		\node [style=empty] (51) at (11.5, -5) {};
		\node [style=empty] (52) at (12, -6) {};
		\node [style=empty] (53) at (15, -8.5) {};
		\node [style=empty] (54) at (3.75, -3.25) {$\TT A\otimes\TT B$};
		\node [style=empty] (55) at (2.5, -5.75) {};
		\node [style=empty] (56) at (5.25, -10.75) {};
		\node [style=empty] (57) at (-3, -3) {$\TT A\otimes\ZZ B$};
		\node [style=empty] (58) at (1, 2) {$\ZZ A\otimes\TT B$};
		\node [style=empty] (59) at (-9.5, 5.25) {};
		\node [style=empty] (60) at (-3.25, -2.5) {};
		\node [style=empty] (61) at (-8.25, 5.5) {};
		\node [style=empty] (62) at (-0.75, 2.25) {};
		\node [style=empty] (63) at (-2.75, -3.5) {};
		\node [style=empty] (64) at (-2, -8.5) {};
		\node [style=empty] (65) at (1.25, 1.5) {};
		\node [style=empty] (66) at (3.5, -2.75) {};
		\node [style=empty] (67) at (-3.75, 4) {$\iota$};
		\node [style=empty] (68) at (-7, 1.25) {$\iota$};
		\node [style=empty] (69) at (-6.5, -16) {$\iota$};
		\node [style=empty] (70) at (17.75, 3) {$\iota$};
		\node [style=empty] (71) at (-0.25, -3.75) {$\iota$};
		\node [style=empty] (72) at (3.25, -0.75) {$\iota$};
		\node [style=empty] (73) at (7, 3) {$\iota$};
		\node [style=empty] (74) at (4.5, 1.25) {$\tau'^\ZZ_{A,\TT B}$};
		\node [style=empty] (75) at (-3.5, -6) {$\tau^\ZZ_{\TT A,B}$};
		\node [style=empty] (76) at (-6.25, -7.75) {$\iota$};
		\node [style=empty] (77) at (-3.25, -12) {$\ZZ\tau'_{A,B}$};
		\node [style=empty] (78) at (1.25, -10.75) {$\iota$};
		\node [style=empty] (79) at (5.25, -7.75) {$\tau_{\TT A,B}$};
		\node [style=empty] (80) at (7.5, -5.25) {$\tau'_{A,\TT B}$};
		\node [style=empty] (81) at (8.5, -12.75) {$\TT\tau'_{A,B}$};
		\node [style=empty] (82) at (12.5, -7.5) {$\TT\tau_{A,B}$};
		\node [style=empty] (83) at (3.25, -14.5) {$\iota$};
		\node [style=empty] (84) at (16.5, -12.75) {$\mu_{A\otimes B}$};
		\node [style=empty] (85) at (16.25, -15.75) {$\mu_{A\otimes B}$};
		\node [style=empty] (86) at (16.5, -5) {$\iota$};
		\node [style=empty] (87) at (9.25, -3.5) {$\iota$};
		\node [style=empty] (88) at (12, -0.5) {$\ZZ\tau_{A,B}$};
		\node [style=empty] (89) at (22, -12) {$\iota_{A\otimes B}$};
		\node [style=empty] (90) at (13, -18.75) {$\iota_{A\otimes B}$};
		\node [style=none] (91) at (1, 4) {\textcolor{blue}{(1)}};
		\node [style=none] (92) at (12.75, 3) {\textcolor{blue}{(2)}};
		\node [style=none] (93) at (-7, -3) {\textcolor{blue}{(3)}};
		\node [style=none] (94) at (-2, 0) {\textcolor{blue}{(4)}};
		\node [style=none] (95) at (6, -2.25) {\textcolor{blue}{(5)}};
		\node [style=none] (96) at (13.25, -3.75) {\textcolor{blue}{(6)}};
		\node [style=none] (97) at (18.5, -5) {\textcolor{blue}{(7)}};
		\node [style=none] (98) at (1, -7) {\textcolor{blue}{(8)}};
		\node [style=none] (99) at (11, -10) {\textcolor{blue}{(9)}};
		\node [style=none] (100) at (-6.5, -12) {\textcolor{blue}{(10)}};
		\node [style=none] (101) at (1.25, -12.75) {\textcolor{blue}{(11)}};
		\node [style=none] (102) at (6.25, -16.5) {\textcolor{blue}{(12)}};
		\node [style=empty] (103) at (2, -5.25) {$\TT A\otimes\TT B$};
		\node [style=empty] (104) at (-2, -3.5) {};
		\node [style=empty] (105) at (0.75, -4.75) {};
		\node [style=empty] (106) at (0.75, 1.5) {};
		\node [style=empty] (107) at (1.75, -4.75) {};
		\node [style=empty] (108) at (0.75, -1.25) {$\iota$};
	\end{pgfonlayer}
	\begin{pgfonlayer}{edgelayer}
		\draw [style=to] (0) to (7);
		\draw [style=to] (7) to (1);
		\draw [style=to] (12) to (13);
		\draw [style=to] (2) to (6);
		\draw [style=to] (10) to (11);
		\draw [style=to] (8) to (9);
		\draw [style=to] (22) to (23);
		\draw [style=to] (24) to (25);
		\draw [style=to] (26) to (27);
		\draw [style=to] (30) to (31);
		\draw [style=to] (32) to (33);
		\draw [style=to] (35) to (34);
		\draw [style=to] (37) to (38);
		\draw [style=to] (39) to (40);
		\draw [style=to] (42) to (43);
		\draw [style=to] (41) to (28);
		\draw [style=to] (45) to (46);
		\draw [style=to] (47) to (48);
		\draw [style=to] (50) to (51);
		\draw [style=to] (52) to (53);
		\draw [style=to] (55) to (56);
		\draw [style=to] (54) to (49);
		\draw [style=to] (59) to (60);
		\draw [style=to] (61) to (62);
		\draw [style=to] (63) to (64);
		\draw [style=to] (65) to (66);
		\draw [style=to] (58) to (41);
		\draw [style=mono] (6) to (3);
		\draw [style=mono] (14) to (15);
		\draw [style=to] (104) to (105);
		\draw [style=to] (106) to (107);
	\end{pgfonlayer}
\end{tikzpicture}

%% file: figures/main-theorem-2-to-3.tikz
\begin{tikzpicture}
	\begin{pgfonlayer}{nodelayer}
		\node [style=empty] (0) at (-4, 2) {$\CC$};
		\node [style=empty] (1) at (1, 2) {$\CC_\TT$};
		\node [style=empty] (2) at (-4, -2) {$\CC_\ZZ$};
		\node [style=empty] (3) at (1, -2) {$Z(\CC_\TT)$};
		\node [style=empty] (4) at (-4, 1.5) {};
		\node [style=empty] (5) at (-4, -1.5) {};
		\node [style=empty] (6) at (1, 1.5) {};
		\node [style=empty] (7) at (1, -1.5) {};
		\node [style=empty] (8) at (-3.5, 1.5) {};
		\node [style=empty] (9) at (0.5, -1.5) {};
		\node [style=empty] (10) at (-1.75, -1.75) {$\cong$};
		\node [style=empty] (11) at (-1.75, -2.5) {$\hat \II$};
		\node [style=empty] (12) at (-4.75, 0) {$\JJ^\ZZ$};
		\node [style=empty] (13) at (1.75, 0) {};
		\node [style=empty] (14) at (-1.5, 2.5) {$\JJ$};
		\node [style=empty] (15) at (-1.025, 0.375) {$\hat \JJ$};
	\end{pgfonlayer}
	\begin{pgfonlayer}{edgelayer}
		\draw [style=to] (4) to (5);
		\draw [style=to] (0) to (1);
		\draw [style=to] (2) to (3);
		\draw [style=to] (8) to (9);
		\draw [style=hook] (7) to (6);
	\end{pgfonlayer}
\end{tikzpicture}

%% file: figures/adjoint-diagram-unique-centre.tikz
\begin{tikzpicture}
	\begin{pgfonlayer}{nodelayer}
		\node [style=empty] (0) at (-6.5, 2) {$\hat\JJ Y$};
		\node [style=empty] (1) at (-6.5, -2) {$\hat\JJ \mathcal RX$};
		\node [style=empty] (2) at (-2.5, -2) {$X$};
		\node [style=empty] (3) at (-6, 1.5) {};
		\node [style=empty] (4) at (-3, -1.5) {};
		\node [style=empty] (5) at (-6.5, 1.5) {};
		\node [style=empty] (6) at (-6.5, -1.5) {};
		\node [style=empty] (7) at (-8, 0) {$\hat\JJ\Phi(f)$};
		\node [style=empty] (8) at (-4, 0.75) {$f$};
		\node [style=empty] (9) at (-4.5, -2.5) {$\varepsilon_X$};
		\node [style=empty] (10) at (6.25, 2) {$Y$};
		\node [style=empty] (11) at (6.25, -2) {$\mathcal RX$};
		\node [style=empty] (12) at (10.25, -2) {$\TT X$};
		\node [style=empty] (13) at (6.75, 1.5) {};
		\node [style=empty] (14) at (9.75, -1.5) {};
		\node [style=empty] (15) at (6.25, 1.5) {};
		\node [style=empty] (16) at (6.25, -1.5) {};
		\node [style=empty] (17) at (5.25, 0) {$\Phi(f)$};
		\node [style=empty] (18) at (8.75, 0.75) {$f$};
		\node [style=empty] (19) at (8.25, -2.5) {$\varepsilon_X$};
	\end{pgfonlayer}
	\begin{pgfonlayer}{edgelayer}
		\draw [style=to] (3) to (4);
		\draw [style=to] (1) to (2);
		\draw [style=dashed-arrow] (5) to (6);
		\draw [style=to] (13) to (14);
		\draw [style=to] (11) to (12);
		\draw [style=dashed-arrow] (15) to (16);
	\end{pgfonlayer}
\end{tikzpicture}

%% file: quantum-simple.tex
\begin{abstract}
	Quantum control is a recent notion in the literature, and many of its facets
	are still poorly understood, especially in terms of programming languages.
	Our goal is to build up solid foundations for the study of quantum control,
	syntactically and semantically. We provide syntax and semantics for a
	simply-typed calculus based on pattern-matching, developed to represent
	quantum reversible operations, and quantum control is ensured with the help
	of a quantum algebraic effect. To enforce reversibility, a syntactic notion
	of orthonormal basis is introduced, called here \emph{orthogonal
	decomposition}. A denotational semantics and an equational theory are
	developed, and we prove that the former is complete with regard to the
	latter.
	
	\paragraph{References.} This work has been the focus of many conversations
	between Kostia Chardonnet, Robin Kaarsgaard, Benoît Valiron and the author.
	It has then been enriched for a paper coauthored with Kinnari Dave, Romain
	Péchoux and Vladimir Zamdzhiev, where this language aimed at quantum control
	is intergrated in a larger language, that also handles classical control.
	That paper is under submission.
\end{abstract}

\section{Introduction}
\label{sec:quantum-simple-intro}

Quantum superposition is an important computational resource that is utilised
by many quantum algorithms and protocols. Therefore, when designing quantum
programming languages and quantum type systems, it is natural to consider how
to introduce quantum superposition into the languages and systems under
consideration. One approach, that we investigate in this chapter, aims to
introduce quantum superposition as a principal feature of the language. That
is, we provide language features that allow us to form superpositions of terms
(\emph{i.e.}~programs) in the language. In particular, instead of adopting a
gate-level view of quantum computation, wherein the computational process is
described through a series of low-level atomic gates, we take an approach that
allows us to abstract away from such details and that allows us to focus on the
linear-algebraic structure of the computation.

Quantum computing embodies several types of operations. The principal
operations are the unitary maps, which are reversible transformations of
quantum states -- \emph{e.g.}, when working with quantum circuits,
\emph{quantum gates} are reversible. State preparation is another kind of
operation, that initialises one or several qubits. Finally, measurement
\emph{breaks} quantum superposition to obtain a classical result with a certain
probability. For example, when measuring a quantum bit $\alpha \ket 0 + \beta
\ket 1$, the output can be $0$ with probability $\abs\alpha ^2$ and $1$ with
probability $\abs\beta ^2$. It is known that, in the design of a quantum
program, measurement can be deferred to the end of the execution
\cite{cross2012topics, nielsenchuang}. Using this principle of deferred
measurement, each program is then divided into two separate parts: the first is
entirely reversible, followed by a measurement at the end. It is then sensible
to focus on \emph{reversible} programming to design a quantum programming
language.

The idea of reversible computation comes from Landauer and
Bennett~\cite{Landauer61, bennett1973logical} with the analysis of its
expressivity, and the relationship between irreversible computing and
dissipation of energy. This leads to an interest in reversible
computation~\cite{bennett2000notes, aman2020foundations}, both with a low-level
approach~\cite{caroe2012design, wille2016syrec, saeedi2013synthesis}, and from
a high-level perspective~\cite{lutz1986janus, yokoyama2007reversible,
yokoyama2016fundamentals, james2012information, james2014theseus,
sabry2018symmetric, yokoyama2011reversible, thomsen2015interpretation,
JacobsenKT18}.

Reversible programming lies on the latter side of the spectrum, and two main
approaches have been followed. Embodied by Janus~\cite{lutz1986janus,
yokoyama2007reversible, yokoyama2010reversible, yokoyama2016fundamentals} and
later R-CORE and R-WHILE~\cite{gluck2019reversible}, the first one focuses on
imperative languages whose control flow is inherently reversible -- the main
issue with this aspect being tests and loops. The other approach is concerned
with the design of functional languages with structured data and related
case-analysis, or \emph{pattern-matching}~\cite{yokoyama2011reversible,
thomsen2015interpretation, james2014theseus, sabry2018symmetric, JacobsenKT18}.
To ensure reversibility, strong constraints have to be established on the
pattern-matching in order to maintain reversibility.

Those developments were utilised to introduce a programming language aimed at
reversible quantum programming \cite{sabry2018symmetric}, which is the work
this chapter builds upon. The goal of that paper is more specific: it aims to
formalise a programming language that performs \emph{quantum control}.
Quantum control, as opposed to classical control, is the ability to realise
the control schemes -- such as \emph{if} statements or \emph{while} loops --
with quantum data, which means, a superposition of states. Informally, one can
say that quantum control is allowing not only superposition of states, but also
superposition of programs. A practical example of quantum control is the
\emph{quantum switch} \cite{chiribella2013quantum}, which works as follows:
given two quantum states $x$ and $y$, and two unitary operations $U,V$, the
quantum switch performs the operation $UV$ on $\ket y$ when $x$ is in the state
$\ket 0$, and $VU$ when $x$ is in the state $\ket 1$. In general, the obtained
operation sends the state $(\alpha \ket 0 + \beta \ket 1) \otimes \ket y$ to
$\alpha \ket 0 \otimes (UV \ket y) + \beta \ket 1 \otimes (VU \ket y)$. Besides
being mathematically feasible, the quantum switch is also doable in a lab
\cite{procopio2014experimental, rubino2017experimental}.

The literature on quantum control is fairly recent, because it was thought to
be not feasible in a realistic quantum computer. Since the introduction of the
quantum switch, quantum control is starting to be studied with a programming
language point of view \cite{sabry2018symmetric, qunity, pablo2022unbounded,
diazcaro2023feasible}.  In some cases, the presentation lacks a proper
denotational study; in some others, it lacks an operational account. One can
summarise by saying that quantum control in programming languages lacks solid
foundations, as much syntactically as mathematically. The work in this chapter
is an early attempt to resolve this issue.

\subsection{Related work} 
\label{sub:qua-related-work}

\paragraph{Reversible computation.} Two successive papers \cite{nous2021invcat,
nous2023invrec} -- that are also the focus of the next chapter -- provide a
categorical semantics of a reversible programming language based on Theseus.
That language is closely related to the one presented in our work, however it
only handles reversibility on classical data. More generally,
\cite{kaarsgaard2017join} details the structure of inverse categories used to
interpret reversible programming.  However, a category interpreting a quantum
programming language also has to take into account quantum states, usually
represented by isometries; this is why it makes sense to work with contractive
maps as morphisms. The category with Hilbert spaces as objects and contraction
as morphisms is not an inverse category. This shows that the work on classical
reversible programming languages cannot be applied to our goal.  Furthermore,
reversible quantum programming is not a monadic effect over reversible
programming: the sensible way of going from an inverse category to Hilbert
spaces is the functor $\ell^2$ \cite{heunen2013l2}, which is not an adjoint
functor, and thus cannot provide a monad.

\paragraph{Quantum control.} 
Our work is based on the paper of Sabry, Valiron and Vizzotto
\cite{sabry2018symmetric} where a functional reversible programming language is
introduced and extended to a quantum programming language handling quantum
control with recursive functions over lists. However, the denotational
semantics given in that paper is not compositional, nor it is proven sound, nor
adequate with regard to the operational semantics. In general, that paper lays
out great ideas on how to work with quantum control, but with few proven
statements. This chapter aims to provide stronger foundations, syntactical and
mathematical. This will hopefully help tackle the question of the denotational
semantics of quantum structural recursion, as introduced in the paper cited.

PUNQ \cite{diazcaro2023feasible} is a programming language which is close to
ours in some aspects: it relies on a notion of orthogonality between terms to
form linear combinations, and its goal is to work with unitaries, to ensure
quantum control. There are some differences: its design is based on linear
logic and is closer to a linear $\lambda$-calculus; the base type is the one of
\emph{bits} -- which is quickly generalised to \emph{qubits} --, even the
authors still work with a specific basis, while one would expect that working
with qubits directly solves this issue (see discussion in
\secref{sec:qu-simple-conclusion}). Finally, to ensure normalisation of
well-typed terms, the authors introduce an orthogonality predicate, akin to one
in \cite{sabry2018symmetric} and in this chapter. However, the former requires
that the terms are computed to check whether they are orthogonal. This means
that the type checking is not static, and therefore not necessarily efficient.

Similarly, there are many approaches to quantum control, but that do not ensure
unitarity of operations, usually related to the $\lambda$-calculus
\cite{diazcaro2022semimodules, diazcaro2022unitary, diazcaro2019unitary,
altenkirch2005functional, chardonnet2022manyworlds}. These developments set
themselves in the long list of papers revolving around algebraic
$\lambda$-calculi \cite{arrighi2017vectorial, arrighi2017lineal,
vaux2009algebraic, selinger2009quantum}. Alejandro D\'iaz-Caro has written a
nice paper on that topic \cite{diaz2022quick}. Note that these approaches
struggle to scale to infinite dimensions. 

Other programming languages handling quantum control have been introduced, such
as Qunity \cite{qunity}. Qunity is based on reversible pattern-matching
\cite{sabry2018symmetric}, like this chapter. However, Qunity assumes built-in
unitary operations, and has then $\lambda$-abstractions that are not proven to
be unitary operations. This last point is tackled with an \emph{error handling}
scheme, which does not appear to be suitable, and there is no semantics to show
how this scheme would behave operationally.

\subsection{Contribution} 
\label{sub:qua-contribution}

We provide a reversible programming language with simple types, inspired by the
direct sum and tensor product of vector spaces, and natural numbers, to show
that it is possible to work with infinite data types. This language relies on
an orthogonality notion between values to define \emph{linear combinations} of
values to represent quantum superposition. A notion of orthonormal basis in the
syntax is introduced, which helps prove that the abstractions are unitary
operations.

Regarding the syntax, the presentation has been improved, some lemmas were
fixed compared to previous presentations \cite{sabry2018symmetric, phd-kostia},
and many more lemmas have been proven concerning the orthogonal decomposition.

This chapter is organised as follows: the first section (see
\secref{sec:simple-language}) outlines the syntax of the language (see
Figure~\ref{fig:qu-syntax}, the grammar of terms, their typing rules, utilising
orthogonality (see Definition~\ref{def:orthogonality}) and orthogonal
decomposition (see Definition~\ref{def:od-ext}). Then, substitutions are
introduced (\secref{sub:qua-substitution}), described in a way that fits the
language, allowing us to write our equivalent of $\beta$-reduction in a
comprehensible manner. We then introduce an equational theory
(\secref{sub:qua-eq-typing}), in the vein of the equational theory of the
simply-typed $\lambda$-calculus, and we prove that it verifies a normalisation
property.  A later section (\secref{sec:qu-simple-maths}) is dedicated to
introducing the mathematical notions and definition required to establish the
denotational semantics, given in \secref{sec:very-simple-semantics}.  Finally,
we prove completeness of the denotational semantics with regard to this
equational theory (\secref{sub:qua-completeness}).

\subsection{Work of the author} 
\label{sub:qua-author}

Within this chapter, the author has contributed to the following points.
\begin{itemize}
	\item A new kind of linear combinations of terms, which is more fit to ensure
		normalisation (see Figure~\ref{fig:qu-syntax}). This required the further
		development of the type system with different conditions than are already
		present in the literature.
	\item A new definition of orthogonality (see
		Definition~\ref{def:orthogonality}), close to the ones in
		\cite{sabry2018symmetric, phd-kostia, nous2021invcat, nous2023invrec,
		chardonnet2023curry}, but it fits a larger collection of terms of the
		language, including linear combinations and function application.
	\item Formalised \emph{well-formed} substitutions (see
		Definition~\ref{def:well-valuation}), to help the denotational semantics.
	\item An equational theory for the language (see \secref{sub:qua-eq-typing}),
		establishing stronger foundations to quantum control than in
		\cite{sabry2018symmetric} where there is a mix of equations and operational
		semantics.
	\item A compositional denotational semantics (see
		\secref{sec:very-simple-semantics}), which is sound and complete with
		regard to the equational theory.
\end{itemize}

\section{The Language}
\label{sec:simple-language}

In this section, we present the syntax of the programming language studied in
this chapter dedicated to simply-typed quantum control. The quantum aspect of
the language is provided as an algebraic effect with the introduction of linear
combinations with the combinator $\Sigma$.

\subsection{Syntax of the Language}
\label{sub:simple-syntax}

\begin{figure}[h]
	\begin{alignat*}{10}
		&\text{(Ground types)}\quad & A, B &~~&&::= ~&\quad& \one \alt A \oplus B \alt A
		\otimes B 
		\alt \nat
		\\
		&\text{(Unitary types)} & T &&&::=&& A \iso B \\[1ex]
		&\text{(Basis Values)} & b &&&::=&& * \alt x \alt \inl{b} \alt \inr{b} \alt
		b \otimes b \alt \zero \alt \suc b \\
		&\text{(Values)} & v &&& ::= && \Sigma_{i \in I} (\alpha_i \cdot b_i) \\
		&\text{(Expressions)} & e &&& ::= && * \alt x \alt \inl{e} \alt \inr{e} \alt
		\pair{e}{e} \alt \zero \alt \suc e \\
		&&&&&&&   \alt \Sigma_{i \in I} (\alpha_i \cdot e_i)
		\\
		&\text{(Unitaries)} & \isoterm &&&::=&& \unibasique \\ 
		&&&&&&& \alt \isoterm \otimes \isoterm \alt \isoterm \oplus \isoterm \alt
		\isoterm \circ \isoterm \alt \isoterm\inv \alt \ctrl \isoterm \\
		&\text{(Terms)} & t &&& ::= && * \alt x \alt \inl{t} \alt \inr{t} \alt
		\pair{t}{t} \alt \zero \alt \suc t \\
		&&&&&&& \alt \omega~t \alt \Sigma_{i \in I}
		(\alpha_i \cdot t_i)
	\end{alignat*}
	\caption{Syntax of simply-typed quantum control.}
	\label{fig:qu-syntax}
\end{figure}

The syntax of the programming language studied in this chapter is described in
a usual way, with grammars, such as the ones in Chapter~\ref{ch:background}
(see the grammars for the simply-typed $\lambda$-calculus in
(\ref{eq:simple-types}) and (\ref{eq:simple-terms})). It is given in
Figure~\ref{fig:qu-syntax}. 

\paragraph{Types.}
The ground types are given by a unit type $\one$ and
the usual connectives $\oplus$ and $\otimes$, which are respectively called
\emph{direct sum} and \emph{tensor product}. We also have the inductive type
$\nat$, as a witness that it is possible to work with infinite data types,
and thus infinite-dimensional spaces in the model. We equip functions, called
\emph{unitaries} in this chapter, with a separate type, written $A \iso B$ when
$A$ and $B$ are two ground types. This double arrow notation is inherited from
\cite{sabry2018symmetric}, as a way of picturing that the operations are indeed
reversible.

\paragraph{Terms.}
The terms of the language are given as follows:
\begin{itemize}
	\item variables $x,y,z,\dots$, given as elements of a set of variables
		$\mathtt{Var}$, assumed to be totally ordered;
	\item a term $*$ called the \emph{unit} corresponding to the unit type $\one$;
	\item usual connectives for the direct sum, $\inl\!$ and $\inr\!$ which are
		respectively called the \emph{left injector} and \emph{right injector};
	\item a connective corresponding to the tensor product, that is also
		written $\otimes$;
	\item terms for natural numbers, $\zero$ and the connective $\suc\!$ that
		gives the successor of a term;
	\item the application of a unitary to a term, written $\omega~t$ when
		$\omega$ is a unitary and $t$ a term;
	\item finally, given a finite set of indices $I$ -- which could be a finite
		set of numbers $\{0,1,2,\dots,n\}$ --, assumed to be totally ordered, a family of
		complex numbers $(\alpha_i)_{i \in I}$ and a family of terms $(t_i)_{i \in
		I}$, one can form the term $\Sigma(I, (\alpha_i)_i, (t_i)_i)$, representing
		the linear combination of the terms with complex scalars. This last term
		construction embodies the quantum effect of the language. In the rest of
		the chapter, we write $\Sigma_{i \in I} (\alpha_i \cdot t_i)$ for
		$\Sigma(I, (\alpha_i)_i, (t_i)_i)$ to make it more readable.
\end{itemize}
In quantum theory, a linear combination of vectors $\sum_{i \in I} \alpha_i
\ket{i}$ is normalised if $\sum_{i \in I} \abs{\alpha_i}^2 = 1$. The family
of real numbers $(\abs{\alpha_i}^2)_{i \in I}$ is then seen as a probability
distribution. This work does not focus on the probabilistic aspect of quantum
theory; however, we want to work with well-formed states, and thus normalised
states. This is why we ensure later that a linear combination of terms is
normalised. Throughout the chapter, a term $\Sigma_{i \in \set{1,2}} (\alpha_i
\cdot t_i)$ might be written $\alpha_1 \cdot t_1 + \alpha_2 \cdot t_2$,
regarded as syntactic sugar. In some examples, a term $\Sigma_{i \in \set *} 1
\cdot t$ can be written $1 \cdot t$ or even $t$ for readability; but note that
$\Sigma_{i \in \set *} 1 \cdot t$ and $t$ are different terms in the syntax.

\begin{remark}
	Since we are working with a programming language which handles complex
	numbers, we might want to ensure that the complex numbers are
	\emph{computable} \cite{turing1937computable}; this is not a terrible
	assumption, since the set of computable complex numbers keeps the structure
	of a field \cite{rice1954recursive}. However, we do not wish to focus on this
	point; we then assume that we work with complex numbers achievable in a given
	quantum hardware, and that those still form a field.
\end{remark}

The terms that are unitary-free -- in the sense that they do not contain any
function application -- and that do not involve linear combinations are called
\emph{basis values}. These terms are naturally classical, as opposed to
quantum. Their name comes from the fact that we use them as a syntactic
representation of the canonical basis of a Hilbert space, as introduced in
\secref{sub:hilb-computing}. Note that basis values are totally ordered.

Values, on the other hand, are linear combinations of basis values. In a value
$\sum_{i \in I} (\alpha_i \cdot b_i)$, we assume that the family $(b_i)_{i \in
I}$ is an increasing sequence, and that none of the scalars are equal to $0$;
this allows us to work with unique normal forms later in the chapter. Since
this definition is restrictive, we need to introduce a more general piece of
syntax, which still does not include unitary application, called an
\emph{expression}. They are used as outputs of functions, that we introduce
below.

\begin{example}[Qubits]
	\label{ex:qua-qubits}
	In our presentation, the type of \emph{qubits} is $\one \oplus \one$. The term
	$\inl *$ represents the quantum state $\ket 0$ and $\inr *$, the other
	element of the canonical basis $\ket 1$. The general state of a qubit is
	given by the term $\alpha \cdot (\inl *) + \beta \cdot (\inr *)$.This type $\one
	\oplus \one$ can also be seen as the type of quantum booleans. The basis value
	$\inl *$ represents \emph{false} and $\inr *$, \emph{true}. 
\end{example}

\paragraph{Unitaries.}
Unitaries are firstly obtained by what we call \emph{unitary abstractions},
written as a set of clauses $\unibasique$ or $\unibreduit$, given a family of
basis values and a family of expressions, both indexed by a set $I$. We will
see that several conditions have to be verified to ensure that this unitary
actually performs a unitary operation (in other words, a reversible operation
between normalised states). The grammar for unitaries also contains operations
such as the direct sum $\oplus$, the tensor product $\otimes$, the inverse
$(-)\inv$ and the qubit control $\ctrl\!$. These operations can be seen as
syntactic sugar, since they can all be performed within unitary abstractions.
Their presence in the grammar is an instance of operations at the level of
unitaries, which are at a \emph{higher order} than the operations at the ground
level. We will see in the next chapter that, provided some conditions, any
calculus can be put on top of unitaries.

Unitary abstractions can be seen as a mix of $\lambda$-abstractions and
pattern-matching; the latter is sometimes written $\mathtt{case}$ or
$\mathtt{match}$ in a functional programming language. We give some examples to
make this intuition clearer.

\begin{equation}
	\label{eq:untyped-uni}
	\set{\mid x \iso x} 
	\qquad 
	\set{\mid x \otimes y \iso y \otimes x}
	\qquad
	\sset{\begin{array}{lcl}
		\mid \inl x & \iso & \inr x \\
		\mid \inr y & \iso & \inl y
	\end{array}
	}
\end{equation}
The example on the left performs the identity; the one in the middle swaps the
two elements of a tensor product; and the last one swaps the two parts of a
direct sum. Part of the conditions for a unitary abstraction to be well-typed,
is to have the same variables on each side of a clause. This is verified by the
examples above. We will see that all three examples are well-typed, with the
typing rules described in the next section. Note that, in those examples, the
terms on the right-hand side of the abstraction should be values, in the form
of a sum. We have simplified notations here, because the sum would involve only
one element.

A simple example that involves quantum superposition is the following:

\[
	\sset{\begin{array}{lcl}
		\mid \inl * & \iso & 
		\frac{1}{\sqrt 2} \cdot (\inl *) + \frac{1}{\sqrt 2} \cdot (\inr *) \\ 
		\mid \inr * & \iso & 
		\frac{1}{\sqrt 2} \cdot (\inl *) - \frac{1}{\sqrt 2} \cdot (\inr *) 
	\end{array}
	}
\]
which operates the well-known \emph{Hadamard} operator on a qubit, given below. 
\begin{equation}
	\label{eq:matrix-had}
	\left(
	\begin{array}{rr}
		\frac{1}{\sqrt 2} & \frac{1}{\sqrt 2} \\
		\frac{1}{\sqrt 2} & - \frac{1}{\sqrt 2}
	\end{array}
	\right)
\end{equation}
This operation is significant in quantum computing: it is the one that allows
one to introduce quantum superposition in a system, as well as entanglement,
when followed by a controlled \emph{not} operator. The Hadamard operator
performs a \emph{change of basis}. It is an important notion in linear algebra,
and is an equally important notion in this chapter, where unitary abstractions
are precisely defined as a change between two bases; the latter are introduced
syntactically, and are more general than linear algebraic bases.


\subsection{Types and Typing Rules}

As usual, a \emph{typing context} consists of a set of pairs of a variable and
a type, written $\Delta$ and generated by $\Delta ::= \emptyset \alt \set{x
\colon A} \cup \Delta$. A comma between two contexts represents the union of
the contexts, \emph{i.e.}~$\Delta, \Delta' = \Delta \cup \Delta'$. The
variables in a context $\Delta$ need to be all different to one another. We
have two levels of judgements: the one for terms, where sequents are written
$\Delta \vdash t \colon A$ and a typing judgement for isos, noted $\entailiso
\omega \colon A \iso B$.  Variables in a context $\Delta$ are strictly linear:
given $\Delta \entail t \colon A$, every element of $\Delta$ has to occur
\emph{exactly once} in the term $t$.

\paragraph{Orthogonality.}
We have seen in \secref{sub:hilb-computing} that a quantum state has to be
normalised. The first focus of this section to ensure normalisation through the
type system. Let us recall that, in quantum computing, there are two necessary
conditions for a superposition of states $\alpha \ket \varphi + \beta \ket
\psi$ to be normalised. The first one is concerned with the probability
distribution condition, $\abs\alpha^2 + \abs\beta^2 = 1$. Secondly, the vectors
$\ket \varphi$ and $\ket \psi$ need to be orthogonal. Indeed, the following
vector:
\[
	\frac{1}{\sqrt 2} \ket 0 + \frac{1}{\sqrt 2} \ket 0 = \sqrt 2 \ket 0
\]
is not normalised. Therefore, its corresponding term 
\[
	\frac{1}{\sqrt 2} \cdot (\inl *) + \frac{1}{\sqrt 2} \cdot (\inl *)
\]
should not be accepted by our type system. To do so, we introduce a notion of
orthogonality for terms in our syntax (see Definition~\ref{def:orthogonality}),
in order to express later a typing rule ensuring normalisation.

Our orthogonality predicate is primarily based on direct sums and injections.
Given two Hilbert spaces $H_1$ and $H_2$ and vectors $x_1 \in H_1$ and $x_2 \in
H_2$, their respective injections into $H_1 \oplus H_2$, namely $(x_1,0)$ and
$(0,x_2)$, are orthogonal. Therefore, given two terms $t_1$ and $t_2$, we fix
that their respective projections $\inl t_1$ and $\inr t_2$ are orthogonal.

Our orthogonality predicate is then generalised on all terms through congruence
rules. The final rule presented in Definition~\ref{def:orthogonality} outlines
the orthogonality obtained with a \emph{change of basis}. For example, in the
Hilbert space $\C^2$, the two vectors:
\[
	\ket + = \frac{1}{\sqrt 2} \ket 0 + \frac{1}{\sqrt 2} \ket 1
	\quad
	\text{and}
	\quad
	\ket - = \frac{1}{\sqrt 2} \ket 0 - \frac{1}{\sqrt 2} \ket 1
\]
are orthogonal.

In the next definition, we introduce the predicate of orthogonality, written
$\bot$. In particular, it means that given $t_1 ~\bot~t_2$, a normalised linear
combination of the two terms can be formed.

\begin{definition}[Orthogonality]
	\label{def:orthogonality}
	We introduce a symmetric binary relation $\bot$ on terms. Given two terms
	$t_1, t_2$, we have $t_1~\bot~t_2$ if it can be derived inductively with the
	rules below; when $t_1~\bot~t_2$ can be derived, we say that $t_1$ and $t_2$
	are orthogonal. The relation $\bot$ is defined as the smallest symmetric
	relation such that:
  \[
    \begin{array}{c}
      \infer{\inl{t_1}~\bot~\inr{t_2}}{}
			\quad
			\infer{\zero~\bot~\suc t}{}
			\quad
			\infer{\suc t_1~\bot~\suc t_2}{t_1~\bot~t_2}
      \quad
      \infer{\pair{t}{t_1}~\bot~\pair{t'}{t_2}}{t_1~\bot~t_2}
      \quad
      \infer{\pair{t_1}{t}~\bot~\pair{t_2}{t'}}{t_1~\bot~t_2}
      \\[1.5ex]
      \infer{\inl{t_1}~\bot~\inl{t_2}}{t_1~\bot~t_2}
      \quad
      \infer{\inr{t_1}~\bot~\inr{t_2}}{t_1~\bot~t_2}
			\quad
			\infer{\omega~t_1~\bot~\omega~t_2}{t_1~\bot~t_2}
 			\quad
 			\infer[(\star)]{
				t~\bot~\Sigma_{i \in I} (\alpha_i \cdot t_i)
 			}{
 				\forall i \in I, t~\bot~t_i
 			}
      \\[1.5ex]
 			\infer{
				t~\bot~\Sigma_{i \in I \cup \set *} (\alpha_i \cdot t_i)
 			}{
 				\forall i \in I, t~\bot~t_i
				&
				\alpha_* = 0
				&
				t_* = t
 			}
			\quad
			\infer[(\star)]{\Sigma_{j \in J} (\alpha_j \cdot t_j)~\bot~\Sigma_{k \in K}
			(\beta_k \cdot t_k)}
				{\forall i\neq j \in I, t_i~\bot~t_j
				&
				J,K \subseteq I
				&
				\sum_{i \in J \cap K} \bar{\alpha_i}\beta_i = 0
				}
    \end{array}
  \]
\end{definition}

\begin{remark}
	Among the two inference rules marked by a $\star$ above, it could seem like the last one
	implies the first one. However, we remind that $\Sigma$ is a contructor, and therefore
	$t$ and $\Sigma_{j \in J} (\alpha_j \cdot t_j)$ are distinct cases of the grammar.
\end{remark}

\begin{remark}[Orthogonality of variables]
	Given two variables $x$ and $y$, the terms $x$ and $y$ are not orthogonal.
	The main reason is that they could be instantiated with the same value. On
	the other hand, $\inl x$ and $\inr y$ are orthogonal, for example.
\end{remark}

\begin{remark}
	We recall the third unitary presented in (\ref{eq:untyped-uni}), which swaps
	$\inl{}$ and $\inr{}$:
	\[
		\sset{\begin{array}{lcl}
			\mid \inl x & \iso & \inr x \\
			\mid \inr y & \iso & \inl y
		\end{array}
		}
	\]
	and we call this unitary $\omega$. The term $\omega~(\inl *)$ represents the
	application of the unitary $\omega$ to $\inl *$. We later show within our
	equational theory (see \secref{sec:qu-simple-eq-theory}), that the terms
	$\omega~(\inl *)$ and $\inr *$ are equal, as expected. However, we cannot
	derive that $\omega~(\inl *)$ and $\inl *$ are orthogonal with the rules of
	orthogonality given above (see Definition~\ref{def:orthogonality}), even if
	$\inl *$ and $\inr *$ are orthogonal. This is because we wish to be able to
	derive orthogonality \emph{statically}.
\end{remark}

\begin{lemma}
	If $t~\bot~t'$, then $t$ and $t'$ are different terms.
\end{lemma}

Note that orthogonality holds without any typing rules or notion of type.
Figure~\ref{fig:typing-values-simple} introduces the typing rules for
expressions, and therefore for basis values and values, thanks to the notion of
orthogonality.

\begin{figure}[!h]
	\[\begin{array}{c}
		\infer{
			\emptyset \entail * \colon \one,
		}
		{}
		\qquad
		\infer{
			x \colon A \entail x \colon A,
		}{}
		\qquad
		\infer{
			\Delta_1,\Delta_2\entail \pair{e_1}{e_2} \colon A\otimes B,
		}{
			\Delta_1\entail e_1 \colon A
			&
			\Delta_2\entail e_2 \colon B
		}
		\\[1.5ex]
		\infer{
			\Delta\entail \inl{e} \colon A\oplus B,
		}{
			\Delta\entail e \colon A
		}
		\qquad
		\infer{
			\Delta\entail \inr{e} \colon A\oplus B,
		}{
			\Delta\entail e \colon B
		}
		\\[1.5ex]
		\infer{\entail \zero \colon \nat,}{}
		\qquad
		\infer{
			\Delta \entail \suc e \colon \nat,
		}{
			\Delta \entail e \colon \nat
		}
		\\[1.5ex]
		\infer{
			\Delta \entail \Sigma_i (\alpha_i \cdot e_i) \colon A.
		}{
			\Delta \entail e_i \colon A
			&
			\Sigma_i \abs{\alpha_i}^2 = 1
			&
			\forall i\not= j, e_i~\bot~e_j
		}
	\end{array}
	\]

	\caption{Typing rules of (basis) values and expressions.}
	\label{fig:typing-values-simple}
\end{figure}

Once we know how basis values are formed in a certain type, we can discuss some
orthogonality properties among specific types. In linear algebra, given a
vector space and an orthogonal basis $B$ of that space, if two elements of the
basis are not orthogonal, then they are equal. Our case is more subtle,
because, for example, $\inr{} (\inl *)$ and $\inr x$ are not equal, but also
not orthogonal. We will see, later in this chapter, that they are linked by
substitution.

\begin{example}
	The values $\inl *$ and $\inr *$ are orthogonal, but of course, they are
	both not orthogonal to $\frac{1}{\sqrt 2} \cdot (\inl *) - \frac{1}{\sqrt 2}
	\cdot (\inr *)$.
\end{example}

\paragraph{Orthogonal decomposition.}
This motivates a syntactic definition for a basis, containing expressions of
the language. Given a set of orthogonal expressions, we want to ensure that
this set spans the whole type, so that it can be seen as an orthonormal basis.
This is given by the notion of \emph{orthogonal decomposition}. We first
introduce this notion in Definition~\ref{def:OD} with basis values only, and we
then extend it to expressions in general in Definition~\ref{def:od-ext}. 

Before outlying the definitions of orthogonal decompositions, we introduce some
notations. Given a set $S = \set{\pair{e_1}{e_1'}, \dots, \pair{e_n}{e_m'}}$,
we define $\pi_1(S) = \set{e_1, \dots, e_n}$ and $\pi_2(S) = \set{e_1', \dots,
e_m'}$.  Finally, we define $S_e^1$ and $S_e^2$ respectively as $\set{e' \alt
\pair{e}{e'} \in S}$ and $\set{e' \alt \pair{e'}{e} \in S}$.

\begin{definition}[Orthogonal Decomposition]
	\label{def:OD}
	We introduce a predicate $\OD{A}{(-)}$ on finite sets of basis values. Given
	a finite set of values $S$, $\OD A S$ holds if it can be derived with the
	following rules. The predicate is defined inductively as the smallest
	predicate such that:
	\[
    \begin{array}{c}
			\infer{\OD{A}{(\set{x})}}{} 
			\qquad 
			\infer{\OD{\one}{(\set{*})}}{} 
			\qquad
    	\infer{
				\OD{A\oplus B}{(\{\inl{s} \alt s \in S \} \cup \{\inr{t} \alt t\in T \})}
			}{
				\OD A S \qquad \OD B T
			} 
    	\mynl
			\infer{
				\OD{\nat}{(\set{\zero} \cup \set{\suc s \alt s \in S})}
			}{
				\OD \nat S
			}
			\qquad
    	\infer{
				\OD{A\otimes B}{S}  
			}{
				\begin{array}{c}
					\OD{B}{(\pi_2(S))} \text{ and } \forall
					b \in \pi_2(S), \OD{A}{(S^2_b)}
				\end{array}
			}
			\mynl
    	\infer{
				\OD{A\otimes B}{S}  
			}{
				\begin{array}{c}
					\OD{A}{(\pi_1(S))} \text{ and } \forall b \in \pi_1(S),
					\OD{B}{(S^1_b)} 
				\end{array}
			}
		\end{array}
	\]
\end{definition}

To simplify the notations in coming proofs, we will write $S \boxplus T$ for
the set $\{\inl{s} \alt s \in S \} \cup \{\inr{t} \alt t\in T \}$, and
$S^{\oplus 0}$ for the set $\set{\zero} \cup \set{\suc s \alt s \in S}$. We
adopt functional programming convention regarding parentheses: when readable,
we write $\OD A S$ instead of $\OD{A}{(S)}$. Also, we call ${\rm OD}$ the
predicate introduced above in general, without precision of type, to facilitate
the later discussions.

\begin{example}
	Following Example~\ref{ex:qua-qubits}, we have $\OD{\one \oplus \one}{\set{\inl *,
	\inr *}}$. The two qubits $\ket 0$ and $\ket 1$ are indeed an orthonormal
	basis for qubit states.
\end{example}

\begin{remark}
	Note that the precondition to derive $\OD{A \otimes B}{S}$ cannot be
	simplified: there are sets $S$ such that $\OD{A}{(\pi_1(S))}$ and for all $b
	\in \pi_1(S)$, $\OD{B}{(S^1_b)}$ but not all $b \in \pi_2(S)$ is such that
	$\OD{B}{(S^2_b)}$, \textit{e.g.} $S = \set{(\inl *) \otimes y, (\inr x)
	\otimes (\inl *), (\inr x) \otimes (\inr *)}$ with the type $(\one \oplus (\one
	\oplus \one)) \otimes (\one \otimes \one)$.
\end{remark}

\begin{example}
	Given any type $A$, we have $\OD{A}{\{ x \}}$, where $x$ is a variable. Since
	it is a variable, it can be substituted with any terms of type $A$, therefore
	it \emph{parses} the whole type. We make precise the notion of subsitution
	for our syntax in \secref{sub:qua-substitution}.
\end{example}

\begin{example}
	\label{ex:od-nat}
	We have $\OD{\nat}{\{\zero, \suc\zero, \suc{} \suc x \}}$. Indeed, any closed
	basis value of type $\nat$ is either $\zero$, or $\suc\zero$, or $\suc{} \suc
	b$.
\end{example}

The predicate ${\rm OD}_A$ defined above ensures that a finite set $S$ of
values represents an orthonormal basis of $A$, that we can view as the
canonical basis. Note that even for $\nat$, the set representing the basis is
finite, \textit{e.g.} $\set{\zero, \suc \zero, \suc (\suc x)}$. With knowledge
of linear algebra, one can say that the bases represented by ${\rm OD}$ are not all
the possible orthonormal bases. A change of basis through a unitary matrix also
provides an orthonormal basis. This is the purpose of the next definition.

\begin{definition}
	\label{def:od-ext}
	We extend the previous definition to general values. $\ODe{A}{(-)}$
	is a predicate on finite sets of values, and 
	$\ODe A S$ holds if it can be derived from the following rule.
	\[ \begin{array}{c}
			\infer{\ODe{A}{(\set{x})}}{} 
			\qquad 
			\infer{\ODe{\one}{(\set{*})}}{} 
			\qquad
    	\infer{
				\ODe{A\oplus B}{(\{\inl{e} \alt e \in S \} \cup \{\inr{e} \alt e\in T \})}
			}{
				\ODe A S \qquad \ODe B T
			} 
    	\mynl
			\infer{
				\ODe{\nat}{(\set{\zero} \cup \set{\suc v \alt v \in S})}
			}{
				\ODe \nat S
			}
			\qquad
    	\infer{
				\ODe{A\otimes B}{S}  
			}{
				\begin{array}{c}
					\ODe{B}{(\pi_2(S))} \text{ and } \forall
					e \in \pi_2(S), \ODe{A}{(S^2_e)}
				\end{array}
			}
			\mynl
    	\infer{
				\ODe{A\otimes B}{S}  
			}{
				\begin{array}{c}
					\ODe{A}{(\pi_1(S))} \text{ and } \forall e \in \pi_1(S),
					\ODe{B}{(S^1_e)} 
				\end{array}
			}
			\qquad
			\infer{
				\ODe{A}{(\set{\Sigma_{e'\in S} (\alpha_{e, e'} \cdot e') \alt e\in S})}
			}{
				\ODe A S
				& (\alpha_{e, e'})_{(e,e') \in S \times S}\text{ is a unitary matrix} 
			}
	\end{array} \]
	Given $\ODe A S$, we say that $S$ is an \emph{orthogonal decomposition} of type $A$.
\end{definition}

To simplify notations, we write $S^\alpha$ for the set
$\set{\Sigma_{e'\in S} (\alpha_{e, e'} \cdot e') \alt e\in S}$.

\begin{example}
	Following Example~\ref{ex:qua-qubits}, we have 
	\[\ODe{\one \oplus
	\one}{\left\{\frac{1}{\sqrt 2} \cdot \inl * + \frac{1}{\sqrt 2} \cdot \inr *, \quad
	\frac{1}{\sqrt 2} \cdot \inl * - \frac{1}{\sqrt 2} \cdot \inr *\right\}}.
	\] 
	The two qubits $\ket +$ and $\ket -$ are indeed an orthonormal basis for
	qubit states.
\end{example}

\begin{example}
	Following Example~\ref{ex:od-nat} and the previous example, we have:
	\[
		\ODe{\nat}{
			\left\{\frac{1}{\sqrt 2} \zero + \frac{1}{\sqrt 2} (\suc \zero), \quad
			\frac{1}{\sqrt 2} \zero - \frac{1}{\sqrt 2} (\suc \zero), \quad
			\suc{} \suc x
			\right\}}.
	\]
\end{example}

The predicate ${\rm OD}$ is defined without the help of orthogonality (see
Definition~\ref{def:orthogonality}), but there is a link: the elements of an
orthogonal decomposition, are in particular pairwise orthogonal.

\begin{lemma}[$\mathrm{OD}$ implies $\bot$]
	\label{lem:od-imp-bot}
	Given $\ODe A S$, for all $t_1 \neq t_2 \in S$, $t_1~\bot~t_2$.
\end{lemma}
\begin{proof}
	The proof is done by induction on ${\rm OD}$.
	\begin{itemize}
		\item $\ODe{A}{\set x}$. There is no pair of different terms in $\set x$.
		\item $\ODe{\one}{\set *}$. There is no pair of different terms in $\set *$.
		\item $\ODe{A \oplus B}{S \boxplus T}$. There are several cases: either both
			terms are of the form $\inl -$, namely $t_1 = \inl t'_1$ and $t_2 = \inl
			t'_2$, with $t'_1$ and $t'_2$ in $S$, in which case the induction
			hypothesis on $\ODe A S$ gives that $t'_1~\bot~t'_2$, and thus $\inl
			t'_1~\bot~\inl t'_2$; or both are of the $\inr -$, the case is similar
			with the induction hypothesis on $\ODe B T$; or $t_1 = \inl t'_1$ and $t_2
			= \inr t'_2$, then we have directly $\inl t'_1~\bot~\inr t'_2$.
		\item $\ODe{A \otimes B}{S}$. Both cases are similar. Suppose that
			$\ODe{A}{\pi_1(S)}$.  We know that $t_1 = t'_1 \otimes t''_1$ and $t'_2
			\otimes t''_2$. The induction hypothesis on $\ODe{A}{\pi_1(S)}$ gives that
			$t'_1~\bot~t'_2$ and thus $t'_1 \otimes t''_1~\bot~t'_2 \otimes t''_2$.
		\item $\ODe{\nat}{S^{\oplus 0}}$. If one of $t_1$ or $t_2$ is $\zero$, the
			conclusion is direct; else, $t_1 = \suc t'_1$ and $t_2 = \suc t'_2$ and
			the induction hypothesis on $\ODe \nat S$ gives that $t'_1~\bot~t'_2$ and
			thus $\suc t'_1~\bot~\suc t'_2$.
		\item $\ODe{A}{S^\alpha}$. By induction hypothesis, all the terms in $S$
			are pairwise orthogonal; and the matrix $\alpha$ is unitary, which means
			that the inner product of its columns is zero when the columns are
			different, which ensures that two different terms in $S^\alpha$ are
			orthogonal.
	\end{itemize}
\end{proof}

Note that while orthogonality ensures non-overlapping, it does not
ensure exhaustivity, only ${\rm OD}_A$ does. The next lemma details
the exhaustivity of ${\rm OD}_A$, and is a consequence of a later result,
that uses the notion of substitution (see \secref{sub:qua-substitution}
and Lemma~\ref{lem:od-output-exhaustive}).

\begin{proposition}
	\label{lem:od-exhaustive}
	If $\ODe A S$ and $\Delta \vdash e \colon A$, then there exists $e' \in S$
	such that $\neg(e~\bot~e')$.
\end{proposition}

This notion of orthogonal decomposition allows us to introduce \emph{unitary}
abstractions in our syntax. A basic unitary has the form $\unibasique$ and is
well-typed only if $(b_i)_i$ forms a basis and $(e_i)_i$ also forms a basis; we
then retrieve the intuition that a unitary should be equivalent to a change
of orthonormal basis. The typing rules for deriving unitaries and unitary
operations are detailed in Figure~\ref{fig:typisos}.

\begin{figure}[!h]
	\[
		\begin{array}{c}
			\infer{
				\entailiso \unibasique \colon A \iso B,
			}{
				\begin{array}{@{}l@{}}
					\Delta_i \entail b_i \colon A \\
					\Delta_i \entail e_i \colon B
				\end{array}
				&
				\begin{array}{@{}l@{}}
					\OD{A}{\{b_1, \dots, b_n\}} \\
					\ODe{B}{\{e_1, \dots, e_n\}}
				\end{array}
			}
			\qquad
			\infer{
				\entailiso \omega\inv \colon B \iso A,
			}{
				\entailiso \omega \colon A \iso B
			}
			\mynl
			\infer{
				\entailiso \omega_2 \circ \omega_1 \colon A \iso C,
			}{
				\entailiso \omega_1 \colon A \iso B
				&
				\entailiso \omega_2 \colon B \iso C
			}
			\qquad
			\infer{
				\entailiso \omega_1 \otimes \omega_2 \colon A_1 \otimes A_2 \iso B_1 \otimes B_2,
			}{
				\entailiso \omega_1 \colon A_1 \iso B_1
				&
				\entailiso \omega_2 \colon A_2 \iso B_2
			}
			\mynl
			\infer{
				\entailiso \omega_1 \oplus \omega_2 \colon A_1 \oplus A_2 \iso B_1 \oplus B_2,
			}{
				\entailiso \omega_1 \colon A_1 \iso B_1
				&
				\entailiso \omega_2 \colon A_2 \iso B_2
			}
			\qquad
			\infer{
				\entailiso \ctrl \isoterm \colon (\one \oplus \one) \otimes A \iso
				(\one \oplus \one) \otimes A.
			}{
				\entailiso \isoterm \colon A \iso A
			}
		\end{array}
	\]
	\caption{Typing rules of unitaries.}
	\label{fig:typisos}
\end{figure}

Terms in our syntax are either expressions or an application of a unitary to a
term, in a similar style to the $\lambda$-calculus; however, we have seen that
abstractions are considered separately in the grammar. The typing rules are
the same as the one for values given in Figure~\ref{fig:typing-values-simple},
with the addition of a rule that enables the application of a unitary to a
term. The details are in Figure~\ref{fig:typterms}. 

\begin{figure}[!h]
	\[\begin{array}{c}
		\infer{
			\emptyset \entail * \colon \one,
		}
		{}
		\qquad
		\infer{
			x \colon A \entail x \colon A,
		}{}
		\qquad
		\infer{
			\Delta_1,\Delta_2\entail \pair{t_1}{t_2} \colon A\otimes B,
		}{
			\Delta_1\entail t_1 \colon A
			&
			\Delta_2\entail t_2 \colon B
		}
		\\[1.5ex]
		\infer{
			\Delta\entail \inl{t} \colon A\oplus B,
		}{
			\Delta\entail t \colon A
		}
		\qquad
		\infer{
			\Delta\entail \inr{t} \colon A\oplus B,
		}{
			\Delta\entail t \colon B
		}
		\\[1.5ex]
		\infer{\vdash \zero \colon \nat,}{}
		\qquad
		\infer{
			\Delta \vdash \suc t \colon \nat,
		}{
			\Delta \vdash t \colon \nat
		}
		\\[1.5ex]
		\infer{
			\Delta \entail \Sigma_i (\alpha_i \cdot t_i) \colon A,
		}{
			\Delta \entail t_i \colon A
			&
			\sum_i \abs{\alpha_i}^2 = 1
			&
			\forall i\not= j, t_i~\bot~t_j
		}
		\qquad
		\infer{
			\Delta \entail \isoterm~t \colon B.
		}{
			\entailiso \isoterm \colon A \iso B
			&
			\Delta \entail t \colon A
		}
	\end{array}
	\]
	\caption{Typing rules of terms.}
	\label{fig:typterms}	
\end{figure}

The application of a unitary to a term is what carries the computational power
of the language. We have seen that in the $\lambda$-calculus, the
$\beta$-reduction reduces an application to a term where a substitution is
performed. A similar mechanism is at play in this syntax; however, we have seen
that the left-hand side in a unitary abstraction can contain several variables.
Hence the introduction of \emph{valuations}, which we use to perform
substitution and to define our equivalent of $\beta$-reduction.

\subsection{Valuations and Substitution}
\label{sub:qua-substitution}

We recall the formalisation proposed in~\cite{sabry2018symmetric}, with the
notion of valuation: a partial map from a finite set of variables (the support)
to a set of values. Given two basis values $b$ and $b'$, we build the smallest
valuation $\sigma$ such that the patterns of $b$ and $b'$ match and such that the
application of the substitution to $b$, written $\sigma(b)$, is equal to $b'$.
We denote the matching of a basis value $b'$ against a pattern $b$ and its
associated valuation $\sigma$ as the predicate $\match{\sigma}{b}{b'}$. Thus,
$\match{\sigma}{b}{b'}$ means that $b'$ matches with $b$ and gives a smallest
valuation $\sigma$, while $\sigma(b)$ is the substitution performed. The
predicate $\match{\sigma}{b}{b'}$ is defined as follows, with $\ini{b}$
being either $\inl{b}$ or $\inr{b}$.
\[
	\begin{array}{c}
		\infer{\match{\sigma}{*}{*}}{}
		\quad
		\infer{\match{\sigma}{x}{b'}}{\sigma = \{ x \mapsto b'\}}
		\quad
		\infer{\match{\sigma}{\ini b}{\ini b'}}{\match{\sigma}{b}{b'}}
		\mynl
  	\infer{
			\match{\sigma}{b_1 \otimes b_2}{b'_1 \otimes b'_2}
		}{
  		\match{\sigma}{b_1}{b'_1}
  		&
  		\match{\sigma}{b_2}{b'_2}
  		&
  		\text{supp}(\sigma_1) \cap \text{supp}(\sigma_2) = \emptyset
  		&
  		\sigma = \sigma_1\cup\sigma_2
    }
		\mynl
		\infer{\match{\sigma}{\zero}{\zero}}{}
		\quad
		\infer{\match{\sigma}{\suc b}{\suc b'}}{\match{\sigma}{b}{b'}}
	\end{array}
\]
Besides basis values, we authorise valuations to replace variables with any
expression, \emph{e.g.} $\set{x \mapsto e}$. Whenever $\sigma$ is a valuation
whose support contains the variables of $t$, we write $\sigma(t)$ for the value
where the variables of $t$ have been replaced with the corresponding terms in
$\sigma$, as follows: 
\begin{itemize}
	\item $\sigma(x) = e$ if $\{x\mapsto e\}\subseteq \sigma$, 
	\item $\sigma(*) = *$,
	\item $\sigma(\inl{t}) = \inl{\sigma(t)}$,
	\item $\sigma(\inr{t}) = \inr{\sigma(t)}$, 
	\item $\sigma(\pair{t_1}{t_2}) = \pair{\sigma(t_1)}{\sigma(t_2)}$, 
	\item $\sigma(\suc t) = \suc \sigma(t)$,
	\item $\sigma(\Sigma_i (\alpha_i \cdot t_i)) = \Sigma_i (\alpha_i \cdot
		\sigma(t_i))$,
	\item $\sigma(\omega~t) = \omega~\sigma(t)$.
\end{itemize}

\begin{remark}
	If $\match{\sigma}{b}{b'}$, then $\sigma(b) = b'$.
\end{remark}

\begin{example}
	Given a valuation $\sigma$ such that $\{ x \mapsto \inl{\inr *} \} \subseteq
	\sigma$, then $\sigma(x)$ is the expression $\inl{\inr *}$.
\end{example}

\begin{example}
	Given a valuation $\sigma$ such that $\{ x \mapsto \inl *, y \mapsto \inr *
	\} \subseteq \sigma$, then $\sigma(x \otimes y)$ is the expression $(\inl *)
	\otimes (\inr *)$.
\end{example}

We can now show the soundness of orthogonality with regard to pattern-matching:
in other words, orthogonality is stable by substitution, and thus the previous
remark ensures there cannot be any match between two basis values if they are
orthogonal.

\begin{lemma}
	\label{lem:ortho-subst-equiv}
	Given two terms $t_1$ and $t_2$, if $t_1~\bot~t_2$, then for all valuations
	$\sigma_1$ and $\sigma_2$, $\sigma_1(t_1)~\bot~\sigma_2(t_2)$.
\end{lemma}
\begin{proof}
	Observe that $\sigma_1(\inl t_1) = \inl{\sigma_1(t_1)}$ and $\sigma_2(\inr
	t_2) = \inr{\sigma_2(t_2)}$ and thus, whatever the valuations are, those two
	terms are orthogonal.  The rest of the proof falls directly by induction on
	the definition of $\bot$.
\end{proof}

If one of the basis values is closed, we observe that there is an equivalence
between matching the pattern and not being orthogonal; which also implies that
different patterns are orthogonal.

\begin{proposition}
	\label{prop:patterns-are-orthogonal}
	Given two well-typed basis values $\Delta \entail b \colon A$ and $~\entail
	b' \colon A$, $\neg(b~\bot~b')$ iff there exists $\sigma$ such that
	$\match{\sigma}{b}{b'}$.
\end{proposition}
\begin{proof}
	This is proven by a direct induction on $\Delta \entail b \colon A$.
\end{proof}

The next two lemmas provide a strong link between orthogonal decompositions and
substitutions.

\begin{lemma}[Exhaustivity and non-overlapping]
	\label{lem:od-substitution}
	Assume that $\OD A S$; then for all closed basis values $~\entail b' \colon
	A$, there exists a unique $b \in S$ and a unique $\sigma$ such that
	$\match{\sigma}{b}{b'}$.
\end{lemma}
\begin{proof}
	This is proven by induction on the derivation of $\OD A S$.
	\begin{itemize}
		\item If $\OD{A}{\set x}$. There is only $x$ in $S$ and $\set{x \mapsto b'}$
			is the only possible substitution.
		\item If $\OD{\one}{\set *}$, we have $b' = *$ and there is nothing to do.
		\item If $\OD{A \oplus B}{S \boxplus T}$, there are two cases:
			\begin{itemize}
				\item either $b' = \inl b'_A$, in which case the induction hypothesis
					gives a unique $b_A \in S$ and a unique $\sigma$ such that
					$\match{\sigma}{b_A}{b'_A}$, and thus $\match{\sigma}{\inl b_A}{b}$
					in a unique way,
				\item or $b' = \inr b'_B$, and a similar argument gives a unique match
					$\match{\sigma}{\inr b_B}{b'}$.
			\end{itemize}
		\item If $\OD{A \otimes B}{S}$, $b' = b'_A \otimes b'_B$, in both cases to
			derive ${\rm OD}$, we get unique $b_A$, $b_B$, $\sigma_A$ and $\sigma_B$
			such that $\match{\sigma_A \cup \sigma_B}{b_A \otimes b_B}{b'}$.
		\item If $\OD{\nat}{S^{\oplus 0}}$, there are two cases: either $b' =
			\zero$, in which case there is nothing to do, or $b' = \suc b''$, and the
			induction hypothesis gives a unique $b$ and a unique $\sigma$ such that
			$\match{\sigma}{b}{b''}$ and thus $\match{\sigma}{\suc b}{b'}$.
	\end{itemize}
\end{proof}

Observe that some of the results in this section focus on \emph{basis values},
and threfore do not involve linear combinations. This is because unitary
abstractions are formed as a set of clauses such as $b_i \iso e_i$, where the
terms on the left can only be basis values, and this allows us to narrow down
the pattern-matching to basis values only.

The definition of valuation $\sigma$ does not involve any condition on types or
type judgements. However, we need this sort of condition to formulate a
substitution lemma, hence the next definition.

\begin{definition}
	\label{def:well-valuation}
	A valuation $\sigma$ is said to be well-formed with regard to a context
	$\Delta$ if for all $(x_i \colon A_i) \in \Delta$, we have $\set{x_i \mapsto
	e_i} \subseteq \sigma$ and $~\entail e_i \colon A_i$. We write $\Delta
	\Vdash \sigma$ for a well-formed valuation with regard to $\Delta$. 
\end{definition}

\begin{remark}
	The valuation $\sigma$ obtained in Lemma~\ref{lem:od-substitution} is
	well-formed iff $b$ is well-typed.
\end{remark}

\begin{lemma}
	\label{lem:qua-subst-type}
	Given a well-typed term $\Delta \entail t \colon A$ and a well-formed
	valuation $\Delta \Vdash \sigma$, then we have $\entail \sigma(t) \colon A$.
\end{lemma}
\begin{proof}
	The proof is done by induction on $\Delta \entail t \colon A$.
	\begin{itemize}
		\item $~\entail * \colon \one$. Direct.
		\item $x \colon A \entail x \colon A$. Since $x \colon A \Vdash \sigma$,
			there is a well-typed term $t$ such that $\set{x \mapsto e} \subseteq
			\sigma$ and thus $\sigma(x)$ is well-formed.
		\item $\Delta_1, \Delta_2 \entail t_1 \otimes t_2 \colon A \otimes B$.  By
			induction hypothesis, $\sigma(t_1)$ and $\sigma(t_2)$ are well-formed,
			and thus $\sigma(t_1 \otimes t_2) = \sigma(t_1) \otimes \sigma(t_2)$ is
			also well-typed.
		\item $\Delta \entail \ini t \colon A_1 \oplus A_2$.
			By induction hypothesis, $\sigma(t)$ is well-typed, and thus
			$\ini{\sigma(t)} = \sigma(\ini t)$ is also well-typed.
		\item $~\entail \zero \colon \nat$. Direct.
		\item $\Delta \entail \suc t \colon \nat$.
			By induction hypothesis, $\sigma(t)$ is well-typed, and thus
			$\suc{\sigma(t)} = \sigma(\suc t)$ is also well-typed.
		\item $\Delta \entail \Sigma_i (\alpha_i \cdot t_i) \colon A$.
			By induction hypothesis, $\sigma(t_i)$ is well-typed for all $i$, and
			thus $\Sigma_i (\alpha_i \cdot \sigma(t_i)) = \sigma(\Sigma_i (\alpha_i
			\cdot t_i))$ is also well-typed, thanks to
			Lemma~\ref{lem:ortho-subst-equiv}.
		\item $\Delta \entail \omega~t \colon B$.
			By induction hypothesis, $\sigma(t)$ is well-typed, and thus
			$\omega~\sigma(t) = \sigma(\omega~t)$ is also well-typed.
	\end{itemize}
\end{proof}

Substitutions, now properly defined, help us formalise how the language handles
operations. In the paper this work is based on \cite{sabry2018symmetric}, the
computational behaviour of the language is presented through an operational
semantics, while terms are considered up to linear algebraic equalities. In
this chapter, we choose to work entirely with an \emph{equational theory},
similar to the ones introduced in the previous chapters (see \secref{sub:ccc},
\secref{sub:sem-lambda-effects} and \secref{sub:theories}), whose details are
outlined in the next section. This equational theory contains both linear
algebraic considerations and the computational aspects of the language.

\section{Equational Theory}
\label{sec:qu-simple-eq-theory}

In this section, we define an equational theory for our language, akin to the
ones presented in the previous chapters for the $\lambda$-calculus (see
Figure~\ref{fig:moggi-rules}), for Moggi's metalanguage (see
Figure~\ref{fig:moggi-monad-rules}) and for the Central Submonad Calculus (see
Figure~\ref{fig:rules}). An equation judgement is written $\Delta \entail t_1 =
t_2 \colon A$, where $\Delta$ is a context, $A$ is a type and $t_1$ and $t_2$
are terms. We do not need to assume that $t_1$ and $t_2$ are well-typed, it is
derived (see Proposition~\ref{prop:equal-to-typed}). We provide the main
equational theory of our language in Figure~\ref{fig:eq-rules-qu-control}, with
the rules of reflexivity, symmetry and transitivity in
Figure~\ref{fig:eq-rules-base}, linear algebraic identities in
Figure~\ref{fig:eq-rules-vector-space}, and congruence identities in
Figure~\ref{fig:eq-rules-qu-control-cong}.

\begin{remark}
	Among the equational rules presented in this section, only the equations in
	Figure~\ref{fig:eq-rules-qu-control} provide an operational account of the
	language. They can be seen as a reduction system from left to right.  On the
	other hand, the equations in Figure~\ref{fig:eq-rules-vector-space} show that
	the algebraic power of Hilbert spaces and isometries can be captured within
	the type system.
\end{remark}

\begin{figure}[!h]
	\[ \begin{array}{c}
		\infer[(refl)]{
			\Delta \entail t = t \colon A
		}{
			\Delta \entail t \colon A
		}
		\qquad
		\infer[(symm)]{
			\Delta \entail t_2 = t_1 \colon A
		}{
			\Delta \entail t_1 = t_2 \colon A
		}
		\mynl
		\infer[(trans)]{
			\Delta \entail t_1 = t_3 \colon A
		}{
			\Delta \entail t_1 = t_2 \colon A
			& \Delta \entail t_2 = t_3 \colon A
		}
	\end{array}
	\]
	\caption{Basic equational rules.}
	\label{fig:eq-rules-base}
\end{figure}

\begin{figure}[!h]
	\[ \begin{array}{c}
		\infer[(perm)]{\Delta \entail \Sigma_i (\alpha_i \cdot t_i) = \Sigma_i
		(\alpha_{\ell(i)} \cdot t_{\ell(i)}) \colon A}{\ell \colon \set{1,\dots,n}
		\to \set{1,\dots,n}\text{ is a bijection} & \Delta \entail \Sigma_i
		(\alpha_i \cdot t_i) \colon A}
		\mynl
		\infer[(0.scal)]{\Delta \entail \Sigma_{i=1}^n (\alpha_i \cdot t_i) =
		\Sigma_{i=1}^{n-1} (\alpha_i \cdot t_i) \colon A}{\Delta \entail
		\Sigma_{i=1}^n (\alpha_i \cdot t_i) \colon A & \alpha_n = 0}
		\qquad
		\infer[(1.scal)]{\Delta \entail \Sigma_{i=1}^1 (1 \cdot t) = t \colon
		A}{\Delta \entail t \colon A}
		\mynl
		\infer[(fubini)]{\Delta \entail \Sigma_i \left( \alpha_i \cdot \Sigma_j (\beta_{ij}
		\cdot t_{j}) \right) = \Sigma_j ( \left( \sum_i \alpha_i \beta_{ij} \right) \cdot t_{j})
		\colon A}{
			\Delta \entail \Sigma_i \left( \alpha_i \cdot \Sigma_j (\beta_{ij} \cdot t_{j})
			\right) \colon A
		}
		\mynl
		\infer[(double)]{\Delta \entail \Sigma_i \left( \alpha_i \cdot \Sigma_j (\beta_{ij}
		\cdot t_{ij}) \right) = \Sigma_{ij} (\alpha_i \beta_{ij}\cdot t_{ij})
		\colon A}{
			\Delta \entail \Sigma_{ij} (\alpha_i \beta_{ij}\cdot t_{ij})
			\colon A
		}
		\mynl
		\infer[(\omega.linear)]{\Delta \entail \omega~\Sigma_i (\alpha_i \cdot
		t_i) = \Sigma_i (\alpha_i \cdot \omega~t_i) \colon B}
		{
			\entailiso \omega \colon A \iso B
			&
			\Delta \entail \Sigma_i (\alpha_i \cdot t_i) \colon A
		}
		\mynl
		\infer[(\iota.linear_1)]{\Delta \entail \inl{} \Sigma_i (\alpha_i \cdot
		t_i) = \Sigma_i (\alpha_i \cdot \inl t_i) \colon A \oplus B}{\Delta \entail
		\Sigma_i (\alpha_i \cdot t_i) \colon A}
		\mynl
		\infer[(\iota.linear_2)]{\Delta \entail \inr{} \Sigma_i (\alpha_i \cdot
		t_i) = \Sigma_i (\alpha_i \cdot \inr t_i) \colon A \oplus B}{\Delta \entail
		\Sigma_i (\alpha_i \cdot t_i) \colon B}
		\mynl
		\infer[(\otimes.linear_1)]{\Delta_1,\Delta_2 \entail t \otimes (\Sigma_i
		(\alpha_i \cdot t_i)) = \Sigma_i (\alpha_i \cdot t \otimes t_i) \colon A
		\otimes B}{\Delta_1 \entail t \colon A & \Delta_2 \entail \Sigma_i
		(\alpha_i \cdot t_i) \colon B}
		\mynl
		\infer[(\otimes.linear_2)]{\Delta_1,\Delta_2 \entail (\Sigma_i (\alpha_i
		\cdot t_i)) \otimes t = \Sigma_i (\alpha_i \cdot t_i \otimes t) \colon A
		\otimes B}{\Delta_1 \entail t \colon B & \Delta \entail \Sigma_i (\alpha_i
		\cdot t_i) \colon A}
		\mynl
		\infer[(S.linear)]{\Delta \entail \suc \Sigma_i (\alpha_i \cdot t_i) =
		\Sigma_i (\alpha_i \cdot \suc t_i) \colon \nat}{\Delta \entail \Sigma_i
		(\alpha_i \cdot t_i) \colon \nat}
	\end{array}
	\]
	\caption{Vector space and linear applications equational rules.}
	\label{fig:eq-rules-vector-space}
\end{figure}

\begin{figure}
	\[ \begin{array}{c}
		\infer[(\iota.eq_1)]{
			\Delta \entail \inl t_1 = \inl t_2 \colon A \oplus B
		}{
			\Delta \entail t_1 = t_2 \colon A
		}
		\qquad
		\infer[(\iota.eq_2)]{
			\Delta \entail \inr t_1 = \inr t_2 \colon A \oplus B
		}{
			\Delta \entail t_1 = t_2 \colon B
		}
		\mynl
		\infer[(\otimes.eq_1)]{
			\Delta, \Delta' \entail \pair{t_1}{t} = \pair{t_2}{t} \colon A \otimes B
		}{
			\Delta \entail t_1 = t_2 \colon A & \Delta' \entail t \colon B
		}
		\qquad
		\infer[(\otimes.eq_2)]{
			\Delta, \Delta' \entail \pair{t}{t_1} = \pair{t}{t_2} \colon A \otimes B
		}{
			\Delta \entail t_1 = t_2 \colon B & \Delta' \entail t \colon A
		}
		\mynl
		\infer[(S.eq)]{
			\Delta \entail \suc t_1 = \suc t_2 \colon \nat
		}{
			\Delta \entail t_1 = t_2 \colon \nat
		}
		\qquad
		\infer[(\omega.eq)]{
			\Delta \entail \omega~t_1 = \omega~t_2 \colon B
		}{
			\Delta \entail t_1 = t_2 \colon A & \entailiso \omega \colon A \iso B
		}
		\mynl
		\infer[(\Sigma.eq)]{
			\Delta \entail \Sigma_i (\alpha_i \cdot t_i) 
			= \Sigma_i (\alpha_i \cdot t'_i) \colon A
		}{
			\Delta \entail \Sigma_i (\alpha_i \cdot t_i) \colon A
			& \forall i, \Delta \entail t_i = t'_i \colon A
		}
	\end{array}
	\]
	\caption{Congruence equational rules of simply-typed quantum control.}
	\label{fig:eq-rules-qu-control-cong}
\end{figure}

\begin{figure}
	\[ \begin{array}{c}
		\infer[(\omega.\beta)]{\entail \unibasique~b' = \sigma(e_i) \colon B
		}{
			\entail b'\colon A & \entailiso \unibasique \colon A \iso B &
			\match{\sigma}{b_i}{b'} 
		}
		\mynl
		\infer[(\omega.inv)]{
			\entail \omega\inv~v = b \colon A
		}{
			\entail \omega~b = v \colon B
		}
		\qquad
		\infer[(\omega.comp)]{
			\Delta \entail (\omega_2 \circ \omega_1)~b = \omega_2~(\omega_1~b) \colon C
		}{
			\entailiso \omega_1 \colon A \iso B & \entailiso \omega_2 \colon B \iso C
			& \Delta \entail b \colon A
		}
		\mynl
		\infer[(\omega.\otimes)]{
			\entail (\omega_1 \otimes \omega_2)~(b_1 \otimes b_2) = (\omega_1~b_1)
			\otimes (\omega_2~b_2) \colon B_1 \otimes B_2
		}{
			\entailiso \omega_1 \colon A_1 \iso B_1 & \entailiso \omega_2 \colon A_2 \iso B_2
			& \entail b_1 \colon A_1 & \entail b_2 \colon A_2
		}
		\mynl
		\infer[(\omega.\oplus_1)]{
			\entail (\omega_1 \oplus \omega_2)~(\inl b) = \inl (\omega_1~b) \colon B_1 \oplus B_2
		}{
			\entailiso \omega_1 \colon A_1 \iso B_1 & \entailiso \omega_2 \colon A_2 \iso B_2
			& \entail b \colon A_1
		}
		\mynl
		\infer[(\omega.\oplus_2)]{
			\entail (\omega_1 \oplus \omega_2)~(\inr b) = \inr (\omega_2~b) \colon B_1 \oplus B_2
		}{
			\entailiso \omega_1 \colon A_1 \iso B_1 & \entailiso \omega_2 \colon A_2 \iso B_2
			& \entail b \colon A_2
		}
		\mynl
		\infer[(\omega.ctrl_1)]{
			\entail (\ctrl \omega)~((\inl *) \otimes b) = (\inl *) \otimes b \colon
			(\one \oplus \one) \otimes A
		}{
			\entailiso \omega \colon A \iso A
			& \entail b \colon A
		}
		\mynl
		\infer[(\omega.ctrl_2)]{
			\entail (\ctrl \omega)~((\inr *) \otimes b) = (\inr *) \otimes (\omega~b)
			\colon (\one \oplus \one) \otimes A
		}{
			\entailiso \omega \colon A \iso A
			& \entail b \colon A
		}
	\end{array}
	\]
	\caption{Computational equational rules of simply-typed quantum control.}
	\label{fig:eq-rules-qu-control}
\end{figure}

\subsection{Equations and typing}
\label{sub:qua-eq-typing}

We start by proving that the equational theory presented is sound with the
typing rules of the language. In other words, we show that if two terms are
equal in our theory, they are both well-typed. To do so, we need to show that
equality between terms preserve orthogonality.

\begin{lemma}
	\label{lem:equal-ortho}
	Given two terms $t_1$ and $t_2$ such that $t_1~\bot~t_2$ and $\Delta \entail
	t_1 = t'_1 \colon A$, then $t'_1~\bot~t_2$.
\end{lemma}
\begin{proof}
	By induction on the rules of the equational theory. 
\end{proof}

The proof of the next proposition heavily relies on the previous lemma. Indeed,
to prove that a linear combination is well-typed, one needs to prove that all
the terms involved in the linear combination are pairwise orthogonal. 

\begin{proposition}
	\label{prop:equal-to-typed}
	If $\Delta \entail t_1 = t_2 \colon A$ is well-formed, then
	$\Delta \entail t_1 \colon A$ and $\Delta \entail t_2 \colon A$
	also are.
\end{proposition}
\begin{proof}
	By induction on the rules of the equational theory. 
\end{proof}

\subsection{Bases}
\label{sub:qua-eq-bases}

As expected, and thanks to the equational theory, an orthogonal decomposition
gives a finite representation of an orthonormal basis. This is proven in
Lemma~\ref{lem:od-output-exhaustive}. We start by proving that any expression
is equal to a combination of basis values.

\begin{lemma}
	\label{lem:values-are-combinations}
	Given a well-typed closed expression $~\entail e \colon A$, there exists
	a set of indices $I$, a family of basis values $(b_i)_{i \in I}$, $(\alpha_i)_{i
	\in I}$ a family of complex numbers, such that $~\entail e =
	\Sigma_i (\alpha_i \cdot b_i) \colon A$.
\end{lemma}
\begin{proof}
	The proof is done by induction on $~\entail e \colon A$.
	\begin{itemize}
		\item $~\entail * \colon \one$. Nothing to do.
		\item $~\entail e_1 \otimes e_2 \colon A \otimes B$. The induction hypothesis
			gives $I$, $(b^1_i)$ and $(\alpha_i)$, and $J$, $(b^2_j)$ and $(\beta_j)$,
			such that $~\entail e_1 = \Sigma_i (\alpha_i \cdot b^1_i) \colon A$
			and $~\entail e_2 = \Sigma_j (\beta_j \cdot b^2_j) \colon B$. Thus,
			we have that 
 			\begin{align*}
 				&\ \entail e_1 \otimes e_2 & \\
 				&= \left( \Sigma_i (\alpha_i \cdot b^1_i) \right)
 				\otimes \left( \Sigma_j (\beta_j \cdot b^2_j) \right) \colon A \otimes B 
				&\text{(induction hypothesis)} \\
 				&= \Sigma_i \left( \alpha_i \cdot b^1_i \otimes \left( \Sigma_j
 				(\beta_j \cdot b^2_j) \right) \right) \colon A \otimes B 
 				&(\otimes.linear_2) \\
 				&= \Sigma_i \left( \alpha_i \cdot \Sigma_j ( \beta_j \cdot 
 				b^1_i \otimes b^2_j ) \right) \colon A \otimes B 
 				& (\otimes.linear_1) \\
 				&= \Sigma_{ij} (\alpha_i \beta_j \cdot b^1_i \otimes b^2_j) \colon A
 				\otimes B
 				& (double)
 			\end{align*}
		\item $~\entail \inl e \colon A \oplus B$. The induction hypothesis
			gives $~\entail e = \Sigma_i (\alpha_i \cdot b_i) \colon A$,
			and observe that $~\entail \inl{} \Sigma_i (\alpha_i \cdot b_i)
			= \Sigma_i (\alpha_i \cdot \inl b_i) \colon A \oplus B$.
		\item $~\entail \inr e \colon A \oplus B$ has the same conclusion.
		\item $~\entail \suc e \colon \nat$ is similar to the previous point.
		\item $~\entail \Sigma_i (\alpha_i \cdot e_i) \colon A$. The induction
			hypothesis gives $(\beta_{ij})$ and $(b_j)$ (the $b$ does not depend in
			$i$ without loss of generality, because $0 \cdot b$ can be added to any
			sum term $t$, as long as $b$ is orthogonal to $t$). Finally, we have
			$~\entail \Sigma_i (\alpha_i \cdot \Sigma_j (\beta_{ij} \cdot b_j)) =
			\Sigma_j ((\sum_i \alpha_i \beta_{ij}) \cdot b_j) \colon A$.
	\end{itemize}
\end{proof}

\begin{remark}
	The resulting term in the previous lemma, written $\Sigma_i (\alpha_i \cdot
	b_i)$ is a value if the basis values are correctly ordered and if all the
	scalars are non zero. Thanks to the $(perm)$ and $(0.scal)$ rules in
	Figure~\ref{fig:eq-rules-vector-space}, we can assume so. Thus, the lemma
	above shows that expressions have a unique normal form. 
\end{remark}

Lemma~\ref{lem:values-are-combinations} shows that any expressions can be
decomposed as a linear combination of elements of the \emph{canonical}
orthogonal decomposition -- namely, made of basis values only. We can
generalise this lemma to any orthogonal decomposition, with the help of
substitutions. We show that, given an orthogonal decomposition $S$, a closed
expression $e$ can be written as a normalised decomposition of elements of $S$,
where variables are substituted. The elements of $S$ can appear several times
in the decomposition. For example, $\{ x \}$ is an orthogonal decomposition of
$\one \oplus \one$, and the term $\frac{1}{\sqrt 2} \cdot (\inl *) +
\frac{1}{\sqrt 2} (\inr *)$ is of the form $\frac{1}{\sqrt 2} \cdot \sigma_1(x)
+ \frac{1}{\sqrt 2} \cdot \sigma_2(x)$, where $\sigma_1$ replaces $x$ with
$\inl *$ and $\sigma_2$ replaces $x$ with $\inr *$.

\begin{lemma}
	\label{lem:od-output-exhaustive}
	Given $\ODe B S$, where all elements of $S$ are well-typed, and a
	well-typed closed expression $~\entail e \colon B$, there exists $I$ a set
	of indices, $(s_i)_{i \in I}$ a family of elements of $S$, $(\alpha_i)_{i \in
	I}$ a family of complex numbers and $(\sigma_i)_{i \in I}$ a family of
	valuations such that $~\entail e = \Sigma_i (\alpha_i \cdot \sigma_i(s_i))
	\colon B$.
\end{lemma}
\begin{proof}
	This is proven by induction on ${\rm OD}$. The previous lemma gives a term
	equal to $e$ in the equational theory written as a finite sum of basis values
	$\Sigma_i (\alpha_i \cdot b_i)$.
	\begin{itemize}
		\item $\ODe{A}{(\set x)}$. The substitution $\sigma = \set{x \mapsto e}$
			is suitable with $e = \sigma(x)$.
		\item $\ODe{\one}{(\set *)}$. Nothing to do.
		\item $\ODe{A \oplus B}{S \boxplus T}$. Each $b_i$ is either $\inl b'_i$,
			with gives a suitable substitution $\sigma_i$ and $s_i \in S$, thus $\inl
			s_i \in S \boxplus T$, or $b_i$ is $\inr b'_i$, giving suitable
			substitution $\sigma_i$ and $s_i \in T$, thus $\inr s_i \in S \boxplus
			T$; all this by induction hypothesis.
		\item $\ODe{A \otimes B}{S}$. Each $b_i$ is of the form $b'_i \otimes b
			''_i$, the induction hypothesis gives suitable $\sigma'_i$, $s'_i$,
			$\sigma''_i$, $s''_i$, that can be assembled into $\sigma_i = \sigma'_i
			\cup \sigma''_i$ and $s_i = s'_i \otimes s''_i$.
		\item $\ODe{\nat}{S^{\oplus 0}}$. Each $b_i$ is either $\zero$, for which
			there is nothing to do, or $\suc b'_i$, in which case the induction
			hypothesis concludes.
		\item $\ODe{B}{S^\beta}$. First, we show that each $s_i \in S$ can
			be written as a linear combination of elements of $S^\beta$. Indeed,
			in the equational theory: 
			\begin{align*}
				\Sigma_{s \in S} (\res\beta_{s,s_i} \cdot
				\Sigma_{s' \in S} (\beta_{s,s'} \cdot s')) 
				&= \Sigma_{s' \in S} \left(\sum_{s \in S} \res\beta_{s,s_i}
				\beta_{s,s'} \right) \cdot s' \\
				&= \Sigma_{s' \in S} (\delta_{s' = s_i} \cdot s') = s_i
			\end{align*}
			and the conclusion is then direct.
	\end{itemize}
\end{proof}

\subsection{Normal Forms}
\label{sub:qua-normalisation}

The following lemma is loosely equivalent to progress for an operational
semantics, and involves both directions: the application of a unitary to a
value reduces to a value, and given a unitary and a value, there exists a value
that is the inverse image of the latter. 

\begin{lemma}
	\label{lem:equational-progress}
	Given $~\entailiso \omega \colon A \iso B$ (see Figure~\ref{fig:typisos}), we
	have the following:
	\begin{itemize}
		\item for all $~\entail e \colon A$, there exists a value judgement
			$~\entail v \colon B$ such that $~\entail \omega~e = v \colon B$;
		\item for all $~\entail e \colon B$, there exists a value judgement
			$~\entail u \colon A$ such that $~\entail \omega~u = e \colon B$.
	\end{itemize}
\end{lemma}
\begin{proof}
	This is proven by induction on the judgement $~\entailiso \omega \colon A
	\iso B$.
	\begin{itemize}
		\item Assume $~\entailiso \unibasique \colon A \iso B$, we have in particular
			that $\OD{A}{(\{b_i\}_{i \leq n})}$. In the case $~\entail e \colon A$,
			Lemma~\ref{lem:od-output-exhaustive} then gives a set $J$ and a
			decomposition of $e$ as follows: $~\entail e = \Sigma_j (\alpha_j \cdot
			\sigma_j(b_{i_j})) \colon A$. Therefore,
			\begin{align*}
				\entail 
				& \unibasique~e & \\
				&= \unibasique~\Sigma_j (\alpha_j \cdot \sigma_j(b_{i_j})) \colon B 
				& (\omega.eq) \\
				&= \Sigma_j (\alpha_j \cdot \unibasique~\sigma_j(b_{i_j})) \colon B 
				& (\omega.linear) \\
				&= \Sigma_j (\alpha_j \cdot \sigma_j(\unibasique~b_{i_j})) \colon B 
				& (\text{definition}) \\
				&= \Sigma_j (\alpha_j \cdot \sigma_j(v_{i_j})) \colon B
				& (\omega.\beta)
			\end{align*}
			The latter term is a closed expression, thus
			Lemma~\ref{lem:values-are-combinations} ensures that there exists a value
			$~\entail v \colon A$ such that $~\entail \unibasique~e = v \colon B$.
			
			On the other hand, we have $\ODe{B}{(\{v_i\}_{i \leq n})}$. Assume $~\entail
			e \colon B$, then Lemma~\ref{lem:od-output-exhaustive} provides a set $K$
			and a decomposition $~\entail e = \Sigma_k (\alpha_k \cdot \sigma_k(v_{i_k}))
			\colon B$. With the same computation as above, we have $~\entail
			\unibasique~u' = e \colon B$ with $u'$ being the expression $\Sigma_k
			(\alpha_k \cdot \sigma_k(b_{i_k}))$. Since it is an expression,
			Lemma~\ref{lem:values-are-combinations} ensures that there is value
			$~\entail u \colon A$ such that $~\entail u' = u \colon A$ and therefore
			$~\entail \unibasique~u = e \colon B$.
		\item Assume $~\entailiso \omega\inv \colon B \iso A$. Given $~\entail e \colon
			B$, the induction hypothesis gives $~\entail u \colon A$ such that
			$~\entail \omega~u = e \colon B$, and thus $~\entail \omega\inv~e = u
			\colon A$.  Moreover, given $~\entail e \colon A$, the induction
			hypothesis gives $~\entail v \colon B$ such that $~\entail \omega~e = v
			\colon B$, and thus $~\entail \omega\inv~v = e \colon A$.
		\item Assume $~\entailiso \omega_2 \circ \omega_1 \colon A \iso C$. The
			induction hypothesis gives us $v_1$ such that $~\entail \omega_1~e = v_1
			\colon B$ and then $v_2$ such that $~\entail \omega_2~v_1 = v_2 \colon
			C$, which ensures the result. A related reasoning proves the second
			point.
		\item Assume $~\entailiso \omega_1 \oplus \omega_2 \colon A_1 \oplus A_2 \iso
			B_1 \oplus B_2$. Lemma~\ref{lem:values-are-combinations} ensures that $e$
			is given as a combination of basis values $~\entail e = \Sigma_i
			(\alpha_i \cdot b_i) \colon A_1 \oplus A_2$. Moreover, we know that
			$~\entail (\omega_1 \oplus \omega_2)~e = \Sigma_i (\alpha_i \cdot
			(\omega_1 \oplus \omega_2)~b_i) \colon B_1 \oplus B_2$. Therefore, it is
			sufficient to consider the case of basis values. There are two similar
			cases, namely $\inl b$ and $\inr b$. In the first case, the induction
			hypothesis gives $v_1$ such that $~\entail \omega_1~b = v_1 \colon B_1$,
			thus $~\entail (\omega_1 \oplus \omega_2)~(\inl b) = \inl v_1 \colon B_1
			\oplus B_2$. The other case is similar. A related reasoning proves the
			second point.
		\item Assume $~\entailiso \omega_1 \otimes \omega_2 \colon A_1 \otimes A_2 \iso
			B_1 \otimes B_2$. Like above, it is sufficient to prove the result for
			basis values. We write $b_1 \otimes b_2$ for $b$, and the induction
			hypothesis provides $v_1$ and $v_2$ such that $~\entail \omega_1~b_1 =
			v_1 \colon B_1$ and $~\entail \omega_2~b_2 = v_2 \colon B_2$, ensuring
			that $~\entail (\omega_1 \otimes \omega_2)~(b_1 \otimes b_2) = v_1
			\otimes v_2 \colon B_1 \otimes B_2$. A related reasoning proves the
			second point.
		\item Assume $~\entailiso \ctrl \omega \colon (\one \oplus \one) \otimes A \iso (\one
			\oplus \one) \otimes A$. Once again, it is sufficient to prove the result
			for basis values. In the case $(\inl *) \otimes b$, there is nothing to
			do. The other case is $(\inr *) \otimes b$, and the induction hypothesis
			gives $v$ such that $~\entail \omega~b = v \colon A$, and then $~\entail
			(\ctrl \omega)~((\inr *) \otimes b) = (\inr *) \otimes v \colon (\one \oplus
			\one) \otimes A$. A related reasoning proves the second point.
	\end{itemize}
\end{proof}

We have proven that unitary applications progress and reduce to values, if one
wishes to have an operational point of view. This allows us to prove, with the
same operational view, that the system admits unique normal forms; this means
that any term $t$ is equal to a single value.

\begin{theorem}
	\label{th:equational-sn}
	Given $~\entail t \colon A$, there exists $~\entail v \colon A$ such that
	$~\entail t = v \colon A$.
\end{theorem}
\begin{proof}
	This is proven by induction on the typing rules of $~\entail t \colon A$.
	\begin{itemize}
		\item The cases $*$, $x$, $\zero$ and sum are straightforward.
		\item In the case $t_1 \otimes t_2$, the induction hypothesis gives
			corresponding $v_1$ and $v_2$, that ensure $~\entail t_1 \otimes t_2 =
			v_1 \otimes v_2 \colon A \otimes B$.
		\item In the case $\ini t$, the induction hypothesis provides $v$ such that
			$~\entail t = v \colon A_i$, thus $~\entail \ini t = \ini v \colon A_1
			\oplus A_2$.
		\item In the case $\suc t$, the induction hypothesis provides $v$ such that
			$~\entail t = v \colon \nat$, thus $~\entail \suc t = \suc v \colon
			\nat$.
		\item In the case $\omega~t$, the induction hypothesis gives $v$ such that
			$~\entail t = v \colon A$. The previous lemma,
			Lemma~\ref{lem:equational-progress}, provides $v'$ such that $~\entail
			\omega~v = v' \colon B$, thus $~\entail t = v' \colon B$.
	\end{itemize}
\end{proof}

\subsection{Discussion: Operational Semantics}

We briefly discuss the operational semantics for terms presented in
\cite{sabry2018symmetric}. This section builds up the comparison with the
$\lambda$-calculus by defining our version of $\beta$-reduction that suits the
language. This reduction is given by the rule $(\omega.\beta)$ in
Figure~\ref{fig:eq-rules-qu-control}, when read left to right.

In \cite{sabry2018symmetric}, values and terms are considered modulo
associativity and commutativity of the addition, and modulo the equational
theory of modules; and they consider the value and term constructs
$\pair{-}{-}$, $\inl(-)$, $\inr(-)$, $\suc -$ and $\omega~-$ distributive over
sum and scalar multiplication, only in this subsection also.

Therefore, in that setting, an expression $e$ is equal to some combination of
basis values $\Sigma_i (\alpha_i \cdot b_i)$; and the application of a unitary
$\omega$ to $e$ is equal to $\Sigma_i (\alpha_i \cdot \omega~b_i)$ thanks to
linearity. Thus, it is sufficient to give a $\beta$-reduction rule for unitaries
applied to basis values, as follows.
\[
	\begin{array}{c}
		\infer{ \unibasique~b' \to \sigma(v_i)}{
			\match{\sigma}{b_i}{b'}
		}
	\end{array}
\]
This rule is the same as $(\omega.\beta)$, this time oriented left to right.
Note that the reduction defined this way can only be applied with a closed
$b'$. However, this formulation is not satisfying, because it requires working
up to linear algebra equalities, which are then mixed with an operational
semantics. Our solution in this chapter is to completely embrace the equational
theory aspect and only work up to equalities, keeping in mind which rules bear
a computational meaning such as the one above.

Another solution is to work only with rewriting rules, which has been the focus
of several papers around algebraic $\lambda$-calculi
\cite{arrighi2017vectorial, arrighi2017lineal, vaux2009algebraic,
selinger2009quantum}, where all the rules, whether they are computational or
linear algebraic, have a direction.

\section{Mathematical Development: Hilbert spaces for semantics}
\label{sec:qu-simple-maths}

This section heavily relies on the notations and definitions in
\secref{sec:back-hilbert}, where introductory notions on Hilbert spaces are
outlined. The goal of this section is to provide the tools to define the
denotational semantics of the programming language given above. To do so, we
work with contractions for convenience, since the main mathematical objects
that we need -- namely isometries and unitaries -- are in particular
contractions.

\paragraph{Sums.}
Given two maps $f,g \colon X \to Y$ in $\Contr$, their linear algebraic sum $f
+ g$ is not necessarily a contraction. We introduce the notion of
\emph{compatibility}. This notion is inherited from the one in restriction and
inverse categories (see Definition~\ref{def:restr-compati}). We use, in
particular, an observation that links zero morphisms to compatibility in
inverse categories (see Lemma~\ref{lem:inv-ortho-compati}); but this is adapted
to Hilbert spaces in this chapter. In this context, the \emph{join} is also
different, as we use the algebraic sum inherited from the vector space
structure.


\begin{definition}[Compatibility]
	\label{def:compati}
	Given $f,g \colon X \to Y$ two maps in $\Contr$, $f$ and $g$ are said to be
	compatible if $(\Ker f)^\bot \bot (\Ker g)^\bot$ and $\iim f \bot \iim g$.
\end{definition}

This definition of compatibility ensures that there is no overlap between the
inputs, and also between the outputs. The next lemma is then direct.

\begin{lemma}
	\label{lem:compati}
	Given two compatible contractive maps $f,g \colon X \to Y$, $f+g$ is
	also contractive.
\end{lemma}
\begin{proof}
	Let $x \in X$. Given the assumptions, there exists $x' \in (\Ker f)^\bot$
	and $x'' \in (\Ker g)^\bot$ such that $(f+g)(x) = f(x') + g(x'')$, $g(x') = 0$,
	$f(x'') = 0$ and $\ip{f(x')}{g(x'')}$. Therefore,
	\begin{align*}
		\Vert (f+g)(x) \Vert 
		&= \Vert f(x') + g(x'') \Vert 
		= \ip{f(x') + g(x'')}{f(x') + g(x'')} \\
		&= \ip{f(x')}{f(x')} + \ip{g(x'')}{g(x'')}
		\leq \ip{x'}{x'} + \ip{x''}{x''} \\ 
		&\leq \ip{x}{x} = \Vert x \Vert.
	\end{align*}
\end{proof}

As mentioned above, the conditions in Def.~\ref{def:compati} can be simplified
in more algebraic expressions, in the spirit of
Lemma~\ref{lem:inv-ortho-compati}. We prove a quick lemma first.

\begin{lemma}
	\label{lem:dagger-ker-im}
	Given $f \colon X \to Y$ in $\Contr$, we have:
	\begin{itemize}
		\item $\Ker(f\dg) = (\iim f)^\bot$;
		\item $\iim(f\dg)^\bot = \Ker f$.
	\end{itemize}
\end{lemma}
\begin{proof}
	Let us prove both points separately.
	\begin{itemize}
		\item We proceed by double inclusion.
			\begin{itemize}
				\item Let $x \in \Ker(f\dg)$. Let $y \in X$. We have \( \ip{x}{f y} =
					\ip{f\dg x}{y} = \ip{0}{y} = 0 \).
					Therefore, $\Ker(f\dg) \subseteq (\iim f)^\bot$.
				\item Let $x \in (\iim f)^\bot$. Thus, for all $y \in X$, we have
					$\ip{x}{f y} = 0$, which implies that $\ip{f\dg x}{y} = 0$. Since it
					is true for all $y$, we have $f\dg x = 0$. Therefore, $(\iim f)^\bot
					\subseteq \Ker(f\dg)$.
			\end{itemize}
		\item We proceed by double inclusion. 
			\begin{itemize}
				\item Let $x \in \iim(f\dg)^\bot$. Thus, for all $y \in Y$,
					$\ip{x}{f\dg y} = 0$; then $\ip{f x}{y} = $ for all $y$, thus $f x = 0$.
					Therefore, $\iim(f\dg)^\bot \subseteq \Ker f$.
				\item Let $x \in \Ker f$. Let $y \in Y$. We have \( \ip{x}{f\dg y} =
					\ip{f x}{y} = 0 \). Therefore, $\Ker f \subseteq \iim(f\dg)^\bot$.
			\end{itemize}
	\end{itemize}
\end{proof}

We can now express a sufficient condition for compatibility in algebraic terms.

\begin{lemma}
	\label{lem:algebra-compati}
	Given two contractive maps $f,g \colon X \to Y$, $f\dg g = 0$ and $fg\dg = 0$
	iff $f$ and $g$ are compatible.
\end{lemma}
\begin{proof}
	We prove this lemma by double implication.
	\begin{itemize}
		\item If $f$ and $g$ are compatible, then for all $x \in X$, $g x$ is in
			$\iim g \subseteq \iim f ^\bot = \Ker (f\dg)$, thus $f\dg g = 0$.
			Similarly, for all $y \in Y$, $g\dg y$ is in $\iim(g\dg) \subseteq
			(\iim(f\dg))^\bot = \Ker f$.
		\item If $f\dg g = 0$ and $fg\dg = 0$. The first equality implies that
			$\iim g \subseteq \Ker (f\dg) = (\iim f)^\bot$ and therefore $\iim g \bot \iim f$.
			The second equality similarly implies that $(\Ker f)^\bot \bot (\Ker g)^\bot$.
	\end{itemize}
\end{proof}

\begin{remark}
	It might seem that the conditions introduced above are not symmetric on $f$
	and $g$. But one can observe that $0\dg = 0$ and $(f\dg g)\dg = g\dg
	f\dg{\dg} = g\dg f$, thus  $f\dg g = 0$ iff $g\dg f = 0$. Similarly, $fg\dg
	= 0$ iff $gf\dg = 0$.
\end{remark}

Lemma~\ref{lem:algebra-compati} introduces a new point of view on
compatibility, through an \emph{orthogonality} between morphisms, as it was
observed for inverse categories in Remark~\ref{rem:loose-inner-product}. This
new point of view of orthogonality is a generalisation of the orthogonality in
Hilbert spaces.  Indeed, two vectors $\ket x$ and $\ket y$ in a Hilbert space
$H$ are orthogonal if $\braket x y = 0$. In our generalised view, $\ket x$ and
$\ket y$ are orthogonal if $\ket x \dg \ket y = 0$. Since $\ket x \dg = \bra
x$, our orthogonality between morphisms generalises the usual notion of
orthogonality.

\begin{example}
	The morphisms $\ketbra 0 0 \colon \C^2 \to \C^2$ and $\ketbra 1 1 \colon \C^2 \to
	\C^2$ are orthogonal in our generalised sense, because the vectors $\ket 0$ and
	$\ket 1$ are orthogonal in the linear algebraic sense. This justifies that
	their linear sum $\ketbra 0 0 + \ketbra 1 1$ is a contraction (and, in this
	case, it is also a unitary).
\end{example}

\paragraph{Direct sum.}
Unsurprisingly, the unit type is to be represented by the one-dimensional
Hilbert space $\C$, the line of complex numbers. In the syntax, orthogonality
and thus pattern-matching, depend on direct sums. The latter are interpreted
as direct sums of Hilbert spaces. We show that this interpretation gives
rise to orthogonality in the sense of contractions.

\begin{definition}
	We write $\iota^{X,Y}_l \colon X \to X \oplus Y$ for the isometry such that
	for all $x \in X$, $\iota^{X,Y}_l x = (x,0)$. We call this the \emph{left
	injection}.  Similarly, the \emph{right injection} is written $\iota^{X,Y}_r
	\colon Y \to X \oplus Y$.
\end{definition}

\begin{lemma}[\cite{heunen2019categories}]
	\label{lem:injection-compati}
	Given two Hilbert spaces $X,Y$, $(\iota^{X,Y}_l)\dg
	\iota^{X,Y}_r = 0$ and $(\iota^{X,Y}_r)\dg \iota^{X,Y}_l = 0$.
\end{lemma}

\begin{example}
	\label{ex:qua-iota}
	The previous lemma ensures that $\iota^{X,Y}_l (\iota^{X,Y}_l)\dg$ and
	$\iota^{X,Y}_r (\iota^{X,Y}_r)\dg$ are compatible. Note that $\iota^{X,Y}_l
	(\iota^{X,Y}_l)\dg + \iota^{X,Y}_r (\iota^{X,Y}_r)\dg = \iid$.
\end{example}

Note that given a complex number $\alpha$ and a contraction $f \colon A \to B$,
the outer product $\alpha \cdot f$ is written $\alpha f$ when it is not
ambiguous. Given a set $S$, we write $(\alpha_i)_{i \in S}$ for a family of
complex numbers indexed by $S$. Given two sets $S$ and $S'$, we write
$(\alpha_{i,j})_{(i,j) \in S \times S'}$ for a matrix of complex numbers
indexed by $S$ and $S'$. The sets of indices can be omitted if there is no
ambiguity, as in Lemma~\ref{lem:normalised-sum}.

\begin{lemma}
	\label{lem:normalised-sum}
	Given a family of pairwise output compatible isometries $f_i \colon A \to B$,
	and a family of complex numbers $\alpha_i$ such that $\sum_i \abs{\alpha_i}^2
	= 1$, $\sum_i \alpha_i f_i$ is an isometry.
\end{lemma}
\begin{proof}
	\begin{align*} 
		&\ (\sum_i \alpha_i f_i)\dg \circ (\sum_j \alpha_j f_j) & \\
		&= (\sum_i \res\alpha_i f_i\dg) \circ (\sum_j \alpha_j f_j) 
		& (\text{dagger and sum commute}) \\
		&= \sum_{i,j} (\res\alpha_i \alpha_j) f_i\dg \circ f_j
		& (\text{composition and sum commute}) \\
		&= \sum_i (\res\alpha_i \alpha_i) f_i\dg \circ f_i
		& (\text{pairwise compatibility}) \\
		&= \left( \sum_i \res\alpha_i \alpha_i \right) \iid 
		& (\text{isometry}) \\
		&= \left( \sum_i \abs{\alpha_i}^2 \right) \iid = \iid. &
	\end{align*}
\end{proof}

We recall that given two maps $f \colon A \to C$ and $g \colon B \to C$ in $\Contr$, 
if $f\dg g = 0_{C,A}$, we say that $f$ and $g$ are orthogonal. We show that this
orthogonality is preserved by postcomposing with an isometry.

\begin{lemma}
	\label{lem:ortho-postiso}
	Given two orthogonal maps $f \colon A \to C$ and $g \colon B \to C$ in
	$\Contr$, and given an isometry $h \colon C \to D$, then $h \circ f \colon A
	\to D$ and $h \circ g \colon B \to D$ are also orthogonal.
\end{lemma}
\begin{proof}
	\begin{align*}
		&\ (h \circ f)\dg \circ h \circ g & \\
		&= f\dg \circ h\dg \circ h \circ g
		& (\text{dagger is contravariant}) \\
		&= f\dg \circ \iid_C \circ g = f\dg \circ g
		& (\text{isometry}) \\
		&= 0_{C,A}.
		& (\text{hypothesis})
	\end{align*}
\end{proof}

In a similar vein, the postcomposition of an isometry with an isometry is still
an isometry. This was already observed when we mentionned that Hilbert spaces
and isometries form a category.

The canonical countably-dimensional Hilbert space is $\ell^2(\N)$, defined in
\secref{sec:back-hilbert}. We recall that we write $\ket n$ for the elements
of the canonical basis in $\ell^2(\N)$. This is an abuse of notation, since the
symbols $\ket 0$ and $\ket 1$ are already used for the canonical basis of
$\C^2$. This is not an issue, since there is an isometric embedding $\C^2 \to
\ell^2(\N)$ which maps $\ket 0$ to $\ket 0$ and $\ket 1$ to $\ket 1$.

\begin{definition}
	We write $\rmsucc \colon \ell^2(\N) \to \ell^2(\N)$ for the linear map
	$\ell^2(\N) \to \ell^2(\N)$ which maps $\ket n$ to $\ket{n+1}$.
\end{definition}

\begin{remark}
	Note that $\rmsucc$ can also be seen as the image of the successor function
	in the natural numbers by the functor $\ell^2$. The linear map $\rmsucc$ is
	an isometry.
\end{remark}

\begin{example}
	\[
		\rmsucc \ket 7 = \ket 8
		\qquad
		\rmsucc \left(\frac{\sqrt 3}{2}\ket 9 + \frac 1 2 \ket{11} \right)
		= \frac{\sqrt 3}{2}\ket{10} + \frac 1 2 \ket{12}
	\]
\end{example}

\paragraph{Unitaries.}
The denotational semantics of our programming language involves unitary maps to
interpret the functions. Those maps live in the category $\Uni$, which is a rig
category: it has bifunctors $\oplus$ and $\otimes$ inherited from $\Hilb$.
Hence the next lemma.

\begin{lemma}
	Given two maps $f \colon A \to B$ and $g \colon C \to D$ in $\Uni$,
	$f \otimes g \colon A \otimes C \to B \otimes D$ is a map in $\Uni$
	and $f \oplus g \colon A \oplus C \to B \oplus D$ is a map in $\Uni$.
\end{lemma}
\begin{proof}
	Direct since the functors $\oplus$ and $\otimes$ are $\dagger$-functors.
\end{proof}

Finally, we present an operation that is common to quantum computing, and thus
preserves the unitary structure.

\begin{lemma}[Controlled unitary]
	Given a unitary map $f \colon A \to A$, there is a unitary map $\rmctrl_A(f)
	\colon (\C \oplus \C) \otimes A \to (\C \oplus \C) \otimes A$ such that
	$\rmctrl_A (f) = \ketbra 0 0 \otimes \iid + \ketbra 1 1 \otimes f$.
\end{lemma}
\begin{proof}
	Direct.
\end{proof}

\section{Denotational Semantics}
\label{sec:very-simple-semantics}

As usual, we write $\sem -$ for the interpretation of types and term
judgements. As mentioned in the previous section, the presentation makes
extensive use of contractions for the denotational semantics. However, values
and terms are directly announced to be isometries, for clarity. It will also
help us highlight the fact that values and terms represent sound quantum
states. In the same vein, the interpretation of unitaries is given as unitary
maps between two Hilbert spaces; but the proof that the semantics of a unitary
abstraction is unitary requires the mathematical development at the level of
contractions.

\subsection{Detailed presentation of the Semantics}
\label{sub:qua-detailed-semantics}

\paragraph{Types.}
The interpretation of a type $A$ is given by a countably-dimensional Hilbert
space. It is given by induction on the grammar of the types. This
interpretation is detailed in
Figure~\ref{fig:-very-simple-type-interpretation}.

\begin{figure}[!h]
	\begin{align*}
		\sem{A} &\colon \Hilb \\
		\sem{\one} &= \C \\
		\sem{A\otimes B} &= \sem{A} \otimes \sem{B} \\
		\sem{A\oplus B} &= \sem{A} \oplus \sem{B} \\
		\sem{\nat} &= \ell^2(\N)
	\end{align*}
	\caption{Interpretation of types.}
	\label{fig:-very-simple-type-interpretation}
\end{figure}

\paragraph{Expressions.}
We start with \emph{expressions}, whose typing rules are introduced in
Figure~\ref{fig:typing-values-simple}. Expressions are terms without unitary
application. The semantics of general terms in given below, once the semantics
of unitaries is defined. Judgements for expressions are first interpreted as
contractions between Hilbert spaces, and we then show that they are isometries.
A judgement is of the form $\Delta \entail e \colon A$, and its interpretation
is written $\sem{\Delta \entail e \colon A}$. Contexts $\Delta = x_1 \colon
A_1 \dots x_n \colon A_n$ are given a denotation $\sem\Delta = \sem{A_1}
\otimes \dots \otimes \sem{A_n}$. When it is not ambiguous, the interpretation
of the judgement $\Delta \entail e \colon A$ is written $\sem e$.

\begin{figure}
	\begin{align*}
		\sem{\Delta \entail e \colon A} &\colon \Iso(\sem	\Delta,\sem A) \\
		\sem{\entail * \colon \one} &= \iid_{\sem \one} \\
		\sem{x \colon A \entail x \colon A} &= \iid_{\sem A} \\
		\sem{\Delta \entail \inl e \colon A\oplus B} &= \iota_l^{\sem A, \sem B}
		\circ \sem{\Delta \entail e\colon A} \\
		\sem{\Delta \entail \inr e \colon A\oplus B} &= \iota_r^{\sem A, \sem B}
		\circ \sem{\Delta \entail e\colon A} \\
		\sem{\Delta_1, \Delta_2 \entail \pair{e}{e'}\colon A\otimes B} &=
		\sem{\Delta_1 \entail e\colon A} \otimes \sem{\Delta_2 \entail
		e'\colon B} \\
		\sem{\entail \zero \colon \nat} &= \ket 0 \\
		\sem{\Delta \entail \suc e \colon \nat} &= \rmsucc \circ \sem{\Delta
		\entail e \colon \nat} \\
		\sem{\Delta \entail \Sigma_{i\leq k} (\alpha_i \cdot  e_i) \colon A} &= \sum_{i\leq
		k} \alpha_i \sem{\Delta \entail e_i \colon A}
	\end{align*}
	\caption{Interpretation of expression judgements as morphisms in $\Iso$.}
	\label{fig:very-simple-value-interpretation}
\end{figure}

\begin{lemma}[Isometry]
	\label{lem:value-isometry}
	If $\Delta \entail e \colon A$ is a well-formed expression judgement, then
	$\sem{\Delta \entail e \colon A}$ is an isometry.
\end{lemma}
\begin{proof}
	Given later with the semantics of terms in general, see
	Lemma~\ref{lem:term-isometry}.
\end{proof}

In quantum physics, the state of a particle is usually described as an
isometry. Showing that our expressions are interpreted as isometries, we can
justify that they are correct \emph{quantum states}. The proof of the previous
lemma is included in one of a larger result, showing that the denotation of all
terms are isometries (see Lemma~\ref{lem:term-isometry}). Moreover, expressions
are used to define the unitary abstractions as a collection of patterns: it is
sensible to prove that these patterns are interpreted with compatible
morphisms, in the sense of Definition~\ref{def:compati}.

\begin{lemma}
	\label{lem:orthogonal-semantics-ortho-value}
	Given two judgements $\Delta_1 \entail e_1 \colon A$ and
	$\Delta_2 \entail e_2 \colon A$, such that $e_1~\bot~e_2$, we have
	$\sem{e_1}\dg \circ \sem{e_2} = 0$.
\end{lemma}
\begin{proof}
	The proof is done by induction on the derivation of $\bot$. It is
	a subproof of the one for Lemma~\ref{lem:orthogonal-semantics-ortho}.
\end{proof}

This result can also be stated for the predicate ${\rm OD}$, with the help of
Lemma~\ref{lem:od-imp-bot}, where it is shown that two values in an orthogonal
decomposition are orthogonal.

\begin{lemma}
	\label{lem:orthogonal-semantics}
	Given two judgements $\Delta_1 \entail e_1 \colon A$ and
	$\Delta_2 \entail e_2 \colon A$, a set of $S$ that contains $e_1$ and $e_2$
	and such that $\ODe A S$, we have $\sem{e_1}\dg \circ
	\sem{e_2} = 0$.
\end{lemma}

An important property of an orthonormal basis in a Hilbert space is the
\emph{resolution of the identity}. We show that, given $\ODe A S$, a similar
property holds. This conforts us in calling $S$ a \emph{syntactic basis}.

\begin{lemma}
	\label{lem:towards-unitary}
	Given $\ODe{A}{\set{e_i}}_{i \leq n}$, and $\Delta_i \entail e_i \colon A$
	for all $i$, we have 
	\[
		\sum_{i \leq n} \sem{e_i} \circ \sem{e_i}\dg =
		\iid_{\sem A}.
	\]
\end{lemma}
\begin{proof}
	The proof is done by induction on ${\rm OD}$.
	\begin{itemize}
		\item $\ODe{A}{\set x}$.
			$\sem x \circ \sem x \dg = \iid_{\sem A} \circ \iid_{\sem A} = \iid_{\sem A}$.
		\item $\ODe{\one}{\set *}$. 
			$\sem * \circ \sem * \dg = \iid_{\sem \one} \circ \iid_{\sem \one} = \iid_{\sem \one}$.
		\item $\ODe{A \oplus B}{S \boxplus T}$. 
			\begin{align*}
				&\ \sum_{e \in S \boxplus T} \sem e \circ \sem e \dg & \\
				&= \sum_{s \in S} \sem{\inl s} \circ \sem{\inl s}\dg
				+
				\sum_{t \in T} \sem{\inr t} \circ \sem{\inr t}\dg
				&\text{(by definition)} \\
				&= \iota_l \circ \left( \sum_{s \in S} \sem s \circ \sem s \dg \right) \circ \iota_l\dg
				+
				\iota_r \circ \left( \sum_{t \in T} \sem t \circ \sem t \dg \right) \circ \iota_r\dg
				&\text{(by linearity)} \\
				&= \iota_l \circ \iid_{\sem A} \circ \iota_l\dg
				+
				\iota_r \circ \iid_{\sem B} \circ \iota_r\dg
				&\text{(by IH)} \\
				&= \iota_l \iota_l\dg + \iota_r \iota_r\dg = \iid_{\sem{A \oplus B}}.
				&\text{(Ex.~\ref{ex:qua-iota})} 
			\end{align*}
		\item $\ODe{A \otimes B}{S}$. Suppose that $\ODe{A}{\pi_1(S)}$ and
			$\ODe{B}{S^1_e}$ for all $e \in \pi_1(S)$.  
			\begin{align*}
 				&\ \sum_{(e \otimes e') \in S} \sem{e \otimes e'} \circ \sem{e \otimes e'}\dg & \\
 				&= \sum_{(e \otimes e') \in S} (\sem{e} \otimes \sem{e'}) \circ
				(\sem{e} \otimes \sem{e'}\dg)
 				&\text{(by definition)} \\
 				&= \sum_{(e \otimes e') \in S} (\sem{e} \circ \sem e \dg) \otimes 
 				(\sem{e'} \circ \sem {e'}\dg )
 				&\text{(by monoidal } \dagger\text{-category)} \\
 				&= \sum_{e \in \pi_1(S)} (\sem{e} \circ \sem e \dg)
 				\otimes \left( \sum_{e' \in S^1_b} \sem{e'} \circ \sem {e'}\dg \right) & \\
 				&= \sum_{e \in \pi_1(S)} (\sem{e} \circ \sem e \dg)
 				\otimes \iid_{\sem B}
 				&\text{(by IH)} \\
 				&= \left( \sum_{e \in \pi_1(S)} \sem{e} \circ \sem e \dg \right)
 				\otimes \iid_{\sem B}
 				&\text{(by linearity)} \\
 				&= \iid_{\sem A}
 				\otimes \iid_{\sem B}
 				= \iid_{\sem{A \otimes B}}. 
 				&\text{(by IH)} 
			\end{align*}
		\item $\ODe{\nat}{S^{\oplus 0}}$. 
			\begin{align*}
				&\ \sum_{e \in S^{\oplus 0}} \sem e \circ \sem e \dg & \\
				&= \sem\zero \circ \sem\zero +
				\sum_{s \in S} \sem{\suc s} \circ \sem{\suc s}\dg
 				&\text{(by definition)} \\
				&= \sem\zero \circ \sem\zero\dg +
				\rmsucc \circ \left( \sum_{s \in S} \sem s \circ \sem s \dg \right)
				\circ \rmsucc \dg
 				&\text{(by linearity)} \\
				&= \sem\zero \circ \sem\zero\dg +
				\rmsucc \circ \iid_{\sem\nat} \circ \rmsucc \dg
				= \iid_{\sem\nat}.
 				&\text{(by IH)} 
			\end{align*}
		\item $\ODe{A}{S^\alpha}$. 
			\begin{align*}
				&\ \sum_{e \in S^\alpha} \sem e \circ \sem e \dg & \\
				&= \sum_{s \in S}
				\sem{\Sigma_{s' \in S} (\alpha_{s,s'} \cdot s')}
				\circ
				\sem{\Sigma_{s' \in S} (\alpha_{s,s'} \cdot s')} \dg
 				&\text{(by definition)} \\
				&= \sum_{s \in S}
				\left(\sum_{s' \in S} \alpha_{s,s'} \sem{s'} \right)
				\circ
				\left(\sum_{s'' \in S} \alpha_{s,s''} \sem{s''} \right) \dg
 				&\text{(by definition)} \\
				&= \sum_{s \in S}
				\sum_{s',s'' \in S}
				\alpha_{s,s'} \res{\alpha_{s,s''}}
				\sem{s'} \circ \sem{s''} \dg
 				&\text{(by linearity)} \\
				&= \sum_{s',s'' \in S}
				\left( \sum_{s \in S} \alpha_{s,s'} \res{\alpha_{s,s''}} \right)
				\sem{s'} \circ \sem{s''} \dg & \\
				&= \sum_{s',s'' \in S}
				\delta_{s' = s''}
				\sem{s'} \circ \sem{s''} \dg 
				= \sum_{s' \in S}
				\sem{s'} \circ \sem{s'} \dg 
				&\text{(by unitarity)} \\
				&= \iid_{\sem A}.
 				&\text{(by IH)}
			\end{align*}
	\end{itemize}
\end{proof}

One final development on the interpretation of values is the link
with substitutions, detailed in the next proposition.

\begin{proposition}
	\label{prop:substitution-interpretation}
	Given a well-typed term $\Delta \entail t \colon A$ and for all $(x_i \colon
	A_i) \in \Delta$, a well-typed expression $~\entail e_i \colon A_i$; if $\sigma =
	\set{x_i \mapsto e_i}_i$, then:
	\[
		\sem{\entail \sigma(t) \colon A} = \sem{\Delta \entail t \colon A} \circ
		\left(\bigotimes_i \sem{\entail e_i \colon A_i}\right).
	\]
	We define then $\sem\sigma = \bigotimes_i \sem{\entail e_i \colon A_i}$.
\end{proposition}
\begin{proof}
	The proof is straightforward by induction on the typing rules for $t$.
\end{proof}

\begin{remark}
	The definition of the interpretation of a substitution above is somewhat
	informal. It would require a lot of care and unnecessary details to make
	the denotation of $\sigma$ fit the denotation of a particular context
	$\Delta$. Since we are working in symmetric monoidal categories, those
	details will be overlooked when working with substitutions. We assume that we
	work up to permutations, and that when an interpretation of a substitution is
	involved, it is with the right permutation.
\end{remark}

Substitutions $\sigma$ emerge from the matching of two basis values, thus
we can prove that the interpretation of the matching gives the interpretation
of the substitution, as stated in the next lemma.

\begin{lemma}
	\label{lem:matching-semantics}
	Given two well-typed basis values $\Delta \entail b \colon A$ and $~\entail
	b' \colon A$, and a substitution $\sigma$, if $\match{\sigma}{b}{b'}$ then
	$\sem b\dg \circ \sem{b'} = \sem\sigma$.
\end{lemma}
\begin{proof}
	The proof is straightforward by induction on $\match{\sigma}{b}{b'}$ (see
	\secref{sub:qua-substitution} for the definition).
\end{proof}

\paragraph{Unitaries.}
The type of unitaries are given as $A \iso B$, and they are first interpreted
as morphisms $\sem A \to \sem B$ in $\Contr$, before showing that their
interpretation actually lies in $\Uni$. We also show that the ${\rm OD}$
conditions ensure that the denotation of (syntactic) untaries is not only a
contractive map, but a unitary between Hilbert spaces. Working with
contractions is necessary to use the notion of compatibility: given a unitary
$\unibasique \colon A \iso B$, we provide an interpretation to each clause
$b_i \iso e_i$ as a contraction $\sem A \to \sem B$, and prove that all
the contractions thus obtained are compatible, and can be summed. Unitary
judgments are of the form $\entailiso \omega \colon A \iso B$, and their
semantics is given by a morphism in $\Uni$:
\[ \sem{\entailiso \omega \colon A \iso B} \colon
\Uni(\sem A,\sem B). \]

Given $\entailiso \unibasique \colon A \iso B$, the interpretation of a clause
$b_i \iso e_i$ is the following contraction: $\sem{\Delta_i \entail e_i \colon
B} \circ \sem{\Delta_i \entail b_i \colon A}\dg$. It should be read as follows:
if the input of type $A$ matches with $b_i$, it provides a substitution through
$\Delta_i$, that is applied to $e_i$. This is better understood through a
diagram:
\[ \begin{tikzcd}
	\sem A && \sem{\Delta_i} && \sem B
	\arrow["\sem{\Delta_i \entail b_i \colon A}\dg", from=1-1, to=1-3]
	\arrow["\sem{\Delta_i \entail e_i \colon B}", from=1-3, to=1-5]
\end{tikzcd} \]

The interpretation of a unitary abstraction is then: \[\sem{\entailiso
\unibasique \colon A\iso B} = \sum_{i\leq n} \sem{\Delta_i \entail e_i \colon
B} \circ \sem{\Delta_i \entail b_i \colon A}\dg.\]

It is left to prove that it is well-defined, and then that it is a proper
unitary operation.

\begin{corollary}
	\label{cor:iso-sem-defined}
	Given $~\entailiso \unibasique \colon A\iso B$,
	its interpretation
	$\sem{\entailiso \unibasique \colon A\iso B}$ 
	is a well-defined morphism in $\Contr$.
\end{corollary}
\begin{proof}
	Given $~\entailiso \unibasique \colon A\iso B$, we know that
	$\ODe{B}{(\set{e_i}_{i\leq n})}$ and $\OD{A}{(\set{b_i}_{i\leq n})}$ hold,
	and for all $i \leq n$, $\Delta_i \entail b_i \colon A$ and $\Delta_i \entail
	e_i \colon B$.

	Since $\OD{A}{(\set{b_i}_{i\leq n})}$ holds, Lemma~\ref{lem:orthogonal-semantics}
	ensures that for all $i \neq j\leq k$, $\sem{b_i}\dg \circ
	\sem{b_j} = 0_{\sem{\Delta_j},\sem{\Delta_i}}$. The same lemma with
	$\ODe{B}{(\set{e_i}_{i\leq n})}$ ensures that for all $i\neq j \leq n$,
	$\sem{e_i}\dg \circ \sem{e_j} =
	0_{\sem{\Delta_j},\sem{\Delta_i}}$. This proves that, for all $i \neq j \leq
	n$, $(\sem{e_i} \circ \sem{b_i}\dg)\dg \circ \sem{e_j}
	\circ \sem{b_j}\dg = 0_{\sem{\Delta_j},\sem{\Delta_i}}$ and
	$\sem{e_i} \circ \sem{b_i}\dg \circ (\sem{e_j}
	\circ \sem{b_j}\dg)\dg = 0_{\sem{\Delta_j},\sem{\Delta_i}}$. This
	proves that for all $i \neq j \leq n$, $\sem{e_i} \circ
	\sem{b_i}\dg$ and $\sem{e_j} \circ \sem{b_j}\dg$
	are compatible, thanks to Lemma~\ref{lem:algebra-compati}. Then,
	Lemma~\ref{lem:compati} ensures that $\sum_{i\leq n} \sem{e_i} \circ
	\sem{b_i}\dg$ is a contraction.
\end{proof}

\begin{theorem}
	Given $~\entailiso \unibasique \colon A \iso B$, its interpretation 
	$\sem{\entailiso \unibasique \colon A \iso B}$ is unitary.
\end{theorem}
\begin{proof}
	Given $~\entailiso \unibasique \colon A\iso B$, we know that
	$\ODe{B}{(\set{e_i}_{i\leq n})}$ and $\OD{A}{(\set{b_i}_{i\leq n})}$ hold,
	and for all $i \leq n$, $\Delta_i \entail b_i \colon A$ and $\Delta_i \entail
	e_i \colon B$.

	First, we prove that $\sem\omega\dg \circ \sem\omega = \iid_{\sem A}$,
	with $\omega = \unibasique$.
	\begin{align*}
		&\ \sem\omega\dg \circ \sem\omega & \\
		&= \left( \sum_{i\leq n} \sem{e_i} \circ \sem{b_i}\dg \right)\dg \circ
		\sum_{j\leq n} \sem{e_j} \circ \sem{b_j}\dg
		& \text{(by definition)} \\
		&= \sum_{i\leq n} (\sem{e_i} \circ \sem{b_i}\dg)\dg \circ
		\sum_{j\leq n} \sem{e_j} \circ \sem{b_j}\dg
		& \text{(dagger distributes over sum)} \\
		&= \sum_{i\leq n} \sem{b_i} \circ \sem{e_i}\dg \circ
		\sum_{j\leq n} \sem{e_j} \circ \sem{b_j}\dg
		& \text{(dagger is contravariant)} \\
		&= \sum_{i,j \leq n} \sem{b_i} \circ \sem{e_i}\dg \circ \sem{e_j} \circ
		\sem{b_j}\dg 
		&\text{(linearity)} \\
		&= \sum_{i \leq n} \sem{b_i} \circ \sem{e_i}\dg \circ \sem{e_i} \circ
		\sem{b_i}\dg 
		&\text{(Lemma~\ref{lem:orthogonal-semantics})} \\
		&= \sum_{i \leq n} \sem{b_i} \circ
		\sem{b_i}\dg 
		&\text{(Lemma~\ref{lem:value-isometry})} \\
		&= \iid_{\sem A}
		&\text{(Lemma~\ref{lem:towards-unitary})}
	\end{align*}
	The other direction $\sem\omega \circ \sem\omega\dg = \iid_{\sem B}$ is similar.
\end{proof}

Note that we have only proven so far that unitary abstractions have a sound
denotational semantics in $\Uni$. The interpretation of operations on unitaries
is given in Figure~\ref{fig:unitary-semantics}. It is explained in
\secref{sec:qu-simple-maths} why this interpretation is in $\Uni$, although
this does not come as a surprise.

\begin{figure}[!h]
	\begin{align*}
		\sem{\entailiso \omega \colon A \iso B} &\colon \Uni(\sem A, \sem B) \\
		\sem{\unibreduit} &= \sum_{i \in I} \sem{e_i} \circ \sem{b_i}\dg \\
		\sem{\omega_2 \circ \omega_1} &= \sem{\omega_2} \circ \sem{\omega_1} \\
		\sem{\omega_1 \otimes \omega_2} &= \sem{\omega_1} \otimes \sem{\omega_2} \\
		\sem{\omega_1 \oplus \omega_2} &= \sem{\omega_1} \oplus \sem{\omega_2} \\
		\sem{\omega\inv} &= \sem{\omega}\dg \\
		\sem{\ctrl \omega} &= {\rm ctrl}_{\sem A} (\sem\omega)
	\end{align*}
	\caption{Interpretation of unitaries in $\Uni$.}
	\label{fig:unitary-semantics}
\end{figure}

\paragraph{Terms.}
One remaining term is the application of a unitary. 
\[\sem{\Delta \entail \omega~t \colon B} = \sem{\entailiso \omega \colon A\iso
B} \circ \sem{\Delta \entail t \colon A}.\] 
The interpretation of all term judgements is found in
Figure~\ref{fig:simple-term-interpretation}. 

\begin{figure}[!h]
	\begin{align*}
		\sem{\Delta \entail t \colon A} &\colon \Iso(\sem	\Delta,\sem A) \\
		\sem{\entail * \colon \one} &= \iid_{\sem \one} \\
		\sem{x \colon A \entail x \colon A} &= \iid_{\sem A} \\
		\sem{\Delta \entail \inl t \colon A\oplus B} &= \iota_l^{\sem A, \sem B}
		\circ \sem{\Delta \entail t\colon A} \\
		\sem{\Delta \entail \inr t \colon A\oplus B} &= \iota_r^{\sem A, \sem B}
		\circ \sem{\Delta \entail t\colon A} \\
		\sem{\Delta_1, \Delta_2 \entail \pair{t}{t'}\colon A\otimes B} &=
		\sem{\Delta_1 \entail t\colon A} \otimes \sem{\Delta_2 \entail
		t'\colon B} \\
		\sem{\entail \zero \colon \nat} &= \ket 0 \\
		\sem{\Delta \entail \suc t \colon \nat} &= \rmsucc \circ \sem{\Delta
		\entail t \colon \nat} \\
		\sem{\Delta \entail \Sigma_{i\leq k} (\alpha_i \cdot t_i) \colon A} &=
		\sum_{i\leq k} \alpha_i \sem{\Delta \entail t_i \colon A} \\
		\sem{\Delta \entail \omega~t \colon B} &= \sem{\entailiso \omega \colon
		A\iso B} \circ \sem{\Delta \entail t \colon A}
	\end{align*}
	\caption{Interpretation of term judgements as morphisms in $\Iso$.}
	\label{fig:simple-term-interpretation}
\end{figure}

We can already show that this interpretation of terms is sound with the sketch
of operational semantics given in the previous section.

\begin{proposition}[Operational Soundness]
	\label{prop:qua-op-soundness}
	Given a well-typed unitary abstraction $~\entailiso \unibasique \colon A
	\iso B$ and a well-typed basis value $~\entail b' \colon A$, if
	$\match{\sigma}{b_i}{b'}$, then 
	\[
		\sem{\entail \unibasique~b' \colon B} 
		= \sem{\entail \sigma(e_i) \colon B}.
	\]
\end{proposition}
\begin{proof}
	First, we deduce from the assumption $\match{\sigma}{b_i}{b'}$ that
	\begin{itemize}
		\item $\sem{b_i}\dg \circ \sem{b'} = \sem\sigma$, thanks to
			Lemma~\ref{lem:matching-semantics}.
		\item for all $j \neq i$, $b_j~\bot~b'$, and thus $\sem{b_j}\dg \circ
			\sem{b'} = 0$, thanks to
			Lemma~\ref{lem:orthogonal-semantics-ortho-value}.
	\end{itemize}
	We can then compute the semantics, with $\omega \defeq \unibasique$:
	\begin{align*}
		&\ \sem{\omega~b'} & \\
		&= \sem\omega \circ \sem{b'} & \text{(by definition)} \\
		&= \left( \sum_j \sem{e_j} \circ \sem{b_j}\dg \right) \circ \sem{b'}
		& \text{(by definition)} \\
		&= \sum_j \sem{e_j} \circ \sem{b_j}\dg \circ \sem{b'}
		& \text{(linearity)} \\
		&= \sem{e_i} \circ \sem{b_i}\dg \circ \sem{b'}
		& \text{(Lemma~\ref{lem:orthogonal-semantics-ortho-value})} \\
		&= \sem{e_i} \circ \sem\sigma
		& \text{(Lemma~\ref{lem:matching-semantics})} \\
		&= \sem{\sigma(e_i)}
		& \text{(Prop.~\ref{prop:substitution-interpretation})}
	\end{align*}
\end{proof}

We prove that, like expressions, terms are indeed interpreted as isometries,
reinforcing the link with quantum physics. This requires several lemmas. The
first lemma shows that the denotational orthogonality is preserved by linear
combinations.

\begin{lemma}
	\label{lem:sum-ortho}
	Given $\Delta_1 \entail t \colon A$ and $\Delta_2 \entail \Sigma_i (\alpha_i
	\cdot t_i) \colon A$ such that for all $i$, $\sem{\Delta_1 \entail t \colon
	A}\dg \circ \sem{\Delta_2 \entail t_i \colon A} =
	0_{\sem{\Delta_2},\sem{\Delta_1}}$; then $\sem{\Delta_1 \entail t \colon
	A}\dg \circ \sem{\Delta_2 \entail \Sigma_i (\alpha_i \cdot t_i) \colon A} =
	0_{\sem{\Delta_2},\sem{\Delta_1}}$.
\end{lemma}
\begin{proof}
	The proof involves few steps, without surprises.
	\begin{align*}
		&\ \sem{\Delta_1 \entail t \colon A}\dg
		\circ \sem{\Delta_2 \entail \Sigma_i (\alpha_i \cdot t_i) \colon A}
		& \\
		&= \sem{\Delta_1 \entail t \colon A}\dg \circ
		\sum_i \alpha_i \sem{\Delta_2 \entail t_i \colon A}
		& \text{(by definition)} \\
		&= \sum_i \alpha_i \sem{\Delta_1 \entail t \colon A}\dg \circ \sem{\Delta_2
		\entail t_i \colon A}
		& \text{(by linearity)} \\
		&= \sum_i \alpha_i\, 0_{\sem{\Delta_2},\sem{\Delta_1}}
		& \text{(hypothesis)} \\
		&= 0_{\sem{\Delta_2},\sem{\Delta_1}}. &
	\end{align*}
\end{proof}

The next lemma states that syntactic orthogonality (defined in
Definition~\ref{def:orthogonality}) implies denotational orthogonality.

\begin{lemma}
	\label{lem:orthogonal-semantics-ortho}
	Given two judgements $\Delta_1 \entail t_1 \colon A$ and
	$\Delta_2 \entail t_2 \colon A$, such that $t_1~\bot~t_2$, we have
	$\sem{t_1}\dg \circ \sem{t_2} = 0_{\sem{\Delta_2},\sem{\Delta_1}}$.
\end{lemma}

The proof of this lemma is interdependent with the proof of the next lemma,
where it is proven that the interpretations of term judgements are normalised
quantum states -- namely, isometries.

\begin{lemma}[Isometry]
	\label{lem:term-isometry}
	If $\Delta \entail t \colon A$ is a well-formed judgement, then $\sem{\Delta
	\entail t \colon A}$ is an isometry.
\end{lemma}
\begin{proof}[Proof of Lemma~\ref{lem:orthogonal-semantics-ortho} and 
	Lemma~\ref{lem:term-isometry}]
	We prove both theorems together, because they are interdependent. We recall
	the whole statement: given a well-formed judgement $\Delta_1 \entail t_1
	\colon A$,
	\begin{itemize}
		\item $\sem{\Delta_1 \entail t_1 \colon A}\dg \circ \sem{\Delta_1 \entail t_1
			\colon A} = \iid_{\sem{\Delta_1}}$;
		\item for all $\Delta_2 \entail t_2 \colon A$ such that
			$t_1 \bot t_2$, $\sem{\Delta_1 \entail t_1 \colon A}\dg \circ \sem{\Delta_2
			\entail t_2 \colon A} = 0_{\sem{\Delta_2},\sem{\Delta_1}}$.
	\end{itemize}
	We prove by induction on the derivation of $\Delta_1 \entail t_1 \colon A$.
	\begin{itemize}
		\item $~\entail * \colon \one$. $\sem * \dg \circ \sem * = \iid_{\sem \one} \circ
			\iid_{\sem \one} = \iid_{\sem \one}$. This term is not orthogonal to any other
			well-typed term.
		\item $x \colon A \entail x \colon A$. 
		$\sem x \dg \circ \sem x = \iid_{\sem A} \circ
			\iid_{\sem A} = \iid_{\sem A}$. This term is not orthogonal to any other.
		\item $\Delta_1, \Delta_2 \entail t_1 \otimes t_2 \colon A \otimes B$.
			We first prove that its denotation is an isometry.
			\begin{align*}
				&\ \sem{t_1 \otimes t_2}\dg \circ \sem{t_1 \otimes t_2} & \\
				&= (\sem{t_1} \otimes \sem{t_2})\dg \circ (\sem{t_1} \otimes \sem{t_2})
				& \text{(by definition)} \\
				&= (\sem{t_1}\dg \otimes \sem{t_2}\dg) \circ (\sem{t_1} \otimes \sem{t_2})
				& \text{(dagger is a monoidal functor)} \\
				&= (\sem{t_1}\dg \circ \sem{t_1}) \otimes (\sem{t_2}\dg \circ \sem{t_2}) 
				& (\otimes \text{ is monoidal)} \\
				&= \iid_{\sem{\Delta_1}} \otimes \iid_{\sem{\Delta_2}}
				& \text{(IH)} \\
				&= \iid_{\sem{\Delta_1} \otimes \sem{\Delta_2}} = \iid_{\sem{\Delta_1
				, \Delta_2}} &
			\end{align*}
			Then we show that, if it is orthogonal to any other well-typed term, say
			$\Delta_3 \entail t_3 \colon C$, then their interpretations are also
			orthogonal. We reason on a case by case basis. The term $t_3$ can be of
			the form $t'_3 \otimes t''_3$, in which case, either $t_1~\bot~t'_3$ or
			$t_2~\bot~t''_3$. In both cases, the result is direct, because the zero
			morphism tensored with any other morphism is still zero. The other cases
			is $t_3$ being a linear combination; this case is covered by
			Lemma~\ref{lem:sum-ortho}.
		\item $\Delta \entail \ini t \colon A_1 \oplus A_2$.
			The denotation of this term is an isometry because the injections $\iota$
			also are and a composition of isometries keeps being an isometry. Now,
			given another term $\Delta_1 \entail t_2 \colon B$ such that
			$\ini t~\bot~t_2$, we have several cases:
			\begin{itemize}
				\item either $t_2$ is of the form $\ini t'_2$ with the same injection
					as $\ini t$, and the result is given by
					Lemma~\ref{lem:ortho-postiso};
				\item or $\ini t = \inl t$ and $t_2 = \inr t'_2$, in which case the conclusion
					is direct thanks to Lemma~\ref{lem:injection-compati} and the induction hypothesis;
				\item or $t_2$ is a linear combination, and Lemma~\ref{lem:sum-ortho}
					concludes.
			\end{itemize}
		\item $~\entail \zero \colon \nat$. $\sem\zero\dg \circ \sem\zero = \braket
			0 0 = 1 = \iid_{\sem \one}$. Moreover, given any natural number $n$,
			$\braket{0}{n+1} = 0$. Lemma~\ref{lem:sum-ortho} concludes for linear
			combinations.
		\item $\Delta \entail \suc t \colon \nat$. The interpretation is an
			isometry by composition of isometries. The orthogonality part is either
			ensured with the last point, or by Lemma~\ref{lem:ortho-postiso} and
			the induction hypothesis.
		\item $\Delta \entail \Sigma_i (\alpha_i \cdot t_i) \colon A$.  The
			interpretation is an isometry by induction hypothesis and
			Lemma~\ref{lem:normalised-sum}. Moreover, the only case of orthogonality
			not yet covered is $t_2 = \Sigma_k (\beta_k \cdot t_k)$ where for all $j
			\neq k$, $t_j \bot t_k$ and $\sum \alpha_j \res\beta_k = 0$.

			given that for all $i\neq j \in I$, $t_i~\bot~t_j$, $J,K
			\subseteq I$, and $\sum_{i \in J \cap K} \bar{\alpha_i}\beta_i = 0$:
			\begin{align*}
				&\ \sem{\Sigma_j (\alpha_j \cdot t_j)}\dg 
				\circ \sem{\Sigma_k (\beta_k \cdot t_k)} & \\
				&= \left( \sum_j \alpha_j \sem{t_j} \right)\dg \circ
				\sum_k \beta_k \sem{t_k}
				& \text{(by definition)} \\
				&= \left( \sum_j \res\alpha_j \sem{t_j}\dg \right) \circ
				\sum_k \beta_k \sem{t_k}
				& \text{(dagger distributes over the sum)} \\
				&= \sum_{j,k} \res\alpha_j \beta_k \sem{t_j}\dg \circ \sem{t_k}
				& \text{(by linearity)} \\
				&= \sum_{i \in J \cap K} \res\alpha_i \beta_i \sem{t_i}\dg \circ
				\sem{t_i}
				& \text{(by induction hypothesis)} \\
				&= \sum_{i \in J \cap K} \res\alpha_i \beta_i \, \iid_{\sem\Delta}
				& \text{(by induction hypothesis)} \\
				&= \left( \sum_{i \in J \cap K} \res\alpha_i \beta_i \right) \iid_{\sem\Delta}
				= 0_{\sem{\Delta_2},\sem{\Delta_1}} &
			\end{align*}
		\item $\Delta \entail \omega~t \colon B$. The interpretation is an isometry
			by composition of isometries. The orthogonality is covered by either
			Lemma~\ref{lem:ortho-postiso}, because a unitary is in particular an
			isometry, or Lemma~\ref{lem:sum-ortho}, similarly to the previous points.
	\end{itemize}
\end{proof}

Providing an interpretation that fits the expectations from quantum physics is
meaningful, but not completely satisfying. We prove a stronger link in the
coming section between the syntax and the semantics through the equational
theory, in the vein of the previous chapter.

\subsection{Completeness}
\label{sub:qua-completeness}

We prove a strong link between the denotational semantics and the equational
theory, namely that an equality statement in one is also an equality statement
in the other. We start by showing soundness, meaning that two terms equal in
the equational theory have the same denotational interpretation.

\begin{proposition}
	\label{prop:qua-soundness}
	Given $\Delta \entail t_1 = t_2 \colon A$, then $\sem{\Delta \entail t_1
	\colon A} = \sem{\Delta \entail t_2 \colon A}$.
\end{proposition}
\begin{proof}
	By induction on the rules of the equational theory. The only non-trivial case
	is done within Proposition~\ref{prop:qua-op-soundness}.
\end{proof}

We prove then completeness, starting with terms without unitary functions.

\begin{lemma}[Completeness of values]
	\label{lem:value-completeness}
	Given $~\entail v_1 \colon A$ and $~\entail v_2 \colon A$, if
	$\sem{\entail v_1 \colon A} = \sem{\entail v_2 \colon A}$, then
	$~\entail v_1 = v_2 \colon A$.
\end{lemma}
\begin{proof}
	This is proven by induction on $~\entail v_1 \colon A$.
	\begin{itemize}
		\item The cases of $*$ and $\zero$ are straightforward.
		\item If $v_1 = b_1 \otimes b'_1$ with type $A \otimes B$, then also $v_2 =
			b_2 \otimes b'_2$, and $\sem{b_1} \otimes \sem{b'_1} = \sem{b_2} \otimes
			\sem{b'_2}$, thus $\sem{b_1} = \sem{b_2}$ and $\sem{b'_1} = \sem{b'_2}$,
			the induction hypothesis ensures that $~\entail b_1 = b_2 \colon A$ and
			$~\entail b'_1 = b'_2 \colon B$, and thus $~\entail b_1 \otimes b'_1 =
			b_2 \otimes b'_2 \colon A \otimes B$, since we are working with basis
			values.
		\item If $v_1 = \ini b_1$ of type $A_1 \oplus A_2$, with $\sem{v_1} =
			\sem{v_2}$, necessarily $v_2 = \ini b_2$, and $\iota^{A,B}_l \circ
			\sem{b_1} = \iota^{A,B}_l \circ \sem{b_2}$, thus $(\iota^{A,B}_l)\dg
			\circ \iota^{A,B}_l \circ \sem{b_1} = (\iota^{A,B}_l)\dg \circ
			\iota^{A,B}_l \circ \sem{b_2}$ which ends with $\sem{b_1} = \sem{b_2}$,
			and the induction hypothesis gives that $~\entail b_1 = b_2 \colon A_i$,
			and thus $~\entail \ini b_1 = \ini b_2 \colon A_i$.
		\item Else, with type $A$, $v_1 = \Sigma_i (\alpha_i \cdot b^1_i)$ with the
			$(b^1_i)$ that are pairwise orthogonal, and $v_2 = \Sigma_j (\beta_j
			\cdot b^2_j)$ with the $(b^2_j)$ that are also pairwise orthogonal. We
			know that \[ \sem{v_1} = \sum_i \alpha_i \sem{b^1_i} = \sum_j \beta_j
			\sem{b^2_j} = \sem{v_2}. \]
			Thus, for all $~\entail b \colon A$, $\sem b\dg \circ \sum_i \alpha_i
			\sem{b^1_i} = \sem b \dg \circ \sum_j \beta_j \sem{b^2_j}$; by
			orthogonality, we have $\sem b\dg \circ \sum_i \alpha_i \sem{b^1_i} =
			\alpha_k \sem{b^1_k}$ for some $k$, and $\sem b\dg \circ \sum_j \beta_j
			\sem{b^2_j} = \beta_{k'} \sem{b^2_{k'}}$ for some $k'$. Thus, $\alpha_k
			\sem{b^1_k} = \beta_{k'} \sem{b^2_{k'}}$, and because they are basis
			values, we have $\alpha_k = \beta_{k'}$ and  $\sem{b^1_k} =
			\sem{b^2_{k'}}$, and the induction hypothesis ensures that $~\entail
			b^1_k = b^2_{k'} \colon A$. Note that this is done for all $~\entail b
			\colon A$, and thus $~\entail v_1 = v_2 \colon A$. 
	\end{itemize}
\end{proof}

We then prove the final result of this section -- namely, completeness on
closed terms --, meaning that the equality statement of the equational theory
and of the denotational semantics on closed terms are equivalent.

\begin{theorem}[Completeness]
	\label{th:qu-completeness}
	\(
		~\entail t_1 = t_2 \colon A 
		\text{ iff } 
		\sem{\entail t_1 \colon A} = \sem{\entail t_2 \colon A}.
	\)
\end{theorem}
\begin{proof}
	We prove both directions.
	\begin{itemize}
		\item For the implication $~\entail t_1 = t_2 \colon A$ to $\sem{\entail
			t_1 \colon A} = \sem{\entail t_2 \colon A}$, see
			Proposition~\ref{prop:qua-soundness}.
		\item The other direction uses Theorem~\ref{th:equational-sn}, the
			equivalent of strong normalisation, that gives $v_1$ and $v_2$ such that
			$~\entail t_1 = v_1 \colon A$ and $~\entail t_2 = v_2 \colon A$. Our
			hypothesis is that $\sem{t_1} = \sem{t_2}$, and thus $\sem{v_1} =
			\sem{t_1} = \sem{t_2} = \sem{v_2}$. Lemma~\ref{lem:value-completeness}
			gives then that $~\entail v_1 = v_2 \colon A$ and transitivity ensures
			that $~\entail t_1 = t_2 \colon A$.
	\end{itemize}
\end{proof}

\section{Discussion and conclusion}
\label{sec:qu-simple-conclusion}

This section concludes with personal comments on the choices made. We discuss
the parts of the language presented in \cite{sabry2018symmetric} that are
\emph{missing} in this chapter, and we give a short explanation on why they
could or could not be included in the chapter. 

\subsection{Inductive types, Higher-order, Recursion}
\label{sub:qu-inductive}

The points discussed in this section are concerned with features that we add to
the language in the next chapter. In that chapter, we work with the reversible
language based on symmetric pattern-matching, but without the quantum effect.
This happens for two reasons. Firstly, it is new in the literature even in the
classical case; and secondly, because the mathematical structure for
\emph{classical} reversibility is more suitable. We explain later in the thesis
the limitations encountered in adding these features to quantum computing.

\paragraph{Inductive types.} 
General inductive types can be added to the language without any issue. In
\cite{sabry2018symmetric}, the quantum language only works with lists, but all
their results regarding the syntax hold even when generalising with a fixed
point combinator for types, usually written $\mu X .  A$ (see the explanations
in \secref{sub:inductive-types}). The denotational semantics of inductive
types, in the context of countably-dimensional Hilbert spaces and isometries,
is not hard to achieve, as observed by Michael Barr \cite[Theorem
3.2]{barr92compact}. This means that, instead of working with a type $\nat$,
the language presented in this chapter could work with a general fixed point
combinator for types.

\paragraph{Higher-order.} 
In this chapter, the functions of the language are called \emph{unitaries}
because of their interpretation as unitary maps between Hilbert spaces. They
are the only operations allowed by quantum mechanics, alongside state
preparation and measurement.  Their treatment in the grammar and in the type
system is separated from terms; even more, unitaries are not terms. This is
different from the traditional $\lambda$-calculus. There are multiple reasons
that justify this choice. We adopt a denotational point of view and comment on
this choice through the mathematical model. Indeed, models of classical
programming languages often involve a \emph{closed} category -- meaning that
function types can be interpreted as objects of the category. However, the
category of countably-dimensional Hilbert spaces and unitary maps is not
closed. We will see in the next chapter how to go around this limitation with
the enrichment of categories. Regarding the syntax, the language can be
extended with \emph{extended values} as they are called in the original paper;
they are expressions that can contain not only variables at the ground level,
but also function variables. These extended values can be used on the
right-hand side of an abstraction, as in the following example:
\[
	\{\mid x \iso \letv{y}{\phi~x}{y} \}
\]
where $\phi$ can be any well-typed function, including a function variable.
Once we include those variables, we can introduce higher-order functions such
as $\lambda \phi . \omega$, and application of those functions. 

\paragraph{Recursion.} 
Once function variables are introduced, one can add a fixed point combinator
for recursion: for example, given a function $\omega$, we can introduce a
function $\ffix \phi . \omega$ that is the fixed point of $\omega$ with regard
to the function variable $\phi$. However, the functions $\ffix \phi . \phi$ and
$\ffix \phi . \{\mid x \iso \letv{y}{\phi~x}{y} \}$ do not terminate, and thus
cannot be interpreted as unitaries between Hilbert spaces. Moreover, we have
yet to find a mathematical interpretation for such fixed points in Hilbert
spaces; and Hilbert spaces might not be the solution. This is discussed in
Chapter~\ref{ch:qu-recursion}.

%

\subsection{Conclusion}
\label{sub:qu-conclusion}

We have presented a programming language equipped with simple types aimed at
quantum control through an algebraic effect. This is done through a syntax that
allows for linear combinations of terms, and a type system which ensures that
the latter are normalised. Then, we have formalised an equational theory,
preserving the typing judgements and handling linear algebra as well as the
computational aspect of the language. Finally, we have provided a denotational
semantics, proven complete with regard to the equational theory.

%
%

%% file: rev-recursion.tex
\begin{abstract}
	This chapter is concerned with the expressivity and denotational semantics of
	a functional higher-order reversible programming language based on Theseus.
	In this language, pattern-matching is used to ensure the reversibility of
	functions.	We then build a sound and adequate categorical semantics based on
	join inverse categories, with additional structures to capture
	pattern-matching. We show how one can encode any Reversible Turing Machine in
	said language. Finally, we derive a full completeness result, stating that
	any computable, partial injective function is the image of a term in the
	language.

	\paragraph{References.} This work is based upon a paper, under submission,
	coauthored with Kostia Chardonnet and Benoît Valiron. The preprint is
	available at \cite{nous2023invrec}.
\end{abstract}

\section{Introduction}
\label{sec:rev-introduction}

As said in the previous chapter, reversible computation has emerged as a
energy-preserving model of computation in which no data is ever erased. This
comes from Laundauer's principle which states that the erasure of information
is linked to the dissipation of energy as heat~\cite{Landauer61,
berut2012experimental}. In reversible computation, given some process $f$,
there always exists an inverse process $f^{-1}$ such that their composition is
equal to the identity: it is always possible to ``\emph{go back in time}'' and
recover the input of your computation. Although this can be seen as very
restrictive, as shown for instance in~\cite{bennett1973logical} non-reversible
computation can be emulated in a reversible setting, by keeping track of
intermediate results.

In order to avoid erasure of information, reversible computation often makes
use of \emph{garbage} or \emph{auxiliary wires}: additional information kept in
order to ensure both reversibility and the non-erasure of information. In
programming languages, this is done by ensuring both \emph{forward} and
\emph{backward} determinism. Forward determinism is almost always ensured in
programming languages: it is about making sure that, given some state of your
system, there is a unique next state that it can go to. Backward determinism on
the other hand makes sure that for given a state, there is only one original
state possible.

Reversible computation has since been shown to be a versatile model. In the
realm of quantum computation, reversible computing is at the root of the
construction of \emph{oracles}, subroutines describing problem instances in
quantum algorithms~\cite{nielsen02quantum}. Most of the research in reversible
circuit design can then been repurposed to design efficient quantum circuits.
On the theoretical side, reversible computing serves as the main ingredient in
several operational models of linear logics, whether through token-based
Geometry of Interaction~\cite{mackie1995geometry} or through the Curry-Howard
correspondence for $\mu$MALL~\cite{chardonnet2023curry, phd-kostia}.

Reversible programming has been approached in two different ways. The first
one, based on Janus and later R-CORE and R-WHILE~\cite{lutz1986janus,
yokoyama2007reversible, gluck2019reversible, yokoyama2016fundamentals},
considers imperative and flow-chart languages. The other one follows a
functional approach~\cite{yokoyama2011reversible, thomsen2015interpretation,
james2014theseus, JacobsenKT18, sabry2018symmetric, chardonnet2023curry}: a
function $A\to B$ in the language represents a function -- a bijection --
between values of type $A$ and values of type $B$. In this approach, types are
typically structured, and functional reversible languages usually feature
pattern-matching to discriminate on values.

One of the issue reversible programming has to deal with is
\emph{non-termination}: in general, a reversible program computes a
\emph{partial injective map}. This intuition can be formalised with the concept
of \emph{inverse categories}~\cite{kastl1979inverse, cockett2002restriction-I,
cockett2003restriction-II, cockett2007restriction-III}: categories in which
every morphism comes with a partial inverse, for which the category $\PInj$ of
sets and partial injective maps is the emblematic concrete instance.

This categorical setting has been successfully used in the study of reversible
programming semantics, whether based on flow-charts~\cite{gluck2017categorical,
kaarsgaard2019condition}, with recursion~\cite{axelsen2016join,
kaarsgaard2017join, kaarsgaard2019inversion, kaarsgaard2019engarde}, with
side-effects~\cite{heunen2015reversible, heunen2018reversible}, \textit{etc}.

Although much work has been dedicated to the categorical analysis of reversible
computation, the \emph{adequacy} of the developed categorical constructs with
reversible functional programming languages has only recently been under
scrutiny, either in \emph{concrete} categories of partial
isomorphisms~\cite{kaarsgaard2019engarde, kaarsgaard2021join}, or for simple,
\emph{non Turing-complete} languages~\cite{nous2021invcat}. A formal, categorical
analysis of a Turing-complete, reversible language is still missing.
Turing-completeness for a reversible programming language might come as a
surprise; however, the literature is already filled with evidence that any
irreversible computation can be simulated in a reversible setting
\cite{bennett1973logical, bennet1982thermo, abramsky2005structural}.

\subsection{Related work} 
\label{sub:rev-related-work}

The work in \cite{kaarsgaard2017join} is foundational in the development of the
semantics of reversible programming languages. In that paper, the authors lay
the foundations for the interpretation of what a programmer needs: data types
and loops. However, no example of practical use is shown for this model. This
is covered in another paper \cite{kaarsgaard2021join}, where a denotational
semantics of the reversible programming language $\mathtt{RFun}$ is provided.
The latter is not typed and is conceptually further from, say, a simply-typed
$\lambda$-calculus. We argue that the language in \cite{sabry2018symmetric} is
a more interesting case study.

The language we bring under scrutiny is the one introduced in
\cite{sabry2018symmetric}. The main goal of that paper was to shed light on the
possibility to program with quantum control at a higher level than the one of
circuits. This was partially achieved: the authors have indeed brought forward a
novel syntax that handles reversibility through pattern-matching, and where
reversible quantum effects can be added. Nevertheless, the denotational
semantics they provide is not satisfactory, as it is not compositional. Without
compositionality, there is no certainty that, for instance, substitution
preserves the interpretation (see the previous chapter for a compositional
denotation semantics for the so-called language with quantum effects).

Our denotational semantics of the classical, reversible language involves
categorical enrichment. However new in the development of reversible
programming languages, this technique has been used in several other instances
\cite{fiore-phd, rennala2018enriched, zamdzhiev2018enriching,
zamdzhiev2020causality, zamdzhiev2021linear, huot2023spaces}.

Other kinds of techniques can be used for the denotation of reversibility, as
compact closed categories \cite{sabry2021compact}, or traced monoidal
categories in the sense of \cite{karvonen-phd}.

\subsection{Contribution}
\label{sub:rev-contribution}

In this chapter, we aim to close the gap: we propose a Turing-complete,
reversible language, together with a categorical semantics. In particular, the
contributions of this paper are as follows.
\begin{itemize}
	\item A Turing-complete, higher-order reversible language with inductive types
		-- this language is described as \emph{classical}, as opposed to a
		\emph{quantum} programming language. Building on the Theseus-based family
		of languages studied in \cite{sabry2018symmetric, nous2021invcat,
		chardonnet2023curry, phd-kostia}, we consider an extension with
		\emph{inductive types}, general \emph{recursion} and \emph{higher-order}
		functions. Note that the Turing-completeness has been the work of Kostia
		Chardonnet, and will not be presented in details here. See the paper
		\cite{nous2023invrec}.
	\item A sound and adequate categorical semantics. We show how the language can
		be interpreted in join inverse rig categories. The result relies on the
		$\DCPO$-enrichments of join inverse rig categories. This part of the
		chapter is entirely the author's work.
	\item A full completeness result for \emph{computable} functions. We finally
		discuss how the interpretation of the language in the category $\PInj$ is
		fully complete in the sense that any computable, partial injective
		set-function on the images of types is realisable within the language.
		This part was produced in collaboration with Kostia Chardonnet.
\end{itemize}

\subsection{Work of the Author}
\label{sub:rev-author}

The author of this thesis has contributed to the following points.
\begin{itemize}
	\item A generalisation of the syntax presented in \cite{sabry2018symmetric}
		with a call-by-name $\lambda$-calculus on top of functions -- called
		\emph{isos} in this chapter. The operational semantics has been updated
		with regard to these changes. The author has also provided the
		corresponding proofs of the usual substitution lemma, subject reduction and
		progress for this system.
	\item A denotational semantics of the language in join inverse rig
		$\dcpo$-categories.
	\item A proof that the denotational semantics is adequate with regard to the
		operational semantics.
	\item A statement -- akin to a completeness result -- saying that any
		computable injection is captured by a function in the language.
\end{itemize}

\section{The Language: Classical Symmetric Pattern-Matching}
\label{sec:rev-language}

In this section, we present a reversible language, unifying and extending the
Theseus-based variants presented in the literature~\cite{sabry2018symmetric,
nous2021invcat, chardonnet2023curry}. In particular, the language we propose
features higher-order (unlike~\cite{nous2021invcat}), pairing, injection,
inductive types (unlike~\cite{sabry2018symmetric}) and general recursion
(unlike~\cite{chardonnet2023curry}). Functions in the language are based on
pattern-matching, following a strict syntactic discipline: term variables in
patterns should be used linearly, and clauses should be non-overlapping on the
left \emph{and} on the right (therefore enforcing non-ambiguity and
injectivity). In~\cite{sabry2018symmetric, nous2021invcat, chardonnet2023curry}
one also requires exhaustivity for totality. In this paper, we drop this
condition in order to allow non-terminating behaviour.

\begin{figure}[!h]
\begin{alignat*}{100}
		&\text{(Base types)} \quad& A, B &&&::=~ && \one \alt A \oplus B
		\alt A \otimes B \alt \mu X. A \alt X \\
    &\text{(Isos)} & T &&&::=&&A\iso B \alt
    T_1 \to T_2 \\[1.5ex]
    &\text{(Values)} & v &&&::=&& * \alt x \alt \inl{v} \alt \inr{v}
    \alt \pair{v_1}{v_2} \alt \fold{v} \\
    &\text{(Patterns)} & p &&&::=&& x \alt \pair{p_1}{p_2} \\
    &\text{(Expressions)} & e &&&::=&& v \alt
    \letv{p_1}{\omega~p_2}{e} \\
    &\text{(Isos)} & \isoterm &&&::=&& \isobasique \alt \ffix
    \isovar.\isoterm \\
		& & &&& && \alt \lambda\isolambdavar.
    \omega \alt \isovar \alt \omega_1~\omega_2\\
    &\text{(Terms)} & t &&&::=&& * \alt x \alt \inl{t} \alt \inr{t}
    \alt
    \pair{t_1}{t_2}  \\
    & & &&& && \alt \fold{t} \alt \isoterm~t \alt
    \letv{p}{t_1}{t_2}
\end{alignat*}
\caption{Terms and types.}
\label{fig:termtypes}
\end{figure}

The language is presented in Figure~\ref{fig:termtypes}. It consists of
two layers.
\begin{itemize}
\item Base types: The base types consist of the unit type $\one$ along with its
	sole constructor $*$, coproduct $A\oplus B$ and tensor product $A\otimes B$
		with their respective constructors, $\inl{(t)}, \inr{(t)}$ and
		$\pair{t_1}{t_2}$.  Finally, the language features inductive types of the
		form $\mu X. A$ where $X$ is a type variable occurring in $A$ and $\mu$ is
		its binder. Its associated constructor is $\fold{(t)}$. The inductive type
		$\mu X. A$ can then be unfolded into $A[\mu X. A/X]$, \emph{i.e.}, substituting
		each occurrence of $X$ by $\mu X. A$ in $A$. Typical examples of inductive
		types that can be encoded this way are the natural number, as $\natT = \mu
		X. (\one\oplus X)$ or the lists of types $A$, noted $[A] = \mu X.
		\one\oplus (A\otimes X)$. Note that we only work with closed types. We
		shall denote term-variables with $x, y, z$.
  
\item Isos types: The language features isos, denoted $\omega$, higher order
  reversible functions whose types $T$ consist either of a pair of base
  type, noted $A\iso B$ or function types between isos, $T_1 \to T_2$. A
  first-order iso of type $A\iso B$ consists of a finite set of
  \emph{clauses}, written $v\iso e$ where $v$ is a value of type $A$ and
  $e$ an expression of type $B$. An expression consists of a succession of
  applications of isos to some argument, described by {\lett} constructions:
  $\letv{(x_1, \dots, x_n)}{\omega~(y_1, \dots, y_n)}{e}$.
  Isos can take other, first-order isos as arguments through the
  $\lambda \phi . \omega$ construction. 
  Finally, isos can also represent \emph{recursive computation} through the $\ffix
  \isovar. \omega$ construction, where $\isovar$ is an
  \emph{iso-variable}.
\end{itemize}

\begin{remark}
	What was called a unitary in the previous chapter is called an iso. Note that
	the term \emph{iso} is the vocabulary in the original paper
	\cite{sabry2018symmetric}. The use of the word \emph{unitary} before now was
	an emphasis on the quantum aspect of the language. The language described in
	this chapter being classical, we choose to use the orginal terminology.
\end{remark}


\paragraph{Formation rules.}
While~\cite{chardonnet2023curry},~\cite{sabry2018symmetric} and the previous
chapter require isos to be exhaustive (\emph{i.e.}~to cover all the possible values of
their input types) and non-overlapping (\emph{i.e.}~two clauses cannot match a same
value), we relax the exhaustivity requirement in this paper, in the spirit of
what was done in~\cite{nous2021invcat}. However, the syntax still depends on a form
of orthogonality between values (resp. expressions) to define isos and to ensure
pattern-matching. 

\begin{remark}
	The definition below mentions terms in the fashion of being as general as
	possible; but the notion of orthogonality is only used within values and
	expressions, because we do not have linear combinations of terms anymore.
\end{remark}

\begin{definition}[Orthogonality]
	\label{def:rev-orthogonality}
	We introduce a binary relation $\bot$ on terms. Given two terms $t_1, t_2$,
	$t_1~\bot~t_2$ holds if it can be derived inductively with the rules below;
	we say that $t_1$ and $t_2$ are orthogonal. The relation $\bot$ is defined as
	the smallest relation such that:
  \[
    \begin{array}{c}
      \infer{\inl{t_1}~\bot~\inr{t_2}}{}
      \qquad
      \infer{\inr{t_1}~\bot~\inl{t_2}}{}
      \qquad
			\infer{C_\bot[t_1]~\bot~C_\bot[t_2]}{t_1~\bot~t_2}
    \end{array}
  \]
	with 
	\begin{align*}
		C_\bot[-] & ::= & - \alt \inl C_\bot[-] \alt \inr C_\bot[-] \alt 
		C_\bot[-] \otimes t \alt t \otimes C_\bot[-] \\
		&& \alt	\fold C_\bot[-] \alt \letv{p}{t}{C_\bot[-]}.
	\end{align*}
\end{definition}

\begin{figure}
	\[
		\begin{array}{c}
			\infer{\Psi;\emptyset\entail * \colon \one}{}
			\qquad
			\infer{\Psi;x \colon A \entail x \colon A}{}
			\mynl
			\infer{\Psi;\Delta\entail \inl{t}\colon A \oplus B}{\Psi;\Delta\entail
			t\colon A}
			\qquad
			\infer{\Psi;\Delta\entail\inr{t}\colon A \oplus B}{\Psi;\Delta\entail t\colon B}
			\mynl
			\infer{
				\Psi;\Delta_1,\Delta_2;\entail \pair{t_1}{t_2} \colon  A \otimes B
			}{
				\Psi;\Delta_1\entail t_1 \colon  A
				&
				\Psi;\Delta_2\entail t_2 \colon  B
			}
			\qquad
			\infer{\Psi; \Delta \entail \fold{t} \colon  \mu X. A }{\Psi; \Delta
			\entail t \colon  A[\mu X. A/X]}
			\mynl
			\infer{
				\Psi;\Delta\entail \isoterm~t \colon  B
			}{
				\Psi\entailiso \isoterm \colon  A \iso B
				&
				\Psi;\Delta\entail t \colon  A
			}
			\mynl
			\infer{\Psi;\Delta_1,\Delta_2\entail \letv{x_1 \otimes \dots \otimes
			x_n}{t_1}{t_2} \colon B}{\Psi;\Delta_1\entail t_1 \colon  A_1 \otimes
			\dots \otimes A_n\qquad\Psi;\Delta_2, x_1 \colon  A_1, \dots, x_n
			\colon  A_n\entail t_2 \colon  B}
		\end{array}
	\]
	\[
		\begin{array}{c}
			\infer{
				\Psi \entailiso
				\isobasique \colon  A \iso B.
			}{
				\begin{array}{l@{\quad}l@{\quad}l@{\qquad}l}
					\Psi;\Delta_1 \entail v_1 \colon  A
					&
					\ldots
					&
					\Psi; \Delta_n\entail v_n \colon  A
					&
					\forall i\not= j, v_i~\bot~v_j \\
					\Psi;\Delta_1\entail e_1 \colon  B
					&
					\ldots
					& \Psi;\Delta_n\entaile e_n \colon  B & \forall i\not= j, e_i~\bot~ e_j
				\end{array}
			}
			\mynl
			\infer{\Psi, \isovar \colon T \entailiso \isovar \colon T}{}
			\qquad
			\infer{
				\Psi\entailiso \ffix \isovar.\omega \colon T
			}
			{
				\Psi, \isovar \colon T \entailiso \omega \colon T
			}
			\qquad
			\infer{
				\Psi\entailiso \lambda \isovar. \omega \colon  T_1 \to T_2
			}
			{
				\Psi, \isovar \colon T_1 \entailiso \omega \colon T_2
			}
			\mynl
			\infer{
				\Psi\entailiso \omega_2~\omega_1 \colon T_2
			}
			{
				\Psi\entailiso \omega_1 \colon T_1
				&
				\Psi\entailiso \omega_2 \colon  T_1 \to T_2
			}
		\end{array}
	\]
	\caption{Typing rules of terms and isos.}
	\label{fig:rev-typterm}
\end{figure}

The typing rules are then given in Figure~\ref{fig:rev-typterm}. While the
rules to form terms do not come as a surprise, note the addition of a context
$\Psi$ which represents the iso-variables. The latter is non-linear, in the
sense that given a well-typed term $\Psi ; \Delta \entail t \colon A$, a
variable $\phi$ in $\Psi$ can appear once, several times, or can also not
appear in $t$. The context $\Delta$ is linear, in the same way as in the last
chapter: a variable $x$ in $\Delta$ is present exactly once in $t$. Note that
the rule for applying an iso to a term $\omega~t$ requires $\omega$ to have an
iso ground type $A \iso B$; this means in particular that the term $(\lambda
\phi . \omega)~t$ cannot be well-typed. 

An iso abstraction $\isobreduit$ is well-typed iff all the $v_i$ and $e_i$
are well-typed and the $v_i$ (resp. the $e_i$) are pairwise orthogonal. This
is necessary to ensure both forward and backward determinism. As mentioned above,
there is no request for exhaustivity, and for example the \emph{empty} iso
$\set \cdot$ is well-typed at all ground iso types. At the same level of isos,
we have a simply-typed $\lambda$-calculus as detailed in \secref{sub:ccc},
this time without the product type.

\begin{remark}
	As a further note, we will observe that our language has a sound and adequate
	denotational semantics in a $\dcpo$-enriched category; and the interpretation
	of the $\lambda$-calculus of isos happens at the dcpo level. This means in
	particular that any language, as long as it is adequately interpreted in the
	cartesian closed category $\dcpo$, can replace the simply-typed
	$\lambda$-calculus of isos. It also means that any usual structure on top
	of a $\lambda$-calculus can be added to the current language.
\end{remark}

Iso abstractions handle the use of several variables, thus our version of
$\beta$-reduction needs to behave accordingly. To do so, we introduce the
notion of valuation, akin to the one in the previous chapter; and the
application of a valuation to a term performs a substitution.

\paragraph{Substitutions.}
We recall the definitions of the last chapter, adapted to the language
presented here. The difference is only embodied by inductive types formalism,
where the $\fold\!$ constructor replaces contructors for natural numbers.
Given two values $v$ and $v'$, we build the smallest
valuation $\sigma$ such that the patterns of $v$ and $v'$ match and that the
application of the substitution to $v$, written $\sigma(v)$, is equal to $v'$.
we denote the matching of a value $v'$ against a pattern $v$ and its
associated valuation $\sigma$ as $\match{\sigma}{v}{v'}$. Thus,
$\match{\sigma}{v}{v'}$ means that $v'$ matches with $v$ and gives a smallest
valuation $\sigma$, while $\sigma(v)$ is the substitution performed. The
predicate $\match{\sigma}{v}{v'}$ it is defined as follows.
\[
	\begin{array}{c}
		\infer{\match{\sigma}{*}{*}}{}
		\quad
		\infer{\match{\sigma}{x}{v'}}{\sigma = \{ x \mapsto v'\}}
		\quad
		\infer{\match{\sigma}{\ini v}{\ini v'}}{\match{\sigma}{v}{v'}}
		\mynl
		\infer{
			\match{\sigma}{v_1 \otimes v_2}{v'_1 \otimes v'_2}
		}{
			\match{\sigma}{v_1}{v'_1}
			&
			\match{\sigma}{v_2}{v'_2}
			&
			\text{supp}(\sigma_1) \cap \text{supp}(\sigma_2) = \emptyset
			&
			\sigma = \sigma_1\cup\sigma_2
		}
		\mynl
		\infer{\match{\sigma}{\fold v}{\fold v'}}{\match{\sigma}{v}{v'}}
	\end{array}
\]
Whenever $\sigma$ is a valuation whose support contains the variables of $t$,
we write $\sigma(t)$ for the value where the variables of $t$ have been
replaced with the corresponding terms in $\sigma$, as follows:
\begin{itemize}
	\item $\sigma(x) = t'$ if $\{x\mapsto t'\}\subseteq \sigma$,
	\item $\sigma(*) = *$,
	\item $\sigma(\inl{t}) = \inl{\sigma(t)}$,
	\item $\sigma(\inr{t}) = \inr{\sigma(t)}$,
	\item $\sigma(\pair{t_1}{t_2}) = \pair{\sigma(t_1)}{\sigma(t_2)}$,
	\item $\sigma(\fold t) = \fold \sigma(t)$,
	\item $\sigma(\omega~t) = \omega~\sigma(t)$,
	\item $\sigma(\letv{p}{t_1}{t_2}) = \letv{p}{\sigma(t_1)}{\sigma(t_2)}$.
\end{itemize}

\begin{remark}
	Definition~\ref{def:well-valuation} is reminded here. Even if it involves a
	slightly different syntax, it still holds. A valuation $\sigma$ is said to be
	well-formed with regard to contexts $\Psi$ and $\Delta$ if for all $(x_i
	\colon A_i) \in \Delta$, we have $\set{x_i \mapsto t_i} \subseteq \sigma$ and
	$\Psi ; \emptyset \entail t_i \colon A_i$. We write $\Psi ; \Delta \Vdash
	\sigma$.
\end{remark}

\begin{lemma}
	\label{lem:rev-term-subst}
	If $\Psi ; \Delta \entail t \colon A$ and $\Psi ; \Delta \Vdash \sigma$ 
	are well-formed, then $\Psi ; \emptyset \entail \sigma(t)
	\colon A$ is well-formed.
\end{lemma}
\begin{proof}
	The proof is done by induction on $\Psi ; \Delta \entail t \colon A$.
	\begin{itemize}
		\item $\Psi ; \emptyset \entail * \colon \one$. Nothing to do.
		\item $\Psi ; x \colon A \entail x \colon A$. 
			Since $x \colon A \Vdash \sigma$ is well-formed, there is $~\entail t
			\colon A$ such that $\set{x \mapsto t} \subseteq \sigma$ and $\sigma(x) =
			t$.
		\item $\Psi ; \Delta_1, \Delta_2 \entail t_1 \otimes t_2 \colon A \otimes B$.
			The induction hypothesis ensures that $\sigma(t_1)$ and $\sigma(t_2)$ are
			well-typed, and thus $\sigma(t_1) \otimes \sigma(t_2) = \sigma(t_1
			\otimes t_2)$ also is.
		\item $\Psi ; \Delta \entail \ini t \colon A_1 \oplus A_2$.
			The induction hypothesis ensures that $\sigma(t)$ is well-typed, and
			thus $\ini \sigma(t) = \sigma(\ini t)$ is.
		\item $\Psi ; \Delta \entail \fold t \colon \mu X . A$.
			The induction hypothesis ensures that $\sigma(t)$ is well-typed, and
			thus $\fold \sigma(t) = \sigma(\fold t)$ is.
		\item $\Psi ; \Delta \entail \omega~t \colon B$.
			The induction hypothesis ensures that $\sigma(t)$ is well-typed, and
			thus $\omega~\sigma(t) = \sigma(\omega~t)$ is.
		\item $\Psi ; \Delta_1, \Delta_2 \entail \letv{p}{t_1}{t_2} \colon B$.
			The induction hypothesis ensures that $\sigma(t_1)$ and $\sigma(t_2)$ are
			well-typed, and thus $\letv{p}{\sigma(t_1)}{\sigma(t_2)}
			= \sigma(\letv{p}{t_1}{t_2})$ also is.
	\end{itemize}
\end{proof}

Once substitutions are defined, we can make our way through the operational
semantics.

\paragraph{Operational semantics.}  
The language is equipped with a small-step operational semantics, that revolves
around a $\beta$-reduction at the level of isos and a reversible equivalent to
$\beta$-reduction at the level of terms. First, we introduce an operational
semantics for isos, similar to a call-by-name reduction strategy in a
$\lambda$-calculus. It is given in Figure~\ref{fig:iso-operational}. The
$\beta$-reduction and the congruence rule are usually for a simply-typed
$\lambda$-calculus. The $\ffix\!$ operator is handled operationally \emph{à la}
PFC \cite{plotkin1977lcf}.

\begin{figure}[!h]
	\[
		\begin{array}{c}
			\infer{
				\ffix \phi . \omega \to \omega[\ffix \phi . \omega / \phi]
			}{}
			\qquad
			\infer{
				(\lambda \phi . \omega_1) \omega_2 \to \omega_1[\omega_2/\phi]
			}{}
			\qquad
			\infer{
				\omega_1\omega_2 \to \omega'_1 \omega_2
			}{
				\omega_1 \to \omega'_1
			}
		\end{array}
	\]
	\caption{Isos operational semantics.}
	\label{fig:iso-operational}
\end{figure}

As in any formal programming language given an operational semantics, progress
requires a notion of value, that is defined below. Note that the iso values are
the ones for closed isos. 

\begin{definition}[Iso values]
	We call \emph{iso values} the following isos:
	\[
		\omega ::= \isobasique \alt \lambda \phi . \omega .
	\]	
\end{definition}

\begin{lemma}
	\label{lem:iso-subst-type}
	If $\Psi, \psi \colon T_1 \entailiso \omega_2 \colon T_2$ and $\Psi
	\entailiso \omega_1 \colon T_1$ are well-formed, then $\Psi \entailiso
	\omega_2[\omega_1/\psi] \colon T_2$.
\end{lemma}
\begin{proof}
	Formally, the inductive definition of an iso judgement also depends on term
	judgements, in other words we also prove the following statement:
	If $\Psi, \psi \colon T_1 ; \Delta \entail t \colon A$ and $\Psi
	\entailiso \omega_1 \colon T_1$ are well-formed, then $\Psi ; \Delta \entail
	t[\omega_1/\psi] \colon A$.
	The proof is done by mutual induction on the term and iso judgements.
	\begin{itemize}
		\item $\Psi, \psi \colon T_1 ; \emptyset \entail * \colon \one$. Direct.
		\item $\Psi, \psi \colon T_1 ; x \colon A \entail x \colon A$. Direct.
		\item $\Psi, \psi \colon T_1 ; \Delta_1, \Delta_2 \entail t_1 \otimes t_2 \colon A \otimes B$.
			We observe that $(t_1 \otimes t_2)[\omega_1/\psi] = t_1[\omega_1/\psi]
			\otimes t_2[\omega_1/\psi]$ and the induction hypothesis concludes.
		\item $\Psi, \psi \colon T_1 ; \Delta \entail \ini t \colon A_1 \oplus A_2$.
			Similar to the previous point.
		\item $\Psi, \psi \colon T_1 ; \Delta \entail \fold t \colon \mu X . A$.
			Similar to the previous point.
		\item $\Psi, \psi \colon T_1 ; \Delta \entail \omega~t \colon B$.
			We observe that $(\omega~t)[\omega_1/\psi] =
			\omega[\omega_1/\psi]~t[\omega_1/\psi]$ and the induction hypothesis
			concludes.
		\item $\Psi, \psi \colon T_1 ; \Delta_1, \Delta_2 \entail \letv{p}{t_1}{t_2} \colon B$.
			With the induction hypothesis, similar to the previous point.
		\item $\Psi, \psi \colon T_1 \entailiso \isobreduit \colon A \iso B$.
			By induction hypothesis, given any iso $\omega$ present in $e_i$,
			$\omega$ is well-typed and $\omega[\omega_1/\psi]$ is also.
		\item $\Psi, \psi \colon T_1, \phi \colon T \entailiso \phi \colon T$. Direct.
		\item $\Psi, \psi \colon T_1 \entailiso \ffix \phi . \omega \colon T$. Note that $(\ffix
			\phi . \omega)[\omega_1/\psi] = \ffix \phi . (\omega[\omega_1/\psi])$,
			and by induction hypothesis $\omega[\omega_1/\psi]$ is well-typed.
		\item $\Psi, \psi \colon T_1 \entailiso \lambda \phi . \omega \colon T_2 \to T'_2$.
			Similar to the previous point.
		\item $\Psi, \psi \colon T_1 \entailiso \omega' \omega \colon T_2$.  Note that $(\omega'
			\omega)[\omega_1/\psi] = \omega' [\omega_1/\psi] \omega[\omega_1/\psi]$,
			and by induction hypothesis, $\omega'[\omega_1/\psi]$ and
			$\omega[\omega_1/\psi]$ are well-typed.
	\end{itemize}
\end{proof}

\begin{lemma}[Iso Subject Reduction]
	\label{lem:iso-subject-reduction}
	If $\Psi \entailiso \omega \colon T$ is well-formed and $\omega \to \omega'$,
	then $\Psi \entailiso \omega' \colon T$.
\end{lemma}
\begin{proof}
	The proof is done by induction on $\to$.
	\begin{itemize}
		\item $\ffix \phi . \omega \to \omega[\ffix \phi . \omega/\phi]$.
			The iso $\ffix \phi . \omega$ is well-typed, thus $\omega$ is also,
			and the previous lemma concludes.
		\item $(\lambda \phi . \omega_1) \omega_2 \to \omega_1[\omega_2/\phi]$.
			For the application to be well-typed, we need both $\lambda \phi .
			\omega_1$ and $\omega_2$ to be well-typed. The former ensures that
			$\omega_1$ is well-typed, and the previous lemma concludes.
		\item $\omega_1 \omega_2 \to \omega'_1 \omega_2$. The induction
			hypothesis ensures that $\omega'_1$ is well-typed, and $\omega_2$
			is also because the application $\omega_1 \omega_2$ is.
	\end{itemize}
\end{proof}

\begin{lemma}[Iso Progress]
	\label{lem:iso-progress}
	If $~\entailiso \omega \colon T$ is well-formed, $\omega$ is either an
	iso value or there exists $\omega'$ such that $\omega \to \omega'$.
\end{lemma}
\begin{proof}
	The proof is done by induction on $~\entailiso \omega \colon T$.
	\begin{itemize}
		\item $~\entailiso \isobreduit \colon A \iso B$ is an iso value.
		\item $~\entailiso \ffix \phi . \omega \colon T$ reduces.
		\item $~\entailiso \lambda \phi . \omega \colon T_1 \to T_2$ is an iso
			value.
		\item $~\entailiso \omega_2 \omega_1 \colon T_2$. By induction hypothesis,
			either $\omega_2$ is a value or it reduces. If it reduces, $\omega_2
			\omega_1$ reduces. If it is a value, it cannot be an iso abstraction:
			being applied to another iso, it must have a type $T_1 \to T_2$. Thus, it
			is of the form $\lambda \phi . \omega'_2$, and $(\lambda \phi .
			\omega'_2) \omega_1$ reduces.
	\end{itemize}
\end{proof}

\begin{corollary}
	\label{cor:iso-progress}
	If $~\entailiso \omega \colon A \iso B$ is well-formed, either there is some
	$\Delta_i \entail v_i \colon A$ and $\Delta_i \entail e_i \colon B$ such that
	$\omega = \isobasique$, or there exists $\omega'$ such that $\omega \to
	\omega'$.
\end{corollary}

We move to an operational semantics for terms of the language, which
requires the introduction of congruence context. The operational semantics is detailed
in Figure~\ref{fig:rev-operational-semantics}.

\begin{figure}[!h]
	\begin{align*}
		C_\to[-] &::= &- \alt \inl C_\to[-] \alt \inr C_\to[-] \alt 
		C_\to[-] \otimes t \alt v \otimes C_\to[-] \\
		&& \alt \isobreduit~C_\to[-] \alt
		\fold C_\to[-] \\
		&& \alt \letv{p}{C_\to[-]}{t}. 
	\end{align*}
	\[
		\begin{array}{c}
			\infer{ \isobasique~v' \to \sigma(e_i)}{
				\match{\sigma}{v_i}{v'}}
			\qquad
			\infer{\letv{p}{v}{t} \to \sigma(t)}{\match{\sigma}{p}{v}}
			\\[1.5ex]
			\infer{C_\to[t_1] \to C_\to[t_2]}{t_1 \to t_2}
			\qquad
			\infer{\omega~t \to \omega'~t}{\omega \to \omega'}
		\end{array}
	\]
	\caption{Term Operational Semantics}
	\label{fig:rev-operational-semantics}
\end{figure}

\begin{lemma}[Subject Reduction]
	\label{lem:rev-subject-reduction}
	If $\Psi; \Delta \entail t \colon A$ is well-formed and $t \to t'$,
	then $\Psi; \Delta \entail t' \colon A$ is also well-formed.
\end{lemma}
\begin{proof}
	The proof is done by induction on $\to$. It revolves around three
	quick observations: Lemma~\ref{lem:iso-subject-reduction},
	Lemma~\ref{lem:rev-term-subst}, and that if $t_2$ and 
	$C_\to[t_1]$ are well-typed, then $C_\to[t_2]$ also is.
\end{proof}

As usual we write $\to^*$ for the reflexive transitive closure of $\to$.  In
particular, the rewriting system follow a \emph{call-by-value} strategy,
requiring that the argument of an iso is fully evaluated to a value before
firing the substitution. Note that unlike~\cite{chardonnet2023curry,
sabry2018symmetric}, we do not require any form of termination and isos are not
required to be exhaustive: the rewriting system can diverge or be stuck.  The
evaluation of an iso applied to a value is dealt with by pattern-matching: the
input value will try to match one of the value from the clauses and potentially
create a substitution if the two values match, giving the corresponding
expression as an output under that substitution.

\begin{example}
	\label{ex:non-termination}
	Observe that $~\entailiso \ffix \phi . \phi \colon A \iso B$ is well-formed
	judgement.  Given any closed term judgement $~\entail t \colon A$, the
	judgement $~\entail (\ffix \phi . \phi)~t \colon B$ is also well-formed, and:
	\begin{align*}
		(\ffix \phi . \phi)~t 
		&\to (\phi[\ffix \phi . \phi / \phi])~t = (\ffix \phi . \phi)~t \\  
		&\to (\phi[\ffix \phi . \phi / \phi])~t = (\ffix \phi . \phi)~t \\  
		&\to (\phi[\ffix \phi . \phi / \phi])~t = (\ffix \phi . \phi)~t \\  
		&\to \cdots
	\end{align*}
	does not terminate. A slightly more subtle instance of non-termination is
	the term $~\entail (\ffix \phi . \set{\mid x \iso \letv{y}{\phi~x}{y}}~v \colon B$,
	which reduces as follows:
	\begin{align*}
		&\ (\ffix \phi . \set{\mid x \iso \letv{y}{\phi~x}{y}})~v \\
		&\to (\set{\mid x \iso \letv{y}{\phi~x}{y}}[\ffix \phi . \set{\mid x \iso
		\letv{y}{\phi~x}{y}} / \phi])~v \\
		&= \set{\mid x \iso \letv{y}{(\ffix \phi .
		\set{\mid x \iso \letv{y}{\phi~x}{y}})~x}{y}}~v \\
		&\to \letv{y}{(\ffix \phi . \set{\mid x \iso \letv{y}{\phi~x}{y}})~v}{y} \\
		&\to \cdots
	\end{align*}
	and it does not terminate.
\end{example}

\begin{example}
	\label{ex:non-reduction}
	Reductions can get stuck, because there is no pattern to match with.  For
	example, $~\entailiso \set{\mid \inr * \iso \inl *} \colon \one \oplus \one \iso \one
	\oplus \one$ is a well-typed iso abstraction. The term $~\entail \set{\mid \inr
	* \iso \inl *}~(\inl *) \colon \one \oplus \one$ does not reduce.
\end{example}

\begin{example}
  \label{ex:isos-map}
  Remember that $[A] = \mu X. \one\oplus (A\otimes X)$.  One can
  define the \emph{map} operator on list with an iso of type
  $(A\iso B)\to [A]\iso [B]$, defined as
  \[
    \lambda \isolambdavar. \ffix \isovar. \left\{
      \begin{array}{l@{~}c@{~}l}
        [~] & {\iso} & [~] \\
        h :: t & {\iso} & \letv{h'}{\isolambdavar~h}{} \letv{t'}{\isovar~t}{}h' :: t'
      \end{array}
    \right\},
  \]
  with the terms $[~] = \fold{(\inl{(*)})}$, representing the empty list, and $h::t =
  \fold{(\inr{(\pair{h}{t})})}$, representing the head and tail of the list. Its inverse
  $\opn{map}^{-1}$ is
  \[
    \lambda\isolambdavar. \ffix \isovar. \left\{
      \begin{array}{l@{~}c@{~}l}
        [~] & {\iso} & [~] \\
        h' :: t' & {\iso} & \letv{t}{\isovar~t'}{} \letv{h}{\isolambdavar h'}{} h :: t
      \end{array}
    \right\}.
\]

Note that in the latter, the variable $\isolambdavar$ has type $B\iso A$.  If
	we consider the inverse of the term $(\opn{map}~\omega)$ we would obtain the
	term $(\opn{map}^{-1}~\omega^{-1})$ where $\omega^{-1}$ would be of type
	$B\iso A$.
\end{example}

\begin{example}[Cantor Pairing]
  \label{ex:cantor}
  One can encode the Cantor Pairing between $\natS \otimes \natS \iso \natS$.
  First recall that the type of natural number $\natT$ is given by $\mu X.
  \one\oplus X$, then define $\overline{n}$ as the encoding of natural numbers
  into a closed value of type $\natT$ as $\overline{0} = \fold{(\inl{*})}$ and
  given a variable $x$ of type $\natT$, its successor is $\overline{S(x)} =
  \fold{(\inr{(x)})}$.  Omitting the $\overline{\ \cdot\ }$ operator for
  readability, the pairing is then defined as:

	\[\begin{array}{ll} \omega_1 : \natT \otimes \natT \iso (\natT \otimes \natT)
		\oplus \one  \\
		= \left\{\begin{array}{lcl}
			\pair{S(i)}{j} & {\iso} & \inl{(\pair{i}{S(j)})} \\
			\pair{0}{S(j)} & {\iso} & \inl{(\pair{j}{0})} \\
		\pair{0}{0} & {\iso} & \inr{(*)} \end{array}\right\},
	\end{array}
	\]
	\[
		\begin{array}{ll} 
			\omega_2 : (\natT \otimes \natT) \oplus \one \iso \natT  \\
			= \left\{\begin{array}{lcl}
				\inl{(x)} & {\iso} & \letv{y}{\isovar~x}{S(y)} \\
				\inr{(*)} & {\iso} & 0
			\end{array}\right\},
		\end{array}\]
	\[\begin{array}{ll}
		\opn{CantorPairing} : \natT\otimes \natT \iso \natT \\
		= \ffix \isovar. \left\{\begin{array}{lcl}
			x & {\iso} & \letv{y}{\omega_1~x}{} \\
			&& \letv{z}{\omega_2~y}{z}
		\end{array}\right\},
	\end{array}\]

  where the variable $\isovar$ in $\omega_2$ is the one being binded
  by the $\ffix\!$ of the $\opn{CantorPairing}$ iso. Intuitively,
  $\omega_1$ realise one step of the Cantor Pairing evaluation, while
  $\omega_2$ check if we reached the end of the computation and either
  apply a recursive call, or stop.
  
	For instance, $\opn{CantorPairing}~\pair{1}{1}$ will match with the first
	clause of $\omega_1$, evaluating into $\inl{\pair{0}{2}}$, and then, inside
	$\omega_2$ the reduction $\opn{CantorPairing}~\pair{0}{2}$ will be triggered
	through the recursive call, evaluating the second clause of $\omega_1$,
	reducing to $\inl{\pair{1}{0}}$, etc.
\end{example}

\section{Denotational semantics}
\label{sec:rev-detonational}

We now show how to build a denotational semantics for the language we presented
thus far. The semantics is akin to the one presented in~\cite{nous2021invcat} but
with extra structure to handle inductive types and recursive functions. The
section is organised as follows.
\begin{itemize}
	\item We first fix the interpretation of types. This requires us to make type
		judgements explicit in the formalisation of the syntax. The semantics of
		type judgement is then given, thanks to the work mentioned in
		\secref{sub:inductive-types}. We then discuss the interpretation of closed
		types, which are the types actually involved in the syntax.
	\item We detail the interpretation of term judgements, that are given
		as Scott continuous maps between the interpretation of the linear context,
		and the dcpo of reversible programs at a certain type.
	\item An interpretation of iso judgements is given, also as Scott continuous
		maps. The development is similar to the denotational semantics of a
		simply-typed $\lambda$-calculus as given in \secref{sub:ccc}, however our
		iso abstractions need more care.
	\item We finish with an interpretation of substitutions, allowing to fromulate
		a soundness and adequacy statement later on.
\end{itemize}

In the whole section, we consider $\CC$ a join inverse rig category (see
Definition~\ref{def:rig}), that is $\dcpo$-enriched (see
Definition~\ref{def:enriched-cat} and \secref{sub:dcpo}) and such that $0$ and
$1$ are distinct objects. 

\subsection{Denotational Semantics of Types}
\label{sub:rev-types}

\paragraph{Term types.}
As explained in the background section (see \secref{sec:restriction}), we can
assume without loss of generality that $\CC$ satisfies the hypothesis of
Theorem~\ref{th:param-alg}. In order to deal with open types, we make use of an
auxiliary judgement for types, of the form $X_1,\ldots,X_n\vDash A$, where
$\{X_i\}_i$ is a subset of the free type variables, non necessarily appearing
in $A$. We interpret this kind of judgement as a $\DCPO$-functor
$\CC^{\abs\Theta}\to\CC$ written $\sem{\Theta\vDash A}$. This is formally
defined as an inductive relation, and the semantics is stated similarly to what
is done in \cite{fiore04axiomatic, zamdzhiev2021linear,
zamdzhiev2021commutative}.

\begin{equation}
	\label{eq:rev-type-formation}
	\begin{array}{c}
		\infer{\Theta, X \vDash X}{}
		\qquad
		\infer{\Theta \vDash \one}{}
		\qquad
		\infer[\star \in \set{\oplus,\otimes}]{\Theta \vDash A \star B}{
			\Theta \vDash A
			&
			\Theta \vDash B
		}
		\qquad
		\infer{\Theta \vDash \mu X . A}{\Theta,X \vDash A}
	\end{array}
\end{equation}

The type judgements are inductively defined as given above. The interpretation
of types is detailed in Figure~\ref{fig:rev-type-interpretation}, where $\Pi
\colon \CC^{\abs\Theta} \to \CC$ is the projection functor on the last
component, $K_1 \colon \CC^{\abs\Theta} \to \CC$ is the constant functor that
outputs $1$ and the $\iid$, $\otimes$ and $\oplus$ are given by the rig
structure, and $(-)\nnoma$ is part of the parameterised initial algebra (see
Definition~\ref{def:para-initial-algebra}), that exists thanks to
Theorem~\ref{th:param-alg} (see details in \secref{sub:inductive-types}). 

\begin{figure}[!h]
  \begin{align*}
    \sem{\Theta \vDash A} &\colon \CC^{\abs\Theta}\to\CC \\
		\sem{\Theta, X \vDash X} &= \Pi \\
    \sem{\Theta \vDash \one} &= K_1 \\
		\sem{\Theta\vDash A\oplus B} &= \oplus\circ\pv{\sem{\Theta\vDash
		A}}{\sem{\Theta\vDash B}} \\
		\sem{\Theta\vDash A\otimes B} &= \otimes\circ\pv{\sem{\Theta\vDash
		A}}{\sem{\Theta\vDash B}} \\
		\sem{\Theta\vDash \mu X.A} &= \left(\sem{\Theta,X\vDash A}\right)^\noma
  \end{align*}
  \caption{Interpretation of types.}
  \label{fig:rev-type-interpretation}
\end{figure}

We remind the typing rule of the term contructor $\fold\!$.
\[
	\infer{\Psi ; \Delta \entail \fold t \colon \mu X . A}{\Psi ; \Delta \entail
	t \colon A[\mu X . A/X]}
\]
This typing rule involves a type substitution. We show in the next lemma
that type substitutions make sense, and we also provide a result about there
interpretation.

\begin{lemma}[Type Substitution]
	\label{lem:rev-type-subst}
	Given well-formed type judgements $\Theta, X \vDash A$ and $\Theta \vDash B$,
	the judgement $\Theta \vDash A[B/X]$ is well-formed and 
	\[
		\sem{\Theta \vDash A[B/X]} = \sem{\Theta, X \vDash A} \circ
		\pv{\iid}{\sem{\Theta \vDash B}}.
	\]
\end{lemma}
\begin{proof}
	The proof that $\Theta \vDash A[B/X]$ is well-formed is direct by
	induction on the formation rules. The semantic equality is also
	proven by induction on the formation rules of $\Theta, X \vDash A$.
	\begin{itemize}
		\item $\Theta \vDash \one$. Nothing to do.
		\item $\Theta, X \vDash X$. Indeed,
			$\sem{\Theta \vDash B} = \Pi \circ \pv{\iid}{\sem{\Theta \vDash B}}$.
		\item $\Theta, X \vDash A_1\star A_2$.
			\begin{align*}
				&\ \sem{\Theta \vDash (A_1\star A_2)[B/X]} & \\
				&= \sem{\Theta \vDash A_1[B/X]\star A_2[B/X]} & \\
				&= \star \circ \pv{\sem{\Theta \vDash A_1[B/X]}}{\sem{\Theta \vDash
				A_2[B/X]}} 
				& \text{(by definition)} \\
				&= \star \circ \pv{\sem{\Theta, X \vDash A_1} \circ
				\pv{\iid}{\sem{\Theta \vDash B}}}{\sem{\Theta, X \vDash A_2} \circ
				\pv{\iid}{\sem{\Theta \vDash B}}}
				& \text{(by IH)} \\
				&= \star \circ \pv{\sem{\Theta, X \vDash A_1}}{\sem{\Theta, X \vDash
				A_1}} \circ \pv{\iid}{\sem{\Theta \vDash B}}
				& \text{(by unicity)} \\
				&= \sem{\Theta, X \vDash A_1\star A_2} \circ \pv{\iid}{\sem{\Theta \vDash
				B}}
				& \text{(by definition)} 
			\end{align*}
		\item $\Theta, X \vDash \mu Y . A$.
			\begin{align*}
				&\ \sem{\Theta \vDash (\mu Y . A)[B/X]} & \\
				&= \sem{\Theta \vDash \mu Y . A[B/X]} & \\
				&= (\sem{\Theta, Y \vDash A[B/X]})\nnoma 
				& \text{(by definition)} \\
				&= (\sem{\Theta, Y, X \vDash A} \circ \pv{\iid}{\sem{\Theta, Y \vDash
				B}})\nnoma
				& \text{(by IH)} \\
 				&= (\sem{\Theta, X, Y \vDash A} \circ (\pv{\iid}{\sem{\Theta \vDash
 				B}} \times \iid))\nnoma
 				& \text{(} Y \text{ is not in } B \text{)} \\
 				&= (\sem{\Theta, X, Y \vDash A})\nnoma \circ \pv{\iid}{\sem{\Theta
 				\vDash B}}
 				& \text{(see \cite[Prop.~4.14]{zamdzhiev2021linear})} \\
 				&= \sem{\Theta, X \vDash \mu Y . A} \circ \pv{\iid}{\sem{\Theta
 				\vDash B}}
 				& \text{(by definition)} 
			\end{align*}
	\end{itemize}
\end{proof}

The previous lemma embodies the link between the fixed point contructor $\mu$
and the parameterised initial algebra, with the following observation:
\[
	\sem{\Theta \vDash A[\mu X . A / X]} = 
	\sem{\Theta, X \vDash A} \circ \pv{\iid}{\sem{\Theta \vDash B}}
	\tilde{\Rightarrow} \sem{\Theta \vDash \mu X . A}
\]
and thus there is a natural isomorphism:
\[ \alpha^{\sem{\Theta \vDash A}} \colon \sem{\Theta, X \vDash A[\mu X .
A/X]} \cong \sem{\Theta \vDash \mu X . A}. \]

In the syntax of this chapter, only closed types are involved in the typing
rules of terms. We sum up the semantics of closed types below.
\begin{equation}
	\label{eq:sem-closed-types}
	\begin{array}{c}
		\den\one = 1
		\qquad
		\den{A\oplus B} = \den A \oplus \den B
		\qquad
		\den{A\otimes B} = \den A \otimes \den B 
		\mynl
		\sem{\mu X.A} \cong \sem{A[\mu X.A/X]}
	\end{array}
\end{equation}

\paragraph{Iso types.}
The basic types of isos are represented by pointed dcpos of morphisms in $\CC$,
written $\den{A\iso B} = \CC(\den A, \den B)$. The rest is given by the usual
type interpretation of a simply-typed $\lambda$-calculus in the cartesian
closed category $\dcpo$ (see \secref{sub:ccc} for the details).

The terms used to build isos are dependent in two contexts: variables in
$\Delta$ and isos in $\Psi$. In general, if $\Delta = x_1 \colon A_1, \dots,
x_m \colon A_m$ and $\Psi = \phi_1 \colon B_1\iso C_1, \dots, \phi_n \colon B_n\iso
C_n$, then we set
  $\den\Delta = \den{A_1}\otimes\dots\otimes\den{A_m}$ and
  $\den\Psi = \CC(\den{B_1},\den{C_1})\times\dots\times\CC(\den{B_n},\den{C_n})$,
with $\otimes$ being the monoidal product in $\CC$ and $\times$ the cartesian
product in $\DCPO$.

\subsection{Denotational Semantics of Terms}
\label{sub:rev-terms}

A well-formed term judgement $\Psi ; \Delta \entail t \colon A$ is given a
semantics 
\[\sem{\Psi ; \Delta \entail t \colon A} \colon \sem\Psi \to
\CC(\sem\Delta, \sem A) \] 
as a Scott continuous map between two dcpos -- in other words, as a morphism in
$\dcpo$. Values do not contain iso variables, thus given a judgement $\Psi ;
\Delta \entail v \colon A$ with $v$ a value, $\sem{\Psi ; \Delta \entail v
\colon A}$ is a constant function, whose output is a morphism $\sem\Delta \to
\sem A$ in $\CC$. The interpretation of values, and therefore of the
corresponding terms, is as follows, for all $g \in \sem\Psi$:
\begin{align*}
	\sem{\Psi ; \Delta \entail t \colon A}(g) &\in \CC(\sem\Delta, \sem A) \\
	\sem{\Psi ; \emptyset \entail * \colon \one}(g) 
	&= \iid_{\sem \one} \\
	\sem{\Psi ; x \colon A \entail x \colon A}(g)
	&= \iid_{\sem A} \\
	\sem{\Psi ; \Delta \entail \inl t \colon A \oplus B}(g)
	&= \iota_l \circ \sem{\Psi ; \Delta \entail t \colon A}(g) \\
	\sem{\Psi ; \Delta \entail \inr t \colon A \oplus B}(g)
	&= \iota_r \circ \sem{\Psi ; \Delta \entail t \colon B}(g) \\
	\sem{\Psi ; \Delta_1, \Delta_2 \entail t_1 \otimes t_2 \colon A \otimes B}(g)
	&= \sem{\Psi ; \Delta_1 \entail t_1 \colon A}(g) \otimes 
	\sem{\Psi ; \Delta_2 \entail t_2 \colon B}(g) \\
	\sem{\Psi ; \Delta \entail \fold t \colon \mu X . A}(g)
	&= \alpha^{\sem{X \vDash A}} \circ \sem{\Psi ; \Delta \entail t \colon A[\mu
	X . A / X]}(g)
\end{align*}

\begin{lemma}
	\label{lem:rev-orthogonal-semantics-values}
	Given two judgements $\Psi ; \Delta_1 \entail v_1 \colon A$ and
	$\Psi ; \Delta_2 \entail v_2 \colon A$, such that $v_1~\bot~v_2$, we have
	that for all $g \in \sem\Psi$:
	\[
		\sem{v_1}(g)\pinv \circ \sem{v_2}(g) = 0_{\sem{\Delta_2},\sem{\Delta_1}}.
	\]
\end{lemma}
\begin{proof}
	This is proven by induction on the definition of $\bot$. The cases $\inl
	v_1~\bot~\inr v_2$ and $\inr v_1~\bot~\inl v_2$ are covered by
	Lemma~\ref{lem:orthogonal-injections}. the other cases involve
	precompositions and tensor products, the result is direct with the induction
	hypothesis.
\end{proof}

Once we have fixed the denotation of the easiest terms, we can cover the
difficult part. Throughout the rest of the section, the semantics of a
well-formed term judgement $\sem{\Psi ; \Delta \entail t \colon A}$ is
obviously a map between sets, and the interesting part is proving that
it is indeed a Scott continuous map between two dcpos. The interpretation
of the remaining terms is given below.
\begin{align*}
		\sem{\Psi ; \Delta \entail t \colon A}
		&\in \dcpo(\sem\Psi,\CC(\sem\Delta,\sem A)) \\
		\sem{\Psi ; \Delta_1,\Delta_2 \entail \letv{p}{t_1}{t_2} \colon B}
		&= \comp~\circ \\
		&\pv{\sem{\Psi ; \Delta_2, p \colon A \entail t_2 \colon
		B}}{(\otimes \circ \pv{\iid_{\sem{\Delta_2}}}{\sem{\Psi ; \Delta_1
		\entail t_1 \colon A}})} \\
		\sem{\Psi ; \Delta \entail \omega~t \colon B}
		&= \comp \circ \pv{\sem{\Psi \entailiso \omega \colon A \iso B}}{\sem{\Psi ;
		\Delta \entail t \colon A}}
\end{align*}
All this is well-defined in $\DCPO$ provided that $\sem{\Psi \entailiso \omega
\colon A \iso B}$ is. This last point is the focus of the next subsection. Note
that the interpretation on terms and iso is thus defined by mutual induction on
the term and iso judgements. This does not cause any difficulty.

Before moving to the denotational semantics of isos, we prove a lemma of
central importance, extending Lemma~\ref{lem:rev-orthogonal-semantics-values},
and showing that the interpretations of two orthogonal expressions are also
orthogonal, in the sense of Remark~\ref{rem:loose-inner-product}.

\begin{lemma}
	\label{lem:rev-orthogonal-semantics}
	Given two judgements $\Psi ; \Delta_1 \entail e_1 \colon A$ and
	$\Psi ; \Delta_2 \entail e_2 \colon A$, such that $e_1~\bot~e_2$, we have
	that for all $g \in \sem\Psi$:
	\[
		\sem{e_1}(g)\pinv \circ \sem{e_2}(g) = 0_{\sem{\Delta_2},\sem{\Delta_1}}.
	\]
\end{lemma}
\begin{proof}
	A large part of this lemma is already proven in
	Lemma~\ref{lem:rev-orthogonal-semantics-values}. It remains to observe that
	the interpretation of $\mathtt{let}$ involves a precomposition, thus the
	induction hypothesis on the definition of $\bot$ is enough to conclude.
\end{proof}

\subsection{Denotational Semantics of Isos}
\label{sub:rev-isos}

Isos do only depend on function variables but they are innately morphisms, so
their denotation will be similar to terms -- a Scott continuous map:
\[
	\den{\Psi\entailiso\omega \colon A\iso B} \colon \den\Psi \to \CC(\den
	A,\den B). 
\]
We define the denotation of an iso by induction on the typing derivation. The
interpretation of an iso-variable is direct: it is the projection on the last
component. The interpretations of evaluations and $\lambda$-abstractions are
usual in a cartesian closed category, in our case, $\dcpo$ (see
\secref{sub:ccc} and \secref{sub:dcpo}). 
\begin{align*}
	\sem{\Psi \entailiso \omega \colon T}
	&\in \dcpo(\sem\Psi, \sem T) \\
	\den{\Psi,\phi \colon T \entailiso \phi \colon T} &=
	\pi_{\sem T} \\
	\sem{\Psi \entailiso \omega_2 \omega_1 \colon T_2}
	&= \rmeval \circ \pv{\sem{\Psi \entailiso \omega_2 \colon T_1 \to
	T_2}}{\sem{\Psi \entailiso \omega_1 \colon T_1}} \\
	\sem{\Psi \entailiso \lambda \phi . \omega \colon T_1 \to T_2} &=
	\rmcurry(\sem{\Psi, \phi \colon T_1 \entailiso \omega \colon T_2}) \\
	\sem{\Psi \entailiso \ffix \phi . \omega \colon T}
	&= \fix (\sem{\Psi, \phi \colon T \entailiso \omega \colon T})
\end{align*}
The remaining rule, that builds an iso abstraction $\isobreduit$, needs more
details. The interpretation of an iso abstraction is close to the one in the
previous chapter (see \secref{sec:very-simple-semantics}), in a related but
different setting.

\begin{lemma}
	\label{lem:compati-clauses}
	Given a well-typed iso abstraction $\Psi \entailiso \isobreduit \colon A
	\iso B$, for all $g \in \sem\Psi$, the morphisms in $\CC$ given by:
	\[
		\sem{\Psi ; \Delta_i \entail e_i \colon B}(g)
		\circ \sem{\Psi ; \Delta_i \entail v_i \colon A}(g)\pinv
	\]
	with $i \in I$, are pairwise inverse compatible.
\end{lemma}
\begin{proof}
	Lemma~\ref{lem:rev-orthogonal-semantics} gives us that for all $g \in
	\sem\Psi$ and $i \neq j \in I$:
	\begin{align*}
		(\sem{e_i}(g)
		\circ \sem{v_i}(g)\pinv)\pinv
		\circ
		(\sem{e_j}(g)
		\circ \sem{v_j}(g)\pinv)
		&= 0_{\sem A,\sem A} \\
		(\sem{e_i}(g)
		\circ \sem{v_i}(g)\pinv)
		\circ
		(\sem{e_j}(g)
		\circ \sem{v_j}(g)\pinv)\pinv 
		&= 0_{\sem B,\sem B}
	\end{align*}
	which are the hypotheses of Lemma~\ref{lem:inv-ortho-compati}. This is enough
	to ensure that for all $g \in \sem\Psi$ and $i \neq j \in I$,
	\[
		\sem{e_i}(g)
		\circ \sem{v_i}(g)\pinv
		\asymp
		\sem{e_j}(g)
		\circ \sem{v_j}(g)\pinv
	\]
	and this last point concludes.
\end{proof}

In a similar vein to the previous chapter, each clause $v_i \iso e_i$, present
in an iso abstraction, is given an interpretation $\sem{e_i} \circ
\sem{v_i}\pinv$.  The previous lemma shows that in the case of an iso
abstraction, the interpretations of all clauses can be joined (in the sense of
Definition~\ref{def:join}), and thus the interpretation of the iso abstraction
is the least upper bound of the interpretation of all the clauses. This least
upper bound is also the one in $\dcpo$, as shown by the lemma below.

\begin{lemma}
	\label{lem:join-scott}
	Given a dcpo $\Xi$, two objects $X$ and $Y$ of $\CC$, a set of indices $I$
	and a family of Scott continuous maps $\xi_i \colon \Xi \to \CC(X,Y)$ that
	are pairwise inverse compatible, the function given by:
	\[
		\left\{
			\begin{array}{ccl}
				\Xi &\to &\CC(X,Y) \\
				x &\mapsto &\bigvee_{i \in I} \xi_i(x)
			\end{array}
		\right.
	\]
	is Scott continuous, and is written $\bigvee_{i \in I} \xi_i$.
\end{lemma}
\begin{proof}
	The function can also be obtained as the join in the dcpo $[\Xi \to
	\CC(X,Y)]$, it is therefore Scott continuous.
\end{proof}

The interpretation of an iso abstraction is then given by:
\begin{equation*}
  \begin{split}
		\den{\Psi\entailiso \isobreduit \colon A\iso B} =
		\bigvee_{i \in I} (\comp \circ \pv{\sem{\Psi ; \Delta_i \entail e_i \colon
		B}}{\sem{\Psi ; \Delta_i \entail v_i \colon A}\pinv})
  \end{split}
\end{equation*}

\begin{proposition}
	Given a well-typed iso abstraction $\Psi \entailiso \isobreduit \colon A \iso B$,
	its interpretation $\sem{\Psi \entailiso \isobreduit \colon A \iso B}$ is
	well-defined as a Scott continuous map between the dcpos $\sem\Psi$ and 
	$\CC(\sem A, \sem B)$.
\end{proposition}
\begin{proof}
	This is a conclusion of Lemmas~\ref{lem:compati-clauses} and
	\ref{lem:join-scott}.
\end{proof}

We complete the denotational semantics of isos with an interpretation of
substitutions. It is not different to the one in a usual $\lambda$-calculus.

\begin{lemma}
	\label{lem:sem-subst-iso}
	Given two well-typed isos $\Psi, \phi \colon T_2 \entailiso \omega_1 \colon
	T_1$ and $\Psi \entailiso \omega_2 \colon T_2$,
	\[
		\sem{\Psi \entailiso \omega_1[\omega_2/\phi] \colon T_1}
		= \sem{\Psi, \phi \colon T_2 \entailiso \omega_1 \colon T_1}
		\circ \pv{\iid}{\sem{\Psi \entailiso \omega_2 \colon T_2}}.
	\]
\end{lemma}
\begin{proof}
	The proof is done by induction on the typing derivation of $\omega_1$.
\end{proof}

\subsection{Denotational Semantics of Valuations and Substitution}
\label{sub:rev-subst-sem}

We provide an interpretation to valuations and to their application to a term
of the syntax. As expected, the result obtained is close to
Proposition~\ref{prop:substitution-interpretation}, that details the semantics
of valuations in the last chapter. 

\begin{proposition}[Substitution lemma]
	\label{prop:rev-substitution-interpretation}
	Given a well-typed term $\Psi ; \Delta \entail t \colon A$ and for all $(x_i
	\colon A_i) \in \Delta$, a well-typed term $\Psi ; \emptyset \entail t_i \colon A_i$; if
	$\sigma = \set{x_i \mapsto t_i}_i$, then for all $g \in \sem\Psi$:
	\[
		\sem{\Psi ; \emptyset \entail \sigma(t) \colon A}(g) = \sem{\Psi ; \Delta
		\entail t \colon A}(g) \circ \left(\bigotimes_i \sem{\Psi ; \emptyset
		\entail t_i \colon A_i}(g) \right).
	\]
	We define then, for all $g \in \sem\Psi$: $\sem\sigma(g) = \bigotimes_i
	\sem{\entail t_i \colon A_i}(g)$.
\end{proposition}
\begin{proof}
	The proof is straightforward by induction on the typing derivation of $t$.
\end{proof}

\begin{remark}
	We remind here that the interpretation of a valuation, given above, is not
	done in the most meticulous way. Indeed, there is no order in which the
	variables occur in a valuation, thus its categorical denotation is
	necessarily \emph{up to permutation}, as the denotation of a context. Since
	$\CC$ is a symmetric monoidal category, we argue that the results of this
	chapter are not impacted by this choice, and that providing an extra care to
	permutations would only introduce more notations and confusion.
\end{remark}

Proposition~\ref{prop:rev-substitution-interpretation} strengthens the
observation that categorical composition is \emph{exactly} substitution in the
Curry-Howard-Lambek correspondence. In addition, in our context of category
equipped with an \emph{inverse} structure, similar to the dagger in the last
chapter, we can show that the \emph{inner product} (see
Remark~\ref{rem:loose-inner-product} to understand the intuition behind this
notion) of two interpretations produces the interpretation of a valuation,
provided that there is a match.

\begin{lemma}
	\label{lem:rev-matching-semantics}
	Given two well-typed values $\Psi ; \Delta \entail v \colon A$ and $\Psi ;
	\emptyset \entail v' \colon A$, and a substitution $\sigma$, if
	$\match{\sigma}{v}{v'}$ then for all $g \in \sem\Psi$: 
	\[ \sem v(g)\pinv \circ
	\sem{v'}(g) = \sem\sigma(g). \]
\end{lemma}
\begin{proof}
	The proof is straightforward by induction on $\match{\sigma}{b}{b'}$.
\end{proof}

The lemmas above can be combined to assert a first step towards soundness: the
interpretations of the left-hand side and right-hand side of the reversible
$\beta$-reduction are equal.

\begin{lemma}
	\label{lem:rev-beta-sound}
	Given a well-typed iso abstraction $~\entailiso \isobreduitj \colon A
	\iso B$ and a well-typed value $~\entail v' \colon A$, if
	$\match{\sigma}{v_i}{v'}$, then
	\[
		\sem{\entail \isobreduitj~v' \colon B}
		= \sem{\entail \sigma(v_i) \colon B}.
	\]
\end{lemma}
\begin{proof}
	First, we deduce from the assumption $\match{\sigma}{v_i}{v'}$ that
	\begin{itemize}
		\item $\neg(v_i~\bot~v')$, and thus $\sem{v_i}\pinv \circ \sem{v'} =
			\sem\sigma$, thanks to Lemma~\ref{lem:rev-matching-semantics}.
		\item for all $j \neq i$, $v_j~\bot~v'$, and thus $\sem{v_j}\pinv \circ
			\sem{v'} = 0$, thanks to Lemma~\ref{lem:rev-orthogonal-semantics}.
	\end{itemize}
	We can then compute the semantics, with $\omega \defeq \isobreduitj$:
	\begin{align*}
		&\ \sem{\omega~v'} & \\
		&= \sem\omega \circ \sem{v'} & \text{(by definition)} \\
		&= \left( \bigvee_j \sem{v_j} \circ \sem{b_j}\pinv \right) \circ \sem{v'}
		& \text{(by definition)} \\
		&= \bigvee_j \sem{v_j} \circ \sem{b_j}\pinv \circ \sem{v'}
		& \text{(composition distributes over join)} \\
		&= \sem{v_i} \circ \sem{b_i}\pinv \circ \sem{v'}
		& \text{(Lemma~\ref{lem:rev-orthogonal-semantics})} \\
		&= \sem{v_i} \circ \sem\sigma
		& \text{(Lemma~\ref{lem:rev-matching-semantics})} \\
		&= \sem{\sigma(v_i)}
		& \text{(Prop.~\ref{prop:rev-substitution-interpretation})}
	\end{align*}
\end{proof}

We have brought forward a categorical interpretation to the programming
language introduced in \secref{sec:rev-language}. This interpretation makes use
of the join inverse rig structure to define the iso abstraction and to perform
a pattern-matching that ensures reversibility; an enrichment in $\dcpo$ allows
to consider recursive isos and inductive data types are represented with the
help or parameterised initial algebras.

Independently of how convincing this model is, it is good practice to prove it
has a strong link with the operational semantics of the language. This link is
called \emph{adequacy}, and is the focus of the next section.

\section{Adequacy}
\label{sec:rev-adequacy}

We show a strong relation between the operational semantics and the
denotational semantics which were introduced in the previous sections. First,
we fix a mathematical interpretation $\sem -$ in a join inverse rig category
$\CC$, that is $\dcpo$-enriched and whose objects $0$ and $1$ are distinct. 

Since the language handles non-termination, our adequacy statement links the
denotational semantics to the notion of termination in the operational
semantics.

\begin{definition}[Terminating]
	Given $~\vdash t\colon A$, $t$ is said \emph{terminating} if there exists a
	value $v$ such that $t\rightarrow^* v$. We either write $t\downarrow$, or $t
	\downarrow v$.
\end{definition}

Since the system is deterministic, if $t \downarrow v$, then $v$ is unique.

The goal of this section is to prove the next theorem.
\begin{theorem}[Adequacy]
	\label{th:rev-adeq}
	Given $~\vdash t\colon A$, $t\downarrow$ iff $\sem{\entail t \colon A} \neq
	0_{1,\sem A}$.
\end{theorem}

Interestingly enough, there are two ways for a term $~\entail t \colon A$ to
have $0_{1, \sem A}$ as its interpretation: either it reduces over and over in
an infinite loop (see Example~\ref{ex:non-termination}), or it is stuck because
of pattern-matching (see Example~\ref{ex:non-reduction}). 

One strategy could be the use of formal approximation relations, introduced by
Plotkin \cite{plotkin1985predomains}. However, because of our two seperate
levels of abstraction (a fixed point calculus at the level of isos, and
inductive types at the level of terms), the definition of the relations and
proving that they exist would be long and of little interest to the reader. We
rather choose a syntactic approach, inspired by the proof in
\cite{pagani2014quantitative}.


\subsection{Soundness}
\label{sub:rev-sound}

We start by showing the simple implication in Theorem~\ref{th:rev-adeq}, that
the denotational semantics is stable w.r.t. computation; in other words,
applying a rule of the operational semantics does not change the mathematical
interpretation. The other direction -- which is made formal later -- is called
\emph{adequacy}, and is usually harder to prove. See \secref{sub:rev-adequacy}
for full details on the adequacy result.

\begin{lemma}
	\label{lem:rev-iso-sound}
	Given a well-formed iso judgement $~\entailiso \omega \colon T$,
	if $\omega \to \omega'$, then 
	\[
		\sem{\entailiso \omega \colon T}
		= \sem{\entailiso \omega' \colon T}.
	\]
\end{lemma}
\begin{proof}
	The proof is done by induction on $\to$.
	\begin{itemize}
		\item $\ffix \phi . \omega \to \omega[\ffix \phi . \omega/\phi]$.
			\begin{align*}
				&\ \sem{\ffix \phi . \omega} & \\
				&= \fix \sem\omega
				& \text{(by definition)} \\
				&= \sem\omega \circ \fix(\sem \omega)
				& \text{(fixed point)} \\
				&= \sem \omega \circ \sem{\ffix \phi . \omega}
				& \text{(by definition)} \\
				&= \sem{\omega[\ffix \phi . \omega/\phi]}
				& \text{(Lem.~\ref{lem:sem-subst-iso})} 
			\end{align*}
		\item $(\lambda \phi . \omega_1) \omega_2 \to \omega_1[\omega_2/\phi]$.
			\begin{align*}
				&\ \sem{(\lambda \phi . \omega_1) \omega_2} & \\
				&= \rmeval \circ \pv{\rmcurry(\sem{ \omega_1})}{\sem{\omega_2}}
				& \text{(by definition)} \\
				&= \sem{\omega_1} \circ \sem{\omega_2}
				& \text{(\secref{sub:ccc})} \\
				&= \sem{\omega_1[\omega_2/\phi]}
				& \text{(Lem.~\ref{lem:sem-subst-iso})} 
			\end{align*}
		\item $\omega_1 \omega_2 \to \omega'_1 \omega_2$. 
			Direct with the induction hypothesis.
	\end{itemize}
\end{proof}

\begin{proposition}[Soundness]
	\label{prop:rev-soundness}
	Given a well-formed term judgement $~\entail t \colon A$,
  if $t \to t'$, then 
	\[
		\sem{\entail t \colon A} = \sem{\entail t' \colon A}.
	\]
\end{proposition}
\begin{proof}
	The proof is done by induction on $\to$.
	\begin{itemize}
		\item $\isobreduit v' \to \sigma(v_j)$. This is covered by Lemma~\ref{lem:rev-beta-sound}.
		\item $\letv{p}{v}{t} \to \sigma(t)$. This is a conclusion of
			Prop.~\ref{prop:rev-substitution-interpretation}.
		\item $\omega~t \to \omega'~t$. We conclude with the previous lemma
			(Lemma~\ref{lem:rev-iso-sound}).
		\item $\ini t \to \ini t'$ when $t \to t'$. The induction hypothesis gives $\sem t = \sem{t'}$
			and then 
			\[
				\sem{\ini t} = \iota_i \sem t = \iota_i \sem{t'} = \sem{\ini t'}.
			\]
		\item $t \otimes t_2 \to t' \otimes t_2$ when $t \to t'$.
			\[
				\sem{t \otimes t_2} = \sem{t} \otimes \sem{t_2} \overset{\text{IH}}{=} 
				\sem{t'} \otimes \sem{t_2} = \sem{t' \otimes t_2}.
			\]
		\item $\isobreduit t \to \isobreduit t'$ when $t \to t'$. In general,
			\[
				\sem{\omega~t} = \sem{\omega} \sem t \overset{\text{IH}}{=} \sem\omega \sem{t'}
				= \sem{\omega~t'}.
			\]
		\item $\fold t \to \fold t'$ when $t \to t'$. 
			\[
				\sem{\fold t} = \alpha \sem t \overset{\text{IH}}{=} \alpha \sem{t'} = \sem{\ini t'}.
			\]
		\item $\letv{p}{t}{t_2} \to \letv{p}{t'}{t_2}$ when $t \to t'$.
			\[
				\sem{\letv{p}{t}{t_2}} = \sem{t_2} (\iid \otimes \sem t)
				\overset{\text{IH}}{=} \sem{t_2} (\iid \otimes \sem{t'})
				= \sem{\letv{p}{t'}{t_2}}.
			\]
	\end{itemize}
\end{proof}

\begin{corollary}
	\label{cor:rev-soundness}
	Given a well-formed term judgement $~\entail t \colon A$ such that
	$t \downarrow$, then 
	\[
		\sem{\entail t \colon A} \neq 0_{\sem A}.
	\]
\end{corollary}
\begin{proof}
	Knowing that there is a value $v$ such that $t \to^* v$,
	Prop.~\ref{prop:rev-soundness} ensures that 
	\[
		\sem t = \sem v \neq 0.
	\]
\end{proof}

However, this was only the simple direction of the main theorem. The proof of the
other implication is the focus of the next section.
 
\subsection{Proof of Adequacy}
\label{sub:rev-adequacy}

Our proof of adequacy involves a finitary sublanguage, where the number of
recursive calls is controlled syntactically. We show the adequacy result for
the finitary terms thanks to strong normalisation, and then show that it
implies adequacy for the whole language; this is achieved by observing that a
normalising finitary term is also normalising in its non-finitary form. 

\paragraph{Finitary sublanguage.}
We introduce the syntax for finitary terms, where the number of possible
reductions is limited by the syntax itself. The grammar of finitary isos
is given by:
\[
	\omega~::=~\isobreduit \alt \lambda \phi . \omega \alt \phi \alt \omega
	\omega \alt \nfix n \phi . \omega
\]
where $n$ is a natural number. The iso $\nfix n \phi . \omega$ has the same
typing rule as $\ffix \phi . \omega$ above. We introduce syntactic sugar,
that denotes an expression that will never reduce, by induction on iso types:
\[
	\begin{array}{c}
		\Omega_{A \iso B} \defeq \set{\mid \cdot} 
		\qquad
		\Omega_{T_1 \to T_2} \defeq \lambda \phi^{T_1} . \Omega_{T_2}
	\end{array}
\]
The syntax of finitary terms does not change compared to the language presented
at the beginning of the chapter, with the addition of a term $\bot$, which
indicates that there was no match when applying an iso abstraction to a value.
The finitary operational semantics is then defined as:
\[
	\begin{array}{c}
		\infer{\nfix 0 \phi^ T . \omega \finto \Omega_T}{}
		\qquad
		\infer{\nfix{n+1} \phi . \omega \finto \omega[\nfix n \phi . \omega /
		\phi]}{}
		\mynl
		\infer{(\lambda \phi. \omega_1) \omega_2 \finto \omega_1[\omega_2/\phi]}{}
		\qquad
		\infer{\omega_1\omega_2 \finto \omega'_1 \omega_2}{\omega_1 \finto \omega'_1}
	\end{array}
\]
\[
	\begin{array}{c}
		\infer{ \isobreduit v' \finto \sigma(e_i)}{
			\match{\sigma}{v_i}{v'}}
		\qquad
		\infer{ \isobreduit v' \finto \bot}{
			\forall i, \neg(\match{\sigma}{v_i}{v'})}
		\\[1.5ex]
		\infer{C_\to[t_1] \finto C_\to[t_2]}{t_1 \finto t_2}
		\qquad
		\infer{C_\to[t] \finto \bot}{t \finto \bot}
		\\[1.5ex]
		\infer{\letv{p}{v}{t} \finto \sigma(t)}{\match{\sigma}{p}{v}}
		\qquad
		\infer{\omega~t \finto \omega'~t}{\omega \finto \omega'}
	\end{array}
\]

\begin{lemma}[Iso Subject Reduction]
	\label{lem:fin-iso-subject-reduction}
	If $\Psi \entailiso \omega \colon T$ is well-formed, $\omega$ is finitary and
	$\omega \finto \omega'$, then $\Psi \entailiso \omega' \colon T$.
\end{lemma}
\begin{proof}
	Strongly similar to Lemma~\ref{lem:iso-subject-reduction}.
\end{proof}

\begin{lemma}[Iso Progress]
	\label{lem:iso-fin-progress}
	If $\Psi \entailiso \omega \colon T$ is well-formed and $\omega$ is finitary,
	$\omega$ is either an iso value or there exists $\omega'$ such that $\omega
	\finto \omega'$.
\end{lemma}
\begin{proof}
	Strongly similar to Lemma~\ref{lem:iso-progress}.
\end{proof}

\begin{lemma}[Subject Reduction]
	\label{lem:fin-subject-reduction}
	If $\Psi; \Delta \entail t \colon A$ is well-formed, $t$ is finitary and $t
	\finto t'$, then $\Psi; \Delta \entail t' \colon A$ is also well-formed.
\end{lemma}
\begin{proof}
	Strongly similar to Lemma~\ref{lem:rev-subject-reduction}.
\end{proof}

\begin{lemma}[Progress]
	\label{lem:fin-progress}
	If $~\entail t \colon A$ and $t$ is finitary, then:
	\begin{itemize}
		\item either $t$ is a value,
		\item or $t \to \bot$,
		\item or there exists $t'$ such that $t \finto t'$.
	\end{itemize}
\end{lemma}
\begin{proof}
	The proof is done by induction on $~\entail t \colon A$.
	\begin{itemize}
		\item $~\entail * \colon \one$ is a value.
		\item $~\entail t_1 \otimes t_2 \colon A \otimes B$.
			By induction hypothesis, either $t_1$ reduces, in which case $t_1 \otimes
			t_2$ reduces too, or $t_1$ is a value. If $t_1$ is a value, by induction
			hypothesis, either $t_2$ reduces, in which case $t_1 \otimes t_2$
			reduces, or $t_2$ is a value, and thus $t_1 \otimes t_2$ is a value.
		\item $\Psi ; \Delta \entail \ini t \colon A_1 \oplus A_2$.
			By induction hypothesis, either $t$ reduces, in which case $\ini t$ too,
			or $t$ is a value, and thus $\ini t$ is a value.
		\item $\Psi ; \Delta \entail \fold t \colon \mu X . A$.
			By induction hypothesis, either $t$ reduces, in which case $\fold t$ too,
			or $t$ is a value, and thus $\fold t$ is a value.
		\item $\Psi ; \Delta \entail \omega~t \colon B$.
			Thanks to Lemma~\ref{lem:iso-fin-progress}, $\omega$ either reduces, in which 
			case $\omega~t$ also reduces, or it is an iso value $\isobreduit$. In the last case,
			in the induction hypothesis gives that either $t$ reduces, in which case
			$\omega~t$ also reduces, or $t$ is value. In that case, either $t$ matches with
			one $v_i$, and $\omega~t$ reduces, or it does not, and $\omega~t \to \bot$.
		\item $\Psi ; \Delta_1, \Delta_2 \entail \letv{p}{t_1}{t_2} \colon B$.
			In both cases of the induction hypothesis, $\letv{p}{t_1}{t_2}$ reduces.
	\end{itemize}
\end{proof}


\paragraph{Strong Normalisation.}
We prove that the reduction $\finto$ is strongly normalising, by observing that
this system is separated in two very distinct systems: one that
reduces the iso $\lambda$-terms, and another that performs the reversible
computations. We show that both those systems can be extended to commute with
each other, which ensures strong normalisation as long as both are strongly
normalising. We start by introducing the system that performs the reductions on
isos.

\[
	\begin{array}{c}
		\infer{\nfix 0 \phi^ T . \omega \isoto \Omega_T}{}
		\qquad
		\infer{\nfix{n+1} \phi . \omega \isoto \omega[\nfix n \phi . \omega /
		\phi]}{}
		\mynl
		\infer{(\lambda \phi. \omega_1) \omega_2 \isoto \omega_1[\omega_2/\phi]}{}
		\qquad
		\infer{\omega_1\omega_2 \isoto \omega'_1 \omega_2}{\omega_1 \isoto \omega'_1}
	\end{array}
\]
\[
	\begin{array}{c}
		\infer{
			\isobreduit \isoto \left\{
				\begin{array}{ll}
					\alt v_i \iso e_i & \text{ if } i \neq j \\
					\alt v_j \iso e'_j & \text{ else}
				\end{array}
			\right\}
		}{
			e_j \isoto e'_j
		}
		\mynl
		\mynl
		\infer{C_\to[t_1] \isoto C_\to[t_2]}{t_1 \isoto t_2}
		\qquad
		\infer{\letv{p}{t'}{t_1} \isoto \letv{p}{t'}{t_2}}{t_1 \isoto t_2}
		\qquad
		\infer{\omega~t \isoto \omega'~t}{\omega \isoto \omega'}
	\end{array}
\]

\begin{lemma}
	\label{lem:iso-sn}
	The reduction system $\isoto$ is strongly normalising.
\end{lemma}
\begin{proof}
	We translate finitary isos and finitary terms into a simply-typed
	$\lambda$-calculus with pairs. We write $\abs t$ the translation of $t$.
	\[
		\begin{array}{c}
			\abs * \defeq *
			\qquad
			\abs x \defeq *
			\qquad
			\abs{\ini t} \defeq \abs t
			\qquad
			\abs{t \otimes t'} \defeq \pv{\abs t}{\abs{t'}}
			\mynl
			\abs{\fold t} \defeq \abs t
			\qquad
			\abs{\omega~t} \defeq \pv{\abs\omega}{\abs t}
			\qquad
			\abs{\letv{p}{t}{t'}} \defeq \pv{\abs t}{\abs{t'}}
			\mynl
			\abs{\isobreduit} \defeq \langle \abs{e_i} \rangle _{i \in I}
			\qquad
			\abs{\lambda \phi . \omega} \defeq \lambda \phi . \abs \omega
			\qquad
			\abs \phi \defeq \phi
			\mynl
			\abs{\omega_2\omega_1} \defeq \abs{\omega_2} \abs{\omega_1}
			\qquad
			\abs{\nfix n \phi . \omega} \defeq \nfix n \phi . \abs \omega
		\end{array}
	\]
	This $\lambda$-calculus is strong normalising; this can be proven with
	candidates of reducibility, \emph{à la} System F \cite[Chapters 11 and
	14]{girard1989proofs}. 
\end{proof}

We then introduce the reductions that perform the reversible computation,
our equivalent of $\beta$-reduction but for our iso language.
\[
	\begin{array}{c}
		\infer{ \isobreduit v' \termto \sigma(e_i)}{
			\match{\sigma}{v_i}{v'}}
		\qquad
		\infer{ \isobreduit v' \termto \bot}{
			\forall i, \neg(\match{\sigma}{v_i}{v'})}
		\\[1.5ex]
		\infer{C_\to[t_1] \termto C_\to[t_2]}{t_1 \termto t_2}
		\qquad
		\infer{C_\to[t] \termto \bot}{t \termto \bot}
		\qquad
		\infer{\letv{p}{v}{t} \termto \sigma(t)}{\match{\sigma}{p}{v}}
	\end{array}
\]
This system is strongly normalising thanks to a decreasing argument:
the number of isos and $\mathtt{let}$ contructors strictly decreases
when applying the reduction $\termto$.
\begin{lemma}
	\label{lem:term-red-sn}
	The reduction system $\termto$ is strongly normalising.
\end{lemma}

\begin{lemma}
	\label{lem:red-commutative}
	$\termto~\isoto~\subseteq~\isoto~\termto$. 
\end{lemma}
	
We say that $\isoto$ commutes \cite{bachmair1986commutation} over $\termto$.
This and the strong normalisation of both systems $\isoto$ and $\termto$
ensures the strong normalisation of them combined $\isoto \cup \termto$
\cite[Theorem 1]{bachmair1986commutation}.

\begin{lemma}
	\label{lem:sub-system-fin}
	$\finto~\subseteq~\isoto \cup \termto$.
\end{lemma}
\begin{proof}
	The proof is direct, by showing that any rule in $\finto$ is either
	in $\isoto$ or $\termto$.
\end{proof}

\begin{theorem}
	\label{th:fin-sn}
	The reduction system $\finto$ is strongly normalising.
\end{theorem}
\begin{proof}
	With \cite[Theorem 1]{bachmair1986commutation} and Lemmas~\ref{lem:iso-sn},
	\ref{lem:term-red-sn}, and~\ref{lem:red-commutative}, we can ensure that
	$\isoto \cup \termto$ is strongly normalising. We conclude then with
	Lemma~\ref{lem:sub-system-fin}, that shows that $\finto$ is a subsystem of a
	strongly normalising system.
\end{proof}

\paragraph{Finitary adequacy.}
We prove adequacy, but for finitary terms. To do so, we also need to introduce
the denotational semantics of finitary isos. The interpratation of $\nfix n$,
instead of being Kleene's fixed point, is the morphism obtained by unfolding
$n$ times. The interpretation of $\Omega_T$ is the bottom element of $\sem T$.

\begin{lemma}
	\label{lem:fin-iso-sound}
	Given a well-formed finitary iso judgement $~\entailiso \omega \colon T$, if
	$\omega \finto \omega'$, then 
	\[
		\sem{\entailiso \omega \colon T}
		= \sem{\entailiso \omega' \colon T}.
	\]
\end{lemma}
\begin{proof}
	Strongly similar to Lemma~\ref{lem:rev-iso-sound}.
\end{proof}

\begin{proposition}[Finitary Soundness]
	\label{prop:fin-soundness}
	Given a well-formed finitary term judgement $~\entail t \colon A$, if $t
	\finto t'$, then 
	\[
		\sem{\entail t \colon A} = \sem{\entail t' \colon A}.
	\]
\end{proposition}
\begin{proof}
	Strongly similar to Prop.~\ref{prop:rev-soundness}.
\end{proof}

\begin{theorem}[Finitary Adequacy]
	\label{th:fin-adeq}
	Given a well-formed finitary term judgement $~\entail t \colon A$,
	$t \downarrow$ iff $\sem{\entail t \colon A} \neq 0_{\sem A}$.
\end{theorem}
\begin{proof}
	We prove both directions of the double implication.
	\begin{itemize}
		\item[$(\Rightarrow)$] Knowing that $t \downarrow$, there exists a value
			$v$ such that $t \finto^* v$, and Prop.~\ref{prop:fin-soundness} ensures
			that $\sem t = \sem v \neq 0$.
		\item[$(\Leftarrow)$] We know that $\finto$ is strongly normalising (see
			Th.~\ref{th:fin-sn}), which means that the reduction from $t$ terminates,
			and Lemma~\ref{lem:fin-progress} ensures that it terminates either on a
			value $v$ or on $\bot$. However, $\sem t \neq 0$, thus it cannot
			terminate on $\bot$ because of Prop.~\ref{prop:fin-soundness}. We have
			then $t \finto^* v$, which concludes.
	\end{itemize}
\end{proof}

\paragraph{Finitary Subterms.}
We conclude in two steps. First we observe that the interpretation of a term is
nothing more than the join of the interpretations of its finitary subterms,
then we show that a reduction $\to^*$ can be linked to a finitary reduction
$\finto^*$. 

\begin{definition}[Finitary Subiso]
	\label{def:fin-subiso}
	Let $\subfin$ be the smallest relation between (finitary or not) isos such that:
	\[
		\begin{array}{c}
			\infer{\nfix n \phi . \omega \subfin \ffix \phi . \omega}{}
			\qquad
			\infer{\omega[\omega_1/\phi] \subfin
			\omega[\omega_2/\phi]}{\omega_1 \subfin \omega_2}
		\end{array}
	\]
\end{definition}

\begin{lemma}
	\label{lem:subfin-to-dcpo}
	Given two well-formed (finitary or not) iso judgements
	$\Psi \entailiso \omega_1 \colon T$ and $\Psi \entailiso \omega_2 \colon T$
	such that $\omega_1 \subfin~\omega_2$, then
	\[
		\sem{\Psi \entailiso \omega_1 \colon T} 
		\leq \sem{\Psi \entailiso \omega_2 \colon T}.
	\]
\end{lemma}
\begin{proof}
	Direct.
\end{proof}

\begin{lemma}
	\label{lem:subfin-join}
	Given a well-formed iso judgement $\Psi \entailiso \omega \colon T$,
	we have:
	\[
		\sem{\Psi \entailiso \omega \colon T} = 
		\bigvee_{\underset{\omega' \text{ finitary}}{\omega'\subfin~\omega}}
		\sem{\Psi \entailiso \omega' \colon T}.
	\]
\end{lemma}
\begin{proof}
	We observe that, by definition:
	\[
		\sem{\Psi \entailiso \ffix \phi . \omega \colon T}
		= \bigvee_{n \in \N} \sem{\Psi \entailiso \nfix n \phi . \omega \colon T}
	\]
	which proves the desired result in the case of $\ffix \phi . \omega$. The
	general conclusion falls by induction.
\end{proof}

We generalise to terms the definition of subisos given above.

\begin{definition}[Finitary Subterm]
	\label{def:fin-subterm}
	Let $\subfin$ be the smallest congruence relation between (finitary or not)
	terms such that:
	\[
		\begin{array}{c}
			\infer{\omega_1~t \subfin~\omega_2~t}{\omega_1 \subfin~\omega_2}
		\end{array}
	\]
\end{definition}

The following lemma follows from the previous definition and
Lemma~\ref{lem:subfin-join}; this is because composition is distributive
with joins (see Definition~\ref{def:join}).

\begin{lemma}
	\label{lem:subfin-join-term}
	Given a well-formed term judgement $\Psi; \Delta \entail t \colon A$,
	we have:
	\[
		\sem{\Psi; \Delta \entail t \colon A} = 
		\bigvee_{\underset{t' \text{ finitary}}{t'\subfin~t}}
		\sem{\Psi; \Delta \entail t' \colon A}.
	\]
\end{lemma}

It is also the right time to observe that if a term has a finitary subterm that
reduces to a value eventually, the former also normalises to the same value.

\begin{lemma}
	\label{lem:fin-red-to}
	Given a well-formed closed term judgement $~\entail t \colon A$, if
	there exists a finitary subterm $t' \subfin t$ and a value such that
	$t' \finto^* v$, then $t \to^* v$.
\end{lemma}
\begin{proof}
	The finitary term $t'$ has the same reduction steps as $t$ up to a point.
	Lemma~\ref{lem:fin-progress} ensures that this end point is either a value or
	$\bot$ in the finitary case. Thus if the reduction from $t'$ gets to a value,
	the reduction from $t$ must also finish on a value. Since the reduction steps
	were exactly the same, both reductions have the same normal form.
\end{proof}

\paragraph{Conclusion.}
We finally have all the tools to conclude with adequacy for closed terms of our
original language.

\begin{proof}[Proof of Adequacy (Theorem~\ref{th:rev-adeq})]
	There are two implications to prove.
	\begin{itemize}
		\item[$(\Rightarrow)$] This first implication is proven as
			Corollary~\ref{cor:rev-soundness}.
		\item[$(\Leftarrow)$] Suppose that $\sem t \neq 0$. Necessarily, thanks to
			Lemma~\ref{lem:subfin-join-term}, there exists a finitary term $t'$ such
			that $t' \subfin t$ and $\sem{t'} \neq 0$. In Theorem~\ref{th:fin-adeq},
			we have proven adequacy for finitary terms, meaning that there exists a
			value $v$ such that $t' \finto^* v$. Lemma~\ref{lem:fin-red-to} ensures
			then that $t \to^* v$, which concludes.
	\end{itemize}
\end{proof}

\section{Expressivity}
\label{sec:expressivity}

This section is devoted to assessing the expressivity of the
language. To that end, we rely on Reversible Turing Machine
(RTM)~\cite{axelsen11rtm}. We describe how to encode an RTM as an iso,
and prove that the iso realises the string semantics of the RTM.

\paragraph{Reference.} The work in this section has been done by Kostia
Chardonnet, firstly as a part of his thesis \cite{phd-kostia}, and in our paper
\cite{nous2023invrec}. Is is presented here out of coherence with the next
section.

\subsection{Recovering duplication, erasure and manipulation of constants}
\label{sec:dup}
Although the language is linear and reversible, since closed values are all
finite, and one can build isos to encode notions of duplication, erasure, and
constant manipulation thanks to partiality.

\begin{definition}[Duplication \cite{nous2023invrec}]
  We define $\dup_A^S$ the iso of type $A\iso A\otimes A$ which can
  duplicate any closed value of type $A$ by induction on $A$, where $S$
  is a set of pairs of a type-variable $X$ and an iso-variable
  $\isovar$, such that for every free-variable $X \subseteq A$, there
  exists a unique pair $(X, \isovar) \in S$ for some $\isovar$.
  The iso is defined by induction on $A$:
  $\dup_\one^S = \{() \iso \pv{()}{()}\}$, and
  \begin{itemize}
    \item $\dup_{A\otimes B}^S = \left\{\begin{array}{lcl}
      \pv{x}{y} & \iso & \letv{\pv{x_1}{x_2}}{\dup_A^S~x}{}
       \letv{\pv{y_1}{y_2}}{\dup_B^S~y}{} \\
			&& \pv{\pv{x_1}{y_1}}{\pv{x_2}{y_2}}
    \end{array}\right\}$;
    \item $\dup_{A\oplus B}^S = \left\{\begin{array}{lcl}
      \inl{(x)} & \iso & \letv{\pv{x_1}{x_2}}{\dup_A^S~x}{}
      \pv{\inl{(x_1)}}{\inl{(x_2)}} \\[0.5em]
      \inr{(y)} & \iso & \letv{\pv{y_1}{y_2}}{\dup_B^S~y}{}
      \pv{\inr{(y_1)}}{\inr{(y_2)}}
    \end{array}\right\}$;
    \item If $(X, \_) \not\in S$: $\dup_{\mu X. A}^S = \ffix\isovar. \left\{\begin{array}{l@{~}c@{~}l}
      \fold{(x)}
      & \iso
      &
        \letv{\pv{x_1}{x_2}}{Dup_{A[\mu X. A/X]}^{S\cup \set{(X , \isovar)}}~x}{} \\
      && \pv{\fold{(x_1)}}{\fold{(x_2)}}
    \end{array}\right\}$;

\item If $(X, \isovar) \in S$:
  $\dup_{\mu X. A}^S = \set{x \iso
    \letv{\pv{x_1}{x_2}}{\isovar~x}{\pv{x_1}{x_2}}}$.
\end{itemize}
\end{definition}

\begin{lemma}[Properties of Duplication \cite{nous2023invrec}]
  \label{lem:duplication-invariant}
  \label{lem:duplication-typed}
  Given a closed type $A$, then $\dup_A^\emptyset$ is well-defined,
  and the iso $\dup_A^\emptyset$ is well typed of type
  $A\iso A\otimes A$.
\end{lemma}

\begin{lemma}[Semantics of Duplication \cite{nous2023invrec}]
  \label{lem:duplication-semantics}
  Given a closed type $A$ and a closed value $v$ of type $A$, then
  $\dup_A^\emptyset~v \to^* \pv{v_1}{v_2}$ and $v = v_1 = v_2$.
\end{lemma}

\begin{definition}[Constant manipulation \cite{nous2023invrec}]
  \label{def:constant}
  We define $\opn{erase}_v \colon A \otimes \Sigma^T \iso A$ which erase
  its second argument when its value is $v$ as $\{\pv{x}{v} \iso
  x\}$. Reversed, it turns any $x$ into $\pv{x}{v}$.
\end{definition}

\subsection{Definition of Reversible Turing Machine}

\begin{definition}[Reversible Turing Machine~\cite{axelsen11rtm}]
  Given a Turing Machine $M = (Q, \Sigma, \delta, b, q_s, q_f)$, where
  $Q$ is a set of states, $\Sigma = \set{b, a_1, \dots, a_n}$ is a
  finite set of tape symbols (in the following, $a_i$ and $b$ always
  refer to elements of $\Sigma$),
  $\delta \subseteq \Delta = (Q\times
  [(\Sigma\times\Sigma)\cup\set{\leftarrow, \downarrow,
    \rightarrow}]\times Q)$ is a partial relation defining the
  transition relation such that there must be no transitions leading
  out of $q_f$ nor into $q_s$, $b$ a blank symbol and $q_s$ and $q_f$
  the initial and final states. We say that $M$ is a \emph{Reversible
    Turing Machine} (RTM) if it is:
  \begin{itemize}
      \item \textit{forward} deterministic: for any two distinct pairs
  of triples $(q_1, a_1, q_1')$ and $(q_2, a_2, q_2')$ in $\delta$, if
  $q_1 = q_2$ then $a_1 = (s_1, s_1')$ and $a_2 = (s_2, s_2')$ and
  $s_1\not= s_2$.

  \item \textit{Backward} deterministic: for any two distinct pairs
  of triples $(q_1, a_1, q_1')$ and $(q_2, a_2, q_2')$ in $\delta$, if
  $q_1' = q_2'$ then $a_1 = (s_1, s_1')$ and $a_2 = (s_2, s_2')$ and
  $s_1'\not= s_2'$.
\end{itemize}
\end{definition}

\begin{definition}[Configurations~\cite{axelsen11rtm}]
  \label{confTM}
  A \emph{configuration} of a RTM is a tuple $(q, (l, s, r)) \in
  \opn{Conf} = Q\times (\Sigma^* \times \Sigma \times \Sigma^*)$ where
  $q$ is the internal state, $l, r$ are the left and right parts of
  the tape (as string) and $s\in \Sigma$ is the current symbol being
  scanned. A configuration is \emph{standard} when the cursor is on
  the immediate left of a finite, blank-free string $s \in
	(\Sigma\setminus\set{b})^*$ and the rest is blank, \emph{i.e.}~it is in
  configuration $(q, (\epsilon, b, s))$ for some $q$, where $\epsilon$
  is the empty string, representing an infinite sequence of blank
  symbols $b$.
\end{definition}

\begin{definition}[RTM Transition~\cite{axelsen11rtm}]
  An RTM $M$ in configuration $C = (q, (l, s, r))$ goes to
  a configuration $C' = (q', (l', s', r'))$, written $T\vdash
  C\rightsquigarrow C'$ in a single step if there exists
  a transition $(q, a, q')\in\delta$ where $a$ is either $(s, s')$,
  and then $l=l'$ and $r=r'$ or $a\in\{\leftarrow,\downarrow,
  \rightarrow\}$, and we have for the case $a=\leftarrow$:
  $l' = l \cdot s$ and for $r = x \cdot r_2$ we have $s' = x$ and
  $r' = r_2$, similarly for the case $a = \rightarrow$ and for the
  case $a =\downarrow$ we have $l' = l$ and $r' = r$ and $s = s'$.
\end{definition}

The semantics of an RTM is given on \textit{standard configurations} of
the form $(q, (\epsilon, b, s))$ where $q$ is a state, $\epsilon$ is
the finite string standing for a blank-filled tape, and $s$ is the
blank-free, finite input of the RTM.

\begin{definition}[String Semantics~\cite{axelsen11rtm}]
  \label{RTM:String-Sem} 
	The semantics of a RTM $M$, written
  $\opn{Sem}(M)$ is defined on standards configurations and is given
  by the set $\opn{Sem}(M) = \set{(s, s') \in ((\Sigma\backslash\set{b})^*
  \times (\Sigma\backslash\set{b})^*) \mid M \vdash (q_s, (\epsilon,
  b, s)) \rightsquigarrow^* (q_f, (\epsilon, b, s'))}$.
\end{definition}

\begin{theorem}[Properties of RTM~\cite{axelsen11rtm}]
  For all RTM $M$, $\opn{Sem}(M)$ is the graph of an injective
  function. Conversely, all injective computable functions (on a
  tape) are realisable by a RTM. Finally, any Turing Machine can be
  simulated by a Reversible Turing Machine.
\end{theorem}

\subsection{Encoding RTMs as Isos}
A RTM configuration is a set-based construction that we can model
using the type constructors available in our language. Because the
transition relation $\delta$ is backward and forward deterministic, it
can be encoded as an iso. Several issues need to be dealt with; we
discuss them in this section.

\paragraph{Encoding configurations.}
The set of states $Q = \{q_1, \dots, q_n\}$ is modelled with the type
$Q^T=\1\oplus\cdots\oplus\1$ ($n$ times). The encoding of the state
$q_i$ is then a closed value $q_i^T$. They are pairwise orthogonal.
The set $\Sigma$ of tape symbols is represented similarly by
$\Sigma^T=\1\oplus\cdots\oplus\1$, and the encoding of the tape symbol
$a$ is $a^T$.
We then define the type of configurations in the obvious manner: a
configuration $C = (q, (l, s, r))$ corresponds to a closed value
$\opn{isos}(C)$ of type
$Q^T \otimes ([\Sigma^T] \otimes \Sigma^T \otimes [\Sigma^T])$.

\paragraph{Encoding the transition relation $\delta$.}
A limitation of our language is that every sub-computation has to be
reversible and does not support infinite data structures such
as streams. In the context of RTMs, the empty string $\epsilon$ is
assimilated with an infinite string of blank symbols. If this can be
formalised in set theory, in our limited model, we cannot emit blank
symbols out of thin air without caution.

In order to simulate an infinite amount of blank symbols on both sides
of the tape during the evaluation, we provide an iso that grows the
size of the two tapes on both ends by blank symbols at each transition
step.
The iso $\opn{growth}$ is shown in
Table~\ref{tab:useful-functions}. It is built using three auxiliary
functions, written in a Haskell-like notation
$\opn{len}$ sends a closed value $[v_1, \dots, v_n]$ to
$\pair{[v_1, \dots, v_n]}{\ov{n}}$.
$\opn{snoc'}$ sends $\pair{[v_1, \dots, v_n]}{v, \ov{n}}$ to
$\pair{[v_1, \dots, v_n, v]}{v, \ov{n}}$.
$\opn{snoc}$ sends $\pv{[v_1, \dots, v_n]}{v})$ to
$\pv{[v_1, \dots, v_n, v]}{v}$.
Finally, $\opn{growth}$ sends
$\pv{[a_1^T, \dots, a_n^T]}{[a_1'^T, \dots, a_m'^T]}$ to
$\pv{[a_1^T, \dots, a_n^T, b^T]}{[a_1'^T, \dots, a_m'^T, b^T]}$.

Now, given a RTM $M = (Q, \Sigma, \delta, b, q_s, q_f)$, a relation
$(q, r, q') \in \delta$ is encoded as a clause between values
$\opn{iso}(q,r,q') = v_1 \iso v_2$ of type $C^T \iso C^T$. These clauses
are defined by case analysis on $r$ as follows.
When $x, x', z, y$ and $y'$ are variables:
\begin{itemize}
\item $\opn{iso}(q, \rightarrow, q') =
  (q^T, (x', z, y :: y')) \iso \letv{(l,
    r)}{\opn{growth}~(x', y')}{(q'^T, (z :: l, y, r))}$,
\item $\opn{iso}(q, \leftarrow, q') =
  (q^T, (x::x', z, y')) \iso \letv{(l, r)}{\opn{growth}~(x',
    y')}{(q'^T, (l, x, z :: r))}$,
\item $\opn{iso}(q, \downarrow, q') =
  (q^T, (x', z, x')) \iso \letv{(l, r)}{\opn{growth}~(x',
    y')}{(q'^T, (l, z, r))}$,
\item $\opn{iso}(q, (s, s'), q') =
  (q^T, (x', s^T, y')) \iso \letv{(l, r)}{\opn{growth}~(x',
    y')}{(q'^T, (l, s'^T, r))}$.
\end{itemize}
The encoding of the RTM $M$ is then the iso $\opn{isos}(M)$ whose
clauses are the encoding of each rule of the transition relation
$\delta$, of type $\opn{Conf}^T\iso\opn{Conf}^T$.

\begin{table}
	\begin{tabular}{@{}l}
		$
		\begin{array}{|l}
			\opn{len}:[A]\iso[A] \otimes \natT\\
			\begin{array}{@{}l@{\,}l@{~}c@{~}l}
				\opn{len} & [~] & {\iso} & ([~], 0) \\
				\opn{len} & h :: t & {\iso} & \letv{(t', n)}{\opn{len}~t}{} \\
				&   & & (h :: t', S(n)) \\
			\end{array}
		\end{array}
		$
		\\[6ex]
		$
		\begin{array}{|l}
			\opn{snoc'} \colon [A] \otimes A \otimes \natT \iso [A] \otimes A \otimes \natT \\
			\begin{array}{@{}l@{\,}l@{~}c@{~}l}
				\opn{snoc'}& ([~], x, 0) & {\iso} & \letv{(x_1, x_2)}{\dup_A^\emptyset x}{}\\
				& & & ([x_1], x_2, 0) \\
				\opn{snoc'} & (h::t, x, S(n)) & {\iso} & \letv{(t', x', n')}{\opn{snoc'} (t, x, n)}{} \\
				&&&(h :: t', x', S(n'))
			\end{array}
		\end{array}
		$
		\\[8ex]
		$
		\begin{array}{|l}
			\opn{growth} \colon [\Sigma^T] \otimes [\Sigma^T] \iso [\Sigma^T]
			\otimes [\Sigma^T]\\
			\begin{array}{@{}l@{\,}l@{~}c@{~}l}
				\opn{growth}
				& (l, r)
				& {\iso}
				& \letv{\pv{l'}{b_1}}{\opn{snoc}\pv{l}{b^T}}{} \\
				&& & \letv{\pv{r'}{b_2}}{\opn{snoc}\pv{r}{b^T}}{} \\
				&& & \letv{l''}{\opn{erase}_b \pv{l'}{b_1}}{} \\
				&& & \letv{r''}{\opn{erase}_b \pv{r'}{b_2}}{} (l'', r'')
			\end{array}
		\end{array}
		$
		\\[8ex]
		$
		\begin{array}{|l}
			\opn{It}: (A\iso A \otimes (\one\oplus\one)) \to (A\iso
			A\otimes \natT)\\
			\begin{array}{@{}l@{\,}l@{~}c@{~}l}
				\opn{It} \isolambdavar
				& x& \iso & \letv{y}{\isolambdavar~x}{}\\
				&&& \letv{z}{
					\left\{\begin{array}{l@{~}c@{~}l}
						(y, \tc)
						&\iso
						& \letv{(z,n)}{(\opn{It}\,\isolambdavar)~y}{(z, S~n)} \\
						(y, \fc) & \iso & (y, 0)
					\end{array}\right\}
					~y}{}
				z
			\end{array}
		\end{array}
		$
		\\[6ex]
		$
		\begin{array}{|l}
			\opn{rmBlank} \colon [\Sigma] \iso [\Sigma] \otimes \natS
			\\
			\begin{array}{@{}l@{\,}l@{~}c@{~}l}
				\opn{rmBlank} & [] & {\iso} & ([], 0) \\
				\opn{rmBlank} & b^T :: t & {\iso} & \letv{(t', n)}{\opn{rmBlank}~t}{(t', S(n))} \\
				\opn{rmBlank} & a_1^T :: t & {\iso} & ((a_1^T :: t), 0) \\
				\vdots& \vdots & \vdots & \vdots \\
				\opn{rmBlank} & a_n^T :: t & {\iso} & ((a_n^T :: t), 0)
			\end{array}
		\end{array}
		$
		\\[10ex]
		$
		\begin{array}{|l}
			\opn{rev}_{\text{aux}} \colon [A] \otimes [A] \iso [A]
			\otimes [A]
			\\
			\begin{array}{@{}l@{\,}l@{~}c@{~}l}
				\opn{rev}_{\text{aux}} & ([], y) & {\iso} & ([], y) \\
				\opn{rev}_{\text{aux}} & (h :: t, y) & {\iso} & \letv{(h_1, h_2)}{\dup_A^\emptyset~h}{} \\
				&& & \letv{(t_1, t_2)}{\isovar (t, h_2 :: y)}{} \\
				&& & (h_1 :: t_1, t_2)
			\end{array}
		\end{array}
		$
		\\[8ex]
		$
		\begin{array}{|l}
			\opn{rev} \colon [A] \iso [A] \otimes [A]\\
			\opn{rev} = \set{x \iso \letv{(t_1,t_2)}{\opn{rev}_{\text{aux}}~(x, [])}{(t_1, t_2)}}
		\end{array}
		$
		\\[4ex]
		$
		\begin{array}{|l}
			\opn{cleanUp}: C^T \otimes \natT \iso C^T \otimes
			\natT \otimes \natT \otimes \natT \otimes [\Sigma^T]
			\\
			\begin{array}{@{}l@{\,}l@{~}c@{~}l}
				\opn{cleanUp}
				&
				((x, (l, y, r)), n)
				& {\iso} &
				\letv{(l', n_1)}{\opn{rmBlank}~l}{}\\
				&&&\letv{(r_{\text{ori}}, r_{\text{rev}})}{\opn{rev}~r}{}\\
				&&&\letv{(r', n_2)}{\opn{rmBlank}~r_{\text{rev}}}{}\\
				&&&((x, (l', y, r')), n, n_1, n_2, r_{\text{ori}})
			\end{array}
		\end{array}
		$
	\end{tabular}
	\caption{Some useful isos for the encoding.}
	\label{tab:useful-functions}
\end{table}

\paragraph{Encoding successive applications of $\delta$.}
The transition $\delta$ needs to be iterated until the final state is
reached. This behavior can be emulated in our language using the iso
$\opn{It}$, defined in Table~\ref{tab:useful-functions}.  The iso
$\opn{It}\,\omega$ is typed with $(A\iso A\otimes \natT)$. Fed with a
value of type $A$, it iterates $\omega$ until $\fc$ is met. It then
returns the result together with the number of iterations.

To iterate $\opn{iso}(M)$, we then only need to modify $\opn{iso}$ to
return a boolean stating whether $q_f$ was met. This can be done straightforwardly, yielding an iso $\opn{isos}_\boolT(M))$ of
type
$
  \opn{Conf}^T\iso \opn{Conf}^T\otimes(\1\oplus\1).
$
With such an iso, given $M$ be a RTM such that $M \vdash (q_s, (\epsilon, b,
s)) \rightsquigarrow^{n+1} (q_f, (\epsilon, b, (a_1, \dots, a_n)))$, then
$\opn{It}(\opn{isos}_\boolT(M))~(q_s^T, ([b^T], b^T, s^T))$ reduces to the
encoding term $((q_f^T, ([b^T, \dots, b^T], b^T, [a_1^T, \dots, a_n^T, b^T,
\dots, b^T])), \overline{n})$. If it were not for the additional blank tape
elements, we would have the encoding of the final configuration.

\paragraph{Recovering a canonical presentation.}
Removing blank states at the \emph{beginning} of a list is easy: it
can for instance, be done with the iso $\opn{rmBlank}$, shown in
Table~\ref{tab:useful-functions}. Cleaning up the tail of the list can
then be done by reverting the list, using, e.g. $\opn{rev}$ in the same
table. By abuse of notation, we use constants in some patterns: an
exact representation would use \Cref{def:constant}. Finally, we can
define the operator $\opn{cleanUp}$, solving the issue raised in the
previous paragraph.
In particular, given a RTM $M$ and an initial configuration $C$ such
that
$M \vdash C \rightsquigarrow C' = (q, (\epsilon, b, (a_1, \dots,
a_n)))$.  Then we have that
$\opn{cleanUp} \opn{It}(\opn{isos}_\boolT(M)) C^T \to^* ((q^T, ([],
b^T, [a_1^T, \dots, a_n^T])), v)$, where $v$ is of type
$\natT \otimes \natT \otimes \natT \otimes [\Sigma^T]$. If we want to
claim that we indeed capture the operational behaviour of RTMS, we need
to get rid of this value $v$.

\begin{figure}
	\begin{center}
		\includegraphics[width=0.8\textwidth]{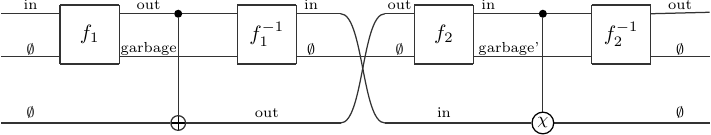}
	\end{center}
	\caption{Reversibly removing additional garbage from some process.}
	\label{fig:garbage-removal}
\end{figure}

\paragraph{Getting rid of the garbage.}
To discard this value $v$, we rely on Bennett's
trick~\cite{bennett1973logical}, shown in \Cref{fig:garbage-removal}.
Given two Turing machines $f_1$ and $f_2$ and some input $\opn{in}$
such that if $f_1(\opn{in}) = \opn{out} \otimes \opn{garbage}$ and
$f_2(\opn{out}) = \opn{in} \otimes \opn{garbage'}$, then the process
consists of taking additional tapes in the Turing Machine in order to
reversibly duplicate (represented by the $\oplus$) or reversibly erase
some data (represented by the $\chi$) in order to recover only the
output of $f_1$, without any garbage.

Given an iso $\omega \colon A\iso B\otimes C$ and
$\omega' \colon B \iso A \otimes C'$ where $C, C'$ represent garbage, we
can build an iso from $A\iso B$ as follows, where the variables
$x, y, z$ (and their indices) respectively correspond to the first,
second, and third wire of~\Cref{fig:garbage-removal}. This operator
makes use of the iso \opn{Dup} discussed in Section~\ref{sec:dup}.
\[
\begin{array}{lcl}
  \garRem{\omega}{\omega'}~x_1
  & \iso
  & \letv{\pv{x_2}{y}}{\omega~x_1}{}
    \letv{\pv{x_3}{z}}{\dup_B^\emptyset~x_2}{}\\
  && \letv{x_4}{\omega^{-1}~\pv{x_3}{y}}{}
   \letv{\pv{z_2}{y_2}}{\omega'~z}{}\\
  && \letv{z_3}{(\dup_B^\emptyset)^{-1}~\pv{z_2}{x_4}}{}
   \letv{z_4}{\omega'^{-1}~\pv{z_3}{y_2}}{}
             z_4.
\end{array}
\]

\begin{theorem}[Capturing the exact semantics of a RTM \cite{nous2023invrec}]
  For all RTM $M$ with standard configurations $C = (q_s,
  (\epsilon, b, s))$ and $C' = (q_f, (\epsilon, b, s'))$ such that
  $M \vdash C \rightsquigarrow^* C'$, we have
  \[\garRem{\opn{cleanUp}(\opn{It}
      (\opn{isos}_\boolT(M)))}{\opn{cleanUp}(\opn{It}
      (\opn{isos}_\boolT(M^{-1})))}~\opn{isos}(C) \to^* \opn{isos}(C')\]
  The behavior of RTMs is thus captured by the language.
\end{theorem}

\section{Semantics preservation}
\label{sec:sem-preservation}

In this section, we fix the interpretation $\sem -$ of the language in
$\PInj$, the category of sets and partial injections. This choice
comes without any loss of generality, and allows us to consider
\emph{computable} set-functions. In this section, we show that given a
computable, reversible set-function $f \colon \interp{A} \to \interp{B}$,
there exists an iso $\omega \colon A\iso B$ such that
$\interp{\omega} = f$. In order to do that, we fix a canonical
flat representation of our types.

\subsection{A Canonical Representation}

We define a canonical representation of closed values of some type $A$ into a new
type $\opn{Enc} = \mathbb{B} \oplus \one \oplus \one \oplus \one \oplus \one \oplus \natT$
(recall that $\mathbb{B} = \one\oplus \one$ and $\natT = \mu X. \one\oplus X$).
For simplicity let us name each the following terms of type $\opn{Enc}$ :
$\ttt = \inl{(\inl{()})}$,
$\fff = \inl{(\inr{()})},$
$S = \inr{(\inl{()})},$
$D^\oplus = \inr{(\inr{(\inl{()})})},$
$D^\otimes = \inr{(\inr{(\inr{(\inl{()})})})},$
$D^\mu = \inr{(\inr{(\inr{(\inr{(\inl{()})})})})}$,
and for every natural number $n$, we write $\tilde{n}$ for the term
$\inr{(\inr{(\inr{(\inr{(\inr{(\inr{(\overline{n})})})})})})}$.
Now, given some closed type $A$, we can define $\floor{-}_A \colon A\iso
[\opn{Enc}]$ the iso that transform any closed value of type $A$ into a list of
$\opn{Enc}$. The iso is defined inductively over $A$: $\floor{-}_\one = \{() \iso [S]\},$ and
  \[\floor{-}_{A\oplus B} = \left\{\begin{array}{lcl}
    \inl{(x)} & \iso & \letv{y}{\floor{x}_A}{D^\oplus :: \fff :: y} \\
    \inr{(x)} & \iso & \letv{y}{\floor{x}_B}{D^\oplus :: \ttt :: y}
  \end{array}\right\},\]
  \[\floor{-}_{A \otimes B} = \left\{\begin{array}{lcl}
    \pv{x}{y} & \iso & \letv{x'}{\floor{x}_A}{}
                       \letv{y'}{\floor{y}_B}{} \\
              & & \letv{\pv{z}{n}}{++~\pv{x'}{y'}}{}
                  D^\otimes :: \tilde{n} :: z
  \end{array}\right\},\]
  \[\floor{-}_{\mu X. A} = \left\{\begin{array}{lcl}
    \fold{x} & \iso & \letv{y}{\floor{x}_{A[\mu X. A/X]}}{}
    D^\mu :: y
  \end{array}\right\},\]
where the iso $++ \colon [A]\otimes [A] \iso [A]\otimes \natT$ is defined as:
\[\ffix f. \left\{\begin{array}{lcl}
  \pv{[]}{x} & \iso & \pv{x}{0} \\
  \pv{h :: t}{x} & \iso & \letv{\pv{y}{n}}{f~\pv{t}{x}}{}
  \pv{h :: y}{S(n)}
\end{array}\right\}.\]

\subsection{Capturing every computable injection}

With this encoding, every iso $\omega \colon A\iso B$ can be turned into another iso
$\floor{\omega} \colon [\opn{Enc}] \iso [\opn{Enc}]$ by composing $\floor{-}_A$,
followed by $\omega$, followed by $\floor{-}_B^{-1}$. This is in particular the
case for isos that are the images of a Turing Machine. We are now ready to see
how every computable function $f$ from $\interp{A} \to \interp{B}$ can be turned
into an iso whose semantics is $f$.
Given a computable function $f \colon \interp{A} \to \interp{B}$, call $M_f$ the
RTM computing $f$. Since $f$ is in $\PInj$, its output
uniquely determines its input. Following~\cite{bennett1973logical}, there exists
another Turing Machine $M_f'$ which, given the output of $M_f$ recovers the
initial input.
In our encoding of a RTM, the iso will have another additional garbage
which consist of a natural number, \emph{i.e.}~the number of steps of the RTM $M_f$.
Using $\garRem{\opn{isos}(M_f)}{\opn{isos}(M_f')}$ we can obtain a single iso,
from the encoding of $A$ to the encoding of $B$, without any garbage left.
This also ensures that $\sem{\garRem{\opn{isos}(M_f)}{\opn{isos}(M_f')}}(x)
= (\sem{\opn{isos}(M_f)}(x))_1$, for any input $x$.

\begin{theorem}[Computable function as Iso]
	\label{th:computable}
	Given a computable function $f\colon \interp{A} \to \interp{B}$, let $g\colon
	\interp{[\opn{Enc}] \otimes [\opn{Enc}]} \to \interp{[\opn{Enc}] \otimes
	[\opn{Enc}]}$ be defined as $g = \interp{\floor{-}_B} \circ f \circ
	\interp{\floor{-}_A^{-1}}$, and let $\omega\colon A \iso B$ be defined as 
	\begin{align*}
		\{x~\iso~
		&\letv{y}{\floor{x}_A}{} \\
		&\letv{y'}{\garRem{\opn{isos}(M_g)}{\opn{isos}(M_g')}~y}{} \\
		&\letv{z}{\floor{y'}_B^{-1}}{z}\}.
	\end{align*}
  Then $\interp{\omega} = f$.
\end{theorem}
\begin{proof}
		 In $\omega$, call the right-hand-side $e$. Notice that 
		 \[\interp{e} =
		 \interp{\floor{-}_B^{-1}} \circ
		 \interp{\garRem{\opn{isos}(M_g)}{\opn{isos}(M_g')}} \circ
		 \interp{\floor{-}_A}. \] 
		 By Prop.~\ref{prop:rev-soundness}, we know that
		 $\interp{\garRem{\opn{isos}(M_g)}{\opn{isos}(M_g')}} = g$.  Therefore,
		 since $g = \interp{\floor{-}_B} \circ f \circ \interp{\floor{-}_A^{-1}}$
		 by definition, we get
			\begin{align*}
				\sem e &= \interp{\floor{-}_B^{-1}} \circ
				\interp{\garRem{\opn{isos}(M_g)}{\opn{isos}(M_g')}} \circ \interp{\floor{-}_A} \\
				&= \interp{\floor{-}_B^{-1}} \circ g \circ \interp{\floor{-}_A} \\
				&= \interp{\floor{-}_B^{-1}} \circ \interp{\floor{-}_B} \circ f \circ
				\interp{\floor{-}_A^{-1}} \circ \interp{\floor{-}_A}
			\end{align*}
		 By Prop.~\ref{prop:rev-soundness} we get that $\interp{\floor{-}_B^{-1}}
		 \circ \interp{\floor{-}_B} = \iid_B$ and $\interp{\floor{-}_A^{-1}} \circ
		 \interp{\floor{-}_A} = \iid_A$.  Therefore $\interp{e} = f$.  Since the
		 left-hand-side of $\omega$ is just a variable we get $\interp{\omega} =
		 \interp{e}\circ \iid\inv = \interp{e} = f$.
\end{proof}

\section{Further notes and conclusion}

In this chapter, we have developed a functional, reversible programming
language, based on \cite{sabry2018symmetric}, with inductive types and general
recursive calls. This language has been proven to have the same expressivity as
a Turing Machine, meaning that any computable function can be performed. We
have provided a mathematical interpretation of the language, based on the
categorical presentation of inverse categories, equipped with:
\begin{itemize}
	\item a join rig structure, to model pattern-matching and iso;
	\item an enrichment in $\dcpo$, to denote recursive calls through least
		fixed points;
	\item parameterised initial algebra, representing inductive data types;
\end{itemize}
for which the category of sets and partial injections, written $\PInj$, is a
concrete model.  This denotational semantics has been proven sound and adequacy
with regard to the operational semantics of the language. The adequacy proof
involves a \emph{finitary} language, based on the original one, where the
number of recursive calls is limited. This adequacy statement, together with
Turing completeness, ensures that any computable partial injection between two
types of the language, has a corresponding iso in the language.

\paragraph{The role of semantics.}
However abstract the development of a denotation semantics seems, it has a role
to play in the formalisation of programming languages. In the original paper
\cite{sabry2018symmetric}, the language at the level of isos would only allow
very specific $\lambda$-abstractions. This \emph{ad hoc} presentation was a
consequence of the authors having an operational understanding of the language
only, without concerns for a denotation one. Once a denotational semantics was
established by the author of this thesis, it was clear that any usual
$\lambda$-calculus -- and really any usual programming language -- could be put
on top of isos, for their semantics lives in the cartesian closed category
$\dcpo$.

The aim of \cite{sabry2018symmetric} is to establish a high-order programming
language that handles quantum control. While the ideas outlined in the paper
are promising, the language in itself suffers from the same issue as observed
above: the denotational understanding of the language is weak. This echoes to
Abramsky's note in \cite{abramsky2020whither}, questioning whether denotational
semantics should lead or follow. We do not argue in favour of one nor the
other; however, we believe that in the presentation of a programming language,
there should be both convincing operational and denotational arguments. This is
what the author hopes he has done successfully all along the thesis.

\paragraph{Quantum control.}
While Chapter~\ref{ch:qu-control} presents an operational and a denotational
account of simply-typed quantum control, a sound and adequate denotational
semantics of the language in \cite{sabry2018symmetric} is yet to be found. With
the development of the language in the current chapter, the desired result
seems to be just around the corner. However, adding reversible quantum effects
to the semantics presented here does not preserve fixed points (a more general
intuition of this point is outlined in the next chapter), and thus an
interpretation of recursive calls in that setting is yet to be found.

%
%

%% file: quantum-recursion.tex
\begin{abstract}
	We present some remarks and ideas on recursion in quantum control. First, we
	outline the limitations of Hilbert spaces as a denotational model of
	programming languages.  Then, we present some results in the semantics of
	guarded recursion applied to quantum programming.

	\paragraph{References.}
	While the different contributions in this chapter are not yet enough to be
	published independently, they appear to the author as potential foundations
	for infinite-dimentional quantum-controlled programming. This work is the
	author's.
\end{abstract}

\section{Introduction}

In this chapter, we tackle the question of quantum recursion with quantum
control -- in other words, in the context of a quantum reversible effect. Let
us recall that, by quantum control, we mean reversible quantum operations,
which are usually unitaries between Hilbert spaces. It is yet unclear whether
general recursion makes sense in that setting. For example, what is a
non-terminating behaviour with unitaries?

While \cite{sabry2018symmetric} brings a syntactic approach to quantum control
with infinite data types, we choose here a mathematical approach, through a
(potential) denotational semantics. Our focus shall be on Hilbert spaces. They
are the most natural candidates for a denotational semantics of quantum
control, because their direct sum allows for quantum superposition.

We have seen in the previous chapters that unitaries can be written as a sum --
or a \emph{decomposition} -- of contractions. Thus, contractions seem to be the
right notion of \emph{partial unitaries}. They can even be described as
\emph{subunitaries} thanks to \cite[Proposition 14]{pablo2022universal} which
shows that a bounded linear map $f \colon H_1 \to H_2$ is contractive if and
only if $f\dg f \sqsubseteq \iid$.

\section{Limitations}

In this section, we present the \emph{things that do not work} when working with
Hilbert spaces and contractions to interpret quantum control. First, we recall
the fact that the reversible quantum effect is not a monad, setting its study
outside of the realm of Moggi. Then, we observe that the enrichment of Hilbert
spaces is not sufficient to use them as a model such as the one in
Chapter~\ref{ch:reversible}, where join inverse rig categories are
$\dcpo$-enriched. Finally, we motivate the fact that the category of Hilbert
spaces and contractions is probably not canonically traced, and therefore
recursion could not studied with that angle.

\subsection{Effects and the functor $\ell^2$}
\label{sub:l2effect}

As shown in \cite[Corollary 4.8]{heunen2013l2}, the functor $\ell^2$ cannot
induce a monad. Reversible quantum operations interpreted as maps between
Hilbert spaces then cannot be seen as an effect in a traditional way
\cite{moggi-lics, moggi}, as laid out in \secref{sub:sem-lambda-effects} and
studied in Chapter~\ref{ch:monads}.  We do not leave out the possibility of
finding another category for which these operations form a monad, but the
author thinks it is unlikely.

In any case, $\ell^2$ is a functor from $\PInj$ to $\Hilb$, and one can show
that a monad over $\PInj$ would not contain any interesting computational
meaning.  Given a monad $(\MM,\eta,\mu)$ over $\PInj$, for all sets $X$, the
multiplication $\mu_X$ is a monomorphism and an epimorphism, and therefore is
an isomorphism between $\MM^2 X$ and $\MM X$.

Instead, $\ell^2$ can induce a \emph{promonad} over $\PInj$, given by $\PP
\defeq \Hilb(\ell^2(-),\ell^2(-)) \colon \PInj^{\rm op} \times \PInj \to \Set$
(see \cite{hughes2000arrows, jacobs2009arrows, alimarine2005arrows,
asade2010arrows} for a detail account on \emph{arrows} -- the programming
language paradigm corresponding to promonads). Promonads in the context of
reversible programming have been studied in details in \cite{heunen2018arrows}.
The language in Chapter~\ref{ch:qu-control} can be seen as working directly
with the promonad $\PP$; however, it is unclear whether this point of view
would improve or facilitate the presentation in that chapter.

\subsection{Hilbert spaces are not properly enriched}

As shown in \cite[Proposition 2.10]{heunen2013l2}, an order on bounded linear
maps between Hilbert spaces can be established; however, this order is not
preserved by composition, and thus does not form an enrichment. This is also
true for the wide subcategories with  isometries, or even contractions, as
morphisms. This hints at the fact that fixed points in a reversible quantum
setting cannot be studied in the same way as in Chapter~\ref{ch:reversible}. We
are left to wonder how the \emph{structurally recursive} fixed points
introduced in \cite{sabry2018symmetric} can be interpreted in a proper,
compositional semantics. While trying to answer this question, the author has
come across the notion of \emph{guarded recursion}, which led to the
observations in \secref{sub:guar-contribution}.

\subsection{Conjecture: infinite-dimensional Hilbert spaces are not
canonically traced}

It has been observed in \cite{bartha2014quantumturing} that the computationally
interesting tensor when working with isometries between Hilbert spaces is the
direct sum $\oplus$, and not the tensor product $\otimes$. It is shown in that
paper that the category of \emph{finite-dimensional} Hilbert spaces and
isometries is traced over the direct sum $\oplus$. This would allow for a
reversible quantum programming language managing finite data type to deal with
finite loops, usual written $\mathtt{for}$ in programming languages. This trace
is the canonical trace in the following sense: given an isometry $g \colon X
\oplus U \to Y \oplus U$, which can be decomposed in blocks, giving
\[
	g = \left( \begin{array}{cc}
		g_{X,Y} & g_{U,Y} \\
		g_{X,U} & g_{U,U}
	\end{array} \right)
\]
the operator
\begin{equation}
	\label{eq:trace}
	\mathrm{Tr}^{X,Y}_U(g) = g_{X,Y} + \sum_{i=1}^{\infty} g_{U,Y} \circ
	g^i_{U,U} \circ g_{X,U}
\end{equation}
is the trace given in \cite{bartha2014quantumturing}. The existence of this
trace depends on the Moore-Penrose generalised inverse
\cite{greville2003inverses, sanchez2020moore-penrose} of $\iid - g_{U,U}$. What
about non-finite dimensions? Pablo Andrés-Martínez \cite[Section
3.3.4]{pablo2022unbounded} raises this as an open question. We push the idea
further by stating the following conjecture: the category of Hilbert spaces and
contractions is not traced with the operator given in (\ref{eq:trace}). The
next lemma is a hint towards this conjecture.

\begin{lemma}
	Let $f$ be the bounded linear map $\ell^2(\N) \to \ell^2(\N)$ such that $f
	\ket n = \frac{1}{n+1} \ket n$. The map $f$ is a contraction and does not
	admit a Moore-Penrose inverse.
\end{lemma}
\begin{proof}
	The Moore-Penrose inverse of $f$ would be $g \colon \ell^2(\N) \to \ell^2(\N)$
	such that $g \ket n = (n+1) \ket n$, which is not bounded.
\end{proof}

We could then potentially provide a $g_{U,U}$ using the map given above, and
obtain a suitable $g$, with the Szőkefalvi-Nagy \textipa{\sffamily ['s\o:kEf6lvi
'n6\textbardotlessj]}
dilation theorem \cite{szokefalvi1954contractions}, which would not be traceable.

This raises more questions than it answers. One may wonder whether the category
we are interested in in Chapter~\ref{ch:qu-control} -- namely,
countably-dimensional Hilbert spaces and isometries -- is traced at all in a
\emph{computationally interesting} way. Making this statement mathematically
precise is a challenge in itself; thus proving or disproving it might require
some effort.

In his thesis, Pablo Andrés-Martínez \cite{pablo2022unbounded} unveiled a
category both traced with regard to its direct sum $\oplus$ and that could
handle non-finite data types; however it has not been proven to be a model for
a sufficiently expressive programming language. The work presented in the next
section is similar in the idea, and has also not been shown to be a sound and
adequate interpretation to a programming language yet. Nevertheless, it is
based on a tried-and-texted paradigm for classical programming, called
\emph{guarded recursion}.

%% file: guarded-recursion.tex
\section{A Foundation for Guarded Quantum Recursion}

Guarded recursion is a framework in which recursive calls are guarded by delay
modalities. This framework is particularly useful to reason about streams in
programming languages. A type system aimed for guarded recursion usually
contains the \emph{later} modality, given by the symbol $\later$. Its
introduction rule is simple:
\[
	\infer{\Theta \vdash~\later\! A}{\Theta \vdash A}
\]
If we take the example of a guarded $\lambda$-calculus -- such as the one
introduced by Nakano \cite{nakano2000recursion} -- a contructor $\nnext\!$ is
added to the syntax, with the following rule:
\[
	\infer{\Gamma \vdash \nnext M \colon \later\! A}{\Gamma \vdash M \colon A}
\]
A term under $\nnext\!$ is \emph{locked} and needs to wait for its time to be
computed. The fixed point combinator introduced in the beginning of the thesis
(see \secref{sub:dcpo}) becomes now $\ffix \colon (\later\! A \to A) \to A$, and given
$\cdot \vdash M \colon \later\! A \to A$, we have the following operational rule:
\[
	\ffix M \to M(\nnext \ffix M)
\]
while does not allow for a infinite reduction in general, because the
$\nnext\!$ will control this behaviour.

This framework admits sound and adequate denotational semantics in the
\emph{topos of trees} $\Set^{\N^{\mathrm op}}$, written $\Scat$ in the
following, which is a cartesian closed category.

\subsection{Work of the author}
\label{sub:author-guarded}

The author has contributed to the following points.
\begin{itemize}
	\item The definition of categories ($\NN$ and $\QQ$, see below) to model
		guarded quantum recursion.
	\item A proof that those categories are $\Scat$-enriched (see
		Lemma~\ref{lem:enriched}), allowing for a similar semantic study as what is
		done in Chapter~\ref{ch:reversible}.
	\item A proof that inductive types can be interpreted in this model (see
		Theorem~\ref{th:param}).
	\item Convincing arguments for this model to interpret a programming language
		with quantum guarded recursion.
\end{itemize}

\subsection{Related work}
\label{sub:guar-related}

The work in \cite{birkedal2012first} can be described as a foundation for
semantics of guarded recursion, and an introduction to what they call
\emph{synthetic guarded domain theory}. Their theory assumes a topos or a
sheaves structure, which will not be the case in an attempt to interpret
quantum control. However, their work is a great source of inspiration for our
denotation of guarded inductive types. They prove that locally contractive
functors have a unique fixed point, which shows that induction and coinduction
have the same interpretation in their model; this has a practical consequence:
the fixed point is obtained as a limit, but can still be manipulated as an
initial algebra.

A first account of solution to solve recursive equations in (pre)sheaves is
given in \cite{gianantonio2004unifying}, and the same authors also observed
that working with contractive functors would unify induction and coinduction
\cite{gianantonio2003unifying}. Their theory is not suited to quantum for the
same reason as above.

The work in \cite{birkedal2010metric} details a denotational semantics to the
guarded $\lambda$-calculus introduced by Nakano in \cite{nakano2000recursion}.
Their model is a category of ultrametric spaces, which is actually included in
the theory of \cite{birkedal2012first}, as justified by the authors themselves.

Several papers have then used synthetic guarded domain theory to develop
refined models of guarded types or guarded recursion
\cite{birkedal2016guarded,mannaa2020ticking,basold2023causal}. Their study,
based on \cite{birkedal2012first}, cannot be used to model quantum computation,
for the same reasons as above.

\subsection{Contribution}
\label{sub:guar-contribution}

This section aims at providing a categorical tool necessary to interpret
guarded recursion with quantum control. Recursive domain equations are here
tackled through locally contractive functors as defined in
\cite{birkedal2012first}. 

We will work with categories of the form $\CC^{\N^{\mathrm op}}$, where $\N$ is
the category of natural numbers starting from $0$, with morphisms defined by
the classical order on natural numbers. Given an object $X$ of
$\CC^{\N^{\mathrm op}}$, meaning a functor $\N^{\mathrm op}\to \CC$, its image
on the morphism $n \leq n+1$ is written $r^X_n \colon X(n+1) \to X(n)$, and is
a morphism in $\CC$. This object $X$ of $\CC^{\N^{\mathrm op}}$ can be
represented with a diagram, such as:
\begin{equation}
	\label{eq:diagram-view}
	\begin{tikzcd}
		X(0) & X(1) & X(2) & X(3) & \cdots
		\arrow["r^X_0"', from=1-2, to=1-1]
		\arrow["r^X_1"', from=1-3, to=1-2]
		\arrow["r^X_2"', from=1-4, to=1-3]
		\arrow[from=1-5, to=1-4]
	\end{tikzcd}
\end{equation}
and a morphism $f \colon X \to Y$ can be pictured with the following diagram
in $\CC$:
\begin{equation*}
	\label{eq:diagram-view-natural}
	\begin{tikzcd}
		X(0) & X(1) & X(2) & X(3) & \cdots \\
		Y(0) & Y(1) & Y(2) & Y(3) & \cdots
		\arrow["r^X_0"', from=1-2, to=1-1]
		\arrow["r^X_1"', from=1-3, to=1-2]
		\arrow["r^X_2"', from=1-4, to=1-3]
		\arrow[from=1-5, to=1-4]
		\arrow["r^Y_0", from=2-2, to=2-1]
		\arrow["r^Y_1", from=2-3, to=2-2]
		\arrow["r^Y_2", from=2-4, to=2-3]
		\arrow[from=2-5, to=2-4]
		\arrow["f_0", from=1-1, to=2-1]
		\arrow["f_1", from=1-2, to=2-2]
		\arrow["f_2", from=1-3, to=2-3]
		\arrow["f_3", from=1-4, to=2-4]
	\end{tikzcd}
\end{equation*}
This way of picturing them will help the intuition throughout the section.
The category
$\Set^{\N^{\mathrm op}}$ is called the \emph{topos of trees}, and is a
cartesian closed category. However, $\CC^{\N^{\mathrm op}}$ in general has a
topos flavour without being cartesian closed. Some subcategories of
$\Coiso^{\N^{\mathrm op}}$ and $\Contr^{\N^{\mathrm op}}$ will be used for our
denotational model; a coisometry is in particular a contraction, thus
$\Coiso^{\N^{\mathrm op}}$ is a subcategory of $\Contr^{\N^{\mathrm op}}$. This
model was inspired by the work in \cite{birkedal2012first} and
\cite{birkedal2016guarded}; however, the motivations, the point of view and the
results differ. We use the following notations:
$ \Scat \defeq \Set^{\N^{\mathrm op}},
\NN \defeq \Coiso^{\N^{\mathrm op}}. $

\begin{lemma}[\cite{maclane2012sheaves}]
	The category $\Scat$ is a cartesian closed category with the following: given
	two objects $X,Y$ in $\Scat$, their product $X \times Y$ is obtained
	pointwise, and $[X \to Y]$ is given by $\Scat(\yo(-) \times X, Y)$ where $\yo
	\colon \N \to \Scat$ is the Yoneda embedding ($\yo$ is the Japanese hiragana
	``yo'' , see Example~\ref{ex:yoneda}).
\end{lemma}

We write $\QQ$, the full subcategory of $\Contr^{\N^{\mathrm op}}$ whose
objects are objects in $\NN$. Note that $\NN$ is embedded in $\QQ$. 
The category $\NN$ will be used to study the semantics of types, whereas $\QQ$
is the category where the terms and the functions are interpreted.

\begin{lemma}
	The categories $\QQ$ and $\NN$ are symmetric monoidal, equipped
	with a pointwise monoidal product. More generally, if $\CC$ is symmetric
	monoidal, so is $\CC^{\N^{\mathrm op}}$.
\end{lemma}
\begin{proof}
	Let $(\CC, \otimes, I, \alpha, \lambda, \rho)$ be a symmetric monoidal
	category. We show that $\CC^{\N^{\mathrm op}}$ is also one. The tensor is obtained
	pointwise: given two objects $X$ and $Y$ in $\CC^{\N^{\mathrm op}}$, their
	tensor product is the following object:
	\[\begin{tikzcd}
		X(0) \otimes Y(0) & X(1) \otimes Y(1) & X(2) \otimes Y(2) & X(3) \otimes
		Y(3) &
		\cdots
		\arrow["r^X_0 \otimes r^Y_0"', from=1-2, to=1-1]
		\arrow["r^X_1 \otimes r^Y_1"', from=1-3, to=1-2]
		\arrow["r^X_2 \otimes r^Y_2"', from=1-4, to=1-3]
		\arrow[from=1-5, to=1-4]
	\end{tikzcd} \]
	We abuse notations and write $X \otimes Y$ for this tensor product.
	The tensor unit of this tensor is the object:
	\[\begin{tikzcd}
		I & I & I & I &
		\cdots
		\arrow["\iid_I"', from=1-2, to=1-1]
		\arrow["\iid_I"', from=1-3, to=1-2]
		\arrow["\iid_I"', from=1-4, to=1-3]
		\arrow[from=1-5, to=1-4]
	\end{tikzcd} \]
	We abuse notations by also writing $I$ for this unit.
	The left unitor $I \otimes X \to X$ is given pointwise as well:
	\[\begin{tikzcd}
		I \otimes X(0) & I \otimes X(1) & I \otimes X(2) & I \otimes X(3) & \cdots \\
		X(0) & X(1) & X(2) & X(3) & \cdots
		\arrow["\iid_I \otimes r^X_0"', from=1-2, to=1-1]
		\arrow["\iid_I \otimes r^X_1"', from=1-3, to=1-2]
		\arrow["\iid_I \otimes r^X_2"', from=1-4, to=1-3]
		\arrow[from=1-5, to=1-4]
		\arrow["r^X_0"', from=2-2, to=2-1]
		\arrow["r^X_1"', from=2-3, to=2-2]
		\arrow["r^X_2"', from=2-4, to=2-3]
		\arrow[from=2-5, to=2-4]
		\arrow["\lambda_{X(0)}"', from=1-1, to=2-1]
		\arrow["\lambda_{X(1)}", from=1-2, to=2-2]
		\arrow["\lambda_{X(2)}", from=1-3, to=2-3]
		\arrow["\lambda_{X(3)}", from=1-4, to=2-4]
	\end{tikzcd} \]
	The right unitor and associator are defined pointwise in the same way. The
	proof of the coherence diagrams is direct.
\end{proof}

Similarly, given two objects $X,Y$ of $\QQ$ (resp. $\NN$), their pointwise
direct sum $X \oplus Y$ is an object of $\QQ$ (resp. $\NN$).
\[\begin{tikzcd}
	X(0) \oplus Y(0) & X(1) \oplus Y(1) & X(2) \oplus Y(2) & X(3) \oplus Y(3) &
	\cdots
	\arrow["r^X_0 \oplus r^Y_0"', from=1-2, to=1-1]
	\arrow["r^X_1 \oplus r^Y_1"', from=1-3, to=1-2]
	\arrow["r^X_2 \oplus r^Y_2"', from=1-4, to=1-3]
	\arrow[from=1-5, to=1-4]
\end{tikzcd} \]
Moreover, in $\QQ$, one can define injections $\iota_l^{X,Y} \colon
X \to X \oplus Y$ as follows:
\[\begin{tikzcd}
	X(0) & X(1) & X(2) & X(3) & \cdots \\
	X(0) \oplus Y(0) & X(1) \oplus Y(1) & X(2) \oplus Y(2) & X(3) \oplus Y(3) &
	\cdots
	\arrow["r^X_0"', from=1-2, to=1-1]
	\arrow["r^X_1"', from=1-3, to=1-2]
	\arrow["r^X_2"', from=1-4, to=1-3]
	\arrow[from=1-5, to=1-4]
	\arrow["r^X_0 \oplus r^Y_0", from=2-2, to=2-1]
	\arrow["r^X_1 \oplus r^Y_1", from=2-3, to=2-2]
	\arrow["r^X_2 \oplus r^Y_2", from=2-4, to=2-3]
	\arrow[from=2-5, to=2-4]
	\arrow["\iota^{X(0),Y(0)}_l", from=1-1, to=2-1]
	\arrow["\iota^{X(1),Y(1)}_l", from=1-2, to=2-2]
	\arrow["\iota^{X(2),Y(2)}_l", from=1-3, to=2-3]
	\arrow["\iota^{X(3),Y(3)}_l", from=1-4, to=2-4]
\end{tikzcd} \]
and $\iota_r^{X,Y} \colon Y \to X \oplus Y$ is defined similarly.


\begin{remark}[Dagger]
	\label{rem:dagger}
	Morphisms in $\QQ$ are natural transformations, whose components
	are morphisms in $\Contr$. In that regard, $\QQ$ inherits some
	of the structure of $\Contr$, but not all. For example,
	$\QQ$ is not a dagger category. However, given $\alpha \colon X \natto Y$
	a morphism in $\QQ$, we will write $\alpha\dg$ for the 
	\emph{componentwise} dagger of $\alpha$, even if it might not 
	be a morphism in $\QQ$.
\end{remark}

\begin{lemma}
	\label{lem:enriched}
	The categories $\QQ$ and $\NN$ are enriched in $\Scat$.
\end{lemma}
\begin{proof}
	The homsets $\NN (X,Y)$ can be seen as objects of $\Scat$,
	with $\NN(X,Y) (n)$ being the set of $n$-th component of
	natural transformations in $\NN(X,Y)$,
	and the image of $n \leq n+1$, written $r^{\NN (X,Y)}_n$, is:
	\[\begin{tikzcd}
		\NN(X,Y)(n) && \NN (X,Y)(n+1).
		\arrow["r^Y_n\circ - \circ (r^X_n)\dg"',from=1-3, to=1-1]
	\end{tikzcd}\]
	Note that elements of $\NN(X,Y)(n)$ are in particular in $\Coiso(X(n),Y(n))$.
	This definition is sound because if $\alpha$ is a natural transformation,
	we have in particular that $\alpha_n\circ r^X_n = r^Y_n \circ \alpha_{n+1}$,
	and by precomposing by $(r^X_n)\dg$, we get
	\begin{align*}
		r^Y_n \circ \alpha_{n+1} \circ (r^X_n)\dg
		&= \alpha_n\circ r^X_n \circ (r^X_n)\dg \\
		&= \alpha_n.
	\end{align*}
	We also have to prove that composition, defined pointwise by the composition
	in $\Coiso$, is a morphism in $\Scat$; formally that for every $X,Y,Z$ objects
	of $\NN$, $\comp_{X,Y,Z} \colon \NN(X,Y) \times \NN(Y,Z) \to \NN(X,Z)$ is a
	morphism in $\Scat$. We need to prove that it is a natural transformation, in
	other words that for all $n$, the diagram:
	\[\begin{tikzcd}
		\NN(X,Y)(n) \times \NN(Y,Z)(n) && \NN(X,Y)(n+1) \times \NN(Y,Z)(n+1) \\
		\NN(X,Z)(n) && \NN(X,Z)(n+1)
		\arrow["r^{\NN(X,Y)}_n \times r^{\NN(Y,Z)}_n"', from=1-3, to=1-1]
		\arrow["r^{\NN(X,Z)}_n", from=2-3, to=2-1]
		\arrow["\comp_n", from=1-1, to=2-1]
		\arrow["\comp_{n+1}", from=1-3, to=2-3]
	\end{tikzcd}\]
	commutes. Indeed:
	\begin{align*}
		(\comp_n \circ (r^{\NN(X,Y)}_n \times r^{\NN(Y,Z)}_n)) (f_{n+1},g_{n+1})
		&= \comp_n (r^Y_n \circ f_{n+1} \circ (r^X_n)\dg,r^Z_n \circ g_{n+1} \circ (r^Y_n)\dg) \\
		&= \comp_n (f_n,g_n) = g_n \circ f_n,
	\end{align*}
	and
	\begin{align*}
		(r^{\NN(X,Z)}_n \circ \comp_{n+1}) (f_{n+1},g_{n+1})
		&= r^{\NN(X,Z)}_n (g_{n+1} \circ f_{n+1}) \\
		&= r^Z_n \circ g_{n+1} \circ f_{n+1} \circ (r^X_n)\dg \\
		&= g_n \circ r^Y_n \circ f_{n+1} \circ (r^X_n)\dg  \\
		&= g_n \circ f_n \circ r^X_n \circ (r^X_n)\dg \\
		&= g_n \circ f_n,
	\end{align*}
	which ensures the commutativity of the diagram above.

	The same observations apply to $\QQ$.
\end{proof}

\begin{remark}
	\label{rem:embedding-S}
	The embedding $E_\NN^\QQ \colon \NN \embed \QQ$ is $\Scat$-enriched. 
\end{remark}

A feature of categories of the form $\CC^{\N^{\mathrm op}}$ is the later
functor. This functor works as some sort of delay operation. It can be used to
keep track of the depth of a term and the number of recursive calls. The way it
works is fairly simple: it shifts the diagram in (\ref{eq:diagram-view}) one step
to the right, and adds a terminal object on the left.

\begin{definition}[Later functor]
	\label{def:later}
	Given any category $\CC$ with terminal object $1$, The later functor is a
	functor $L\colon \CC^{\N^{\mathrm op}} \to \CC^{\N^{\mathrm op}}$, such that
	given a functor $X \colon \N^{\mathrm op} \to \CC$, $LX(0) = 1$ and $LX(n+1)
	= X(n)$.  Given $\alpha\colon X\natto Y$, $L\alpha_0 = ~!$ (the terminal map),
	and $(L\alpha)_{n+1} = \alpha_n$.
\end{definition}

\begin{remark}
	Note that this endofunctor can be defined on $\Scat$, $\QQ$ and $\NN$; where
	the terminal object in the latter categories is the zero-dimentional Hilbert
	space $\set 0$. We will use the same letter $L$ when it is not ambiguous. If
	any ambiguity arises, the notations $L^\Scat$, $L^\QQ$ and $L^\NN$ will be
	used.
\end{remark}

\begin{lemma}[\cite{birkedal2012first}]
	The functor $L^\Scat \colon \Scat \to \Scat$ is a strict monoidal functor.
\end{lemma}

\begin{remark}
	The functors $L^\NN$ and $L^\QQ$ preserve the monoidal structure
	in the sense that $L(X \otimes Y) = LX \otimes LY$, however $L$
	does not map the tensor unit to the tensor unit in those categories.
\end{remark}

\begin{lemma}
	\label{lem:later-enriched-view}
	Given $X,Y$ objects in $\NN$, we have $L^\Scat \NN(X,Y) \cong \NN(L^\NN X, L^\NN Y)$.
\end{lemma}
\begin{proof}
	Let $X$ and $Y$ be two objects in $\NN$. Remember that the homset $\NN (X,Y)$
	can be seen as an object of $\Scat$, with $\NN(X,Y) (n)$ being the set of
	$n$-th component of natural transformations in $\NN(X,Y)$, and the image of
	$n \leq n+1$, written $r^{\NN (X,Y)}_n$, is:
	\[\begin{tikzcd}
		\NN(X,Y)(n) && \NN (X,Y)(n+1).
		\arrow["r^Y_n\circ - \circ (r^X_n)\dg"',from=1-3, to=1-1]
	\end{tikzcd}\]
	We have $L^\Scat \NN(X,Y) = \{ * \}$
	and $\NN(L^\NN X, L^\NN Y) = \{ 0_{\{ 0 \} \to \{ 0 \}} \}$. They are both
	singletons. Moreover, $L^\Scat \NN(X,Y)_{n+1} = \NN(X,Y)_n$ is the set of
	$n$-th components of natural transformations from $X$ to $Y$, which is
	exactly like $\NN(L^\NN X, L^\NN Y)_{n+1}$. Note also that the functor $L$
	does not change the morphisms image of $n \leq n+1$, except shifting them.
	Therefore, $L^\Scat \NN(X,Y) \cong \NN(L^\NN X, L^\NN Y)$.
\end{proof}

The delay embodied by the functor $L$ can be introduced by a natural
transformation, called \emph{next}. This natural transformation helps us
introduce the delay in a programming language, as the denotational semantics
of a delayed program.

\begin{definition}[Next]
	\label{def:next}
	Given any category $\CC$ with terminal object $1$, and $L \colon
	\CC^{\N^{\mathrm op}} \to \CC^{\N^{\mathrm op}}$ the later functor,
	there is a natural transformation $\nu\colon id\natto L$, defined as
	$\nu_{X,0} = ~!$ and $\nu_{X,n+1} = r^X_n$. Or as a diagram,
	it maps the functor $X$ to the functor $LX$ as follows:
	\[\begin{tikzcd}
		X(0) & X(1) & X(2) & X(3) & \cdots \\
		1 & X(0) & X(1) & X(2) & \cdots
		\arrow["r^X_0"', from=1-2, to=1-1]
		\arrow["r^X_1"', from=1-3, to=1-2]
		\arrow["r^X_2"', from=1-4, to=1-3]
		\arrow[from=1-5, to=1-4]
		\arrow["!"', from=1-1, to=2-1]
		\arrow["r^X_0", from=1-2, to=2-2]
		\arrow["r^X_1", from=1-3, to=2-3]
		\arrow["r^X_2", from=1-4, to=2-4]
		\arrow["!", from=2-2, to=2-1]
		\arrow["r^X_0", from=2-3, to=2-2]
		\arrow["r^X_1", from=2-4, to=2-3]
		\arrow[from=2-5, to=2-4]
	\end{tikzcd}\]
	Which gives, in the categories $\QQ$ and $\NN$:
	\[\begin{tikzcd}
		X(0) & X(1) & X(2) & X(3) & \cdots \\
		\set 0 & X(0) & X(1) & X(2) & \cdots
		\arrow["r^X_0"', from=1-2, to=1-1]
		\arrow["r^X_1"', from=1-3, to=1-2]
		\arrow["r^X_2"', from=1-4, to=1-3]
		\arrow[from=1-5, to=1-4]
		\arrow["0"', from=1-1, to=2-1]
		\arrow["r^X_0", from=1-2, to=2-2]
		\arrow["r^X_1", from=1-3, to=2-3]
		\arrow["r^X_2", from=1-4, to=2-4]
		\arrow["0", from=2-2, to=2-1]
		\arrow["r^X_0", from=2-3, to=2-2]
		\arrow["r^X_1", from=2-4, to=2-3]
		\arrow[from=2-5, to=2-4]
	\end{tikzcd}\]
\end{definition}

\begin{remark}
	Note that $\nu$ is a \emph{high-level} natural transformation. It is an
	informal way to say that it is a natural transformation at the level of
	$\CC^{\N^{\mathrm op}}$. For all functors $X \colon \N^{\mathrm op} \to \CC$,
	$\nu_X$ is a morphism in $\CC^{\N^{\mathrm op}}$.  A morphism in
	$\CC^{\N^{\mathrm op}}$ is a natural transformation -- this time,
	\emph{low-level} -- at the level of $\CC$. Thus, if $n$ is a natural number,
	then $\nu_{X,n}$ is a morphism in $\CC$.

	Once again, and similarly to the functor $L$, the natural transformation
	$\nu$ can be defined in $\Scat$ as well as in $\NN$ and $\QQ$.
	If any ambiguity arises, the notations $\nu^\Scat$, $\nu^\QQ$ and $\nu^\NN$
	will be used. Observe that the components of $\nu^\QQ$ and $\nu^\NN$
	are the same.
\end{remark}

A direct consequence of Lemma~\ref{lem:enriched} is the possibility to move
back and forth from $\NN$ to its enrichment $\Scat$ with the natural
transformation $\nu$.

\begin{lemma}
	\label{lem:next-enriched-view}
	Given $X,Y$ two objects of $\NN$, we have
	$\nu^\Scat_{\NN(X,Y)} = \nu^\NN_Y \circ - \circ \left( \nu^\NN_X \right)\dg$.
\end{lemma}
\begin{proof}
	This follows from Lemma~\ref{lem:enriched},
	Lemma~\ref{lem:later-enriched-view} and Definition~\ref{def:next}.
\end{proof}

Lemmas~\ref{lem:later-enriched-view} and \ref{lem:next-enriched-view}
also allow for the following observation.

\begin{corollary}
	The functor $L^\NN$ is $\Scat$-enriched.
\end{corollary}

Note that, as explained in Remark~\ref{rem:dagger}, $\left( \nu^\NN_X
\right)\dg$ is not a natural transformation, and thus the notation above is
loose. However, it can be used in some cases, as shown by the next lemma.

\begin{lemma}
	\label{lem:dagger-use}
	Given a morphism $f \colon X \to Y$ in $\QQ$, we have that $\nu^\QQ_Y \circ f
	\circ \left( \nu^\QQ_X \right)\dg \colon LX \to LY$ is a morphism in $\QQ$.
\end{lemma}
\begin{proof}
	Remember that $\nu$ is defined as $\nu_{X,n} = r^X_{n-1}$. Let us
	proceed:
	\begin{align*}
		r_n^Y \circ f_{n+1} \circ (r^X_n)\dg \circ r^X_n
		&= f_n \circ r^X_n \circ (r^X_n)\dg \circ r^X_n \\
		&= f_n \circ r^X_n \\
		&= r^Y_n \circ f_{n+1} \\
		&= r^Y_n \circ f_{n+1} \circ r^X_{n+1} \circ (r^X_{n+1})\dg \\
		&= r^Y_n \circ r^Y_{n+1} \circ f_{n+2} \circ (r^X_{n+1})\dg.
	\end{align*}
\end{proof}

To interpret inductive types, we need the help of fixed points to solve
recursive domain equations. First, we recall some notions on contractive
morphisms in $\Scat$, introduced in \cite{birkedal2012first}.

\begin{definition}[\cite{birkedal2012first}]
	\label{def:contractive-S}
	A morphism $f \colon X \to Y$ in $\Scat$ is contractive if
	there exists a morphism $g \colon LX \to Y$ such that
	$f = g \circ \nu_X$. A morphism $f \colon X \times Y \to Z$ is
	contractive in the first variable if there exists $g \colon LX \times Y \to Z$
	such that $f = g \circ (\nu_X \times \iid_Y)$.
\end{definition}

\begin{lemma}[\cite{birkedal2012first}]
	\label{lem:contractive-S}
	The following assertions hold.
	\begin{itemize}
		\item Given $f \colon X \to Y$, $g \colon Y \to Z$, if 
			either $f$ or $g$ is contractive, then $gf$ is contractive.
		\item Given $f \colon X \to Y$ and $g \colon X' \to Y'$
			contractive, so is $f \times g \colon X \times X' \to Y \to Y'$.
		\item A morphism $h \colon X \times Y \to Z$ is contractive in the first
			variable iff $\mathrm{curry}(h) \colon X \to Z^Y$ is contractive.
	\end{itemize}
\end{lemma}

\begin{theorem}[\cite{birkedal2012first}]
	\label{th:contractive-S}
	There exists a natural family of morphisms $\mathrm{fix}_X \colon
	(LX \to X) \to X$ which computes unique fixed points: given
	$f \colon X \times Y \to X$ contractive in the first variable 
	and $g \colon LX \times Y \to X$ the resulting morphism, then
	$\mathrm{fix}_X \circ \mathrm{curry}(g)$ is the unique $h \colon Y \to X$
	such that $f \circ \pv{h}{\iid_Y} = h$.
\end{theorem}

\begin{remark}
	Unsurprisingly, the previous theorem can be used for the interpretation of a
	fixed point combinator.
\end{remark}

The definition of a contractive morphism in $\Scat$ help define functors
that have a fixed point in $\NN$. We precise what we mean by fixed point
of a functor.

\begin{definition}[Fixed Point]
	\label{def:fixed-point}
	A fixed point of an endofunctor $T  \colon \NN \to \NN$ is a pair $(X,\alpha
	\colon TX \to X)$ such that $\alpha$ is an isomorphism. 
\end{definition}

We define locally contractive functors, which are functor that admit a fixed
points. We will see that, as its name suggests, a locally contractive functor
has a unique fixed point, up to isomorphism.

\begin{definition}[Locally Contractive functor]
	\label{def:lacally-contractive-functor}
	An $\Scat$-functor $F\colon \NN \to \NN$ is said to be \emph{locally
	contractive} if its morphism mapping $F_{X,Y}\colon \NN(X,Y) \to
	\NN(FX,FY)$ is contractive; meaning it factorises through $\nu$:
	for all $X,Y$ there is a morphism mapping $G_{X,Y} \colon L( \NN(X,Y) )
	\to \NN(FX,FY)$ such that:
	\[\begin{tikzcd}
		\NN(X,Y) && \NN(FX,FY) \\
		& L(\NN(X,Y))
		\arrow["F_{X,Y}", from=1-1, to=1-3]
		\arrow["\nu_{\NN(X,Y)}"', from=1-1, to=2-2]
		\arrow["G_{X,Y}"', from=2-2, to=1-3]
	\end{tikzcd}\]
	and such that $G$ behaves \emph{like a functor}; formally, such that
	the diagrams:
	\[\begin{tikzcd}
		L\NN(Y,Z)\times L\NN(X,Y) & L(\NN(Y,Z)\times \NN(X,Y)) & L\NN(X,Z) \\
		\NN(FY,FZ)\times \NN(FX,FY) && \NN(FX,FZ)
		\arrow["\cong", from=1-1, to=1-2]
		\arrow["L(\mathrm{comp})", from=1-2, to=1-3]
		\arrow["\mathrm{comp}", from=2-1, to=2-3]
		\arrow["G_{X,Y}\times G_{Y,Z}"', from=1-1, to=2-1]
		\arrow["G_{X,Z}", from=1-3, to=2-3]
	\end{tikzcd}\]
	\[\begin{tikzcd}
		\set{\star} & L\NN(X,X) \\
		& \NN(FX,FX)
		\arrow["L(\mathrm{id})", from=1-1, to=1-2]
		\arrow["\mathrm{id}"', from=1-1, to=2-2]
		\arrow["G_{X,X}", from=1-2, to=2-2]
	\end{tikzcd}\]
	commute in $\Scat$, for all objects $X,Y,Z$.
\end{definition}

\begin{remark}
	The definition above was inspired by a similar one in
	\cite{birkedal2012first}, where they use the fact that their category is
	closed to draw the diagrams with objects and morphisms of the category.
	We manage to do the same with the enrichment of $\NN$ in $\Scat$.
\end{remark}

There are some direct examples of locally contractive functors.

\begin{lemma}
	The functor $L^\NN$ is locally contractive.
\end{lemma}
\begin{proof}
	This is a direct conclusion of Lemma~\ref{lem:later-enriched-view} and
	Definition~\ref{def:lacally-contractive-functor}.
\end{proof}

\begin{lemma}[\cite{birkedal2012first}]
	\label{lem:contractive-comp}
	Given $F,G \colon \NN \to \NN$ two $\Scat$-enriched functors and
	such that either $F$ or $G$ is locally contractive, then $FG$ is 
	locally contractive.
\end{lemma}

\begin{remark}[Initial Algebra and Final Coalgebra]
	\label{rem:canon}
	Given a locally contractive endofunctor $F \colon \NN \to \NN$, a fixed point
	of $F$ is an initial algebra and a final coalgebra. Indeed, given a fixed
	point $\alpha \colon FX \cong X$ and an algebra $\beta \colon FY \natto Y$, an
	algebra morphism $\gamma \colon X \natto Y$ is given by a fixed point of $H
	\colon \gamma \mapsto \beta \circ F_{X,Y}(\gamma) \circ \alpha^{-1}$. Note
	that $H$ is then a contractive morphism in $\Scat$, because $F$ is locally
	contractive. It is proven in \cite{birkedal2012first} that such a morphism
	has a unique fixed point. Thus the algebra morphism $\gamma$ is unique; and
	this makes $\alpha$ an initial algebra. The proof of the coalgebra part is
	similar.
\end{remark}

\begin{remark}
	This basically means that induction and coinduction in our system are the
	same. This observation was made long before us (see related work
	\secref{sub:guar-related}). In this setting, it is a matter a choice whether the
	syntax should use $\mu$ or $\nu$ as a notation for fixed points. Our focus is
	on inductive data types, we then use $\mu$.
\end{remark}

Given a morphism $\alpha\colon X \to Y$ in $\NN$, one says that
$\alpha$ is an $n$-isomorphism if the first $n$ components of $\alpha$
(that is to say, $\alpha_0,\dots,\alpha_{n-1}$) are isomorphisms.

\begin{lemma}
	\label{lem:niso}
	A locally contractive functor $F \colon \NN \to \NN$ maps an $n$-isomorphism
	to an $n+1$-isomorphism.
\end{lemma}
\begin{proof}
	This is the purpose of Definition~\ref{def:lacally-contractive-functor}. The
	proof is direct by observing that $L$ shifts all components one step to the
	right. 
\end{proof}

This observation allows to prove the next theorem, with the same proof strategy
as in \cite{birkedal2012first}.

\begin{theorem}
	\label{th:fixedpoint}
	Any locally contractive endofunctor $T \colon \NN \to \NN$ has a fixed point.
\end{theorem}
\begin{proof}
	Note that the category $\NN$ has a terminal object, written $Z$:
	\[\begin{tikzcd}
		\set 0 & \set 0 & \set 0 & \dots
		\arrow["0"', from=1-2, to=1-1]
		\arrow["0"', from=1-3, to=1-2]
		\arrow[from=1-4, to=1-3]
	\end{tikzcd}\]
	and we will write $!$ the unique map to $Z$. Given a
	functor $T \colon \NN \to \NN$, let us have a look at the sequence:
	\begin{equation}
		\label{eq:terminal-sequence}
		\begin{tikzcd}
			TZ & T^2Z & T^3Z & T^4Z & \dots
			\arrow["T!"', from=1-2, to=1-1]
			\arrow["T^2!"', from=1-3, to=1-2]
			\arrow["T^3!"', from=1-4, to=1-3]
			\arrow[from=1-5, to=1-4]
		\end{tikzcd}
	\end{equation}
	and we obtain the fixed point of $T$ through the limit
	of the above diagram. The limit can be built by observing the
	following diagram:
	\[\begin{tikzcd}
		TZ(0) & T^2Z(0) & T^3Z(0) & T^4Z(0) & \dots \\
		TZ(1) & T^2Z(1) & T^3Z(1) & T^4Z(1) & \dots \\
		TZ(2) & T^2Z(2) & T^3Z(2) & T^4Z(2) & \dots \\
		TZ(3) & T^2Z(3) & T^3Z(3) & T^4Z(3) & \dots \\
		\vdots & \vdots & \vdots & \vdots
		\arrow["T!_0"', from=1-2, to=1-1]
		\arrow["T^2!_0"', from=1-3, to=1-2]
		\arrow["T^3!_0"', from=1-4, to=1-3]
		\arrow[from=1-5, to=1-4]
		\arrow["T!_1"', from=2-2, to=2-1]
		\arrow["T^2!_1"', from=2-3, to=2-2]
		\arrow["T^3!_1"', from=2-4, to=2-3]
		\arrow[from=2-5, to=2-4]
		\arrow["T!_2"', from=3-2, to=3-1]
		\arrow["T^2!_2"', from=3-3, to=3-2]
		\arrow["T^3!_2"', from=3-4, to=3-3]
		\arrow[from=3-5, to=3-4]
		\arrow["T!_3"', from=4-2, to=4-1]
		\arrow["T^2!_3"', from=4-3, to=4-2]
		\arrow["T^3!_3"', from=4-4, to=4-3]
		\arrow[from=4-5, to=4-4]
		\arrow["r^1_0", from=2-1, to=1-1]
		\arrow["r^2_0", from=2-2, to=1-2]
		\arrow["r^3_0", from=2-3, to=1-3]
		\arrow["r^4_0", from=2-4, to=1-4]
		\arrow["r^1_1", from=3-1, to=2-1]
		\arrow["r^2_1", from=3-2, to=2-2]
		\arrow["r^3_1", from=3-3, to=2-3]
		\arrow["r^4_1", from=3-4, to=2-4]
		\arrow["r^1_2", from=4-1, to=3-1]
		\arrow["r^2_2", from=4-2, to=3-2]
		\arrow["r^3_2", from=4-3, to=3-3]
		\arrow["r^4_2", from=4-4, to=3-4]
		\arrow[from=5-1, to=4-1]
		\arrow[from=5-2, to=4-2]
		\arrow[from=5-3, to=4-3]
		\arrow[from=5-4, to=4-4]
	\end{tikzcd}\]
	which is simply Diagram~(\ref{eq:terminal-sequence}) expended with the diagram
	view of objects of $\NN$, see Diagram~(\ref{eq:diagram-view}). Let us consider
	the object of $\NN$ made of the diagonal elements of the last diagram above,
	and we call this new object $\Omega$. Its image on objects is $\Omega(n) =
	T^{n+1}Z(n)$ and its image on morphisms can be read on the diagram (it does
	not matter which one is chosen because the diagram commutes).
	
	Moreover, Lemma~\ref{lem:niso} ensures that all $T^n!_k$, where $k<n$, is an
	isomorphism; so they can be soundly reversed, (in other words, having those
	arrows in the other direction does not break the commutativity of the
	diagram) which gives that:
	\[\begin{tikzcd}
		TZ & T^2Z & T^3Z & T^4Z & \dots & \Omega
		\arrow["T!"', from=1-2, to=1-1]
		\arrow["T^2!"', from=1-3, to=1-2]
		\arrow["T^3!"', from=1-4, to=1-3]
		\arrow[from=1-5, to=1-4]
		\arrow[from=1-6, to=1-5]
		\arrow[bend left=20, from=1-6, to=1-4]
		\arrow[bend left=30, from=1-6, to=1-3]
		\arrow[bend left=40, from=1-6, to=1-2]
		\arrow[bend left=50, from=1-6, to=1-1]
	\end{tikzcd}\]
	commutes.
	This new object $\Omega$ is a limit of Diagram~(\ref{eq:terminal-sequence}):
	it is made of elements of the diagram, thus any object that is mapped to the
	diagram has a single way to be mapped to $\Omega$. For the same reason,
	$T\Omega$ is a limit of the diagram as well, thus $\Omega \cong T\Omega$.
	Note that this can all be viewed as a consequence of Lemma~\ref{lem:niso}.
\end{proof}

\begin{remark}
	The category $\NN$ is \emph{contractively complete} in the
	sense of \cite{birkedal2012first}.
\end{remark}

We generalise the notion of locally contractive endofunctors to
functors $\NN^k \to \NN$, to facilitate the discussion to come.

\begin{definition}
	An $\Scat$-functor $T \colon \NN^k \to \NN$ is \emph{locally contractive}
	if it is serarately locally contractive in each variable; formally,
	given any vector $\ov F(-)$ of objects of $\NN$ with one hole
	(\emph{e.g.}, $\ov F(-) = F_1,\dots,F_{j-1},-,F_{j+1},\dots,F_k$),
	the functor $T(\ov F(-)) \colon \NN \to \NN$ is locally contractive.
\end{definition}

\begin{theorem}[Parameterised Fixed Point]
	\label{th:param}
	A locally contractive functor $T \colon \NN^{k+1} \to \NN$
	admits a parameterised fixed point; in details,
	a pair $(T^\noma,\phi^T)$ such that:
	\begin{itemize}
		\item $T^\noma \colon \NN^k \to \NN$ is a locally contractive functor,
		\item $\phi^T \colon T \circ \pv{\mathrm{id}}{T^\noma} \natto T^\noma$ is a
			natural isomorphism,
		\item for every object $\ov F$ in $\NN^k$,
			$(T^\noma \ov F,\phi^T)$ is the fixed point of $T(\ov F,-)$.
	\end{itemize}
\end{theorem}
\begin{proof}
	The proof is the same as the one of \cite[Theorem
	7.5]{birkedal2012first}. Given $\ov F$ an object
	in $\NN^k$, $T(\ov F,-) \colon \NN \to \NN$ is a contractive functor, and
	thus has a fixed point $(\Omega(\ov F),\alpha^{\ov F})$. The next
	step is to prove that the statement $\Omega(-)$ induces a functor $\NN^k \to
	\NN$. Remember that $\Omega(\ov F)(n) = T(\ov F, -)^{n+1} Z(n)$, which has a
	functor flavour; given $\beta \colon \ov F \natto \ov G$, $\Omega(\beta)_n =
	(T(\beta,-)^{n+1}Z)_n$ makes $\Omega(-)$ a functor (it preserves the identity
	and composition because $T$ does). Also, in the formula, $T$ is applied at
	least once and is locally contractive, thus $\Omega$ is locally contractive.
	The natural transformation $T\circ\pv{\mathrm{id}}{\Omega}\natto \Omega$ is
	obtained by looking at the square diagram in the first part of the proof;
	its isomorphic nature is inherited from Lemma~\ref{lem:niso}.
\end{proof}

Note that a parameterised fixed point provides a natural isomorphism, whose
components are also isomorphisms. These components are coisometries, and an
isomorphic coisometry is a unitary. 

\paragraph{As a denotational semantics.}
Contractive functors and contractive morphisms would respectively be used to
interpret inductive data types and guarded recursion functions in a programming
language. We can imagine a language akin to the one in
Chapter~\ref{ch:reversible} with quantum superpositions, with the
introduction of the modality $\later$ and the combinator $\nnext\!$.  The
interpretation of a type judgement $\Theta \vdash A$ is given by a functor
$\NN^{\abs\Theta} \to \NN$, and the interpretation of $\Theta \vdash \mu X . A$
is given by $\sem{\Theta, X \vdash A}\nnoma$ thanks to Theorem~\ref{th:param}.
The fixed point combinator would have the following typing rule:
\[
	\infer{
		\Psi \entailiso \ffix \phi . \omega \colon T
	}{
		\Psi, \phi \colon \later\! T \entailiso \omega \colon T
	}
\]
which ensures that the interpretation $\sem{\Psi, \phi \colon \later\! T
\entailiso \omega \colon T}$ is contractive (see
Definition~\ref{def:contractive-S}) and therefore admits a unique fixed point
(see Theorem~\ref{th:contractive-S}). 

\subsection{Conclusion on Guarded Quantum Recursion}
\label{sub:guar-conclusion}

We have laid out convincing grounds for guarded quantum recursion based on its
denotational semantics in $\NN$ and $\QQ$, two categories enriched in the topos
of trees $\Scat$, therefore allowing us to extract properties of the latter for
the interpretation of the functions in the language.

The details on the corresponding programming language capturing both quantum
control as an algebraic effect and guarded recursion is left as future work.

%% file: conclusion.tex
\vspace{2cm}

In this thesis, we are concerned with the question of the semantics of effects in
programming languages. This study is conducted through the prism of formal
programming languages, which are directly a $\lambda$-calculus or inspired by
it. We have put forward new perspectives on various effects and their
commutativity via the question of centrality, and we have thouroughly studied a
reversible algebraic effect aimed at quantum computing, allowing for a quantum
control of the program flow. We have also questioned the semantics of inductive
types and recursion in the context of reversibility, to potentially adapt it to
the quantum case.

In Chapter~\ref{ch:monads}, we have studied the question of centrality for
\emph{monads}. In particular, we have provided three equivalent conditions for
a monad to be centralisable. Monads have been shown to be the right structure
to model effects in category theory. However, some effects cannot be captured
as a monad, such as the reversible effect we focus on in
Chapter~\ref{ch:qu-control}. The question of centrality of effects should
then be generalised to \emph{promonads}, which are to monads what relations are
to functions. Moreover, commutativity can be studied more broadly than with
centres, and a whole theory of \emph{centralisers} might be developed.

In Chapter~\ref{ch:qu-control}, we have laid foundations for the semantics of
quantum computing seen as a reversible effect. This point of view allows us to
work with a \emph{quantum control flow}. This is especially meaningful since
some quantum-controlled operations, such as the quantum switch, cannot be
performed with classical control. Once this simply-typed quantum control and
its semantics are properly presented, we have wondered about possible additions
to the language, such as infinite data types and recursion.

To do so, we have first studied the question of infinite data types and
recursion in reversible computation, in Chapter~\ref{ch:reversible}. In
particular, we have provided a sound and adequate denotational semantics of the
same language, without the quantum effect, and with inductive types and
recursion. This is done hoping that it is possible to generalise to the
language incorporating the quantum effect.

However, we show in Chapter~\ref{ch:qu-recursion} that this generalisation is
not that simple. On a more positive note, we outline a potential categorical
model for quantum control on infinite data types and recursion, with the help
of \emph{guarded recursion}. This model strongly echoes with the category
$\mathbf{LSI}_\leq$ introduce by Pablo Andrés-Martínez in his thesis
\cite{pablo2022unbounded}. In the conclusion of his thesis, he wrote:
\begin{changemargin}{2cm}{2cm}
	``Quantum computer scientists tend to dismiss unbounded iteration in quantum
	algorithms as an uninteresting field of work: in the case of classical
	control flow due to the assumption that testing a termination condition on
	every iteration would destroy any achievable quantum speed-up and, in the
	case of quantum control flow, due to the tech- nical obstacles that
	interference and the possibility of infinitely many execution paths would
	entail [\dots]''
\end{changemargin}
Even if these concerns were verified, these comments apply to classical
computing also, since the memory of our computer is finite. However, it did not
prevent computer scientists from studying infinite data types, infinite loops
and their semantics. Modern programming languages contain types such as natural
numbers and lists, which are by essence infinite, even if their representation
in the architecture is necessarily finite.

Pablo Andrés-Martínez finishes his conclusion with the following sentence.
\begin{changemargin}{2cm}{2cm}
	``It is my hope that further study on this field will yield quantum programming
	languages supporting quantum control flow and new algorithms that make use of
	unbounded iteration.''
\end{changemargin}
I share his hope.